\documentclass[english,12pt]{article}
\usepackage[T1]{fontenc}
\usepackage[latin9]{inputenc}
\usepackage{geometry}
\usepackage{afterpage}
\usepackage[round]{natbib}
\usepackage{amsmath}
\usepackage{amsfonts}
\usepackage{amssymb}
\usepackage{amsthm}
\usepackage{lmodern}
\usepackage{rotating}
\usepackage{color}
\usepackage{dcolumn}
\usepackage{setspace}
\usepackage[english]{babel}
\usepackage{eurosym}
\usepackage{esint}
\usepackage{float}
\usepackage{graphicx}
\usepackage{epstopdf}
\usepackage{hyperref}
\usepackage[none]{hyphenat}
\usepackage{booktabs}
\usepackage{threeparttable}
\usepackage{authblk}
\usepackage{pdflscape}
\usepackage{longtable}
\usepackage{multirow}
\usepackage{dpfloat, booktabs}
\usepackage{xfrac}
\usepackage{enumerate}
\newtheorem{prop}{Proposition}
\newtheorem{thm}{Theorem}
\newtheorem{lem}{Lemma}
\newtheorem{cor}{Corollary}
\newtheorem{assum}{Assumption}

\newtheorem{con}{Condition}
\providecommand{\keywords}[1]{\textbf{Keywords:} #1}
\begin{document}

\begin{spacing}{1.2}

\title{Decomposition of Differences in Distribution under Sample Selection and the Gender Wage Gap}

\author{Santiago Pereda-Fern\'{a}ndez\thanks {Departamento de Econom\'{i}a, Universidad de Cantabria, Avenida de los Castros, s/n, 39005 Santander, Spain. I would like to thank Manuel Arellano, Domenico Depalo, Franco Peracchi, Alexandra Sober\'{o}n, Paolo Zacchia and seminar participants at CEMFI, Universidad de Cantabria, Universit\`{a} di Bologna and Universit\`{a} di Roma Tor Vergata. This work is part of the I+D+i project Ref. TED2021-131763A-I00 financed by MCIN/AEI/10.13039/501100011033 and by the European Union NextGenerationEU/PRTR. I gratefully acknowledge financial support from the Spanish Ministry of Universities and the European Union-NextGenerationEU (RMZ-18). All remaining errors are my own. I can be reached via email at santiagopereda@gmail.com.}}
\affil{Universidad de Cantabria}
\maketitle

\date

\begin{abstract}
\noindent
I address the decomposition of the differences between the distribution of outcomes of two groups when individuals self-select themselves into participation. I differentiate between the decomposition for participants and the entire population, highlighting how the primitive components of the model affect each of the distributions of outcomes. Additionally, I introduce two ancillary decompositions that help uncover the sources of differences in the distribution of unobservables and participation between the two groups. The estimation is done using existing quantile regression methods, for which I show how to perform uniformly valid inference. I illustrate these methods by revisiting the gender wage gap, finding that changes in female participation and self-selection have been the main drivers for reducing the gap.
\end{abstract}



\end{spacing}

\pagebreak{}

\section{Introduction}\label{sec:intro}

Decomposition methods have been widely used to understand the sources of differences between the outcomes of two groups. Following the seminal works by \cite{Oaxaca1973} and \cite{Blinder1973}, mean differences between two groups have been attributed to either differences in mean covariates (endowments effect/composition component) or differences in the slope parameters (coefficients effect/structural component). Subsequently, several methods have been proposed to account for differences between other features of the distribution, such as the unconditional quantiles. These decompositions typically rely on the ignorability assumption (selection on observables).

However, individuals present in the decomposition often constitute a self-selected sample. This requires accounting for it to estimate the structural parameters using sample selection methods. Moreover, the decomposition of a feature of the outcomes between two groups could depend on differences the unobservables that depend on the amount of self-selection and participation for each group. This is the case if one considers the distribution of actual outcomes rather than the distribution of potential outcomes that would be observed if everyone participated.

Additionally, sample selection raises the question of which population is the target of the analysis: participants or the entire population. In the first case, the comparison is between those individuals with non-zero outcomes. \textit{E.g.}, one could be interested in understanding the differences in labor earnings of two groups of employed workers. However, it can be argued that including non-participants provides a more comprehensive comparison. This would amount to include unemployed workers and those out of the labor force in the analysis.

This paper makes the following contributions. First, I propose several decompositions according to the target population and the distributional feature of interest. Specifically, I consider decompositions for participants and the entire population for actual outcomes. Second, I propose two ancillary decompositions of the propensity score and the average value of the unobservable variable that affects the outcome, which is interpreted as ability. Third, I estimate the functionals that are used in the decomposition using the Quantile Regression with Selection (QRS) estimator proposed by \cite{Arellano2017}. I show the asymptotic properties of the estimated components of the decomposition, as well as how to carry uniformly valid inference. Fourth, I obtain new estimates of the evolution of the gender wage gap, finding that the unobserved ability of female workers has increased more than that of males. This, together with the fall in the gender participation gap have been two major factors in reducing the gender wage gap.

In the absence of self-selection, individual actual and potential outcomes coincide. In contrast, actual outcomes equal zero for non-participants when there is self-selection. Hence, there are two potentially relevant target populations to consider: just participants, who have a non-zero outcome, and the entire population, which also includes non-participants. The decomposition for both populations are related and, depending on the context, either could be more appropriate to analyze. I consider both of them, highlighting their differences and similarities in a general framework and under some simplifying assumptions. Furthermore, because differences in the unobservables can account for differences in the actual outcomes, the type of decompositions presented here differ from those that focus on the potential outcomes when every individual is assumed to participate.

Under ignorability, differences in outcomes can be split into differences in covariates and differences in the returns to the covariates. However, when there is sample selection on unobservables, there are two additional components that can explain the differences between the two groups. The first one is the selection component, which reflects how, \emph{caeteris paribus}, the unobserved characteristics of a group may be more or less positively selected relative to those of the other group. The other one is the participation component, which reflects differences in outcomes that can be attributed to the differences in the participation rates between the two groups. This extends the decomposition into three components in a triangular model with a binary treatment considered in \cite{Pereda2022}.

Most decompositions in a cross-sectional data setting have imposed no selection on unobservables.\footnote{There exist works that explicitly account for time-invariant unobservables in a panel data setting. See \cite{Fortin2011} and references therein for further details.} This problem has been highlighted, \textit{e.g.}, by \cite{Kunze2008} or \cite{Huber2015c}. Some exceptions include \cite{Neuman2003,Neuman2004} and \cite{Mora2008}, who used the Heckit correction to obtain their estimates corrected for selection, or \cite{Cukrowska2016}, who used a multinomial correction model. More recently, \cite{Huber2020} used mediation analysis tools to decompose the mean gap with sample selection under the assumption that the effect of the unobservables is homogeneous (\textit{e.g.}, if the model is additively separable).

A common feature of these works is that they are only concerned with decompositions of mean differences. In contrast, the focus gradually expanded to analyze other features, including unconditional distributions. The latter are particularly relevant, as many statistics of interest can be expressed as a function of it. Regarding estimation, there exists a wide menu of available methods, including those based on quantile regression \citep{Machado2005,Melly2006,Chernozhukov2013}, on distribution regression \cite{Chernozhukov2013}, and on reweighting \citep{Dinardo1996,Firpo2018}. A thorough review of the advantages and disadvantages of each method is found in \cite{Fortin2011}.

A decomposition into four components has previously been considered by \cite{Chernozhukov2019c}, who also used it in a sample selection framework with a distributional regression estimator. The main difference with respect to the one considered in this paper regards the structural functions employed to model the outcomes. In their paper, they use bivariate Gaussian distributions to model the sample selection locally, whereas the approach in this paper is to model the unobservables with a copula globally. This allows to interpret the unobservables as the conditional ranks in the distribution of potential outcomes. Moreover, it is straightforward to relate differences in the value of these unobservables between the two groups to differences in the copula, the propensity score and the distribution of covariates, which may be of independent interest for the researcher.

Other related works are \cite{Maasoumi2017} and \cite{Maasoumi2019}. In the first one they estimate the differences between functions of the distributions of the two groups for the selected sample and they bound the differences for the entire population. In the second one they analyze the evolution of the gender wage gap in the US using the estimator proposed by \cite{Arellano2017}, including a decomposition of the gap between the two distributions of potential wages for the entire population. There are some differences relative to these two papers. First, I consider differences between the actual distributions and relate them to the different primitives of the model. Second, I consider four different sources of variability, which can help explain which are the determinants of the gap. The latter can be a problem not just if sample selection is ignored, but also if one analyzes the distributions of potential outcomes ignoring participation.

For estimation purposes, I consider a nonseparable model with a univariate unobservable variable in the outcome equation, which naturally points at quantile regression methods. The QRS estimator can estimate the structural function that relates potential wages to covariates and unobservables and the copula that captures the amount of self-selection. This, together with an estimator of the propensity score and the sample distribution of covariates are all the ingredients needed to estimate the different components of the decomposition. Moreover, I discuss under which simplifying assumptions other alternative estimators can be used.

The decomposition methods have been applied to a wide variety of topics, including test scores gaps between genders \citep{Sohn2008}, schools \citep{Krieg2006} or countries \citep{Mcewan2004}, differences in students' enrollment \citep{Borooah2005}, differences in health insurance coverage between different demographic groups \citep{Bustamante2009}, or gender differences in smoking behavior \citep{Bauer2007}. However, most decomposition studies revolved around wage gaps between genders \cite{Kunze2008}, race \citep{Barsky2002}, unionization status \citep{Card1996}, or public-private employees \citep{Depalo2020}.

I apply the methods presented in this paper to decompose the gender wage gap. I revisit the estimates by \cite{Maasoumi2019}, which are based on the Current Population Survey for the period 1976-2013. I perform several decompositions, considering different population targets (actual earnings for participants and the full population, as well as potential outcomes for the full population) and several statistics (mean, unconditional distributions and generalized entropy indices). I find that considering just employed workers understates the wage gap relative to consider the entire population. However, the gender gap has substantially diminished over the period considered in both cases. The main contributing factors for this reduction are the increase in female participation, and the increase in the average value of both observed and unobserved characteristics of employed females.

The rest of the paper is organized as follows: Section~\ref{sec:model} presents the model and the main functions of interest. The decompositions of interest are presented in Section~\ref{sec:decom}, including particular cases of interest that are nested in the more general model. Section~\ref{sec:est} describes the estimation of the decompositions for the general model, describing its asymptotic properties and showing how to conduct inference. Section~\ref{sec:emp} revisits the estimates of the evolution of the gender wage gap, and Section~\ref{sec:conc} concludes.
\section{The Model}\label{sec:model}

Consider the following selection model:
\begin{align}
Y&=g_{D}\left(X,U\right)S\label{eq:y}\\
S&=\mathbf{1}\left(\pi_{D}\left(Z\right)-V>0\right)\label{eq:s}
\end{align}
where $Y$ denotes the continuous outcome of interest, $X$ a set of predetermined covariates, $Z\equiv\left(Z_{1},X'\right)$ is composed of the instrument $Z_{1}$ and the predetermined covariates, $S$ is a binary indicator for participation, $D$ is an observed variable that denotes the group to which the individuals of the population can belong, and $U$ and $V$ are two unobservable random variables. Equations~\ref{eq:y}-\ref{eq:s} conform a sample selection model that allows for a broad class of differences between individuals of different groups. The set of groups is denoted by $\mathcal{D}$, and its dimension is assumed to be finite. For expositional purposes, the leading case is two groups, \textit{i.e.}, $D\in\mathcal{D}=\left\{0,1\right\}$.\footnote{This framework is reminiscent of those considered for the study of the marginal treatment effect (MTE; \citealp{Bjorklund1987}). Such framework can be rationalized by a generalized Roy model with imperfect information \citep{Pereda2022}.}

This system can be used to model several economic phenomena, such as labor wages, denoted by $Y$. These could be different for people belonging to different demographic groups, like gender ($D$). Equation~\ref{eq:s} reflects the fact that only employed workers ($S=1$) have a wage. These are modeled by Equation~\ref{eq:y}, and they depend on their observed and unobserved characteristics, respectively given by $X$ and $U$. The latter, together with $V$ are the unobservables of the model, which can be interpreted as the ability of the worker and the unobserved propensity to be unemployed, respectively.\footnote{One could also consider work intensity as an additional dependent variable in the model. \textit{E.g.}, one could consider weekly hours or worked weeks, as in \cite{Fernandez2022}. Such extensions could also be considered by appropriately modeling the unobservable variables, possibly adding some extra ones. Such comprehensive models could be informative about features that are not captured by a simpler model, at the cost of making them computationally more complicated. Their study is left for future work.}

Following \cite{Heckman2005}, the distribution of $V$ can be normalized to be uniformly distributed on the unit interval. This is convenient, as $\pi_{d}$ can now be interpreted as the propensity score.\footnote{Note that the propensity score is sensitive to the instrument used. Therefore, using different instruments would lead to different results. See, \textit{e.g.}, \cite{Mogstad2021} for further details on the interpretation of treatment effects in with multiple instruments, a framework related to the one considered in this paper.} Moreover, if the same normalization is applied to the distribution of $U$, Equation~\ref{eq:y} uses the Skorohod representation, allowing us to interpret $g_{d}\left(x,u\right)$ as the structural quantile function (SQF). These normalizations offer two advantages: the joint distribution of the unobservables conditional on $D=d$ and $X=x$, denoted by $C_{d,x}\left(u,v\right)\equiv\mathbb{P}\left(U\leq u,V\leq v|D=d,X=x\right)$, can be interpreted as a copula and, by taking the inverse of the SQF with respect to $U$, one obtains the conditional distribution of potential outcomes if all were participants.\footnote{This constitutes a conceptual difference relative to \cite{Chernozhukov2019c}. They model the relation between the outcome variable and the participation decision using a bivariate model that is observationally equivalent. However, it is not possible to interpret the distribution of the unobservables as a copula with their representation. Hence, the interpretation of their unobservables is less clear, as they are locally defined.}

The copula also conveys some important information for the policy maker regarding some counterfactual scenarios. It captures the amount of self-selection on unobservables, linking the participation decision to the outcome. Negative amounts of correlation are associated with positive selection of individuals into participation. Consequently, the more negative the amount of correlation, the lower the potential outcomes of non-participants relative to participants.

To show that these normalizations are without loss of generality, consider an alternative data generating process determined by functions $\tilde{g}_{D}$ and $\tilde{\pi}_{D}$, as well as the distribution of the unobservables $\tilde{F}_{\tilde{U},\tilde{V}|D,X}$. Their observational equivalence is established by Lemma~\ref{lem:normalize}:
\begin{lem}\label{lem:normalize}
Let $Y=\tilde{g}_{D}\left(X,\tilde{U}\right)S$ and $S=\mathbf{1}\left(\tilde{\pi}_{D}\left(Z\right)-\tilde{V}>0\right)$, where the distribution of the unobservables is given by $\tilde{F}_{\tilde{U},\tilde{V}|D,X}\left(\tilde{U},\tilde{V}|D,X\right)$, with marginal distributions $\tilde{F}_{\tilde{U}|D,X}$ and $\tilde{F}_{\tilde{V}|D,X}$, that may depend non-trivially on $X$. Then, there exist $g_{D}$, $\pi_{D}$ such that the model given by Equations~\ref{eq:y}-\ref{eq:s}, where $U|D,X\sim U\left[0,1\right]$ and $V|D,X\sim U\left[0,1\right]$, generates the same distribution of $\left(Y,S,D,Z\right)$.
\end{lem}

The identification conditions are listed in \cite{Arellano2017}.\footnote{Note that the identification is based on instrumental variables, rather than on a control function approach. A similar model was considered using control functions was considered by \cite{Fernandez2022}. The two models are non-nested, and in particular \cite{Fernandez2022} works under the assumption that the unobservables of $Z$, including $X$. This contrasts with the assumptions in \cite{Arellano2017}, which allows for non-trivial dependence between the unobservables and $X$. A comparison of the conditions required for identification with instrumental variables and control functions is provided, \textit{e.g.}, by \cite{Heckman2007b}.} Apart from some standard assumptions (exclusion restriction, continuous outcomes, propensity score strictly within the unit interval and a well-defined continuous copula for the unobservables), they require either identification at infinity, or a continuous instrument and that the copula be real analytic with respect to its second argument.\footnote{The set of assumptions required for identification by \cite{Chernozhukov2019c} is different, though they are related. They provide a comparison of both sets of assumptions in their paper.} One could relax the latter by imposing a parametric assumption, a possibility considered in Appendix~\ref{app:param}. Additionally, these assumptions can be relaxed if the copula is homogeneous with respect to the covariates and one uses variation in the covariates as a source of exogenous variation.
\section{Decompositions}\label{sec:decom}

Decompositions typically vary according to the distributional feature that is decomposed, such as the difference of the means or of the unconditional quantiles. Sample selection introduces two additional margins of choice regarding which decomposition to make. First, because only a fraction of the population participates, one could consider a decomposition for participants or for the entire population, assigning a value of zero for non-participants. Second, the distribution of the unobservables for participants and non-participants differ, making the distributions of actual and potential outcomes differ, too. The focus here is on actual outcomes of both participants and the full population.

\subsection{Primitives of the Decompositions}

The decomposition of the outcomes for either participants or the entire population requires to account for the following primitive functions: the SQF, $g_{d}\left(x,u\right)$, the propensity score, $\pi_{d}\left(z\right)$, the marginal distribution of the covariates, $F_{Z}^{d}\left(z\right)$, and the conditional copula of the unobservables, $C_{d,x}\left(u,v\right)$, for $d=0,1$. This contrasts with decompositions of potential outcomes, which account for selection only to estimate the structural parameters, and then base the decomposition on the SQF and the propensity score.

The leading feature that is decomposed is the mean outcome. This requires assigning a value to the counterfactual mean outcome for the population of interest when a combination of the distribution of the covariates or the estimated structural functions are changed to match those of the other group. Following the previous discussion, the mean outcome for participants with the distribution of the observables of group $h$, the SQF of group $j$, the copula of group $k$ and the propensity score of group $l$ is given by
\begin{align}\label{eq:meanys}
\mathbb{E}\left[Y^{hjkl}|S=1\right]\equiv\int_{\mathcal{Z}}\int_{0}^{1}g_{j}\left(x,u\right)dG_{k,x}\left(u,\pi_{l}\left(z\right)\right)dF_{Z}^{h}\left(z\right)
\end{align}
where $G_{d,x}\left(u,v\right)\equiv\mathbb{P}\left(U\leq u|D=d,X=x,V\leq v\right)=\frac{C_{d,x}\left(u,v\right)}{v}$ denotes the copula conditional on participation and $\mathcal{Z}$ denotes the support of $Z$. This is the channel through which the unobservables affect the mean outcome. Similarly, the counterfactual value for the entire population is given by
\begin{align}\label{eq:meany}
\mathbb{E}\left[Y^{hjkl}\right]\equiv\int_{\mathcal{Z}}\int_{0}^{1}g_{j}\left(x,u\right)dC_{k,x}\left(u,\pi_{l}\left(z\right)\right)dF_{Z}^{h}\left(z\right)
\end{align}

More generally, one can obtain equivalent expressions for the distributions for each group, that constitute the building blocks for other functionals of interest, notably unconditional quantiles. The unconditional cumulative distribution functions for participants and the entire population are respectively given by
\begin{align}
F_{Y|S=1}^{hjkl}\left(y\right)&\equiv\int_{\mathcal{Z}}\int_{0}^{1}\mathbf{1}\left(g_{j}\left(x,u\right)\leq y\right)dG_{k,x}\left(u,\pi_{l}\left(z\right)\right)dF_{Z}^{h}\left(z\right)\label{eq:fys}\\
F_{Y}^{hjkl}\left(y\right)&\equiv\int_{\mathcal{Z}}\left[\int_{0}^{1}\mathbf{1}\left(g_{j}\left(x,u\right)\leq y\right)dC_{k,x}\left(u,\pi_{l}\left(z\right)\right)+\left(1-\pi_{l}\left(z\right)\right)\right]dF_{Z}^{h}\left(z\right)\label{eq:fy}
\end{align}

Consequently, the unconditional quantile functions are given by\footnote{Note that if $Y$ is continuous and there is no bunching at any particular value, Equation~\ref{eq:qys} equals the inverse of the cumulative distribution function, but the same is not true for Equation~\ref{eq:qy}, which has a mass point at $Y=0$.}
\begin{align}
Q_{Y|S=1}^{hjkl}\left(\tau\right)&\equiv\inf\left\{y:\tau\leq F_{Y|S=1}^{hjkl}\left(y\right)\right\}\label{eq:qys}\\
Q_{Y}^{hjkl}\left(\tau\right)&\equiv\inf\left\{y:\tau\leq F_{Y}^{hjkl}\left(y\right)\right\}\label{eq:qy}
\end{align}

\subsection{Main Decompositions}

The decomposition in a sample selection model is more intricate than in the regular one. To see this, note that the covariates can have an impact on the final outcome through three different channels. First, they affect the distribution of potential outcomes through the SQF, the only channel present when there is no sample selection. Second, they affect the propensity to participate, making some individuals more likely to participate. Third, they affect the amount of self-selection through the distribution of the unobservables, which is indexed by $x$ in general. In other words, some characteristics may display some correlation with the unobservables, either reinforcing or mitigating the differences between the two groups.

Consequently, to better assess how the covariates affect the propensity score and the selection on unobservables channels, it is useful to report two ancillary decompositions, presented in Section~\ref{sec:ancillary}. Moreover, the decomposition is path dependent. I present one of the possible orders for the decomposition of mean differences for participants and describe each component in turn.\footnote{Specifically, there are a total of 24 decompositions, and the size of each component may vary in each of them. See \cite{Fortin2011} for further details on the limits of path dependent decompositions.}

Note that, by definition, $\mathbb{E}\left[Y|D=1,S=1\right]=\mathbb{E}\left[Y^{1111}|S=1\right]$ and $\mathbb{E}\left[Y|D=0,S=1\right]=\mathbb{E}\left[Y^{0000}|S=1\right]$. The difference between these two can be decomposed as:
\begin{align}\label{eq:decommeanys}
\mathbb{E}\left[Y|D=1,S=1\right]-\mathbb{E}\left[Y|D=0,S=1\right]&=\underbrace{\mathbb{E}\left[Y^{1111}|S=1\right]-\mathbb{E}\left[Y^{0111}|S=1\right]}_{\text{endowments component}}\nonumber\\
&+\underbrace{\mathbb{E}\left[Y^{0111}|S=1\right]-\mathbb{E}\left[Y^{0011}|S=1\right]}_{\text{coefficients component}}\nonumber\\
&+\underbrace{\mathbb{E}\left[Y^{0011}|S=1\right]-\mathbb{E}\left[Y^{0001}|S=1\right]}_{\text{selection component}}\nonumber\\
&+\underbrace{\mathbb{E}\left[Y^{0001}|S=1\right]-\mathbb{E}\left[Y^{0000}|S=1\right]}_{\text{participation component}}
\end{align}

Equation~\ref{eq:decommeanys} decomposes the mean difference between the two groups into four components. The first two are those present in decompositions under the ignorability assumption. The endowments component captures how differences in the covariates between the two groups lead to differences in mean outcomes. The other term present in this type of decompositions is the coefficients component. It reflects how differences in the SQF between the two groups are related to differences in mean outcomes, which may be partly driven by discrimination. For example, in the Oaxaca-Blinder decomposition, this term equals the difference in the OLS coefficients between the two groups, scaled by the average covariates of one of them.

The remaining two terms arise in a sample selection framework. The participation component has the easiest interpretation, as it relates differences in the probability of participating for both groups to differences in mean outcomes. To get some intuition, assume that more able individuals (\textit{i.e}, those with high $U$) are also more prone to participate (\textit{i.e.}, low $V$). Then, as the propensity score increases from zero to one, the average ability of participants gets smaller towards its mean value, reducing the mean wage for participants.

The selection component links differences in the amount of self-selection into participation to differences in mean outcomes. The interpretation is slightly different from that of the participation component, as it affects the distribution of the unobservables without affecting the participation. Therefore, \textit{caeteris paribus}, the higher the level of selection, the higher the average ability of participants and, consequently, the higher their mean outcome.

Note that because of the linearity of the expectation operator, each component can be written in terms of the difference of the primitive functions. For instance, the coefficients component can be written as
\begin{align*}
\mathbb{E}\left[Y^{0111}|S=1\right]-\mathbb{E}\left[Y^{0011}|S=1\right]=\int_{\mathcal{Z}}\int_{0}^{1}\left(g_{1}\left(x,u\right)-g_{0}\left(x,u\right)\right)dG_{1,x}\left(u,\pi_{1}\left(z\right)\right)dF_{Z}^{0}\left(z\right)
\end{align*}

This convenient property does not generally hold, and in particular for the decomposition of the unconditional quantiles. Moreover, in some particular cases, the expression for some components can be simplified, as shown in Appendix~\ref{sec:cases}. The decomposition for the entire population is analogous to the one for participants:
\begin{align}\label{eq:decommeany}
\mathbb{E}\left[Y|D=1\right]-\mathbb{E}\left[Y|D=0\right]&=\underbrace{\mathbb{E}\left[Y^{1111}\right]-\mathbb{E}\left[Y^{0111}\right]}_{\text{endowments component}}+\underbrace{\mathbb{E}\left[Y^{0111}\right]-\mathbb{E}\left[Y^{0011}\right]}_{\text{coefficients component}}\nonumber\\
&+\underbrace{\mathbb{E}\left[Y^{0011}\right]-\mathbb{E}\left[Y^{0001}\right]}_{\text{selection component}}+\underbrace{\mathbb{E}\left[Y^{0001}\right]-\mathbb{E}\left[Y^{0000}\right]}_{\text{participation component}}
\end{align}

The interpretation of the different components is similar, with one notable difference: because the outcome for non-participants equals zero by definition, the mean outcome for the entire population is the mean outcome for participants multiplied by the propensity score. Therefore, the endowments and participation components operate through two channels: first, by changing the proportion of participants; second, by affecting the average outcome of participants. In contrast, the coefficients and selection components operate exclusively through the second channel.

The last two considered decompositions are those of the unconditional quantiles of the outcome for both target populations.\footnote{Additional decompositions can be constructed analogously. \textit{E.g.}, if one is interested in differences in inequality, one could construct equivalent counterfactual values for the Gini index and apply the decomposition. As long as the statistic of interest can be expressed in terms of the same primitives of the model, and under some regularity conditions, they will be well-behaved. See, \textit{e.g.}, \cite{Chernozhukov2013} for further details.} They consist of the same four components as the mean decompositions, although there are some important differences. From a mathematical standpoint, the unconditional quantile function is not a linear operator. Therefore, their expressions are more convoluted even under some simplifying assumptions. From a policy perspective, they are more informative, as they allow to assess which segments of the population present the largest differences. This makes them particularly relevant in cases in which the differences are of opposite signs for different quantiles of the distribution.

Formally, the decomposition of the unconditional quantile distribution for participants at quantile $\tau$ equals
\begin{align}\label{eq:decomqys}
Q_{Y|D=1,S=1}\left(\tau\right)-Q_{Y|D=0,S=1}\left(\tau\right)&=\underbrace{Q_{Y|S=1}^{1111}\left(\tau\right)-Q_{Y|S=1}^{0111}\left(\tau\right)}_{\text{endowments component}}+\underbrace{Q_{Y|S=1}^{0111}\left(\tau\right)-Q_{Y|S=1}^{0011}\left(\tau\right)}_{\text{coefficients component}}\nonumber\\
&+\underbrace{Q_{Y|S=1}^{0011}\left(\tau\right)-Q_{Y|S=1}^{0001}\left(\tau\right)}_{\text{selection component}}+\underbrace{Q_{Y|S=1}^{0001}\left(\tau\right)-Q_{Y|S=1}^{0000}\left(\tau\right)}_{\text{participation component}}
\end{align}
while the decomposition of the unconditional quantile distribution for the entire population at quantile $\tau$ is given by
\begin{align}\label{eq:decomqy}
Q_{Y|D=1}\left(\tau\right)-Q_{Y|D=0}\left(\tau\right)&=\underbrace{Q_{Y}^{1111}\left(\tau\right)-Q_{Y}^{0111}\left(\tau\right)}_{\text{endowments component}}+\underbrace{Q_{Y}^{0111}\left(\tau\right)-Q_{Y}^{0011}\left(\tau\right)}_{\text{coefficients component}}\nonumber\\
&+\underbrace{Q_{Y}^{0011}\left(\tau\right)-Q_{Y}^{0001}\left(\tau\right)}_{\text{selection component}}+\underbrace{Q_{Y}^{0001}\left(\tau\right)-Q_{Y}^{0000}\left(\tau\right)}_{\text{participation component}}
\end{align}

\subsection{Ancillary Decompositions}\label{sec:ancillary}

The previous discussion highlighted the multiple channels through which the covariates can affect the final outcome in a sample selection framework. To better assess their role, researchers could perform two additional decompositions of the propensity score for the entire population and the average value of the unobservables for participants. These decompositions could be presented alongside the main decomposition, complementing the analysis.\footnote{It is worth stressing that, because the variable that is decomposed is not the main outcome, the components of the ancillary decompositions are not comparable to those of the main decompositions, although they are related to each other.}

The first one is the participation decomposition. This is a regular means decomposition without selection using the propensity score as the dependent variable. As such, differences in the propensity score would be split into the usual endowments and coefficients components.

The second one is the self-selection decomposition. This is more reminiscent to the decompositions of the mean outcomes. To see this, note that the average value of $U$ for participants with the distribution of the observables of group $h$, the copula of group $l$ and the propensity score of group $m$ is given by
\begin{align}\label{eq:meanu}
\mathbb{E}\left[U^{hkl}|S=1\right]\equiv\int_{\mathcal{Z}}\int_{0}^{1}udG_{k,x}\left(u,\pi_{l}\left(z\right)\right)dF_{Z}^{h}\left(z\right)
\end{align}

Then, the difference of this statistic between the two groups can be decomposed as:
\begin{align}\label{eq:decommeanus}
\mathbb{E}\left[U|D=1,S=1\right]-\mathbb{E}\left[U|D=0,S=1\right]&=\underbrace{\mathbb{E}\left[U^{111}|S=1\right]-\mathbb{E}\left[U^{011}|S=1\right]}_{\text{endowments component}}\nonumber\\
&+\underbrace{\mathbb{E}\left[U^{011}|S=1\right]-\mathbb{E}\left[U^{001}|S=1\right]}_{\text{selection component}}\nonumber\\
&+\underbrace{\mathbb{E}\left[U^{001}|S=1\right]-\mathbb{E}\left[U^{000}|S=1\right]}_{\text{participation component}}
\end{align}

Similarly to the main decompositions, the endowments component captures differences in the copula and the propensity score due to differences in the covariates. On the other hand, the selection and participation components capture differences in the unobserved ability between the two groups due to differences in the copula and the propensity score, respectively.
\section{Estimation}\label{sec:est}

For expositional convenience, I present the estimators of the decompositions for participants. The estimators of the decomposition for the entire population are similarly constructed using the analogy principle. Their exact expressions and their asymptotic properties are presented in Appendix~\ref{app:additional}. Throughout the entire section, the following assumptions are maintained:
\begin{assum}\label{assum:samp}
$\left(Y_{i},S_{i},D_{i},Z_{i}'\right)'$ are \textit{iid} for $i=1,...,n$, defined on the probability space $\left(\Omega,\mathfrak{F},\mathbb{P}\right)$ and take values in a compact set.
\end{assum}
\begin{assum}\label{assum:gsize}
Denote the sample size for the d-th group by $n_{d}$. It is non decreasing in the total sample size, $n\equiv\sum_{d}n_{d},$ and $\sfrac{n}{n_{d}}\overset{P}{\rightarrow}\lambda_{d}\in\left(1,\infty\right)\forall d\in\mathcal{D}$ as $n\rightarrow\infty$.
\end{assum}
\begin{assum}\label{assum:propensity}
Let $\pi_{d}\left(Z\right)\equiv \pi_{d}\left(Z;\gamma_{d}\right)$, with $\dim\left(\gamma_{d}\right)<\infty$. $\pi_{d}\left(Z;\gamma_{d}\right)$ is continuously differentiable with respect to $\gamma_{d}$. Moreover, there exists an asymptotically linear estimator $\hat{\gamma}_{d}\forall d\in\mathcal{D}$ that admits the following representation: $\hat{\gamma}_{d}-\gamma_{d}=-B_{d}^{-1}\frac{1}{n_{d}}\sum_{i=1}^{n}\mathbf{1}\left(D=d\right)\cdot b_{d}\left(S_{i},Z_{i};\gamma\right)+o_{P}\left(\frac{1}{\sqrt{n}}\right)$.
\end{assum}
\begin{assum}\label{assum:bound}
$Y$ has conditional density that is bounded from above and away from zero, a.s. on a compact set $\mathcal{Y}$. The density is given by $f_{Y|Z,S=1}^{d}\left(y|z\right)$ $\forall d\in\mathcal{D}$.
\end{assum}

Assumption~\ref{assum:samp} states the sampling process of the data. Assumption~\ref{assum:gsize} further restricts it to ensure that the number of individuals in each group converges to a fixed proportion with respect to the sample size. The propensity score is required to satisfy some mild regularity conditions given by Assumption~\ref{assum:propensity}. It ensures that its asymptotic distribution is well-behaved, and it is satisfied by many estimation methods, such as maximum likelihood. Finally, Assumption \ref{assum:bound} ensures that the dependent variable has a finite conditional density for the entirety of its support.

\subsection{Main Decompositions}

Let $\upsilon_{m}\left(z,\tau,\pi,f\right)\equiv\left(g_{j}\left(x,\tau\right),c_{k,x}\left(\tau,\pi\right),\pi_{l}\left(z\right),\int_{\mathcal{Z}}fdF_{Z}^{h}\right)$ be the vector that contains the structural functions, where $m\equiv\left(h,j,k,l\right)\in\mathcal{M}\equiv\left\{\left(h,j,k,l\right)\in\mathcal{D}^{4}\right\}$ is used to keep notation compact.\footnote{Strictly speaking, the last component of this vector is not of the structural functions, which would be $F_{Z}^{h}$. However, the former is more convenient to derive the asymptotic theory, so I refer to $\upsilon_{\ell}$ as the vector of structural functions.} The counterfactual mean outcome is estimated as:
\begin{align}\label{eq:meanyhat}
\hat{\mathbb{E}}\left[Y^{m}|S=1\right]=\frac{1}{n_{h}}\sum_{i=1}^{n}\int_{\varepsilon}^{1-\varepsilon}\hat{g}_{j}\left(X_{i},u\right)d\hat{G}_{k,x}\left(u,\hat{\pi}_{l}\left(Z_{i}\right)\right)\mathbf{1}\left(D_{i}=h\right)
\end{align}

Each of the effects is computed as $\hat{\mathbb{E}}\left[Y^{m}|S=1\right]-\hat{\mathbb{E}}\left[Y^{m'}|S=1\right]$ for the appropriate choice of $m,m'$.\footnote{An alternative approach to the one presented in this paper would be to provide bounds of the quantities of interest. \cite{Maasoumi2017} do this in a similar setting, whereas \cite{Blundell2007c} obtained estimates of changes in the distribution of wages by gender using bounds based a variety of assumptions. Unfortunately, as pointed out by the latter, the bounds on differentials tend to be much larger than the bounds on the quantities of interest. Therefore, while in principle it would be possible to bound the components presented in this paper, the width of these bounds could be too large to be informative in most empirical applications.} Similarly, the estimator of the unconditional distribution equals:
\begin{align}\label{eq:qyhat}
\hat{Q}_{Y|S=1}^{m}\left(\tau\right)=\inf\left\{y:\tau\leq \hat{F}_{Y|S=1}^{m}\left(y\right)\right\}
\end{align}
where $\hat{F}_{Y|S=1}^{m}\left(y\right)=\frac{1}{n_{h}}\sum_{i=1}^{n}\left[\varepsilon+\int_{\varepsilon}^{1-\varepsilon}\mathbf{1}\left(\hat{g}_{j}\left(X_{i},u\right)\leq y\right)d\hat{G}_{k,x}\left(u,\hat{\pi}_{l}\left(Z_{i}\right)\right)\right]\mathbf{1}\left(D_{i}=h\right)$.

Similarly, the effects of interest are computed as $\hat{Q}_{Y|S=1}^{m}\left(\tau\right)-\hat{Q}_{Y|S=1}^{m'}\left(\tau\right)$. To derive the asymptotic distribution of these estimators, let $\mathcal{U},\overline{\mathcal{P}}$ be closed subsets of $\left(0,1\right)$, and $\mathcal{F}$ denote the class of measurable functions that includes $\left\{F_{Y|Z}\left(y|z\right):y\in\mathcal{Y},z\in\mathcal{Z}\right\}$ as well as the indicators of all the rectangles in $\overline{\mathbb{R}}^{d_{z}}$, such that $\mathcal{F}$ is totally bounded under the metric $\sigma\left(f,\tilde{f}\right)=\left[\int\left(f-\tilde{f}\right)^{2}dF_{Z}\right]^{\sfrac{1}{2}}$, and denote the space of real-valued bounded functions defined on the index set by the supremum norm by $\ell^{\infty}$. The estimator $\hat{\upsilon}_{m}$ needs to satisfy the following condition:
\begin{con}\label{con:genmod}
The estimator of the components of the decomposition for the general model, $\hat{\upsilon}_{m}\left(z,\tau,\pi,f\right)$, satisfies $\sqrt{n}\left(\hat{\upsilon}_{m}\left(z,\tau,\pi,f\right)-\upsilon_{m}\left(z,\tau,\pi,f\right)\right)\Rightarrow\mathbb{Z}_{\upsilon_{m}}\left(z,\tau,\pi,f\right)$, a stochastic process in $\ell^{\infty}\left(\mathcal{Z}\mathcal{U}\overline{\mathcal{P}}\mathcal{F}\mathcal{M}\right)$, where $\mathbb{Z}_{\upsilon_{m}}\left(z,\tau,\pi,f\right)$ is a tight zero-mean Gaussian process with a.s. uniformly continuous paths is $\mathcal{Z}\mathcal{U}\overline{\mathcal{P}}\mathcal{F}\mathcal{M}$.
\end{con}

To keep notation compact, denote the difference between two counterfactual values of a statistic indexed by $m$ and $m'$ by $\Delta^{m,m'}\left(\cdot\right)$. Moreover, denote the support of $Y$ by $\mathcal{Y}$. The asymptotic distribution of the components of the means and unconditional quantile decomposition are established in the following theorem:
\begin{thm}\label{thm:asymgc}
Let the estimator $\hat{\upsilon}_{m}\left(z,\tau,\pi\right)$ satisfy Condition~\ref{con:genmod}. Under Assumptions~\ref{assum:samp}-\ref{assum:bound}, the following hold:
\begin{align*}
\sqrt{n}\Delta^{m,m'}\left(\hat{\mathbb{E}}\left[Y|S=1\right]-\mathbb{E}\left[Y|S=1\right]\right)\Rightarrow\mathbb{Z}_{\Delta^{mm'}Y|S=1}
\end{align*}
where $\mathbb{Z}_{\Delta^{mm'}Y|S=1}$ is a zero-mean Gaussian random variable $\forall m,m'\in\mathcal{M}$, defined in Appendix~\ref{app:proofs}, and
\begin{align*}
\sqrt{n}\Delta^{m,m'}\left(\hat{Q}_{Y|S=1}\left(\tau\right)-Q_{Y|S=1}\left(\tau\right)\right)\Rightarrow\mathbb{Z}_{Q|S=1,mm'}\left(\tau\right)
\end{align*}
a stochastic process in metric space $\ell^{\infty}\left(\mathcal{T}\mathcal{M}\right)$, where $\mathcal{T}\subset\left(0,1\right)$ and $\mathbb{Z}_{Q|S=1,mm'}\left(\tau\right)$ is a zero-mean tight Gaussian process, defined in Appendix~\ref{app:proofs}, with a.s. uniformly continuous paths in $\mathcal{T}\mathcal{M}$.
\end{thm}

Theorem~\ref{thm:asymgc} establishes the asymptotic gaussianity of the counterfactuals of interest when the estimator of the structural functions satisfies Condition~\ref{con:genmod}. As a consequence, one needs to consider an estimator of the structural functions for which this condition holds. One such example is the estimator proposed by \cite{Arellano2017}. This estimator is adapted to account for different groups, and it has to satisfy the following assumptions:
\begin{assum}\label{assum:linear}
$g_{d}\left(x,\tau\right)=x'\beta_{d}\left(\tau\right)$ $\forall d\in\mathcal{D}$, where $\beta_{d}$ is continuous and such that $g_{d}\left(x,\tau\right)$ is increasing in its last argument.
\end{assum}
\begin{assum}\label{assum:parcop}
$C_{d,x}\left(u,v\right)\equiv C_{d,x}\left(u,v;\theta_{d}\right)$, with $\dim\left(\theta_{d}\right)<\infty\forall d\in\mathcal{D}$. $C_{d,x}\left(u,v;\theta_{d}\right)$ is uniformly continuous and differentiable with respect to its arguments \textit{a.e.} Its density, $c_{d,x}\left(u,v;\theta_{d}\right)$, is well-defined and finite.
\end{assum}
\begin{assum}\label{assum:compact}
Let $\beta\left(\tau\right)$, $\theta$, and $\gamma$ be the concatenation of $\beta_{d}\left(\tau\right)$, $\theta_{d}$, and $\gamma_{d}$ $\forall d\in\mathcal{D}$. $\forall\tau\in\mathcal{U}$, $\left(\beta\left(\tau\right)',\theta',\gamma'\right)'\in int\mathcal{B}\times\Theta\times\Gamma$, where $\mathcal{B}\times\Theta\times\Gamma
$ is compact and convex, and $\mathcal{U}=\left[\varepsilon,1-\varepsilon\right]$, for some constant $\varepsilon$ that is used to avoid the estimation of extreme quantiles.
\end{assum}
\begin{assum}\label{assum:fullrank}
Matrices of derivatives of the moments $J_{\beta d}\left(\tau\right)$, $\tilde{J}_{\beta d}\left(\tau\right)$, $J_{\gamma d}\left(\tau\right)$, $\tilde{J}_{\gamma d}\left(\tau\right)$, $J_{\theta d}\left(\tau\right)$, $\tilde{J}_{\theta d}\left(\tau\right)$ $\forall d\in\mathcal{D}$, as defined in Appendix~\ref{app:proofs}, are continuous and have full rank, uniformly over $\mathcal{B}\times\Theta\times\Gamma\times\mathcal{U}$, and $\forall d\in\mathcal{D}$.
\end{assum}
\begin{assum}\label{assum:prop}
Denote the support of $\pi_{d}\left(Z\right)$ conditional on $X=x$ by $\mathcal{P}_{d,x}$. $\forall x\in\mathcal{X}$ and $\forall d\in\mathcal{D}$, $\mathcal{P}_{d,x}\subset\left[0,1\right]$ is an open interval, and $\overline{\mathcal{P}}$ is the closure of $\cup_{d\in\mathcal{D}}\mathcal{P}_{d,x}$.
\end{assum}

Assumption~\ref{assum:linear} is made for convenience. Despite restricting the shape of the distribution of $Y$, it does not force conditional quantile curves to be parallel to each other for different values of the covariates. Moreover, it does not suffer from the curse of dimensionality as other more flexible alternatives, such as the partially linear model \citep{Lee2003}. Assumption~\ref{assum:parcop} is also imposed for practical reasons, and it allows to use the most common parametric copulas, such as the Clayton, the Gaussian, or a Bernstein copula of a fixed order.\footnote{Note that Assumption~\ref{assum:parcop} does not provide a way to select the most appropriate copula given some data. This could be done, \textit{e.g.}, by choosing the copula that minimizes the criterion function given in Equation~\ref{eq:thetahat} whenever the candidate copulas have the same number of parameters. Alternatively, it could be selected using cross validation as in \cite{Pereda2022}.} Importantly, it allows the copula to be dependent on the covariates. Assumption~\ref{assum:compact} is a regularity condition and Assumption~\ref{assum:fullrank} ensures that the moments needed to derive the asymptotic distribution of the estimator have full rank. Lastly, Assumption~\ref{assum:prop} ensures that the participation decision is not deterministic for any individual.

The following steps describe how to compute the estimator, and how to implement it to obtain an estimate of the decompositions:
\begin{enumerate}
\item Estimate the propensity score by $\hat{\pi}_{d}\left(Z_{i}\right)\equiv\pi_{d}\left(Z_{i},\hat{\gamma}_{d}\right)$.
\item Fix a value of $t\in\Theta$. For $\forall d\in\mathcal{D}$ and $\tau\in\mathcal{U}$, compute $\hat{\beta}_{d}\left(\tau;t\right)$ as
\begin{align}\label{eq:betahat}
\hat{\beta}_{d}\left(\tau;t\right)\equiv\arg\min_{b\in\mathcal{B}}\sum_{i=1}^{n}\mathbf{1}\left(D_{i}=d\right)S_{i}\rho_{G_{d,x}\left(\tau,\hat{\pi}_{d}\left(Z_{i}\right);t\right)}\left(Y_{i}-X_{i}'b\right)
\end{align}
where $\rho_{u}\left(x\right)\equiv xu\mathbf{1}\left(x\geq 0\right)-\left(1-u\right)x\mathbf{1}\left(x<0\right)$ denotes the check function.
\item Estimate the copula parameters $\forall d\in\mathcal{D}$ by minimizing over $t\in\Theta$:
\begin{align}\label{eq:thetahat}
\hat{\theta}_{d}\equiv\arg\min_{t\in\Theta}\left\|\sum_{i=1}^{n}\int_{\varepsilon}^{1-\varepsilon}\mathbf{1}\left(D_{i}=d\right)S_{i}\varphi\left(\tau,Z_{i}\right)\left[\mathbf{1}\left(Y_{i}\leq X_{i}'\hat{\beta}_{d}\left(\tau;t\right)\right)-G_{d,x}\left(\tau,\hat{\pi}\left(Z_{i}\right);t\right)\right]d\tau\right\|
\end{align}
where $\varphi\left(\tau,Z_{i}\right)$ is an instrument function.\footnote{\textit{E.g.}, a polynomial of the propensity score \citep{Arellano2017}.}
\item The slope parameters are obtained by $\hat{\beta}_{d}\left(\tau\right)\equiv\hat{\beta}_{d}\left(\tau;\hat{\theta}_{d}\right)$ $\forall d\in\mathcal{D}$.
\item The SQF and the copula are respectively given by $\hat{g}_{d}\left(x,\tau\right)=x'\hat{\beta}_{d}\left(\tau\right)$ and $\hat{C}_{d,x}\left(\tau,\pi\right)=C_{d,x}\left(\tau,\pi;\hat{\theta}_{d}\right)$.
\end{enumerate}

Step 1 is standard and can be done, \textit{e.g.}, by logit or probit. Step 2 is a rotated quantile regression conditional on a particular value of the copula. Note that in practice one needs to set a grid of values of $\tau$, such as $\tau=\left\{0.01,...,0.99\right\}$. The third step is computationally expensive, as it involves the minimization over a non-convex space. The most common parametric copulas depend on few parameters, so that a grid search may be feasible. For Bernstein copulas, one can use the algorithm proposed in \cite{Pereda2022}. The last two steps are immediate, and they yield the slope parameters, which are those estimated in step 2 for the value of the copula estimated in step 3, the SQF and the copula.

The following theorem establishes the uniform asymptotic distribution of the \cite{Arellano2017} estimator:
\begin{thm}\label{thm:asym}
Let $\hat{\vartheta}_{m}\left(\tau\right)\equiv\left(\hat{\beta}_{m}\left(\tau\right)',\hat{\theta}_{m}',\hat{\gamma}_{m}'\right)'$. Under Assumptions~\ref{assum:samp}-\ref{assum:prop}, their joint asymptotic distribution is given by $\sqrt{n}\left(\hat{\vartheta}_{m}\left(\tau\right)-\vartheta_{m}\left(\tau\right)\right)\Rightarrow \mathbb{Z}_{\vartheta_{m}}\left(\tau\right)$, a stochastic process in $\ell^{\infty}\left(\mathcal{U}\mathcal{M}\right)$, where $\mathbb{Z}_{\vartheta_{m}}\left(\tau\right)$ is a zero-mean tight Gaussian process with a.s. uniformly continuous paths in $\mathcal{U}\mathcal{M}$ and covariance function $\Sigma_{\vartheta_{m}}\left(\tau,\tau'\right)$, where:
\begin{align*}
&\Sigma_{\vartheta_{m}}\left(\tau,\tau'\right)=H_{m}\left(\tau\right)\Sigma_{R_{m}}\left(\tau,\tau'\right)H_{m'}\left(\tau'\right)'\\
&H_{m}\left(\tau\right)=F_{m}^{I}\left(\tau\right)\left[C_{m}\left(\tau\right)+\left(I-\int_{\varepsilon}^{1-\varepsilon}D_{m}\left(u\right)F_{m}^{I}\left(u\right)du\right)^{-1}\int_{\varepsilon}^{1-\varepsilon}D_{m}\left(u\right)F_{m}^{I}\left(u\right)C_{m}\left(u\right)du\right]\\
&\Sigma_{R_{m}}\left(\tau,\tau'\right)=\mathbb{E}\left[\mathbb{Z}_{R_{m}}\left(\tau\right)\mathbb{Z}_{R_{m'}}\left(\tau'\right)'\right]
\end{align*}
and functions $C_{m}\left(\tau\right)$, $D_{m}\left(\tau\right)$, $F_{m}^{I}\left(\tau\right)$ and $\mathbb{Z}_{R_{m}}\left(\tau\right)$ are defined in Appendix~\ref{app:proofs}.
\end{thm}

It remains to verify that it satisfies Condition~\ref{con:genmod}. This is done in the following corollary:
\begin{cor}\label{cor:asym}
Let $\hat{g}_{d}\left(x,\tau\right)=x'\hat{\beta}_{d}\left(\tau\right)$, $\hat{c}_{d,x}\left(\tau,\pi\right)=c_{d,x}\left(\tau,\pi;\hat{\theta}_{d}\right)$, $\hat{\pi}_{d}\left(z\right)=\pi_{d}\left(z;\hat{\gamma}_{d}\right)$ and $\hat{F}_{Z}^{d}\left(z\right)=\frac{1}{n_{d}}\sum_{i=1}^{n}\mathbf{1}\left(D_{i}=d\right)\mathbf{1}\left(Z_{i}\leq z\right)$. They satisfy Condition~\ref{con:genmod}.
\end{cor}

\subsection{Inference}

The expressions of the asymptotic variance of the different estimators are complex and depend on several density functions. Therefore, using resampling methods to obtain standard errors is preferable to obtaining closed-form expressions. In this paper I consider the weighted bootstrap \citep{Ma2005}. The following assumption defines the weights used to obtain all the bootstrap estimates:

\begin{assum}\label{assum:weights}
Let $W_{i}$ be an \textit{iid}\ sample of positive weights, such that $\mathbb{E}\left(W_{i}\right)=1$, $Var\left(W_{i}\right)=\omega_{0}>0$ and is independent of $\left(Y_{i},D_{i},S_{i},Z_{i}'\right)'$ for $i=1,...,n$.
\end{assum}

For the \cite{Arellano2017} estimator, the weighted bootstrap is implemented as follows:
\begin{itemize}
\item For each repetition $r=1,...,R$, draw the weights $W_{i,r}$ for $i=1,...,n$ that satisfy Assumption~\ref{assum:weights}.
\item Estimate the propensity score using the weights for each observation. Denote the estimate by $\hat{\pi}_{d,r}^{*}\left(z_{i}\right)\equiv\pi_{d}\left(z_{i},\hat{\gamma}_{d,r}^{*}\right)$.
\item Estimate the slope and copula parameters by adding the weights to Equations~\ref{eq:betahat}-\ref{eq:thetahat}:
\begin{align*}
&\hat{\beta}_{d,r}^{*}\left(\tau;t\right)\equiv\arg\min_{b\in\mathcal{B}}\sum_{i=1}^{n}W_{i,r}\mathbf{1}\left(D_{i}=d\right)S_{i}\rho_{G_{d,x}\left(\tau,\hat{\pi}_{d,r}^{*}\left(Z_{i}\right);t\right)}\left(Y_{i}-X_{i}'b\right)\\
&\hat{\theta}_{d,r}^{*}\equiv\arg\min_{t\in\Theta}\left\|\sum_{i=1}^{n}\int_{\varepsilon}^{1-\varepsilon}W_{i,r}\mathbf{1}\left(D_{i}=d\right)S_{i}\varphi\left(\tau,Z_{i}\right)\left[\mathbf{1}\left(Y_{i}\leq X_{i}'\hat{\beta}_{d,r}^{*}\left(\tau;t\right)\right)-G_{d,x}\left(\tau,\hat{\pi}_{d,r}^{*}\left(Z_{i}\right);t\right)\right]d\tau\right\|\\
&\hat{\beta}_{d,r}^{*}\left(\tau\right)\equiv\hat{\beta}_{d,r}^{*}\left(\tau;\hat{\theta}_{d,r}^{*}\right)
\end{align*}
\item Estimate the counterfactual mean outcomes and unconditional distributions as:
\begin{align*}
&\hat{\mathbb{E}}^{*}\left[Y_{r}^{m}|S=1\right]=\frac{1}{n_{h}}\sum_{i=1}^{n}W_{i,r}\int_{\varepsilon}^{1-\varepsilon}\hat{g}_{j,r}^{*}\left(X_{i},u\right)d\hat{G}_{k,x,r}^{*}\left(u,\hat{\pi}_{l,r}^{*}\left(Z_{i}\right)\right)\mathbf{1}\left(D_{i}=h\right)\\
&\hat{F}_{Y_{r}|S=1}^{\ell,*}\left(y\right)=\frac{1}{n_{h}}\sum_{i=1}^{n}W_{i,r}\left[\varepsilon+\int_{\varepsilon}^{1-\varepsilon}\mathbf{1}\left(\hat{g}_{j,r}^{*}\left(X_{i},u\right)\leq y\right)d\hat{G}_{k,x,r}\left(u,\hat{\pi}_{l,r}^{*}\left(Z_{i}\right)\right)\right]\mathbf{1}\left(D_{i}=h\right)
\end{align*}
where $\hat{g}_{d,r}^{*}\left(x,u\right)\equiv x'\hat{\beta}_{d,r}^{*}\left(u\right)$ and $\hat{G}_{d,x,r}^{*}\left(u,\pi\right)\equiv G_{d,x}\left(u,\pi;\hat{\theta}_{d,r}^{*}\right)$.
\item The remaining counterfactual unconditional quantile function and the components of the decomposition are computed as described in the text. They are denoted by $\hat{Q}_{Y_{r}|S=1}^{m,*}\left(\tau\right)$, $\Delta^{m,m'}\hat{\mathbb{E}}^{*}\left[Y_{r}|S=1\right]$, and $\Delta^{m,m'}\hat{Q}_{Y_{r}|S=1}^{*}\left(\tau\right)$.
\item Once all $r$ estimates have been obtained, estimate the variance of each of the statistics as $\frac{q_{0.75}\left(\tau\right)-q_{0.25}\left(\tau\right)}{z_{0.75}-z_{0.25}}$, where $z_{p}$ is the $p$-th quantile of the standard normal distribution, and $q_{p}\left(\tau\right)$ is the $p$-th quantile of the distribution of the statistic, for $r=1,...,R$.
\end{itemize}

This bootstrap estimator of the variance is based on the one presented in \cite{Chernozhukov2013}. Even though it is possible to use the variance of the estimator across repetitions of the bootstrap to obtain the standard errors, it would require additional conditions for it to be valid \citep{Kato2011}. This estimator only requires that the bootstrap converges in distribution to the asymptotic distribution of the sample estimator, which is established in the following theorem:
\begin{thm}\label{thm:bootstrap}
Under Assumptions~\ref{assum:samp}-\ref{assum:weights}, the weighted bootstrap estimators are denoted by $\hat{\vartheta}_{m,r}^{*}\left(\tau\right)$, $\Delta^{m,m'}\hat{\mathbb{E}}^{*}\left[Y_{r}|S=1\right]$, and $\Delta^{m,m'}\hat{Q}_{Y_{r}|S=1}^{*}\left(\tau\right)$. They consistently estimate the limiting laws of $\hat{\vartheta}_{m}\left(\tau\right)$, $\Delta^{m,m'}\hat{\mathbb{E}}\left[Y|S=1\right]$, and $\Delta^{m,m'}\hat{Q}_{Y|S=1}\left(\tau\right)$. Moreover,
\begin{align*}
\sqrt{\frac{n}{\omega_{0}}}\left(\hat{\vartheta}_{m,r}^{*}\left(\tau\right)-\hat{\vartheta}_{m}\left(\tau\right)\right)&\Rightarrow\mathbb{Z}_{\vartheta_{m}}\left(\tau\right)
\end{align*}
a stochastic process in metric space $\ell^{\infty}\left(\mathcal{U}\mathcal{M}\right)$,
\begin{align*}
\sqrt{\frac{n}{\omega_{0}}}\left(\Delta^{m,m'}\hat{\mathbb{E}}^{*}\left[Y_{r}|S=1\right]-\Delta^{m,m'}\hat{\mathbb{E}}\left[Y|S=1\right]\right)&\Rightarrow\mathbb{Z}_{\Delta^{mm'}Y|S=1}
\end{align*}
$\forall m,m'\in\mathcal{M}$,
\begin{align*}
\sqrt{\frac{n}{\omega_{0}}}\left(\Delta^{m,m'}\hat{Q}_{Y_{r}|S=1}^{*}\left(\tau\right)-\Delta^{m,m'}\hat{Q}_{Y|S=1}\left(\tau\right)\right)&\Rightarrow\mathbb{Z}_{Q|S=1,mm'}\left(\tau\right)
\end{align*}
a stochastic process in metric space $\ell^{\infty}\left(\mathcal{T}\mathcal{M}\right)$,
\end{thm}

On top of providing uniform confidence bands for the functionals of interest and the intermediate functions, the weighted bootstrap can be used to carry out uniform inference using, \textit{e.g.}, a Kolmogorov-Smirnov test to any of the components of the decomposition of the unconditional quantiles. For example, one could test the null hypothesis that one of the components equals a specific value, $\Delta^{mm'}Q_{Y|S=1}\left(\tau\right)$. The test statistic would be given by
\begin{align*}
KS_{n}=\sup_{\tau\in\mathcal{T}}\sqrt{n}\hat{\Sigma}_{Q|S=1,mm'}\left(\tau\right)^{-\sfrac{1}{2}}\left|\Delta^{m,m'}\left(\hat{Q}_{Y|S=1}\left(\tau\right)-Q_{Y|S=1}\left(\tau\right)\right)\right|
\end{align*}
where $\hat{\Sigma}_{Q|S=1,mm'}\left(\tau\right)$ is an estimator of the asymptotic variance of $\Delta^{m,m'}\hat{Q}_{Y|S=1}\left(\tau\right)$, such as the one proposed in the weighted bootstrap algorithm. The critical value $c_{1-\alpha}$ would be $1-\alpha$ quantile of the distribution of the bootstrapped of the $KS_{n}$ statistic. Similar uniform confidence bands can be constructed for other functionals of interest.

\subsection{Ancillary Decompositions}

The ancillary decompositions are also based on the vector $\upsilon_{\ell}\left(z,\tau,\pi,f\right)$. The counterfactual values of the mean propensity score and the mean value of the unobservables are given by
\begin{align}
\hat{\mathbb{E}}\left[\pi^{m}|S=1\right]&=\frac{1}{n_{h}}\sum_{i=1}^{n}\hat{\pi}_{l}\left(Z_{i}\right)\mathbf{1}\left(D_{i}=h\right)\label{eq:pihat}\\
\hat{\mathbb{E}}\left[U^{m}|S=1\right]&=\frac{1}{n_{h}}\sum_{i=1}^{n}\int_{\varepsilon}^{1-\varepsilon}ud\hat{G}_{k,x}\left(u,\hat{\pi}_{l}\left(Z_{i}\right)\right)\mathbf{1}\left(D_{i}=h\right)\label{eq:uhat}
\end{align}

The asymptotic distribution of the components of the two ancillary distributions is established in the following theorem:
\begin{thm}\label{thm:asymanc}
Let the estimator $\hat{\upsilon}_{m}\left(z,\tau,\pi\right)$ satisfy Condition~\ref{con:genmod}. Under Assumptions~\ref{assum:samp}-\ref{assum:bound}, the following hold for all $\left(mm'\right)$:
\begin{align*}
&\sqrt{n}\Delta^{m,m'}\left(\hat{\mathbb{E}}\left[\pi|S=1\right]-\mathbb{E}\left[\pi|S=1\right]\right)\Rightarrow\mathbb{Z}_{\Delta^{mm'}\pi|S=1}\\
&\sqrt{n}\Delta^{m,m'}\left(\hat{\mathbb{E}}\left[U|S=1\right]-\mathbb{E}\left[U|S=1\right]\right)\Rightarrow\mathbb{Z}_{\Delta^{mm'}U|S=1}
\end{align*}
where $\mathbb{Z}_{\Delta^{\ell\ell'}\pi|S=1}$ and $\mathbb{Z}_{\Delta^{\ell\ell'}U|S=1}$ are zero-mean Gaussian random variables.
\end{thm}
\section{Evolution of the Gender Wage Gap}\label{sec:emp}

I study the evolution of the gender gap between earnings distributions using the Current Population Survey (CPS) dataset. I extend the analysis in \cite{Maasoumi2019} by decomposing several features of the distribution of actual earnings for employed workers and the entire population. In addition, I do the two ancillary decompositions regarding differences in participation and self-selection between men and women.

To preserve the comparability to \cite{Maasoumi2019}, the sample covers the 1976-2013 period, and the regressions are done on a year-by-year basis, abstracting from any dynamics. I restrict the analysis to individuals between 18 and 64 years old, who work for wages and salary, do not live in group quarters and worked at least for 20 weeks and 25 hours per week in the previous year.\footnote{Note that this analysis excludes part-time workers, seasonal workers, self-employed, or the total number of hours worked (which affects the intensive margin). Despite this, the gaps considered in this paper are relevant for a large fraction of the working age population, and the elements analyzed could also be relevant for the analysis of a more comprehensive gender gap that accounts for the aforementioned factors. Such analysis is beyond the scope of this paper.} Like them, I estimate the propensity score using the probit estimator. I use the same regressors they used, \textit{i.e.}, a third degree polynomial of age, four levels of education, four regional dummies, marital status, an indicator for white race, and the interactions between age and the other listed covariates, plus another variable they did not use: the number of children.

The dependent variable is mean log hourly wages.\footnote{As in \cite{Maasoumi2019}, it is computed as the logarithm of the total wage and salary income divided by the number of week and hours worked during the previous year and adjusted for inflation using the 1999 consumer price index adjustment factor.} The specification for the QRS estimator uses the same set of variables (except for the interactions between age and the remaining covariates), the same instrument (number of children below 5 years old) and two parametric copulas: the Frank and the Gaussian copulas.\footnote{A discussion of the validity of the exclusion restriction is provided in \cite{Maasoumi2019}. They acknowledge that the instrument may be stronger for women than for men, as well as for earlier years than for the latter. In addition, it may be stronger for married workers than for those who are not.} However, there are some slight differences along several dimensions: I include the regressor number of children, the quantile grid, the propensity score, and the objective function used to estimate the copula.\footnote{For precision, the quantile grid I used for the estimation is $\left(0.01,0.02,...,0.99\right)$, while the one used by \cite{Maasoumi2019} was $\left(0.3,0.4,...,0.7\right)$; I use the propensity score as the instrument $\varphi\left(u,z\right)=\hat{\pi}\left(z\right)$ as suggested in \cite{Arellano2017}, whereas \cite{Maasoumi2019} use $\varphi\left(u,z\right)=\sqrt{u\left(1-u\right)}\hat{\pi}\left(z\right)$, which puts less weight on values that are further away from the median; the objective function equals Equation~\ref{eq:thetahat}, while the objective function used by \cite{Maasoumi2019} is $\frac{1}{N}\sum_{i=1}^{N}\sum_{j=1}^{J}\left(\varphi\left(\tau_{j},z_{i}\right)\left(\mathbf{1}\left(y_{i}\leq x_{i}'\hat{\beta}\left(\tau_{j}\right)\right)-G\left(\tau_{j},\hat{\pi}\left(z_{i}\right);\hat{\theta}\right)\right)\right)$; the implemented quantile regression estimates from Stata in \cite{Maasoumi2019} were in some cases numerically slightly worse than the ones in Matlab, \textit{i.e.}, the value of the check function was smaller for the latter. See Appendix~\ref{app:add} for further details.} Additionally, I also allow for more flexible specifications that separately estimate the main equation according to race (white vs non-white), level of education (college vs less than college) and marital status (married vs non-married), which are reported in Section~\ref{sec:hetcop}, and the estimates using the same specification as in \cite{Maasoumi2019}, which are reported in Appendix~\ref{app:add}.\footnote{For completeness, I also report the results of the decompositions of potential outcomes and the estimates of the general entropy measures considered by \cite{Maasoumi2019} in Appendix~\ref{app:add}.}

\subsection{Evolution of Participation}

To analyze the evolution of the gender earnings gap, I begin by analyzing the evolution of labor market participation for both genders. Table~\ref{tab:prop} reports the average estimated propensity score by gender for the entire period. There has been a marked catch-up between female and male participation: in 1976, female participation was roughly one third, steadily increasing to over one half in the early 2000s, to fall slightly in the aftermath of the financial crisis. Meanwhile, male participation has been more stable: it has been equal to around two thirds until the financial crisis, falling to about 60\% afterwards.

\begin{table}[htbp]
  \centering
  \caption{Average propensity to work by gender}\label{tab:prop}
		\begin{threeparttable}
    \begin{tabular}{ccc|ccc}
		\hline
    Year  & Male  & Female & Year  & Male  & Female \\
		\hline
    1976  & 68.0  & 35.0  & 1995  & 66.7  & 48.0 \\
    1977  & 68.6  & 36.1  & 1996  & 67.2  & 49.1 \\
    1978  & 68.6  & 37.5  & 1997  & 67.1  & 49.7 \\
    1979  & 69.8  & 39.2  & 1998  & 67.7  & 50.3 \\
    1980  & 69.1  & 40.7  & 1999  & 68.7  & 51.1 \\
    1981  & 67.6  & 40.4  & 2000  & 68.8  & 52.0 \\
    1982  & 64.8  & 38.7  & 2001  & 69.2  & 52.5 \\
    1983  & 61.6  & 37.9  & 2002  & 68.3  & 51.5 \\
    1984  & 61.5  & 38.9  & 2003  & 66.3  & 50.3 \\
    1985  & 63.1  & 40.8  & 2004  & 65.2  & 49.5 \\
    1986  & 64.4  & 41.7  & 2005  & 65.3  & 49.3 \\
    1987  & 64.7  & 42.6  & 2006  & 65.7  & 50.0 \\
    1988  & 66.6  & 46.0  & 2007  & 66.3  & 50.4 \\
    1989  & 67.0  & 46.8  & 2008  & 65.8  & 50.9 \\
    1990  & 68.8  & 47.7  & 2009  & 64.0  & 49.8 \\
    1991  & 67.8  & 47.6  & 2010  & 59.8  & 47.7 \\
    1992  & 66.7  & 47.8  & 2011  & 59.2  & 46.8 \\
    1993  & 65.6  & 47.4  & 2012  & 60.1  & 46.7 \\
    1994  & 65.5  & 47.2  & 2013  & 60.8  & 46.9 \\
		\hline
    \end{tabular}\begin{tablenotes}[para,flushleft]
\begin{spacing}{1}
{\footnotesize Notes: average estimated propensity score by year and gender; coefficients scaled by 100.}
\end{spacing}
\end{tablenotes}
\end{threeparttable}
\end{table}

Consequently, the gender participation gap has more than halved during the period, from 33\% to about 14\% (Table~\ref{tab:dprop}). Its decomposition shows that almost the entirety of the gap is explained by the coefficients components. In words, there has been a structural increase in female participation into employment for women unrelated to gender differences in covariates. On the other hand, the endowments component has been either statistically not significant or slightly negative for most of the period. Indeed, only between 1987 and 1989 was this component positive and barely significant. Hence, the catch up in college education rates for female workers has not contributed to an increased participation in the labor force relative to men.

\begin{table}[htbp]
  \centering
  \caption{Participation decomposition}\label{tab:dprop}
		\begin{threeparttable}
    \begin{tabular}{cc@{\,}lc@{\,}lc@{\,}l|cc@{\,}lc@{\,}lc@{\,}l}
		\hline
    Year  & \multicolumn{2}{c}{Total} & \multicolumn{2}{c}{EC} & \multicolumn{2}{c}{CC} & Year  & \multicolumn{2}{c}{Total} & \multicolumn{2}{c}{EC} & \multicolumn{2}{c}{CC} \\
		\hline
    1976  & 33.0  & ***   & -0.6  & ***   & 33.7  & ***   & 1995  & 18.7  & ***   & -0.4  & **    & 19.1  & *** \\
    1977  & 32.5  & ***   & -0.7  & ***   & 33.2  & ***   & 1996  & 18.0  & ***   & -0.2  &       & 18.3  & *** \\
    1978  & 31.1  & ***   & -0.6  & ***   & 31.7  & ***   & 1997  & 17.4  & ***   & -0.2  &       & 17.6  & *** \\
    1979  & 30.7  & ***   & -0.5  & ***   & 31.2  & ***   & 1998  & 17.4  & ***   & -0.3  &       & 17.7  & *** \\
    1980  & 28.4  & ***   & -0.4  & **    & 28.8  & ***   & 1999  & 17.6  & ***   & -0.1  &       & 17.7  & *** \\
    1981  & 27.2  & ***   & -0.1  &       & 27.3  & ***   & 2000  & 16.8  & ***   & 0.0   &       & 16.8  & *** \\
    1982  & 26.1  & ***   & -0.1  &       & 26.2  & ***   & 2001  & 16.7  & ***   & 0.0   &       & 16.8  & *** \\
    1983  & 23.8  & ***   & 0.0   &       & 23.8  & ***   & 2002  & 16.8  & ***   & -0.2  & *     & 17.0  & *** \\
    1984  & 22.5  & ***   & 0.0   &       & 22.5  & ***   & 2003  & 15.9  & ***   & -0.2  &       & 16.1  & *** \\
    1985  & 22.3  & ***   & 0.0   &       & 22.4  & ***   & 2004  & 15.7  & ***   & -0.4  & ***   & 16.1  & *** \\
    1986  & 22.6  & ***   & 0.1   &       & 22.5  & ***   & 2005  & 16.1  & ***   & -0.3  & **    & 16.3  & *** \\
    1987  & 22.1  & ***   & 0.3   & *     & 21.8  & ***   & 2006  & 15.7  & ***   & -0.2  & *     & 15.9  & *** \\
    1988  & 20.6  & ***   & 0.3   & *     & 20.3  & ***   & 2007  & 15.8  & ***   & -0.4  & ***   & 16.2  & *** \\
    1989  & 20.3  & ***   & 0.3   & *     & 20.0  & ***   & 2008  & 14.9  & ***   & -0.4  & ***   & 15.4  & *** \\
    1990  & 21.1  & ***   & 0.1   &       & 21.0  & ***   & 2009  & 14.2  & ***   & -0.7  & ***   & 15.0  & *** \\
    1991  & 20.2  & ***   & -0.3  & **    & 20.5  & ***   & 2010  & 12.1  & ***   & -0.8  & ***   & 12.9  & *** \\
    1992  & 18.9  & ***   & -0.2  &       & 19.0  & ***   & 2011  & 12.5  & ***   & -0.8  & ***   & 13.3  & *** \\
    1993  & 18.2  & ***   & -0.2  &       & 18.4  & ***   & 2012  & 13.4  & ***   & -0.8  & ***   & 14.2  & *** \\
    1994  & 18.3  & ***   & -0.2  &       & 18.5  & ***   & 2013  & 13.9  & ***   & -1.0  & ***   & 14.8  & *** \\
		\hline
    \end{tabular}\begin{tablenotes}[para,flushleft]
\begin{spacing}{1}
{\footnotesize Notes: Total, EC and CC respectively denote total difference, endowments component and coefficients component; coefficients scaled by 100; *, ** and ** respectively denote statistical significance at the 90\%, 95\% and 99\% confidence level.}
\end{spacing}
\end{tablenotes}
\end{threeparttable}
\end{table}

\subsection{Evolution of Self-Selection}

The second feature of interest is the evolution of differences in self selection, which depends mainly on the estimated copula. Because the values for different copulas are not directly comparable, I report the Kendall's $\tau$ correlation coefficients.\footnote{Unlike to the more common Spearman's $\rho$ correlation coefficient, Kendall's $\tau$ is invariant to the distribution of the marginals.} Table~\ref{tab:kendalls} reports the baseline estimates for each year and gender, both with the Frank and the Gaussian copulas. These coefficients indicate that the amount of selection into employment has steadily increased for female workers: until the early 80s, there used to be negative selection that turned positive afterwards.\footnote{Recall that a negative (positive) coefficient implies positive (negative) selection into employment.} This result reinforces the dynamics on selection into employment found by \cite{Mulligan2008}. On the other hand, the amount of selection for male workers has fluctuated more over time, being either above or below that of females depending on the year. The results are very similar with both copulas, which suggests that the choice of the parametric copula is of secondary importance. To assess the sensitivity of these results to the model used, I report in Table~\ref{tab:heck} in Appendix~\ref{app:tabs} a comparison of the correlation coefficients with those of the Heckman 2-stage estimator \citep{Heckman1979}. The findings show that the copula estimated by both models are very similar throughout the entire period.

\begin{table}[htbp]
  \centering
  \caption{Kendall's $\tau$ correlation coefficients}\label{tab:kendalls}
		\begin{threeparttable}
    \begin{tabular}{cr@{\,}lr@{\,}lr@{\,}lr@{\,}l}
		\hline
							& \multicolumn{4}{c}{Frank copula} & \multicolumn{4}{c}{Gaussian copula} \\
    Year  & \multicolumn{2}{c}{Male} & \multicolumn{2}{c}{Female} & \multicolumn{2}{c}{Male} & \multicolumn{2}{c}{Female} \\
		\hline
    1976  & 0.28  & ***   & 0.15  & ***   & 0.28  & ***   & 0.13  & *** \\
    1977  & 0.16  &       & 0.09  & *     & 0.16  &       & 0.08  & * \\
    1978  & 0.21  & *     & -0.01 &       & 0.21  & *     & -0.01 &  \\
    1979  & 0.23  & **    & -0.02 &       & 0.24  & **    & -0.02 &  \\
    1980  & 0.22  & **    & 0.01  &       & 0.24  & **    & 0.01  &  \\
    1981  & 0.21  & ***   & 0.01  &       & 0.23  & ***   & 0.01  &  \\
    1982  & 0.14  &       & -0.03 &       & 0.15  &       & -0.03 &  \\
    1983  & 0.01  &       & -0.10 & *     & 0.02  &       & -0.09 & ** \\
    1984  & -0.09 &       & 0.02  &       & -0.08 &       & 0.02  &  \\
    1985  & 0.13  &       & -0.05 &       & 0.12  &       & -0.05 &  \\
    1986  & 0.09  &       & -0.11 & **    & 0.09  &       & -0.10 & *** \\
    1987  & 0.13  &       & -0.17 & ***   & 0.11  &       & -0.17 & *** \\
    1988  & -0.06 &       & -0.10 & ***   & -0.05 &       & -0.10 & *** \\
    1989  & -0.10 &       & -0.10 & **    & -0.12 &       & -0.10 & ** \\
    1990  & -0.07 &       & -0.22 & ***   & -0.09 &       & -0.23 & *** \\
    1991  & -0.06 &       & -0.10 & ***   & -0.05 &       & -0.10 & ** \\
    1992  & -0.08 &       & -0.11 & ***   & -0.09 &       & -0.11 & *** \\
    1993  & -0.28 & **    & -0.13 & ***   & -0.28 & **    & -0.13 & *** \\
    1994  & -0.21 &       & -0.17 & ***   & -0.22 & *     & -0.17 & *** \\
    1995  & -0.30 & ***   & -0.21 & ***   & -0.27 & **    & -0.21 & *** \\
    1996  & 0.07  &       & -0.21 & ***   & 0.07  &       & -0.22 & *** \\
    1997  & 0.01  &       & -0.21 & ***   & 0.01  &       & -0.21 & *** \\
    1998  & -0.19 &       & -0.20 & ***   & -0.22 &       & -0.21 & *** \\
    1999  & 0.08  &       & -0.20 & ***   & 0.06  &       & -0.22 & *** \\
    2000  & 0.09  &       & -0.12 & *     & 0.10  &       & -0.13 & * \\
    2001  & 0.12  &       & -0.12 & **    & 0.13  &       & -0.12 & ** \\
    2002  & 0.04  &       & -0.13 & **    & 0.04  &       & -0.13 & ** \\
    2003  & 0.05  &       & -0.25 & ***   & 0.06  &       & -0.25 & *** \\
    2004  & 0.08  &       & -0.21 & ***   & 0.10  &       & -0.21 & *** \\
    2005  & 0.08  &       & -0.21 & ***   & 0.08  &       & -0.22 & *** \\
    2006  & 0.28  & ***   & -0.17 & ***   & 0.29  & ***   & -0.17 & *** \\
    2007  & 0.10  &       & -0.28 & ***   & 0.10  &       & -0.30 & *** \\
    2008  & 0.13  &       & -0.30 & ***   & 0.13  &       & -0.31 & *** \\
    2009  & -0.05 &       & -0.25 & ***   & -0.06 &       & -0.26 & *** \\
    2010  & -0.21 &       & -0.27 & ***   & -0.18 &       & -0.28 & *** \\
    2011  & -0.03 &       & -0.31 & ***   & -0.01 &       & -0.32 & *** \\
    2012  & -0.28 & **    & -0.28 & ***   & -0.26 & *     & -0.28 & *** \\
    2013  & 0.06  &       & -0.25 & ***   & 0.06  &       & -0.27 & *** \\
		\hline
    \end{tabular}\begin{tablenotes}[para,flushleft]
\begin{spacing}{1}
{\footnotesize Notes: Kendall's $\tau$ correlation coefficients of the copula estimates by year and gender; *, ** and ** respectively denote statistical significance at the 90\%, 95\% and 99\% confidence level.}
\end{spacing}
\end{tablenotes}
\end{threeparttable}
\end{table}

A more informative way to understand these estimates is to compare the average value of the unobservable $u$ across genders and periods. This is shown in Table~\ref{tab:meanu} and Figure~\ref{fig:meanu}. For the entire period considered, the average value of $u$ for full-time employed females steadily increased from slightly above 40 to around 60. The average value for employed males has slightly increased over time. However, it displayed much more fluctuation over time, which may be partly related to the lack of strength of instrument for male workers. Still, there was a positive trend until the mid-nineties, steadily decreasing until the eve of the financial crisis, when it soared for a few years.\footnote{These trends can be seen more clearly if one smooths the copula estimates over time. In Figure~\ref{fig:meanuma} in Appendix~\ref{app:tabs} I compare the baseline estimates to a 5 year moving average of the copula.} In other words, along with the increase in participation, there has been an increase in the amount of self-selection into employment for women. Hence, potential earnings of non-employed women are lower than actual earnings of those employed, given the same observed characteristics.

\begin{table}[htbp]
  \centering
  \caption{Mean value of $u$ for participants}\label{tab:meanu}
		\begin{threeparttable}
    \begin{tabular}{ccccc|ccccc}
		\hline
          & \multicolumn{2}{c}{Frank copula} & \multicolumn{2}{c}{Gaussian copula} &       & \multicolumn{2}{c}{Frank copula} & \multicolumn{2}{c}{Gaussian copula} \\
    Year  & Male  & Female & Male  & Female & Year  & Male  & Female & Male  & Female \\
		\hline
    1976  & 43.4  & 42.7  & 43.6  & 43.5  & 1995  & 56.5  & 57.0  & 55.5  & 56.5 \\
    1977  & 46.0  & 45.5  & 46.2  & 45.9  & 1996  & 47.9  & 57.1  & 48.0  & 56.8 \\
    1978  & 44.8  & 49.8  & 45.1  & 49.8  & 1997  & 49.2  & 56.8  & 49.2  & 56.5 \\
    1979  & 44.7  & 50.3  & 44.8  & 50.3  & 1998  & 53.9  & 56.5  & 54.3  & 56.2 \\
    1980  & 44.8  & 49.2  & 44.6  & 49.2  & 1999  & 47.7  & 56.4  & 48.2  & 56.5 \\
    1981  & 44.7  & 49.0  & 44.5  & 49.2  & 2000  & 47.5  & 53.6  & 47.5  & 53.7 \\
    1982  & 46.0  & 50.6  & 45.9  & 50.5  & 2001  & 47.0  & 53.3  & 46.8  & 53.3 \\
    1983  & 49.2  & 53.7  & 49.0  & 53.2  & 2002  & 48.6  & 54.1  & 48.6  & 53.8 \\
    1984  & 51.9  & 48.7  & 51.6  & 48.7  & 2003  & 48.3  & 58.1  & 48.1  & 57.5 \\
    1985  & 46.2  & 51.6  & 46.5  & 51.5  & 2004  & 47.5  & 56.8  & 47.3  & 56.5 \\
    1986  & 47.3  & 53.9  & 47.4  & 53.5  & 2005  & 47.5  & 57.1  & 47.6  & 56.8 \\
    1987  & 46.3  & 56.5  & 47.0  & 55.8  & 2006  & 42.8  & 55.4  & 43.0  & 54.9 \\
    1988  & 50.9  & 53.4  & 50.6  & 52.9  & 2007  & 47.2  & 59.2  & 47.2  & 59.0 \\
    1989  & 51.9  & 53.1  & 52.2  & 52.9  & 2008  & 46.4  & 59.8  & 46.5  & 59.4 \\
    1990  & 51.0  & 57.5  & 51.4  & 57.3  & 2009  & 50.8  & 58.2  & 50.9  & 58.0 \\
    1991  & 50.8  & 53.3  & 50.6  & 53.0  & 2010  & 55.5  & 59.2  & 54.3  & 58.8 \\
    1992  & 51.5  & 53.5  & 51.5  & 53.2  & 2011  & 50.4  & 60.9  & 49.7  & 60.4 \\
    1993  & 56.1  & 54.2  & 56.0  & 54.2  & 2012  & 57.3  & 59.6  & 56.4  & 59.3 \\
    1994  & 54.5  & 55.7  & 54.6  & 55.3  & 2013  & 47.7  & 58.6  & 47.8  & 58.8 \\
		\hline
    \end{tabular}\begin{tablenotes}[para,flushleft]
\begin{spacing}{1}
{\footnotesize Notes: coefficients scaled by 100.}
\end{spacing}
\end{tablenotes}
\end{threeparttable}
\end{table}

\begin{figure}[htbp]
\caption{Mean value of $u$ for participants}
\includegraphics[width=16.5cm]{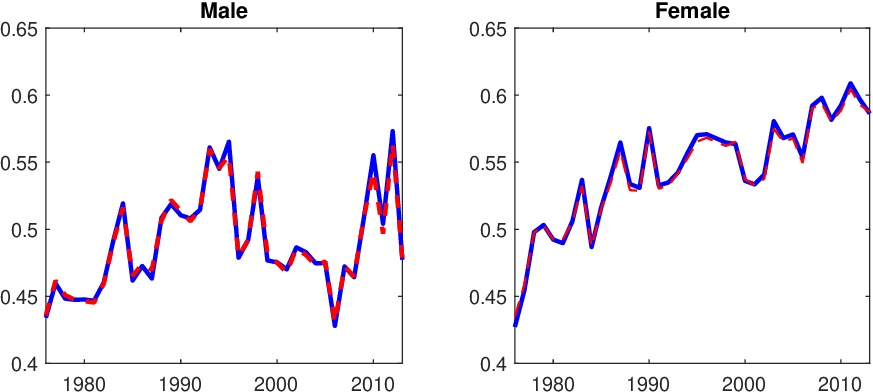}\label{fig:meanu}

{\footnotesize Notes: the solid blue line denotes the estimate with the Frank copula for all participants; the dashed red line denotes the estimate with the Gaussian copula for all participants; coefficients scaled by 100.}
\end{figure}

The average selection difference between male and female workers has consequently become negative, decreasing by about 11.6 percentage points, although it also reflects the oscillation of the estimates for males. The participation components experienced a decrease of about 6 percentage points, roughly half of the difference. Hence, even if the amount of self-selection (copula) had been the same for both genders, the increase of female employment rates contributed to the increase in average unobserved ability, setting a gap relative to male workers. On the other hand, the selection component is less precisely estimated and it displays an unstable evolution. This is a direct consequence of the oscillating behavior of the male copula estimates. Regardless, the long term trend also points at an increase of the gap in favor of employed women. Lastly, because the baseline estimates are homogeneous with respect to the covariates, the endowments component is negligible.\footnote{Note that it is not exactly zero because there is variation in the copula because of differences in the propensity score, deriving themselves from differences in the distribution of $Z$ between genders.}

\begin{table}[htbp]
  \centering
  \caption{Self-selection decomposition (Frank copula)}\label{tab:dug}
		\begin{threeparttable}
    \begin{tabular}{cc@{\,}lc@{\,}lc@{\,}lc@{\,}l|cc@{\,}lc@{\,}lc@{\,}lc@{\,}l}
		\hline
     Year  & \multicolumn{2}{c}{Total} & \multicolumn{2}{c}{EC} & \multicolumn{2}{c}{SC} & \multicolumn{2}{c}{PC} & Year  & \multicolumn{2}{c}{Total} & \multicolumn{2}{c}{EC} & \multicolumn{2}{c}{SC} & \multicolumn{2}{c}{PC} \\
		\hline
    1976  & 0.7   &       & -0.1  & **    & -2.7  &       & 3.5   & ***   & 1995  & -0.5  &       & 0.1   &       & 2.2   &       & -2.7  & *** \\
    1977  & 0.5   &       & -0.1  &       & -1.5  &       & 2.1   & *     & 1996  & -9.2  & ***   & 0.0   &       & -6.5  & ***   & -2.7  & *** \\
    1978  & -5.0  & *     & -0.1  & *     & -4.7  & **    & -0.1  &       & 1997  & -7.6  & **    & 0.0   &       & -5.0  & *     & -2.5  & *** \\
    1979  & -5.6  & **    & -0.1  & **    & -5.1  & **    & -0.4  &       & 1998  & -2.6  &       & 0.0   &       & -0.1  &       & -2.5  & *** \\
    1980  & -4.5  & **    & -0.1  & *     & -4.5  & **    & 0.1   &       & 1999  & -8.7  & ***   & 0.0   &       & -6.2  & **    & -2.5  & *** \\
    1981  & -4.3  &       & 0.0   &       & -4.5  & *     & 0.2   &       & 2000  & -6.1  & *     & 0.0   &       & -4.6  & *     & -1.4  & * \\
    1982  & -4.6  & *     & 0.0   &       & -4.1  & *     & -0.5  &       & 2001  & -6.3  & ***   & 0.0   &       & -5.0  & **    & -1.4  & *** \\
    1983  & -4.5  &       & 0.0   &       & -2.9  &       & -1.6  & *     & 2002  & -5.4  & **    & 0.0   &       & -3.8  & *     & -1.6  & ** \\
    1984  & 3.3   &       & 0.0   &       & 2.9   &       & 0.3   &       & 2003  & -9.8  & ***   & 0.0   &       & -7.0  & ***   & -2.8  & *** \\
    1985  & -5.5  & *     & 0.0   &       & -4.6  & *     & -0.8  &       & 2004  & -9.3  & ***   & 0.0   &       & -7.0  & **    & -2.3  & *** \\
    1986  & -6.7  & ***   & 0.0   &       & -5.0  & **    & -1.7  & ***   & 2005  & -9.6  & ***   & 0.0   &       & -7.2  & **    & -2.4  & *** \\
    1987  & -10.1 & ***   & 0.0   &       & -7.5  & ***   & -2.6  & ***   & 2006  & -12.6 & ***   & 0.0   &       & -10.7 & ***   & -1.9  & *** \\
    1988  & -2.5  &       & 0.0   &       & -1.1  &       & -1.5  & ***   & 2007  & -12.0 & ***   & 0.0   &       & -8.8  & ***   & -3.2  & *** \\
    1989  & -1.2  &       & 0.0   &       & 0.2   &       & -1.3  & **    & 2008  & -13.4 & ***   & 0.0   &       & -10.1 & ***   & -3.2  & *** \\
    1990  & -6.5  & ***   & 0.0   &       & -3.3  &       & -3.2  & ***   & 2009  & -7.4  & **    & 0.0   &       & -4.8  &       & -2.6  & *** \\
    1991  & -2.5  &       & 0.0   &       & -1.0  &       & -1.5  & ***   & 2010  & -3.7  &       & 0.1   &       & -1.5  &       & -2.4  & *** \\
    1992  & -2.0  &       & 0.0   &       & -0.6  &       & -1.4  & ***   & 2011  & -10.5 & **    & 0.0   &       & -7.7  & *     & -2.8  & *** \\
    1993  & 1.9   &       & 0.0   &       & 3.5   &       & -1.6  & ***   & 2012  & -2.3  &       & 0.2   & **    & 0.2   &       & -2.7  & *** \\
    1994  & -1.2  &       & 0.0   &       & 0.9   &       & -2.2  & ***   & 2013  & -10.9 & ***   & 0.0   &       & -8.3  & ***   & -2.5  & *** \\
		\hline
    \end{tabular}\begin{tablenotes}[para,flushleft]
\begin{spacing}{1}
{\footnotesize Notes: Total, EC, SC and PC respectively denote total difference, endowments component, selection component and participation component; coefficients scaled by 100; *, ** and ** respectively denote statistical significance at the 90\%, 95\% and 99\% confidence level.}
\end{spacing}
\end{tablenotes}
\end{threeparttable}
\end{table}

\subsection{Evolution of Labor Earnings}

Next consider the distributions of actual earnings for participants and the entire population by gender. Specifically, I report their means and the value of their 10th, 25th, 50th, 75th and 90th percentiles in Tables~\ref{tab:y1g}-\ref{tab:y2g}. The actual earnings distribution for participants shows a small decrease in mean earnings for male workers and a slightly larger increase for mean female earnings.

\begin{table}[htbp]
  \centering
  \caption{Actual earnings distributions for participants by gender (Frank copula)}\label{tab:y1g}
		\begin{threeparttable}
    \begin{tabular}{ccccccc|cccccc}
		\hline
          & \multicolumn{6}{c}{Male}                      & \multicolumn{6}{c}{Female} \\
    Year  & Mean  & P10   & P25   & P50   & P75   & P90   & Mean  & P10   & P25   & P50   & P75   & P90 \\
		\hline
    1976  & 2.72  & 2.03  & 2.40  & 2.77  & 3.09  & 3.36  & 2.29  & 1.74  & 2.00  & 2.30  & 2.60  & 2.86 \\
    1977  & 2.73  & 2.02  & 2.40  & 2.78  & 3.10  & 3.37  & 2.32  & 1.77  & 2.03  & 2.33  & 2.62  & 2.88 \\
    1978  & 2.74  & 2.03  & 2.40  & 2.78  & 3.12  & 3.38  & 2.31  & 1.76  & 2.02  & 2.33  & 2.63  & 2.89 \\
    1979  & 2.75  & 2.05  & 2.40  & 2.79  & 3.13  & 3.40  & 2.33  & 1.78  & 2.04  & 2.33  & 2.63  & 2.90 \\
    1980  & 2.73  & 2.03  & 2.40  & 2.78  & 3.12  & 3.38  & 2.32  & 1.78  & 2.04  & 2.32  & 2.62  & 2.89 \\
    1981  & 2.68  & 1.97  & 2.34  & 2.73  & 3.07  & 3.34  & 2.29  & 1.74  & 2.00  & 2.29  & 2.59  & 2.87 \\
    1982  & 2.67  & 1.95  & 2.31  & 2.71  & 3.06  & 3.35  & 2.27  & 1.72  & 1.98  & 2.28  & 2.59  & 2.86 \\
    1983  & 2.65  & 1.91  & 2.28  & 2.70  & 3.07  & 3.36  & 2.28  & 1.70  & 1.98  & 2.29  & 2.60  & 2.88 \\
    1984  & 2.65  & 1.90  & 2.28  & 2.70  & 3.07  & 3.38  & 2.30  & 1.70  & 1.99  & 2.31  & 2.63  & 2.91 \\
    1985  & 2.66  & 1.88  & 2.28  & 2.71  & 3.08  & 3.38  & 2.31  & 1.69  & 1.98  & 2.32  & 2.65  & 2.94 \\
    1986  & 2.66  & 1.88  & 2.28  & 2.71  & 3.09  & 3.39  & 2.32  & 1.69  & 1.99  & 2.33  & 2.67  & 2.97 \\
    1987  & 2.68  & 1.89  & 2.29  & 2.73  & 3.11  & 3.42  & 2.35  & 1.70  & 2.01  & 2.36  & 2.71  & 3.02 \\
    1988  & 2.68  & 1.89  & 2.29  & 2.72  & 3.11  & 3.43  & 2.36  & 1.69  & 2.02  & 2.38  & 2.72  & 3.03 \\
    1989  & 2.68  & 1.90  & 2.30  & 2.72  & 3.11  & 3.42  & 2.36  & 1.69  & 2.02  & 2.38  & 2.73  & 3.04 \\
    1990  & 2.67  & 1.89  & 2.29  & 2.71  & 3.10  & 3.43  & 2.37  & 1.69  & 2.03  & 2.39  & 2.75  & 3.07 \\
    1991  & 2.64  & 1.86  & 2.25  & 2.67  & 3.07  & 3.40  & 2.36  & 1.69  & 2.01  & 2.38  & 2.74  & 3.05 \\
    1992  & 2.63  & 1.84  & 2.23  & 2.67  & 3.06  & 3.39  & 2.36  & 1.69  & 2.00  & 2.38  & 2.74  & 3.06 \\
    1993  & 2.62  & 1.83  & 2.22  & 2.66  & 3.06  & 3.40  & 2.37  & 1.69  & 2.01  & 2.39  & 2.75  & 3.07 \\
    1994  & 2.60  & 1.80  & 2.19  & 2.64  & 3.04  & 3.39  & 2.36  & 1.66  & 2.00  & 2.38  & 2.75  & 3.08 \\
    1995  & 2.60  & 1.80  & 2.20  & 2.64  & 3.05  & 3.41  & 2.36  & 1.65  & 1.98  & 2.38  & 2.76  & 3.11 \\
    1996  & 2.61  & 1.80  & 2.19  & 2.63  & 3.04  & 3.40  & 2.36  & 1.65  & 1.99  & 2.37  & 2.76  & 3.10 \\
    1997  & 2.62  & 1.82  & 2.20  & 2.63  & 3.04  & 3.40  & 2.37  & 1.66  & 2.00  & 2.38  & 2.76  & 3.11 \\
    1998  & 2.64  & 1.84  & 2.23  & 2.66  & 3.06  & 3.43  & 2.39  & 1.70  & 2.02  & 2.40  & 2.78  & 3.12 \\
    1999  & 2.67  & 1.87  & 2.26  & 2.68  & 3.08  & 3.46  & 2.41  & 1.71  & 2.04  & 2.43  & 2.81  & 3.16 \\
    2000  & 2.68  & 1.87  & 2.26  & 2.69  & 3.10  & 3.48  & 2.43  & 1.71  & 2.05  & 2.44  & 2.83  & 3.18 \\
    2001  & 2.71  & 1.90  & 2.27  & 2.70  & 3.11  & 3.53  & 2.46  & 1.74  & 2.08  & 2.47  & 2.85  & 3.21 \\
    2002  & 2.70  & 1.89  & 2.28  & 2.70  & 3.11  & 3.53  & 2.48  & 1.75  & 2.09  & 2.48  & 2.86  & 3.23 \\
    2003  & 2.70  & 1.90  & 2.28  & 2.70  & 3.12  & 3.53  & 2.49  & 1.76  & 2.10  & 2.49  & 2.88  & 3.25 \\
    2004  & 2.69  & 1.88  & 2.26  & 2.69  & 3.12  & 3.52  & 2.49  & 1.76  & 2.11  & 2.50  & 2.89  & 3.26 \\
    2005  & 2.69  & 1.87  & 2.25  & 2.68  & 3.11  & 3.53  & 2.48  & 1.74  & 2.10  & 2.49  & 2.88  & 3.25 \\
    2006  & 2.69  & 1.88  & 2.24  & 2.67  & 3.10  & 3.52  & 2.47  & 1.72  & 2.08  & 2.47  & 2.87  & 3.25 \\
    2007  & 2.68  & 1.87  & 2.24  & 2.67  & 3.10  & 3.52  & 2.48  & 1.72  & 2.08  & 2.48  & 2.89  & 3.28 \\
    2008  & 2.68  & 1.88  & 2.25  & 2.67  & 3.11  & 3.52  & 2.49  & 1.75  & 2.10  & 2.50  & 2.89  & 3.28 \\
    2009  & 2.67  & 1.84  & 2.23  & 2.66  & 3.10  & 3.53  & 2.47  & 1.72  & 2.08  & 2.47  & 2.88  & 3.26 \\
    2010  & 2.68  & 1.86  & 2.25  & 2.68  & 3.12  & 3.54  & 2.48  & 1.73  & 2.09  & 2.49  & 2.90  & 3.29 \\
    2011  & 2.67  & 1.84  & 2.22  & 2.67  & 3.11  & 3.53  & 2.47  & 1.72  & 2.07  & 2.49  & 2.89  & 3.29 \\
    2012  & 2.65  & 1.81  & 2.20  & 2.65  & 3.10  & 3.53  & 2.46  & 1.70  & 2.06  & 2.47  & 2.88  & 3.28 \\
    2013  & 2.65  & 1.80  & 2.19  & 2.64  & 3.09  & 3.52  & 2.46  & 1.70  & 2.06  & 2.47  & 2.88  & 3.28 \\
		\hline
    \end{tabular}
\end{threeparttable}
\end{table}

This catch-up, however, masks an increase in the inequality within each gender, as interquantile ranges (IQR) increased both for male and female distributions: the 90-10 IQR increased from 132 to 172 percentage points for men, and from 112 to 158 for women; the 75-25 IQR similarly increased from 69 to 90 and 60 to 82 percentage points for male and female workers, respectively. Still, the evolution has been quite heterogeneous across the distributions of earnings. For male workers earnings followed a long term decrease for percentiles below the 75th, and a steady increase for those at the top of the distribution; for female workers there has been a gain for those above the 25th percentile, more pronounced at the top. Despite this catch-up, there is still a gap in favor of men at all quantiles.

By construction, earnings at any given quantile are smaller on the distribution for the full population than on that for participants, resulting in a smaller mean for both genders. Nonetheless, including non-participants increases the earnings gap. Following the increase in female labor participation, this distribution has steadily increased for females, reducing the gap relative to males by a bigger fraction than for the distribution of participants. On the other hand, mean earnings for males have oscillated across time following the changes in participation and average earnings for participants. Note that the fall in the male participation rate in the last years of the sample has been starker than that of females, prompting a decrease in mean earnings for both genders, along with a decrease of the gap.

\begin{table}[htbp]
  \centering
  \caption{Actual earnings distributions for the full population by gender (Frank copula)}\label{tab:y2g}
		\begin{threeparttable}
    \begin{tabular}{ccccccc|cccccc}
		\hline
          & \multicolumn{6}{c}{Male}                      & \multicolumn{6}{c}{Female} \\
    Year  & Mean  & P10   & P25   & P50   & P75   & P90   & Mean  & P10   & P25   & P50   & P75   & P90 \\
		\hline
    1976  & 1.89  & 0.00  & 0.00  & 2.50  & 2.97  & 3.28  & 0.81  & 0.00  & 0.00  & 0.00  & 2.09  & 2.60 \\
    1977  & 1.90  & 0.00  & 0.00  & 2.51  & 2.99  & 3.30  & 0.84  & 0.00  & 0.00  & 0.00  & 2.13  & 2.62 \\
    1978  & 1.91  & 0.00  & 0.00  & 2.50  & 3.00  & 3.31  & 0.87  & 0.00  & 0.00  & 0.00  & 2.16  & 2.64 \\
    1979  & 1.94  & 0.00  & 0.00  & 2.52  & 3.01  & 3.32  & 0.92  & 0.00  & 0.00  & 0.00  & 2.20  & 2.65 \\
    1980  & 1.91  & 0.00  & 0.00  & 2.50  & 3.00  & 3.30  & 0.95  & 0.00  & 0.00  & 0.00  & 2.23  & 2.66 \\
    1981  & 1.84  & 0.00  & 0.00  & 2.42  & 2.94  & 3.26  & 0.93  & 0.00  & 0.00  & 0.00  & 2.19  & 2.63 \\
    1982  & 1.76  & 0.00  & 0.00  & 2.34  & 2.91  & 3.25  & 0.89  & 0.00  & 0.00  & 0.00  & 2.15  & 2.61 \\
    1983  & 1.67  & 0.00  & 0.00  & 2.24  & 2.88  & 3.25  & 0.87  & 0.00  & 0.00  & 0.00  & 2.14  & 2.63 \\
    1984  & 1.66  & 0.00  & 0.00  & 2.23  & 2.89  & 3.26  & 0.91  & 0.00  & 0.00  & 0.00  & 2.19  & 2.67 \\
    1985  & 1.71  & 0.00  & 0.00  & 2.27  & 2.90  & 3.27  & 0.95  & 0.00  & 0.00  & 0.00  & 2.23  & 2.71 \\
    1986  & 1.74  & 0.00  & 0.00  & 2.30  & 2.92  & 3.28  & 0.98  & 0.00  & 0.00  & 0.00  & 2.26  & 2.74 \\
    1987  & 1.76  & 0.00  & 0.00  & 2.32  & 2.94  & 3.31  & 1.01  & 0.00  & 0.00  & 0.00  & 2.30  & 2.79 \\
    1988  & 1.81  & 0.00  & 0.00  & 2.35  & 2.95  & 3.32  & 1.10  & 0.00  & 0.00  & 0.00  & 2.37  & 2.82 \\
    1989  & 1.82  & 0.00  & 0.00  & 2.36  & 2.95  & 3.32  & 1.12  & 0.00  & 0.00  & 0.00  & 2.39  & 2.84 \\
    1990  & 1.85  & 0.00  & 0.00  & 2.38  & 2.94  & 3.32  & 1.14  & 0.00  & 0.00  & 0.00  & 2.41  & 2.86 \\
    1991  & 1.81  & 0.00  & 0.00  & 2.33  & 2.91  & 3.29  & 1.14  & 0.00  & 0.00  & 0.00  & 2.40  & 2.85 \\
    1992  & 1.77  & 0.00  & 0.00  & 2.29  & 2.89  & 3.28  & 1.15  & 0.00  & 0.00  & 0.00  & 2.40  & 2.86 \\
    1993  & 1.74  & 0.00  & 0.00  & 2.26  & 2.88  & 3.27  & 1.15  & 0.00  & 0.00  & 0.00  & 2.42  & 2.87 \\
    1994  & 1.73  & 0.00  & 0.00  & 2.23  & 2.86  & 3.27  & 1.14  & 0.00  & 0.00  & 0.00  & 2.41  & 2.87 \\
    1995  & 1.76  & 0.00  & 0.00  & 2.26  & 2.87  & 3.29  & 1.16  & 0.00  & 0.00  & 0.00  & 2.41  & 2.89 \\
    1996  & 1.78  & 0.00  & 0.00  & 2.27  & 2.87  & 3.28  & 1.18  & 0.00  & 0.00  & 0.00  & 2.42  & 2.89 \\
    1997  & 1.78  & 0.00  & 0.00  & 2.28  & 2.87  & 3.27  & 1.20  & 0.00  & 0.00  & 0.00  & 2.44  & 2.90 \\
    1998  & 1.81  & 0.00  & 0.00  & 2.31  & 2.89  & 3.30  & 1.22  & 0.00  & 0.00  & 1.16  & 2.46  & 2.92 \\
    1999  & 1.86  & 0.00  & 0.00  & 2.36  & 2.93  & 3.33  & 1.25  & 0.00  & 0.00  & 1.38  & 2.50  & 2.95 \\
    2000  & 1.87  & 0.00  & 0.00  & 2.37  & 2.94  & 3.36  & 1.28  & 0.00  & 0.00  & 1.52  & 2.53  & 2.98 \\
    2001  & 1.90  & 0.00  & 0.00  & 2.38  & 2.95  & 3.39  & 1.31  & 0.00  & 0.00  & 1.57  & 2.55  & 3.00 \\
    2002  & 1.87  & 0.00  & 0.00  & 2.37  & 2.95  & 3.40  & 1.30  & 0.00  & 0.00  & 1.48  & 2.55  & 3.01 \\
    2003  & 1.82  & 0.00  & 0.00  & 2.33  & 2.94  & 3.39  & 1.27  & 0.00  & 0.00  & 1.22  & 2.55  & 3.03 \\
    2004  & 1.79  & 0.00  & 0.00  & 2.30  & 2.93  & 3.38  & 1.25  & 0.00  & 0.00  & 0.00  & 2.55  & 3.02 \\
    2005  & 1.79  & 0.00  & 0.00  & 2.29  & 2.92  & 3.38  & 1.24  & 0.00  & 0.00  & 0.00  & 2.54  & 3.03 \\
    2006  & 1.79  & 0.00  & 0.00  & 2.28  & 2.91  & 3.37  & 1.26  & 0.00  & 0.00  & 0.83  & 2.54  & 3.02 \\
    2007  & 1.80  & 0.00  & 0.00  & 2.29  & 2.92  & 3.38  & 1.27  & 0.00  & 0.00  & 1.19  & 2.55  & 3.04 \\
    2008  & 1.80  & 0.00  & 0.00  & 2.29  & 2.92  & 3.38  & 1.29  & 0.00  & 0.00  & 1.37  & 2.57  & 3.05 \\
    2009  & 1.74  & 0.00  & 0.00  & 2.25  & 2.90  & 3.38  & 1.25  & 0.00  & 0.00  & 0.00  & 2.54  & 3.03 \\
    2010  & 1.64  & 0.00  & 0.00  & 2.15  & 2.89  & 3.37  & 1.21  & 0.00  & 0.00  & 0.00  & 2.54  & 3.05 \\
    2011  & 1.62  & 0.00  & 0.00  & 2.11  & 2.87  & 3.36  & 1.19  & 0.00  & 0.00  & 0.00  & 2.53  & 3.04 \\
    2012  & 1.63  & 0.00  & 0.00  & 2.12  & 2.87  & 3.37  & 1.18  & 0.00  & 0.00  & 0.00  & 2.51  & 3.02 \\
    2013  & 1.65  & 0.00  & 0.00  & 2.13  & 2.86  & 3.36  & 1.19  & 0.00  & 0.00  & 0.00  & 2.51  & 3.03 \\
		\hline
    \end{tabular}
\end{threeparttable}
\end{table}

\subsection{Main Decompositions}

Tables~\ref{tab:dy1g}-\ref{tab:dy2g} report the decompositions of the mean earnings gap for the two populations considered.\footnote{The estimates of the decompositions of mean earnings with the Heckman 2-stage estimator are very similar to those found with the QRS estimator. These are shown in Figures~\ref{fig:dech1}-\ref{fig:dech2} in Appendix~\ref{app:tabs}.} The mean gap for participants has more than halved during the period: from over 40\% gender gap for workers, it was equal to 19\%. Out of the four components, the largest one in every period has been the coefficients component. Its size displays some yearly variation driven by the fluctuation of the slope and copula parameters for male workers.\footnote{The erratic behavior of this and the selection components could be linked to the strength of the instrument, which is weaker for male workers. To see this, note that the autocorrelation of the total gap is 0.99. Out of the four components, it is also large for the endowments (0.99) and participation components (0.79), whereas for the coefficients and selection components are significantly smaller: 0.47 and 0.37, respectively. However, the autocorrelation of the sum of these two components is equal to 0.93. Therefore, it is plausible that the lack of strength of the instrument may be responsible for the large variations in the size of these two components even in consecutive years.} If one considers the moving average estimates of the copula parameters, it can be seen a decrease of this component until the mid-nineties, increasing afterwards.\footnote{See Figures~\ref{fig:meanyma1}-\ref{fig:meanyma2} in Appendix~\ref{app:tabs}.} Analogously, the selection component also displays an erratic behavior, which is again a consequence of the estimates of the copula for males. The sign of this component is often negative, and when it is positive it is not significant at the 95\% confidence level. Moreover, it displays a slightly downward trend, thus increasingly helping in the reduction of the mean gap (by 14 percentage points in the last year of analysis).

\begin{table}[htbp]
  \centering
  \caption{Mean decomposition, actual earnings for participants (Frank copula)}\label{tab:dy1g}
		\begin{threeparttable}
    \begin{tabular}{cr@{\,}lr@{\,}lr@{\,}lr@{\,}lr@{\,}l}
		\hline
    Year  & \multicolumn{2}{c}{Total} & \multicolumn{2}{c}{EC} & \multicolumn{2}{c}{CC} & \multicolumn{2}{c}{SC} & \multicolumn{2}{c}{PC} \\
		\hline
    1976  & 0.43  & ***   & 0.01  & ***   & 0.41  & ***   & -0.03 &       & 0.05  & *** \\
    1977  & 0.41  & ***   & 0.01  & ***   & 0.39  & ***   & -0.02 &       & 0.03  & * \\
    1978  & 0.42  & ***   & 0.01  & ***   & 0.48  & ***   & -0.06 & **    & 0.00  &  \\
    1979  & 0.42  & ***   & 0.01  & ***   & 0.48  & ***   & -0.07 & **    & -0.01 &  \\
    1980  & 0.41  & ***   & 0.01  & **    & 0.47  & ***   & -0.06 & **    & 0.00  &  \\
    1981  & 0.40  & ***   & 0.00  &       & 0.45  & ***   & -0.06 & *     & 0.00  &  \\
    1982  & 0.40  & ***   & 0.01  & **    & 0.46  & ***   & -0.05 & *     & -0.01 &  \\
    1983  & 0.38  & ***   & 0.01  & **    & 0.43  & ***   & -0.04 &       & -0.02 & * \\
    1984  & 0.35  & ***   & 0.01  & **    & 0.30  & ***   & 0.04  &       & 0.00  &  \\
    1985  & 0.35  & ***   & 0.01  & ***   & 0.42  & ***   & -0.07 & *     & -0.01 &  \\
    1986  & 0.34  & ***   & 0.01  & ***   & 0.43  & ***   & -0.07 & **    & -0.02 & *** \\
    1987  & 0.34  & ***   & 0.01  & **    & 0.48  & ***   & -0.11 & ***   & -0.04 & *** \\
    1988  & 0.32  & ***   & 0.01  & ***   & 0.35  & ***   & -0.02 &       & -0.02 & *** \\
    1989  & 0.32  & ***   & 0.01  & ***   & 0.33  & ***   & 0.00  &       & -0.02 & ** \\
    1990  & 0.30  & ***   & 0.01  & ***   & 0.39  & ***   & -0.05 &       & -0.05 & *** \\
    1991  & 0.28  & ***   & 0.00  &       & 0.31  & ***   & -0.02 &       & -0.02 & ** \\
    1992  & 0.27  & ***   & 0.00  & *     & 0.29  & ***   & -0.01 &       & -0.02 & *** \\
    1993  & 0.25  & ***   & 0.00  &       & 0.22  & ***   & 0.05  &       & -0.03 & *** \\
    1994  & 0.23  & ***   & 0.00  &       & 0.25  & ***   & 0.01  &       & -0.03 & *** \\
    1995  & 0.24  & ***   & 0.00  &       & 0.25  & ***   & 0.03  &       & -0.04 & *** \\
    1996  & 0.25  & ***   & 0.00  &       & 0.41  & ***   & -0.10 & ***   & -0.04 & *** \\
    1997  & 0.25  & ***   & -0.01 & **    & 0.38  & ***   & -0.08 & *     & -0.04 & *** \\
    1998  & 0.24  & ***   & -0.01 & ***   & 0.30  & ***   & 0.00  &       & -0.04 & *** \\
    1999  & 0.26  & ***   & -0.01 & **    & 0.40  & ***   & -0.10 & **    & -0.04 & *** \\
    2000  & 0.25  & ***   & -0.01 & **    & 0.35  & ***   & -0.07 & *     & -0.02 &  \\
    2001  & 0.24  & ***   & -0.01 & ***   & 0.36  & ***   & -0.08 & **    & -0.02 & ** \\
    2002  & 0.23  & ***   & -0.01 & ***   & 0.33  & ***   & -0.06 & *     & -0.03 & ** \\
    2003  & 0.21  & ***   & -0.02 & ***   & 0.39  & ***   & -0.12 & ***   & -0.05 & *** \\
    2004  & 0.20  & ***   & -0.02 & ***   & 0.38  & ***   & -0.12 & **    & -0.04 & *** \\
    2005  & 0.21  & ***   & -0.02 & ***   & 0.40  & ***   & -0.12 & **    & -0.04 & *** \\
    2006  & 0.22  & ***   & -0.03 & ***   & 0.45  & ***   & -0.18 & ***   & -0.03 & *** \\
    2007  & 0.20  & ***   & -0.03 & ***   & 0.45  & ***   & -0.16 & ***   & -0.06 & *** \\
    2008  & 0.20  & ***   & -0.03 & ***   & 0.46  & ***   & -0.18 & ***   & -0.06 & *** \\
    2009  & 0.20  & ***   & -0.04 & ***   & 0.37  & ***   & -0.08 &       & -0.04 & *** \\
    2010  & 0.20  & ***   & -0.04 & ***   & 0.30  & ***   & -0.03 &       & -0.04 & *** \\
    2011  & 0.20  & ***   & -0.04 & ***   & 0.42  & ***   & -0.14 & *     & -0.05 & *** \\
    2012  & 0.19  & ***   & -0.03 & ***   & 0.27  & ***   & 0.00  &       & -0.05 & *** \\
    2013  & 0.19  & ***   & -0.04 & ***   & 0.41  & ***   & -0.14 & ***   & -0.04 & *** \\
		\hline
    \end{tabular}\begin{tablenotes}[para,flushleft]
\begin{spacing}{1}
{\footnotesize Notes: Total, EC, CC, SC and PC respectively denote total difference, endowments component, coefficients component, selection component and participation component; *, ** and ** respectively denote statistical significance at the 90\%, 95\% and 99\% confidence level.}
\end{spacing}
\end{tablenotes}
\end{threeparttable}
\end{table}

The dynamics of the remaining two components has been more stable: they were initially positive, and they eventually became negative, therefore reducing the mean gender gap. Moreover, the magnitude of these two components has been more modest than that of the other two: the endowments components changed from 1 to -4 percentage points, whereas the participation component experienced a more pronounced fall (from 5 to -4 percentage points).

One way to highlight the importance of accounting for selection is to compare the previous decomposition to the standard Oaxaca-Blinder decomposition. This is shown in Figure~\ref{fig:deco1}. Both components of the latter display a different evolution over time. Specifically, the endowments component retains the downward trend, but its magnitude is much larger, whereas the coefficients component displays a steady downward trend. Hence, the latter appears to be a contributor to the decrease in the gender wage gap, in contrast with the analysis that accounts for sample selection.

\begin{figure}[htbp]
\caption{Actual earnings decomposition for participants, OLS}
\includegraphics[width=16.5cm]{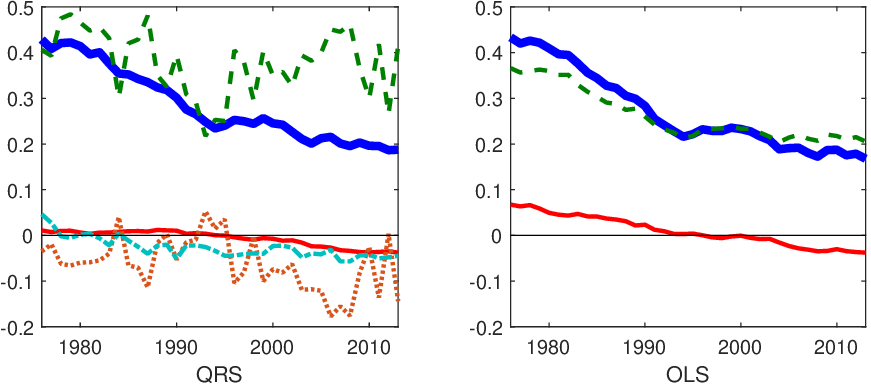}\label{fig:deco1}

{\footnotesize Notes: QRS and OLS stand for Quantile Regression with Selection and Ordinary Least Squares estimators; the solid thick blue line denotes the total gap between male and female workers; the solid thin red line denotes the endowments component; the dashed green line denotes the coefficients component; the dotted orange line denotes the selection component; the dashed-dotted cyan line denotes the participation component.}
\end{figure}

The gap for the entire population has followed a similar trend, steadily decreasing to less than a half of the gap in 1976. Specifically, from a gap of over 100\%, it fell to less than 50\%. However, its magnitude has always been larger, owing to the gender participation gap. Indeed, the participation component constitutes the lion share of the gap, and its reduction has been responsible for the majority of the reduction of the gap: from an initial 81 percentage points, it fell to 33. The remaining three components display a similar behavior to the one found in the decomposition of actual earnings for participants, although their size is slightly scaled down.

\begin{table}[htbp]
  \centering
  \caption{Mean decomposition, actual earnings for the full population (Frank copula)}\label{tab:dy2g}
		\begin{threeparttable}
    \begin{tabular}{cr@{\,}lr@{\,}lr@{\,}lr@{\,}lr@{\,}l}
		\hline
    Year  & \multicolumn{2}{c}{Total} & \multicolumn{2}{c}{EC} & \multicolumn{2}{c}{CC} & \multicolumn{2}{c}{SC} & \multicolumn{2}{c}{PC} \\
		\hline
    1976  & 1.07  & ***   & 0.00  &       & 0.29  & ***   & -0.02 &       & 0.81  & *** \\
    1977  & 1.06  & ***   & -0.01 &       & 0.28  & ***   & -0.01 &       & 0.79  & *** \\
    1978  & 1.03  & ***   & 0.00  &       & 0.34  & ***   & -0.04 & **    & 0.73  & *** \\
    1979  & 1.03  & ***   & 0.00  &       & 0.35  & ***   & -0.04 & **    & 0.72  & *** \\
    1980  & 0.96  & ***   & 0.00  &       & 0.33  & ***   & -0.04 & **    & 0.67  & *** \\
    1981  & 0.91  & ***   & 0.01  &       & 0.31  & ***   & -0.03 & *     & 0.63  & *** \\
    1982  & 0.87  & ***   & 0.01  &       & 0.30  & ***   & -0.03 & *     & 0.59  & *** \\
    1983  & 0.79  & ***   & 0.01  & *     & 0.28  & ***   & -0.02 &       & 0.53  & *** \\
    1984  & 0.75  & ***   & 0.01  & **    & 0.20  & ***   & 0.02  &       & 0.52  & *** \\
    1985  & 0.75  & ***   & 0.01  & **    & 0.27  & ***   & -0.04 & *     & 0.50  & *** \\
    1986  & 0.76  & ***   & 0.01  & ***   & 0.28  & ***   & -0.04 & **    & 0.50  & *** \\
    1987  & 0.75  & ***   & 0.02  & ***   & 0.31  & ***   & -0.06 & ***   & 0.48  & *** \\
    1988  & 0.71  & ***   & 0.02  & ***   & 0.24  & ***   & -0.01 &       & 0.46  & *** \\
    1989  & 0.69  & ***   & 0.02  & ***   & 0.22  & ***   & 0.00  &       & 0.45  & *** \\
    1990  & 0.71  & ***   & 0.01  & **    & 0.27  & ***   & -0.03 &       & 0.46  & *** \\
    1991  & 0.67  & ***   & 0.00  &       & 0.22  & ***   & -0.01 &       & 0.46  & *** \\
    1992  & 0.62  & ***   & 0.00  &       & 0.20  & ***   & -0.01 &       & 0.43  & *** \\
    1993  & 0.59  & ***   & 0.00  &       & 0.15  & ***   & 0.03  &       & 0.41  & *** \\
    1994  & 0.59  & ***   & 0.00  &       & 0.17  & ***   & 0.01  &       & 0.41  & *** \\
    1995  & 0.60  & ***   & -0.01 &       & 0.18  & ***   & 0.02  &       & 0.42  & *** \\
    1996  & 0.60  & ***   & -0.01 &       & 0.27  & ***   & -0.06 & ***   & 0.40  & *** \\
    1997  & 0.59  & ***   & -0.01 &       & 0.26  & ***   & -0.05 & *     & 0.38  & *** \\
    1998  & 0.59  & ***   & -0.01 & **    & 0.21  & ***   & 0.00  &       & 0.39  & *** \\
    1999  & 0.61  & ***   & 0.00  &       & 0.28  & ***   & -0.06 & **    & 0.39  & *** \\
    2000  & 0.58  & ***   & 0.00  &       & 0.24  & ***   & -0.04 & *     & 0.39  & *** \\
    2001  & 0.59  & ***   & -0.01 & *     & 0.25  & ***   & -0.05 & **    & 0.39  & *** \\
    2002  & 0.58  & ***   & -0.01 & ***   & 0.23  & ***   & -0.04 & *     & 0.40  & *** \\
    2003  & 0.55  & ***   & -0.01 & ***   & 0.26  & ***   & -0.07 & ***   & 0.37  & *** \\
    2004  & 0.53  & ***   & -0.02 & ***   & 0.25  & ***   & -0.07 & **    & 0.37  & *** \\
    2005  & 0.54  & ***   & -0.02 & ***   & 0.26  & ***   & -0.07 & **    & 0.37  & *** \\
    2006  & 0.53  & ***   & -0.02 & ***   & 0.29  & ***   & -0.10 & ***   & 0.36  & *** \\
    2007  & 0.53  & ***   & -0.03 & ***   & 0.29  & ***   & -0.09 & ***   & 0.36  & *** \\
    2008  & 0.51  & ***   & -0.03 & ***   & 0.30  & ***   & -0.10 & ***   & 0.34  & *** \\
    2009  & 0.49  & ***   & -0.04 & ***   & 0.24  & ***   & -0.05 &       & 0.34  & *** \\
    2010  & 0.43  & ***   & -0.04 & ***   & 0.19  & ***   & -0.01 &       & 0.30  & *** \\
    2011  & 0.43  & ***   & -0.04 & ***   & 0.25  & ***   & -0.07 & *     & 0.30  & *** \\
    2012  & 0.45  & ***   & -0.04 & ***   & 0.17  & ***   & 0.00  &       & 0.32  & *** \\
    2013  & 0.47  & ***   & -0.05 & ***   & 0.25  & ***   & -0.07 & ***   & 0.33  & *** \\
		\hline
    \end{tabular}\begin{tablenotes}[para,flushleft]
\begin{spacing}{1}
{\footnotesize Notes: Total, EC, CC, SC and PC respectively denote total difference, endowments component, coefficients component, selection component and participation component; *, ** and ** respectively denote statistical significance at the 90\%, 95\% and 99\% confidence level.}
\end{spacing}
\end{tablenotes}
\end{threeparttable}
\end{table}

The same decompositions are performed for the unconditional distributions. I present the estimates for a number of years (1976, 1984, 1992, 2000, 2007, 2013) in Figures~\ref{fig:dec1}-\ref{fig:dec2}. Additionally, I report the estimates for several quantiles in Tables~\ref{tab:dy1p10g}-\ref{tab:dy2p90g} in Appendix~\ref{app:tabs}.

Figure~\ref{fig:dec1} shows the evolution of the gap for the entire quantile process. Several changes have taken place. First, the gap increases monotonically with the quantiles of the distribution in every year, with the exception of the extreme top quantiles. Second, there has been a generalized reduction at all quantiles and, the decrease in the gap in absolute value has also been larger for higher quantiles. Namely, the gap at the 10th percentile fell from 29\% to 10\%, whereas the gap at the 90th percentile decreased from 49\% to 24\%. As it was the case for the mean decomposition, the coefficients component has been the largest one for almost every quantile and every year. In contrast, the selection component has been relatively flat across quantiles, although it displayed great variation across years, switching sign several times.

\begin{figure}[htbp]
\caption{Unconditional quantiles decompositions, actual earnings for participants (Frank copula)}
\includegraphics[width=16.5cm]{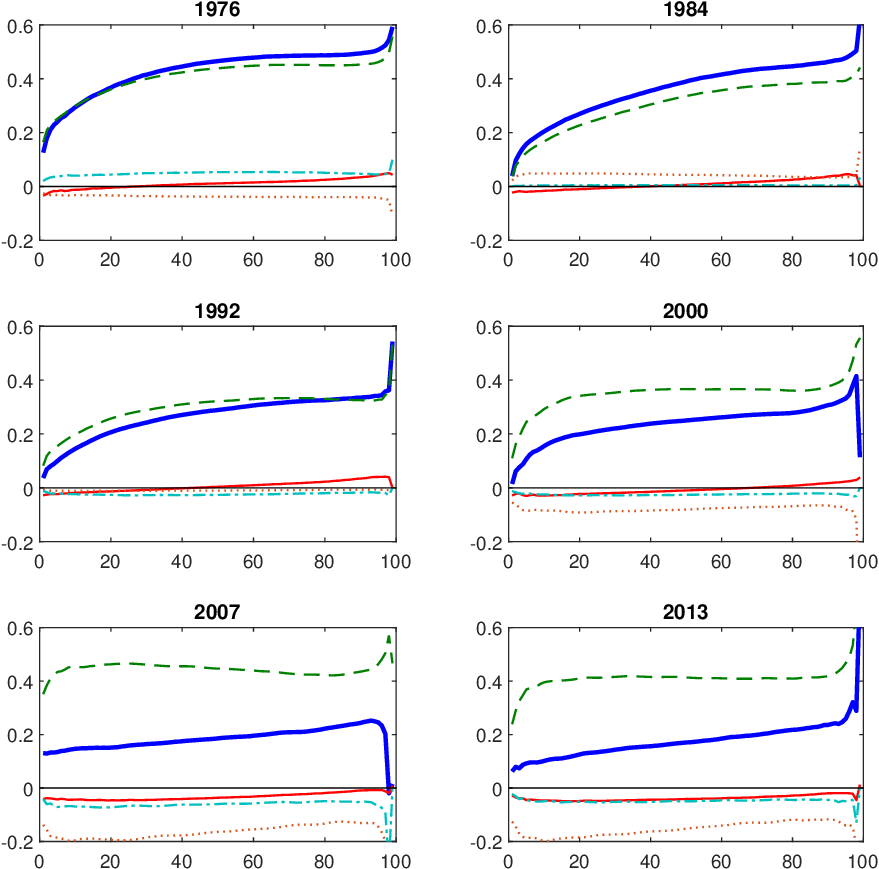}\label{fig:dec1}

{\footnotesize Notes: the solid thick blue line denotes the total gap between male and female workers; the solid thin red line denotes the endowments component; the dashed thin green line denotes the coefficients component; the dotted thin orange line denotes the selection component; the dashed-dotted thin cyan line denotes the participation component.}
\end{figure}

Regarding the participation and endowments components, their behavior is quite similar: their magnitude is smaller than that of the other two components, they were initially positive for the majority of the distribution and had a mild upward slope, and they have progressively flattened out, becoming negative and therefore reducing the gender gap. Moreover, their magnitude is similar to that of the decomposition of the mean.

The distributional gap for the entire population has an unconventional shape, as it displays a thick spike for a large part of the distribution. Its width equals the difference in the participation rates between men and women, and its height equals earnings of male workers from the left tail of their gender distribution. Therefore, the width of this spike has progressively diminished over time with the reduction of the participation gap. However, it still remains the main factor of difference between the two distributions. Second, since the fraction of non-participants is positive for both men and women, the lower tail of the gap equals zero, as workers of both genders on that tail do not have any labor earnings.

Additionally, the participation component has a decreasing shape after the end of the spike, reflecting a shifting in the distribution between male and female worker. To see this, denote by $\tau_{f}$ the quantile at which women earnings becomes positive, \textit{i.e.}, $\tau_{f}=1-\mathbb{E}\left[\pi_{f}\left(Z\right)\right]$. Men above $\tau_{f}$ represent those above a certain level of earnings, which would be equivalent to comparing the distribution of earnings of employed women at a given quantile to the distribution of employed men of a higher quantile.

The other three components have a similar shape to the one found in the decomposition for participants, although their size is somewhat smaller. Consequently, the coefficients component is the second largest one, remaining an important determinant of the gap. Finally, the endowments component and selection components are flatter, relatively small, and they have become negative over time reducing the gap at all quantiles.

\begin{figure}[htbp]
\caption{Unconditional quantiles decompositions, actual earnings for participants (Gaussian copula)}
\includegraphics[width=16.5cm]{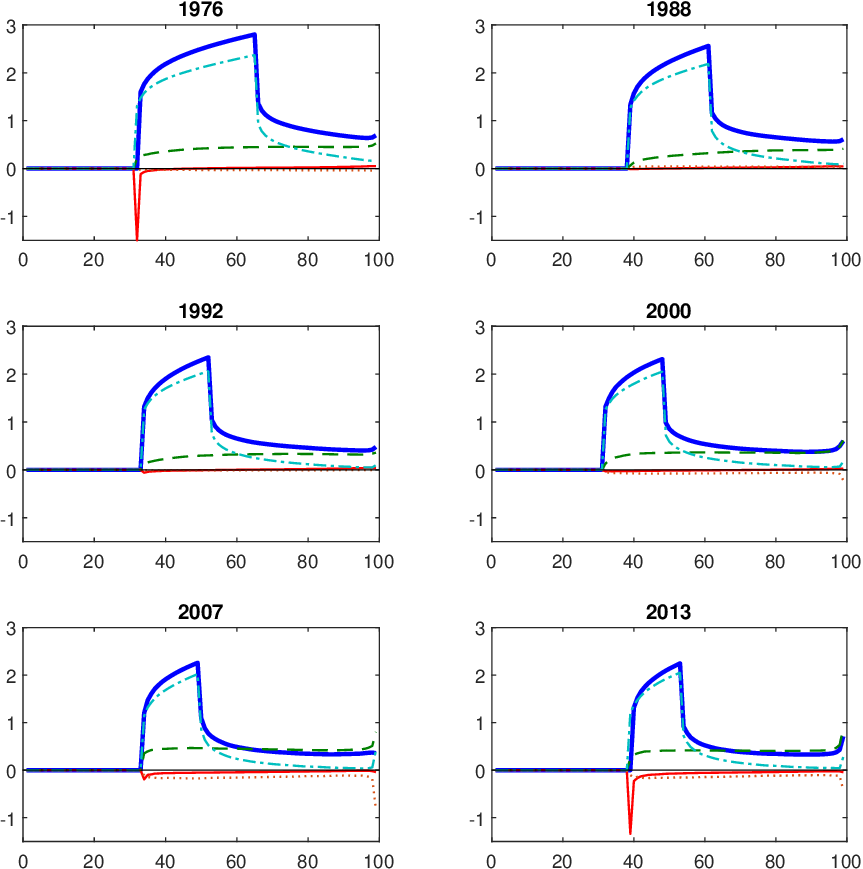}\label{fig:dec2}

{\footnotesize Notes: the solid thick blue line denotes the total gap between male and female workers; the solid thin red line denotes the endowments component; the dashed thin green line denotes the coefficients component; the dotted thin orange line denotes the selection component; the dashed-dotted thin cyan line denotes the participation component.}
\end{figure}

These results are related to the findings in \cite{Olivetti2008}. In their cross-country comparison, they found that countries with the highest participation gaps tended to have lower wage gaps for participants. However, their estimated gaps accounting for self-selection showed that gaps increased much more in countries with larger participation gaps. This is captured by the participation component in the decomposition of the wage gap: for a given level of self-selection, narrowing the participation gap would make the participation component shrink to zero. However, this does not eliminate the impact of self-selection on unobservables, as differences in the amount of self-selection could be due to differences in the copula. Indeed, the estimates in this paper showed how, as participation gaps decreased over the considered period, the intensity of self-selection increased much more for women. Therefore, a more comprehensive cross-country analysis could account for this factor to explain how differences in self-selection intensity have determined different wage gaps in different countries.

\subsection{Heterogeneous copulas}\label{sec:hetcop}

One potentially strong assumption regards the fact that the baseline copulas used in the estimation are homogeneous across the covariates. This limits any selection differences across some of these characteristics to the propensity score channel. To address this issue, I repeat the estimation separately for three different categories: race (white vs non-white), education level (college graduates vs. less than college) and marital status (married vs unmarried).

Some of the estimates are slightly sensitive to the heterogeneous copulas. The mean value of unobserved ability is most similar when one obtains the estimates by race (Figure~\ref{fig:meanu2}, top panel). However, it has been in general larger for white males relative to the baseline estimates, although the estimates are more volatile. In contrast, the estimates for non-white males show a smaller value than the baseline for almost the entire period. The estimates for white females show a more stable evolution of the average value of unobserved ability, whereas for non-white female workers, the evolution has been positive, very closely to the baseline estimates.

\begin{figure}[htbp]
\caption{Mean value of $u$ for participants}
\includegraphics[width=16.5cm]{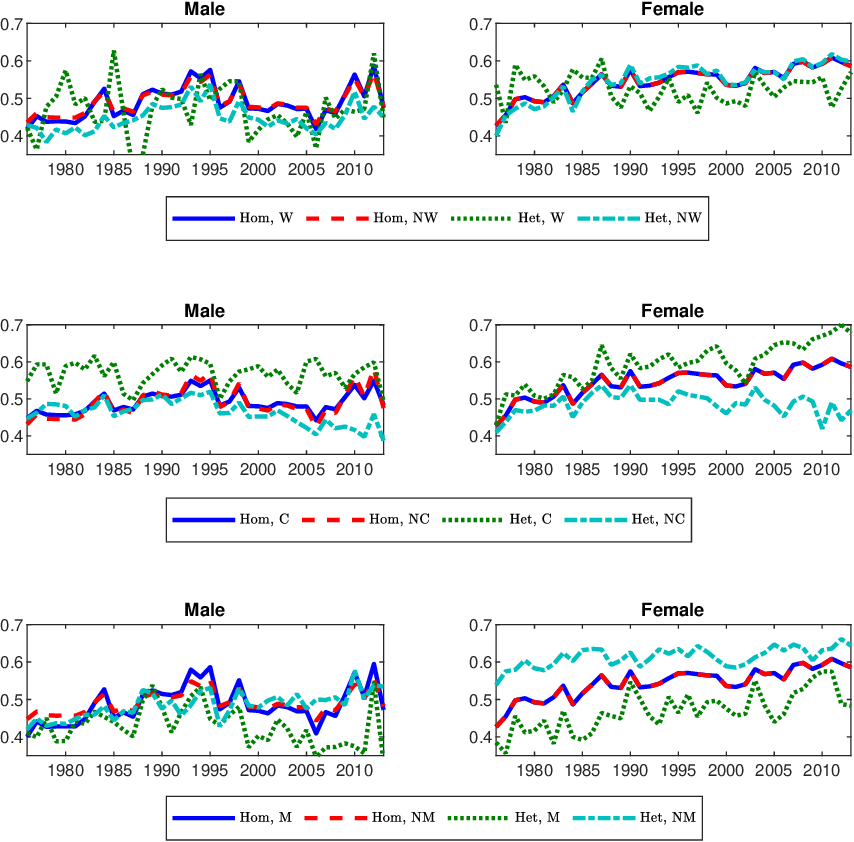}\label{fig:meanu2}

{\footnotesize Notes: Hom, Het, W, NW, C, NC, M and NM respectively denote homogeneous copula, heterogeneous copula, white, non-white, college graduates, less than college, married and unmarried.}
\end{figure}

A more evident difference relative to the baseline case appears in the estimates split by education level (Figure~\ref{fig:meanu2}, central panel), showing a great divide in the average level of unobserved ability between those with a college degree and those without. For the first group, this level has been higher and relatively stable, whereas for the latter it has been lower and decreasing steadily since the mid-nineties. The same divide can be observed for female workers although the average level of unobserved ability steadily increased for college-educated women, and has remained stable for lower-educated female workers.

Finally, if one allows the copula to be different for married and unmarried workers (Figure~\ref{fig:meanu2}, bottom panel), we find the largest differences for women: in particular, unmarried female workers tend to have a much higher level of unobserved ability, although the gap with married women has decreased in the last two decades. In contrast, the estimates for men are much similar to the baseline ones. The main difference corresponds to the period beginning in the late nineties, in which the average level of unobserved ability for married male workers is smaller.

Despite these differences in the amount of self-selection, the mean earnings gap for participants remains largely unaltered in these specifications (Figure~\ref{fig:dY1_het}). However, the decompositions do vary slightly. The most noticeable differences relative to the baseline estimates arise in the coefficients and selection components. In particular, the coefficients components is generally larger for the estimates with heterogeneous copulas by race and marital status, and smaller for the estimates that are heterogeneous by education level. The opposite is true for the selection component, as these two almost cancel each other out entirely. This reinforces the hypothesis that the instrument is weak for male workers. The only exception is the model with heterogeneous copula by marital status, for which the the participation components is more negative during the entire period, \textit{i.e.}, in favor of female workers.

\begin{figure}[htbp]
\caption{Mean decomposition, actual earnings for participants with heterogeneous copulas}
\includegraphics[width=16.5cm]{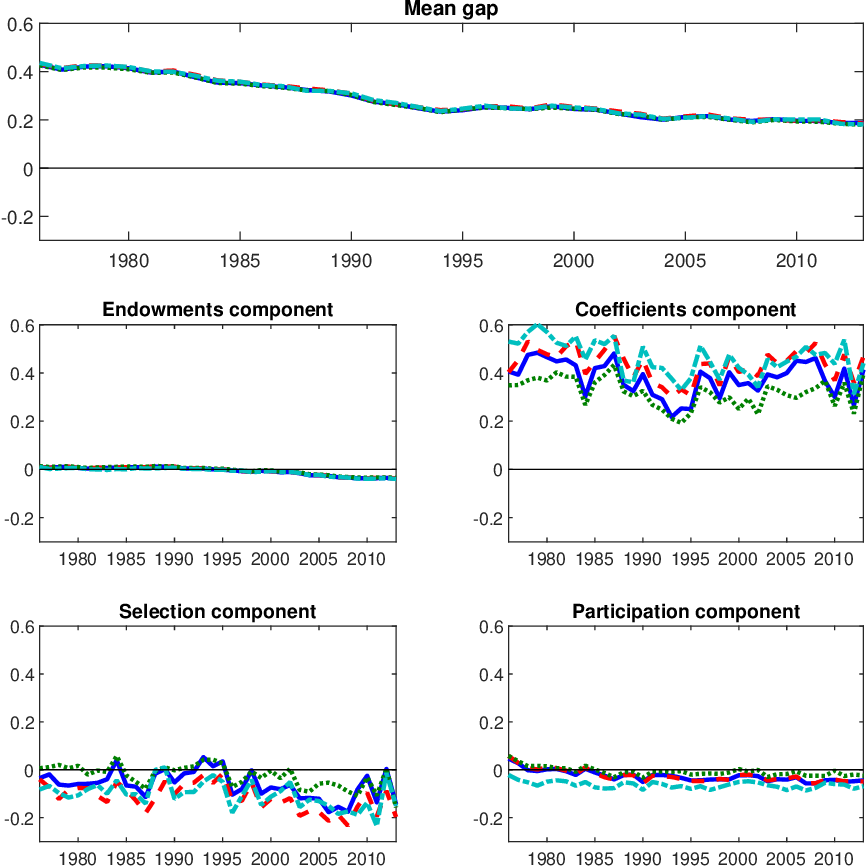}\label{fig:dY1_het}

{\footnotesize Notes: the solid blue line indicates the gap with the homogeneous copulas; the dashed red line indicates the gap with the heterogeneous copulas by race; the dotted green line indicates the gap with the heterogeneous copulas by education level; the dashed-dotted cyan line indicates the gap with the heterogeneous copulas by marital status.}
\end{figure}
\section{Conclusion}\label{sec:conc}

In this paper I have introduced a new way to decompose differences in outcomes between two groups when there is self-selection into participation. In particular, rather than considering the decomposition of potential outcomes for the entire population, I decompose differences in actual outcomes for both participants and the entire population. These differences are decomposed into four components: endowments, coefficients, participation and selection. Moreover, I propose to provide two additional ancillary decompositions regarding differences in participation and self-selection.

I apply this methodology to analyze the labor earnings gap between males and females. I find that increases in female labor market participation and improvements in self-selection that led to an increase in unobserved ability for females have been responsible for a large share of the fall of the gap. Moreover, considering the gap for participants greatly underestimates the gap for the entire population, due to the still existing gender participation gap.

\bibliographystyle{chicago}
\begin{spacing}{1.0}

\bibliography{decomposition}
\end{spacing}
\pagebreak

\appendix

\section{Mathematical proofs}\label{app:proofs}

\subsection{Proof of Lemma~\ref{lem:normalize}}

$\forall d\in\mathcal{D}$, let $V=\tilde{F}_{V|d,X}\left(\tilde{V}|d,X\right)$. By definition, $V\sim U\left(0,1\right)$. Moreover, $\tilde{V}<\tilde{\pi}_{d}\left(Z\right)\Leftrightarrow \tilde{F}_{V|d,X}\left(\tilde{V}|d,X\right)<\tilde{F}_{V|d,X}\left(\tilde{\pi}_{d}\left(Z\right)|d,X\right)\equiv\pi_{d}\left(Z\right)$. Hence, $\tilde{F}_{V|S,X}=F_{V|S,X}$.

Similarly, let $U=\tilde{F}_{U|d,X}\left(\tilde{U}|d,X\right)$, which is also uniformly distributed on the unit interval. It follows that $Y=\tilde{g}_{d}\left(X,\tilde{U}\right)S=\tilde{g}_{d}\left(X,\tilde{F}_{U|d,X}^{-1}\left(U|d,X\right)\right)S\equiv g_{d}\left(X,U\right)S$. The joint distribution of $\left(\tilde{U},\tilde{V}\right)$ can be written as
\begin{align*}
\mathbb{P}\left(\tilde{U}\leq\tau,\tilde{V}\leq\tilde{\pi}_{d}\left(z\right)|Z=z\right)&=\mathbb{P}\left(\tilde{F}_{U|d,x}^{-1}\left(U|d,x\right)\leq\tau,\tilde{F}_{V|d,x}^{-1}\left(V|d,x\right)\leq\tilde{\pi}_{d}\left(z\right)|Z=z\right)\\
&=\mathbb{P}\left(U\leq\tilde{F}_{U|d,x}\left(\tau|x\right),V\leq\pi_{d}\left(z\right)|Z=z\right)\\
&=C_{d,x}\left(\tilde{F}_{U|d,x}\left(\tau|d,x\right),\pi_{d}\left(z\right)\right)
\end{align*}
where the first equality follows by the invertibility of $\tilde{U}$ and $\tilde{V}$, the second one by the first result of the lemma, and the third one by definition of the copula.

Define $\tilde{G}_{d,x}\left(\tau,\pi_{d}\left(z\right)\right)\equiv\mathbb{P}\left(\tilde{U}\leq\tau|S=1,Z=z\right)$. It can be expressed as
\begin{align}
\tilde{G}_{d,x}\left(\tau,\tilde{\pi}_{d}\left(z\right)\right)&\equiv\frac{\mathbb{P}\left(\tilde{U}\leq\tau,\tilde{V}\leq\tilde{\pi}_{d}\left(z\right)|Z=z\right)}{\mathbb{P}\left(\tilde{V}\leq\tilde{\pi}_{d}\left(z\right)|Z=z\right)}\nonumber\\
&=\frac{\mathbb{P}\left(\tilde{F}_{U|d,x}^{-1}\left(U_{1}|d,x\right)\leq\tau,\tilde{F}_{V|d,x}^{-1}\left(V|d,x\right)\leq\tilde{\pi}_{d}\left(z\right)|Z=z\right)}{\mathbb{P}\left(\tilde{F}_{V|d,x}^{-1}\left(V|d,x\right)\leq\tilde{\pi}_{d}\left(z\right)\right|Z=z)}\nonumber\\
&=\frac{C_{d,x}\left(\tilde{F}_{U|d,x}\left(\tau|d,x\right),\pi_{d}\left(z\right)\right)}{\pi_{d}\left(z\right)}=G_{d,x}\left(\tilde{F}_{U|d,x}\left(\tau|x\right),\pi_{d}\left(z\right)\right)\label{eq:cond0}
\end{align}

Then, the distribution of $Y$, conditional on $S=1$ and $Z=z$ equals
\begin{align*}
\mathbb{P}\left(Y\leq y|S=1,Z=z\right)&=\int\mathbf{1}\left(\tilde{g}_{d}\left(x,\tilde{u}\right)\leq y\right)d\tilde{G}_{d,x}\left(\tilde{u},\tilde{\pi}_{d}\left(z\right)\right)\\
&=\int\mathbf{1}\left(\tilde{g}_{d}\left(x,\tilde{u}\right)\leq y\right)dG_{d,x}\left(\tilde{F}_{U|d,x}\left(\tilde{u}|x\right),\pi_{d}\left(z\right)\right)\\
&=\int\mathbf{1}\left(\tilde{g}_{d}\left(x,\tilde{F}_{U|d,x}^{-1}\left(u\right)\right)\leq y\right)dG_{d,x}\left(u,\pi_{d}\left(z\right)\right)\\
&=\int\mathbf{1}\left(g_{d}\left(x,u\right)\leq y\right)dG_{d,x}\left(u,\pi_{d}\left(z\right)\right)\\
\end{align*}
where the second equality follows by Equation~\ref{eq:cond0}, the third one by the invertibility of $\tilde{U}$, and the fourth one by the definition of $\tilde{g}_{d}$, completing the proof.

\subsection{Proof of Theorem~\ref{thm:asymgc}}

Define
\begin{align*}
\mathbb{Z}_{\upsilon_{m}}\left(z,\tau,\pi,f\right)\equiv\begin{pmatrix}
\mathbb{Z}_{g_{j}}\left(\tau,x\right)\\
\mathbb{Z}_{c_{k}}\left(\tau,\pi\right)\\
\mathbb{Z}_{\pi_{l}}\left(z\right)\\
\sqrt{\lambda_{h}}\mathbb{Z}_{Z_{h}}\left(f\right)
\end{pmatrix}
\end{align*}

By Assumptions~\ref{assum:samp}-\ref{assum:gsize}, it is possible to apply Lemma E.4 in \cite{Chernozhukov2013}. This, together with Lemma~\ref{lem:hadamarde} and the functional delta method yields\footnote{Recall that $\hat{F}_{Z}^{h}\left(z\right)=\frac{1}{n_{h}}\sum_{i=1}^{n}\mathbf{1}\left(Z_{i}\leq z\right)\mathbf{1}\left(D_{i}=h\right)$, so that $\sqrt{n}\int_{\mathcal{Z}}fd\left(\hat{F}_{Z}^{h}-F_{Z}^{h}\right)\Rightarrow\sqrt{\lambda_{h}}\mathbb{Z}_{Z_{h}}\left(f\right)$.}
\begin{align*}
\sqrt{n}&\left(\hat{\mathbb{E}}\left[Y^{m}|S=1\right]-\mathbb{E}\left[Y^{m}|S=1\right]\right)\\
&=\sqrt{n}\int_{\mathcal{Z}}\int_{\varepsilon}^{1-\varepsilon}\left(\hat{g}_{j}\left(x,\tau\right)-g_{j}\left(x,\tau\right)\right)dG_{k,x}\left(\tau|\pi_{l}\left(z\right)\right)sdF_{Z}^{h}\left(z\right)\\
&+\sqrt{n}\int_{\mathcal{Z}}\int_{\varepsilon}^{1-\varepsilon}g_{j}\left(x,\tau\right)d\left(\hat{G}_{k,x}\left(\tau|\pi_{l}\left(z\right)\right)-G_{k,x}\left(\tau|\pi_{l}\left(z\right)\right)\right)sdF_{Z}^{h}\left(z\right)\\
&+\sqrt{n}\int_{\mathcal{Z}}\int_{\varepsilon}^{1-\varepsilon}g_{j}\left(x,\tau\right)d\left(\nabla_{\pi}G_{k,x}\left(\tau|\pi_{l}\left(z\right)\right)\right)\left(\hat{\pi}_{l}\left(z\right)-\pi_{l}\left(z\right)\right)sdF_{Z}^{h}\left(z\right)\\
&+\sqrt{n}\int_{\mathcal{Z}}\int_{\varepsilon}^{1-\varepsilon}g_{j}\left(x,\tau\right)dG_{k,x}\left(\tau|\pi_{l}\left(z\right)\right)d\left(\hat{F}_{Z}^{h}\left(z\right)-F_{Z}^{h}\left(z\right)\right)+o_{P}\left(1\right)\\
&\Rightarrow\int_{\mathcal{Z}}\int_{\varepsilon}^{1-\varepsilon}\mathbb{Z}_{g_{j}}\left(\tau,x\right)dG_{k,x}\left(\tau|\pi_{l}\left(z\right)\right)dF_{Z}^{h}\left(z\right)\\
&+\int_{\mathcal{Z}}\int_{\varepsilon}^{1-\varepsilon}g_{j}\left(x,\tau\right)\frac{1}{\pi_{l}\left(z\right)}\mathbb{Z}_{c_{k}}\left(\tau,\pi_{l}\left(z\right)\right)d\tau dF_{Z}^{h}\left(z\right)\\
&+\int_{\mathcal{Z}}\int_{\varepsilon}^{1-\varepsilon}g_{j}\left(x,\tau\right)\frac{\nabla_{v}c_{k,x}\left(u,\pi_{l}\left(z\right)\right)\pi_{l}\left(z\right)-c_{k,x}\left(\tau,\pi_{l}\left(z\right)\right)}{\pi_{l}\left(z\right)^{2}}\mathbb{Z}_{\pi_{l}}\left(z\right)d\tau dF_{Z}^{h}\left(z\right)\\
&+\sqrt{\lambda_{h}}\mathbb{Z}_{Z_{h}}\left(\int_{\varepsilon}^{1-\varepsilon}g_{j}\left(\cdot,\tau\right)dG_{k,x}\left(\tau|\pi_{l}\left(\cdot\right)\right)\right)\equiv\mathbb{Z}_{Y^{m}|S=1}
\end{align*}
in $\ell^{\infty}\left(\mathcal{M}\right)$. Apply the functional delta method once more to obtain
\begin{align*}
\sqrt{n}\Delta^{m,m'}\left(\hat{\mathbb{E}}\left[Y|S=1\right]-\mathbb{E}\left[Y|S=1\right]\right)\Rightarrow\mathbb{Z}_{Y^{m}|S=1}-\mathbb{Z}_{Y^{m'}|S=1}\equiv\mathbb{Z}_{\Delta^{mm'}Y|S=1}
\end{align*}

To show the asymptotic distribution of $\Delta^{m,m'}\hat{Q}_{Y|S=1}\left(\tau\right)$, I firstly show some intermediate steps. First, note that we can write
\begin{align*}
&\hat{F}_{Y|S=1}^{m}\left(y|z\right)=\frac{1}{n_{h}}\sum_{i=1}^{n}\left[\varepsilon+\int_{\varepsilon}^{1-\varepsilon}\mathbf{1}\left(\hat{g}_{j}\left(X_{i},u\right)\leq y\right)d\hat{G}_{k,x}\left(u,\hat{\pi}_{l}\left(Z_{i}\right)\right)\right]\mathbf{1}\left(D_{i}=h\right)\\
&=\int_{\mathcal{Z}}\left[\varepsilon+\int_{\varepsilon}^{1-\varepsilon}\mathbf{1}\left(\hat{g}_{j}\left(x,u\right)\leq y\right)\frac{\hat{c}_{k,x}\left(u,\hat{\pi}_{l}\left(z\right)\right)}{\hat{\pi}_{l}\left(z\right)}du\right]d\hat{F}_{Z}^{h}\left(z\right)\equiv\int_{\mathcal{Z}}\hat{F}_{Y|Z,S=1}^{m}\left(y|z\right)d\hat{F}_{Z_{m}}\left(z\right)
\end{align*}

Next, consider the joint asymptotic distribution of $\left(\hat{F}_{Y|Z,S=1}^{m}\left(y|z\right),\int_{\mathcal{Z}}fd\hat{F}_{Z_{m}}\left(z\right)\right)$. By Lemma~\ref{lem:hadamardf} and the functional delta method, it follows that
\begin{align*}
\sqrt{n}&\left(\hat{F}_{Y|Z,S=1}^{m}\left(y|z\right)-F_{Y|Z,S=1}^{m}\left(y|z\right)\right)\\
&=\sqrt{n}\int_{\varepsilon}^{1-\varepsilon}\left(\mathbf{1}\left(\hat{g}_{j}\left(x,u\right)\leq y\right)-\mathbf{1}\left(g_{j}\left(x,u\right)\leq y\right)\right)\frac{c_{k,x}\left(u,\pi_{l}\left(z\right)\right)}{\pi_{l}\left(z\right)}du\\
&+\sqrt{n}\int_{\varepsilon}^{1-\varepsilon}\mathbf{1}\left(g_{j}\left(x,u\right)\leq y\right)\frac{1}{\pi_{l}\left(z\right)}\left(\hat{c}_{k,x}\left(u,\pi_{l}\left(z\right)\right)-c_{k,x}\left(u,\pi_{l}\left(z\right)\right)\right)du\\
&+\sqrt{n}\int_{\varepsilon}^{1-\varepsilon}\mathbf{1}\left(g_{j}\left(x,u\right)\leq y\right)\frac{\nabla_{v}c_{k,x}\left(u,\pi_{l}\left(z\right)\right)\pi_{l}\left(z\right)-c_{k,x}\left(u,\pi_{l}\left(z\right)\right)}{\pi_{l}\left(z\right)^{2}}\left(\hat{\pi}_{l}\left(z\right)-\pi_{l}\left(z\right)\right)du+o_{P}\left(1\right)\\
&\Rightarrow -f_{Y|Z,S=1}\left(y|z\right)\mathbb{Z}_{g_{j}}\left(G_{j,x}^{-1}\left(F_{Y|Z,S=1}^{j}\left(y|z\right),\pi_{j}\left(z\right)\right),x\right)\\
&+\int_{\varepsilon}^{1-\varepsilon}\mathbf{1}\left(g_{j}\left(x,u\right)\leq y\right)\frac{1}{\pi_{l}\left(z\right)}\mathbb{Z}_{c_{k}}\left(u,\pi_{l}\left(z\right)\right)du\\
&+\int_{\varepsilon}^{1-\varepsilon}\mathbf{1}\left(g_{j}\left(x,u\right)\leq y\right)\frac{\nabla_{v}c_{k,x}\left(u,\pi_{l}\left(z\right)\right)\pi_{l}\left(z\right)-c_{k,x}\left(u,\pi_{l}\left(z\right)\right)}{\pi_{l}\left(z\right)^{2}}\mathbb{Z}_{\pi_{l}}\left(z\right)du\\
&\equiv\mathbb{Z}_{F_{Y}^{m}|Z,S=1}\left(y,z\right)
\end{align*}
in $\ell^{\infty}\left(\mathcal{Y}\mathcal{Z}\mathcal{M}\right)$. Therefore,
\begin{align*}
\sqrt{n}\begin{pmatrix}
\hat{F}_{Y|Z,S=1}^{m}\left(y|z\right)-F_{Y|Z,S=1}^{m}\left(y|z\right)\\
\int_{\mathcal{Z}}fd\left(\hat{F}_{Z}^{h}\left(z\right)-F_{Z}^{h}\left(z\right)\right)
\end{pmatrix}\Rightarrow\begin{pmatrix}
\mathbb{Z}_{F_{Y}^{m}|Z,S=1}\left(y,z\right)\\
\sqrt{\lambda_{h}}\mathbb{Z}_{Z_{h}}\left(f\right)
\end{pmatrix}
\end{align*}
in $\ell^{\infty}\left(\mathcal{Y}\mathcal{Z}\mathcal{F}\mathcal{M}\right)$.

The next step is to show the asymptotic distribution of the estimator of the unconditional distribution $\hat{F}_{Y|S=1}^{m}\left(y\right)$. By the functional delta method,
\begin{align*}
\sqrt{n}\left(\hat{F}_{Y|S=1}^{m}\left(y\right)-F_{Y|S=1}^{m}\left(y\right)\right)&=\sqrt{n}\int_{\mathcal{Z}}\left(\hat{F}_{Y|Z,S=1}^{m}\left(y|z\right)-F_{Y|Z,S=1}^{m}\left(y|z\right)\right)dF_{Z}^{h}\left(z\right)\\
&+\sqrt{n}\int_{\mathcal{Z}}F_{Y|Z,S=1}^{m}\left(y|z\right)d\left(\hat{F}_{Z}^{h}\left(z\right)-F_{Z}^{h}\left(z\right)\right)\\
&\Rightarrow\int_{\mathcal{Z}}\mathbb{Z}_{F_{Y}^{m}|Z,S=1}\left(y|z\right)dF_{Z}^{h}\left(z\right)+\sqrt{\lambda_{h}}\mathbb{Z}_{Z_{h}}\left(F_{Y|Z,S=1}^{m}\left(y|z\right)\right)\\
&\equiv\mathbb{Z}_{F_{Y|S=1}^{m}}\left(y\right)
\end{align*}
uniformly in $\ell^{\infty}\left(\mathcal{Y}\mathcal{M}\right)$.

The final step is to show the asymptotic distribution of $\hat{Q}_{Y|S=1}^{m}\left(\tau\right)$. By Lemma E.4 in \cite{Chernozhukov2013} and the functional delta method,
\begin{align*}
\sqrt{n}\left(\hat{Q}_{Y|S=1}^{m}\left(\tau\right)-Q_{Y|S=1}^{m}\left(\tau\right)\right)&=-\frac{\sqrt{n}\left(\hat{F}_{Y|S=1}^{m}\left(Q_{Y|S=1}^{m}\left(\tau\right)\right)-F_{Y|S=1}^{m}\left(Q_{Y|S=1}^{m}\left(\tau\right)\right)\right)}{f_{Y|S=1}^{m}\left(Q_{Y|S=1}^{m}\left(\tau\right)\right)}+o_{P}\left(1\right)\\
&\Rightarrow-\frac{\mathbb{Z}_{F_{Y|S=1}^{m}}\left(Q_{Y|S=1}^{m}\left(\tau\right)\right)}{f_{Y|S=1}^{m}\left(Q_{Y|S=1}^{m}\left(\tau\right)\right)}\\
&\equiv\mathbb{Z}_{Q_{Y|S=1}^{m}}\left(\tau\right)
\end{align*}
uniformly in $\ell^{\infty}\left(\mathcal{T}\mathcal{M}\right)$, where I have used the Hadamard differentiability of the quantile operator \citep{Chernozhukov2010}. $\tau\rightarrow Q_{Y|S=1}^{m}\left(\tau\right)$ is a.s. uniformly continuous by Assumption~\ref{assum:bound}, and together with the a.s. uniform continuity of $\mathbb{Z}_{F_{Y|S=1}^{m}}\left(y\right)$, it follows that $\mathbb{Z}_{Q_{Y|S=1}^{m}}\left(\tau\right)$ is a.s. uniformly continuous with respect to $\tau$.

Finally, note that
\begin{align*}
\Delta^{m,m'}\left(\hat{Q}_{Y|S=1}\left(\tau\right)-Q_{Y|S=1}\left(\tau\right)\right)=\hat{Q}_{Y|S=1}^{m}\left(\tau\right)-Q_{Y|S=1}^{m}\left(\tau\right)-\left(\hat{Q}_{Y|S=1}^{m'}\left(\tau\right)-Q_{Y|S=1}^{m'}\left(\tau\right)\right)
\end{align*}

Hence, $\sqrt{n}\Delta^{m,m'}\left(\hat{Q}_{Y|S=1}\left(\tau\right)-Q_{Y|S=1}\left(\tau\right)\right)\Rightarrow\mathbb{Z}_{Q_{Y|S=1}^{m}}\left(\tau\right)-\mathbb{Z}_{Q_{Y|S=1}^{m'}}\left(\tau\right)\equiv\mathbb{Z}_{Q|S=1,mm'}\left(\tau\right)$, finishing the proof.

\subsection{Proof of Theorem~\ref{thm:asym}}

Let $E\equiv\left(Y,S,D,Z\right)$. Moreover, define
\begin{align*}
R_{m}\left(E,\beta,\theta,\gamma,\tau\right)\equiv
\begin{bmatrix}
\mathbf{1}\left(D=j\right)SX\zeta_{G_{j,x}\left(\tau,\pi_{j}\left(Z;\gamma\right),\theta\right)}\left(Y-X'\beta\right)\\
\int_{\varepsilon}^{1-\varepsilon}\mathbf{1}\left(D=k\right)S\varphi\left(u,Z\right)\zeta_{G_{k,x}\left(\tau,\pi_{k}\left(Z;\gamma\right),\theta\right)}\left(Y-X'\beta\right)du\\
\mathbf{1}\left(D=l\right)b_{l}\left(S,Z;\gamma\right)\\
\end{bmatrix}
\end{align*}
\begin{align*}
q_{m}\left(E,\beta,\theta,\gamma,\tau\right)\equiv
\begin{bmatrix}
\mathbf{1}\left(D=j\right)SX\rho_{G_{j,x}\left(\tau,\pi_{j}\left(Z;\gamma\right),\theta\right)}\left(Y-X'\beta\right)\\
\mathbf{1}\left(D=k\right)S\int_{\varepsilon}^{1-\varepsilon}\varphi\left(u,Z\right)\rho_{G_{k,x}\left(\tau,\pi_{k}\left(Z;\gamma\right),\theta\right)}\left(Y-X'\beta\right)du\\
\mathbf{1}\left(D=l\right)\tilde{b}_{l}\left(S,Z;\gamma\right)\\
\end{bmatrix}
\end{align*}
for some $\tilde{b}_{l}\left(S,Z;\gamma\right)$ such that $\frac{\partial}{\partial\gamma}\tilde{b}_{l}\left(S,Z;\gamma\right)=b_{l}\left(S,Z;\gamma\right)$, $f\mapsto\mathbb{E}_{n}\left[f\left(E\right)\right]\equiv\frac{1}{n}\sum_{i=1}^{n}f\left(E\right)$, $f\mapsto\mathbb{G}_{n}\left[f\left(E\right)\right]\equiv\frac{1}{\sqrt{n}}\sum_{i=1}^{n}f\left(E\right)-\mathbb{E}\left(f\left(E\right)\right)$, $Q_{m,n}\left(\beta,\theta,\gamma,\tau\right)\equiv\mathbb{E}_{n}\left[q_{m}\left(E,\beta,\theta,\gamma,\tau\right)\right]$, and $Q_{m}\left(\beta,\theta,\gamma,\tau\right)\equiv\mathbb{E}\left[q_{m}\left(E,\beta,\theta,\gamma,\tau\right)\right]$, where $\rho_{\tau}\left(u\right)\equiv\left(\tau-\mathbf{1}\left(u<0\right)\right)u$, $\zeta_{\tau}\left(u\right)\equiv\left(\mathbf{1}\left(u<0\right)-\tau\right)$, $\epsilon_{d}\left(\tau\right)\equiv Y-X'\beta_{d}\left(\tau\right)$, and $\hat{\epsilon}_{d}\left(\tau\right)\equiv Y-X'\hat{\beta}_{d}\left(\tau\right)$.

First I show the consistency of the estimator. By Assumptions~\ref{assum:propensity} to~\ref{assum:parcop}, $Q_{m}\left(\beta,\theta,\gamma,\tau\right)$ is continuous over $\mathcal{B}\times\Theta\times\Gamma\times\mathcal{U}$. By Lemma~\ref{lem:se}, $\sup_{\left(\beta,\theta,\gamma,\tau\right)\in\mathcal{B}\times\Theta\times\Gamma\times\mathcal{U}}\left\|Q_{m,n}\left(\beta,\theta,\gamma,\tau\right)-Q_{m}\left(\beta,\theta,\gamma,\tau\right)\right\|\overset{P}{\rightarrow}0$. Thus, by Lemma~\ref{lem:argmax}, $\sup_{\tau\in\mathcal{U}}\left\|\hat{\vartheta}_{m}\left(\tau\right)-\vartheta_{m}\left(\tau\right)\right\|\overset{P}{\rightarrow}0$.

Next, I show its asymptotic distribution. By Theorem 3 in \cite{Koenker1978}, it is possible to show that
\begin{align*}
O\left(\frac{1}{\sqrt{n}}\right)=\sqrt{n}\mathbb{E}_{n}\left[\mathbf{1}\left(D=d\right)SX\zeta_{G_{d,x}\left(\tau,\pi_{d}\left(Z;\hat{\gamma}_{d}\right),\hat{\theta}_{d}\right)}\left(\hat{\epsilon}_{d}\left(\tau\right)\right)\right]
\end{align*}

By Lemma~\ref{lem:se} and Assumption~\ref{assum:fullrank}, the following expansion holds in $\ell^{\infty}\left(\mathcal{U}\right)$:
\begin{align*}
O\left(\frac{1}{\sqrt{n}}\right)&=\mathbb{G}_{n}\left[\mathbf{1}\left(D=d\right)SX\zeta_{G_{d,x}\left(\tau,\pi_{d}\left(Z;\hat{\gamma}_{d}\right);\hat{\theta}_{d}\right)}\left(\hat{\epsilon}_{d}\left(\tau\right)\right)\right]+\sqrt{n}\mathbb{E}\left[\mathbf{1}\left(D=d\right)SX\zeta_{G_{d,x}\left(\tau,\pi_{d}\left(Z;\hat{\gamma}_{d}\right);\hat{\theta}_{d}\right)}\left(\hat{\epsilon}_{d}\left(\tau\right)\right)\right]\\
&=\mathbb{G}_{n}\left[\mathbf{1}\left(D=d\right)SX\zeta_{G_{d,x}\left(\tau,\pi_{d}\left(Z;\gamma_{d}\right);\theta_{d}\right)}\left(\epsilon_{d}\left(\tau\right)\right)\right]+o_{P}\left(1\right)\\
&+\sqrt{n}\mathbb{E}\left[\mathbf{1}\left(D=d\right)SX\zeta_{G_{d,x}\left(\tau,\pi_{d}\left(Z;\hat{\gamma}_{d}\right);\hat{\theta}_{d}\right)}\left(\hat{\epsilon}_{d}\left(\tau\right)\right)\right]\\
&=\mathbb{G}_{n}\left[\mathbf{1}\left(D=d\right)SX\zeta_{G_{d,x}\left(\tau,\pi_{d}\left(Z;\gamma_{d}\right);\theta_{d}\right)}\left(\epsilon_{d}\left(\tau\right)\right)\right]+J_{\beta d}\left(\tau\right)\sqrt{n}\left(\hat{\beta}_{d}\left(\tau\right)-\beta_{d}\left(\tau\right)\right)\\
&-J_{\gamma d}\left(\tau\right)\sqrt{n}\left(\hat{\gamma}_{d}-\gamma_{d}\right)-J_{\theta d}\left(\tau\right)\sqrt{n}\left(\hat{\theta}_{d}-\theta_{d}\right)+o_{P}\left(1\right)
\end{align*}
where
\begin{align*}
J_{\beta d}\left(\tau\right)\equiv\frac{\partial\mathbb{E}\left[\mathbf{1}\left(D=d\right)SX\zeta_{G_{d,x}\left(\tau,\pi_{d}\left(Z;\gamma_{d}\right);\theta_{d}\right)}\left(\epsilon_{d}\left(\tau\right)\right)\right]}{\partial\beta_{d}}
\end{align*}
\begin{align*}
J_{\gamma d}\left(\tau\right)\equiv-\frac{\partial\mathbb{E}\left[\mathbf{1}\left(D=d\right)SX\zeta_{G_{d,x}\left(\tau,\pi_{d}\left(Z;\gamma_{d}\right);\theta_{d}\right)}\left(\epsilon_{d}\left(\tau\right)\right)\right]}{\partial\gamma_{d}}
\end{align*}
\begin{align*}
J_{\theta d}\left(\tau\right)\equiv-\frac{\partial\mathbb{E}\left[\mathbf{1}\left(D=d\right)SX\zeta_{G_{d,x}\left(\tau,\pi_{d}\left(Z;\gamma_{d}\right);\theta_{d}\right)}\left(\epsilon_{d}\left(\tau\right)\right)\right]}{\partial\theta_{d}}
\end{align*}

Rearranging and solving for $\sqrt{n}\left(\hat{\beta}_{d}\left(\tau\right)-\beta_{d}\left(\tau\right)\right)$,
\begin{align}\label{eq:beta1}
\sqrt{n}\left(\hat{\beta}_{d}\left(\tau\right)-\beta_{d}\left(\tau\right)\right)&=-J_{\beta d}\left(\tau\right)^{-1}\left\{\mathbb{G}_{n}\left[\mathbf{1}\left(D=d\right)SX\zeta_{G_{d,x}\left(\tau,\pi_{d}\left(Z;\gamma_{d}\right);\theta_{d}\right)}\left(\epsilon_{d}\left(\tau\right)\right)\right]\right.\nonumber\\
&\left.-J_{\gamma d}\left(\tau\right)\sqrt{n}\left(\hat{\gamma}_{d}-\gamma_{d}\right)-J_{\theta d}\left(\tau\right)\sqrt{n}\left(\hat{\theta}_{d}-\theta_{d}\right)\right\}+o_{P}\left(1\right)
\end{align}
in $\ell^{\infty}\left(\mathcal{U}\right)$.

Using Theorem 3 in \cite{Koenker1978} again, it is possible to show that
\begin{align*}
O\left(\frac{1}{\sqrt{n}}\right)=\sqrt{n}\mathbb{E}_{n}\left[\int_{\varepsilon}^{1-\varepsilon}\mathbf{1}\left(D=d\right)S\varphi\left(u,Z\right)\zeta_{G_{d,x}\left(u,\pi_{d}\left(Z;\hat{\gamma}_{d}\right);\hat{\theta}_{d}\right)}\left(\hat{\epsilon}_{d}\left(u\right)\right)du\right]
\end{align*}

By Lemma~\ref{lem:se} and Assumption~\ref{assum:fullrank}, the following expansion holds:
\begin{align*}
O\left(\frac{1}{\sqrt{n}}\right)&=\mathbb{G}_{n}\left[\int_{\varepsilon}^{1-\varepsilon}\mathbf{1}\left(D=d\right)S\varphi\left(u,Z\right)\zeta_{G_{d,x}\left(u,\pi_{d}\left(Z;\hat{\gamma}_{d}\right);\hat{\theta}_{d}\right)}\left(\hat{\epsilon}_{d}\left(u\right)\right)du\right]\\
&+\sqrt{n}\int_{\varepsilon}^{1-\varepsilon}\mathbb{E}\left[\mathbf{1}\left(D=d\right)S\varphi\left(u,Z\right)\zeta_{G_{d,x}\left(u,\pi_{d}\left(Z;\hat{\gamma}_{d}\right);\hat{\theta}_{d}\right)}\left(\hat{\epsilon}_{d}\left(u\right)\right)\right]du\\
&=\mathbb{G}_{n}\left[\int_{\varepsilon}^{1-\varepsilon}\mathbf{1}\left(D=d\right)S\varphi\left(u,Z\right)\zeta_{G_{d,x}\left(u,\pi_{d}\left(Z;\gamma_{d}\right);\theta_{d}\right)}\left(\epsilon_{d}\left(u\right)\right)du\right]+o_{P}\left(1\right)\\
&+\sqrt{n}\int_{\varepsilon}^{1-\varepsilon}\mathbb{E}\left[\mathbf{1}\left(D=d\right)S\varphi\left(u,Z\right)\zeta_{G_{d,x}\left(u,\pi_{d}\left(Z;\hat{\gamma}_{d}\right);\hat{\theta}_{d}\right)}\left(\hat{\epsilon}_{d}\left(u\right)\right)\right]du\\
&=\mathbb{G}_{n}\left[\int_{\varepsilon}^{1-\varepsilon}\mathbf{1}\left(D=d\right)S\varphi\left(u,Z\right)\zeta_{G_{d,x}\left(u,\pi_{d}\left(Z;\gamma_{d}\right);\theta_{d}\right)}\left(\epsilon_{d}\left(u\right)\right)du\right]\\
&+\sqrt{n}\int_{\varepsilon}^{1-\varepsilon}\tilde{J}_{\beta d}\left(u\right)\left(\hat{\beta}_{d}\left(u\right)-\beta_{d}\left(u\right)\right)du-\sqrt{n}\int_{\varepsilon}^{1-\varepsilon}\tilde{J}_{\theta d}\left(u\right)du\left(\hat{\theta}_{d}-\theta_{d}\right)\\
&-\sqrt{n}\int_{\varepsilon}^{1-\varepsilon}\tilde{J}_{\gamma d}\left(u\right)du\left(\hat{\gamma}_{d}-\gamma_{d}\right)+o_{P}\left(1\right)
\end{align*}
where
\begin{align*}
\tilde{J}_{\beta d}\left(\tau\right)\equiv\frac{\partial\mathbb{E}\left[\mathbf{1}\left(D=d\right)S\varphi\left(\tau,Z\right)\zeta_{G_{d,x}\left(\tau,\pi_{d}\left(Z;\gamma_{d}\right);\theta_{d}\right)}\left(\epsilon_{d}\left(\tau\right)\right)\right]}{\partial\beta_{d}}
\end{align*}
\begin{align*}
\tilde{J}_{\gamma d}\left(\tau\right)\equiv-\frac{\partial\mathbb{E}\left[\mathbf{1}\left(D=d\right)S\varphi\left(\tau,Z\right)\zeta_{G_{d,x}\left(\tau,\pi_{d}\left(Z;\gamma_{d}\right);\theta_{d}\right)}\left(\epsilon_{d}\left(\tau\right)\right)\right]}{\partial\gamma_{d}}
\end{align*}
\begin{align*}
\tilde{J}_{\theta d}\left(\tau\right)\equiv-\frac{\partial\mathbb{E}\left[\mathbf{1}\left(D=d\right)S\varphi\left(\tau,Z\right)\zeta_{G_{d,x}\left(\tau,\pi_{d}\left(Z;\gamma_{d}\right);\theta_{d}\right)}\left(\epsilon_{d}\left(\tau\right)\right)\right]}{\partial\theta_{d}}
\end{align*}

Rearranging and solving for $\sqrt{n}\left(\hat{\theta}_{d}-\theta_{d}\right)$,
\begin{align}\label{eq:theta1}
\sqrt{n}\left(\hat{\theta}_{d}-\theta_{d}\right)&=\left[\int_{\varepsilon}^{1-\varepsilon}\tilde{J}_{\theta d}\left(u\right)du\right]^{-1}\left\{\mathbb{G}_{n}\left[\int_{\varepsilon}^{1-\varepsilon}\mathbf{1}\left(D=d\right)S\varphi\left(u,Z\right)\zeta_{G_{d,x}\left(u,\pi_{d}\left(Z;\gamma_{d}\right);\theta_{d}\right)}\left(\epsilon_{d}\left(u\right)\right)du\right]\right.\nonumber\\
&+\sqrt{n}\int_{\varepsilon}^{1-\varepsilon}\tilde{J}_{\beta d}\left(u\right)\left(\hat{\beta}_{d}\left(u\right)-\beta_{d}\left(u\right)\right)du-\left.\sqrt{n}\int_{\varepsilon}^{1-\varepsilon}\tilde{J}_{\gamma d}\left(u\right)du\left(\hat{\gamma}_{d}-\gamma_{d}\right)\right\}+o_{P}\left(1\right)
\end{align}

Now define
\begin{align*}
A_{d}\left(\tau\right)\equiv\hat{\vartheta}_{d}\left(\tau\right)-\vartheta_{d}\left(\tau\right)
\end{align*}
\begin{align*}
C_{d}\left(\tau\right)\equiv\begin{bmatrix}-J_{\beta d}\left(\tau\right)^{-1} & 0 & 0\\
0 & \left[\int_{\varepsilon}^{1-\varepsilon}\tilde{J}_{\theta d}\left(u\right)du\right]^{-1} & 0\\
0 & 0 & -B_{d}^{-1}\end{bmatrix}
\end{align*}
\begin{align*}
D_{d}\left(\tau\right)\equiv\begin{bmatrix}0 & 0 & 0\\
\left[\int_{\varepsilon}^{1-\varepsilon}\tilde{J}_{\theta d}\left(u\right)du\right]^{-1}\tilde{J}_{\beta d}\left(\tau\right) & 0 & -\left[\int_{\varepsilon}^{1-\varepsilon}\tilde{J}_{\theta d}\left(u\right)du\right]^{-1}\tilde{J}_{\gamma d}\left(\tau\right)\\
0 & 0 & 0\end{bmatrix}
\end{align*}
\begin{align*}
F_{d}\left(\tau\right)\equiv\begin{bmatrix}0 & J_{\beta d}\left(\tau\right)^{-1}J_{\theta d}\left(\tau\right) & J_{\beta d}\left(\tau\right)^{-1}J_{\gamma d}\left(\tau\right)\\
0 & 0 & 0\\
0 & 0 & 0\end{bmatrix}
\end{align*}
\begin{align*}
\psi_{d}\left(\tau\right)\equiv R_{d}\left(E,\beta\left(\tau\right),\theta,\gamma,\tau\right)
\end{align*}

Combining Equations~\ref{eq:beta1} and~\ref{eq:theta1} yields
\begin{align}\label{eq:fredholm}
A_{d}\left(\tau\right)=F_{d}\left(\tau\right)A_{d}\left(\tau\right)+\int_{\varepsilon}^{1-\varepsilon}D_{d}\left(u\right)A_{d}\left(u\right)du+C_{d}\left(\tau\right)\frac{1}{\sqrt{n}}\mathbb{G}_{n}\psi_{d}\left(\tau\right)+o_{P}\left(\frac{1}{\sqrt{n}}\right)
\end{align}
in $\ell^{\infty}\left(\mathcal{U}\right)$. Equation~\ref{eq:fredholm} is a particular case of a Fredholm integral equation of the second kind. The solution to this type of equations is a Liouville-Neumann series. By Lemma~\ref{lem:fredholm}, the solution to this equation is given by:
\begin{align}\label{eq:a}
\sqrt{n}A_{d}\left(\tau\right)&=F_{d}^{I}\left(\tau\right)\left(I-\int_{\varepsilon}^{1-\varepsilon}D_{d}\left(u\right)F_{d}^{I}\left(u\right)du\right)^{-1}\int_{\varepsilon}^{1-\varepsilon}D_{d}\left(u\right)F_{d}^{I}\left(u\right)C_{d}\left(u\right)\mathbb{G}_{n}\psi_{d}\left(u\right)du\nonumber\\
&+F_{d}^{I}\left(\tau\right)C_{d}\left(\tau\right)\mathbb{G}_{n}\psi_{d}\left(\tau\right)+o_{P}\left(1\right)
\end{align}
in $\ell^{\infty}\left(\mathcal{U}\right)$, where $F_{d}^{I}\left(\tau\right)\equiv\left(I-F_{d}\left(\tau\right)\right)^{-1}=I+F_{d}\left(\tau\right)$. Using the Functional Delta Method and Lemmas~\ref{lem:hadamard} and~\ref{lem:se}, it follows that 
$\sqrt{n}\left(\hat{\vartheta}_{d}\left(\tau\right)-\vartheta_{d}\left(\tau\right)\right)\Rightarrow H_{d}\left(\tau\right)\mathbb{Z}_{R_{d}}\left(\tau\right)\equiv\mathbb{Z}_{\vartheta_{d}}\left(\tau\right)$, a zero-mean tight Gaussian process with covariance $\Sigma_{\vartheta_{d}}\left(\tau,\tau'\right)=H_{d}\left(\tau\right)\Sigma_{R_{d}}\left(\tau,\tau'\right)H_{d}\left(\tau\right)$. Using the same arguments, conclude that $\sqrt{n}\left(\hat{\vartheta}_{m}\left(\tau\right)-\vartheta_{m}\left(\tau\right)\right)\Rightarrow H_{m}\left(\tau\right)\mathbb{Z}_{R_{m}}\left(\tau\right)\equiv\mathbb{Z}_{\vartheta_{m}}\left(\tau\right)$.

\subsection{Proof of Corollary~\ref{cor:asym}}

The first step is to show the asymptotic distribution of the first three components of $\left(\hat{g}_{j}\left(x,\tau\right),\hat{c}_{k,x}\left(\tau,\pi\right),\hat{\pi}_{l}\left(z\right),\int_{\mathcal{Z}}fd\hat{F}_{Z}^{h}\right)$. By Theorem~\ref{thm:asym} and the functional delta method,
\begin{align*}
\sqrt{n}\left(\left(\hat{g}_{j}\left(x,\tau\right),\hat{c}_{k,x}\left(\tau,\pi\right),\hat{\pi}_{l}\left(z\right)\right)-\left(g_{j}\left(x,\tau\right),c_{k,x}\left(\tau,\pi\right),\pi_{l}\left(z\right)\right)\right)\Rightarrow J_{\upsilon_{m}}\left(z,\tau,\pi\right)\mathbb{Z}_{\vartheta_{m}}\left(\tau\right)
\end{align*}
where
\begin{align*}
J_{\upsilon_{m}}\left(z,\tau,\pi\right)=
\begin{pmatrix}
x' & 0 & 0\\
0 & \nabla_{\theta_{l}}c_{k,x}\left(\tau,\pi;\theta_{l}\right) & 0\\
0 & 0 & \nabla_{\gamma_{k}}\pi_{l}\left(z;\gamma_{k}\right)
\end{pmatrix}
\end{align*}
and I have also used the fact that the mapping $b\rightarrow x'b\left(u\right)$ is linear and therefore Hadamard differentiable at $b\left(\cdot\right)=\beta_{d}\left(\cdot\right)$ tangentially to a set $\mathcal{B}$ with derivative equal to $\phi_{\beta}\left(h\right)=x'h\left(\tau\right)$. By Lemma E.4 in \cite{Chernozhukov2013}, $\sqrt{n_{h}}\int_{\mathcal{Z}}fd\left(\hat{F}_{Z}^{h}-F_{Z}^{h}\right)\Rightarrow\mathbb{Z}_{Z_{h}}\left(f\right)$. Taking these together yields $\sqrt{n}\left(\hat{\upsilon}_{m}\left(z,\tau,\pi,f\right)-\upsilon_{m}\left(z,\tau,\pi,f\right)\right)\Rightarrow\mathbb{Z}_{\upsilon_{m}}\left(z,\tau,\pi,f\right)$, where
\begin{align*}
\mathbb{Z}_{\upsilon_{m}}\left(z,\tau,\pi,f\right)\equiv\begin{pmatrix}
J_{\upsilon_{m}}\left(z,\tau,\pi\right)\mathbb{Z}_{\vartheta_{m}}\left(\tau\right)\\
\sqrt{\lambda_{h}}\mathbb{Z}_{Z_{h}}\left(f\right)
\end{pmatrix}
\end{align*}
in $\ell^{\infty}\left(\mathcal{Z}\mathcal{U}\overline{\mathcal{P}}\mathcal{F}\mathcal{M}\right)$.

\subsection{Proof of Theorem~\ref{thm:bootstrap}}

First, I show the distribution of $\vartheta_{d}^{*}\left(\tau\right)$. Using the same arguments used in Theorem~\ref{thm:asym} and Assumption~\ref{assum:weights}, it follows that
\begin{align}\label{eq:thetastar}
\sqrt{n}\left(\hat{\vartheta}_{d}^{*}\left(\tau\right)-\vartheta_{d}\left(\tau\right)\right)&=F_{d}^{I}\left(\tau\right)\left(I-\int_{\varepsilon}^{1-\varepsilon}D_{d}\left(u\right)F_{d}^{I}\left(u\right)du\right)^{-1}\int_{\varepsilon}^{1-\varepsilon}D_{d}\left(u\right)F_{d}^{I}\left(u\right)C_{d}\left(u\right)\mathbb{G}_{n}^{*}\psi_{d}\left(u\right)du\nonumber\\
&+F_{d}^{I}\left(\tau\right)C_{d}\left(\tau\right)\mathbb{G}_{n}^{*}\psi_{d}\left(\tau\right)+o_{P}\left(1\right)
\end{align}
uniformly in $\ell^{\infty}\left(\mathcal{U}\right)$, where $f\mapsto\mathbb{G}_{n}^{*}\left[f\left(E\right)\right]\equiv\frac{1}{\sqrt{n}}\sum_{i=1}^{n}W_{i,r}f\left(E\right)-\mathbb{E}\left(f\left(E\right)\right)$. Therefore, $\sqrt{n}\left(\hat{\vartheta}_{d}^{*}-\vartheta_{d}\right)\Rightarrow\mathbb{Z}_{\vartheta_{d}}^{*}\left(\tau\right)\equiv\sqrt{\omega_{0}+1}\mathbb{Z}_{\vartheta_{d}}\left(\tau\right)$, a zero-mean Gaussian process with covariance $\left(\omega_{0}+1\right)\Sigma_{\vartheta_{d}}\left(\tau,\tau'\right)$.

Now subtract Equation~\ref{eq:a} from Equation~\ref{eq:thetastar} to get
\begin{align}
&\sqrt{n}\left(\hat{\vartheta}_{d}^{*}\left(\tau\right)-\hat{\vartheta}_{d}\left(\tau\right)\right)\nonumber\\
&=F_{d}^{I}\left(\tau\right)\left(I-\int_{\varepsilon}^{1-\varepsilon}D_{d}\left(u\right)F_{d}^{I}\left(u\right)du\right)^{-1}\int_{\varepsilon}^{1-\varepsilon}D_{d}\left(u\right)F_{d}^{I}\left(u\right)C_{d}\left(u\right)\frac{1}{\sqrt{n}}\sum_{i}^{n}\left(W_{i,r}-1\right)\psi_{d}\left(u\right)du\nonumber\\
&+F_{d}^{I}\left(\tau\right)C_{d}\left(\tau\right)\frac{1}{\sqrt{n}}\sum_{i}^{n}\left(W_{i,r}-1\right)\psi_{d}\left(\tau\right)+o_{P}\left(1\right)
\end{align}
uniformly in $\ell^{\infty}\left(\mathcal{U}\right)$. 

By Assumption~\ref{assum:weights}, it follows that $\sqrt{\frac{n}{\omega_{0}}}\left(\hat{\vartheta}_{d,r}^{*}\left(\tau\right)-\hat{\vartheta}_{d}\left(\tau\right)\right)\Rightarrow\mathbb{Z}_{\vartheta_{d}}\left(\tau\right)$. By the functional delta method, Theorem~\ref{thm:asymgc}, and the previous result, it is straightforward to show that:
\begin{align*}
\sqrt{n}\left(\Delta^{m,m'}\hat{\mathbb{E}}^{*}\left[Y_{r}|S=1\right]-\Delta^{m,m'}\mathbb{E}\left[Y|S=1\right]\right)&\Rightarrow\sqrt{\omega_{0}+1}\mathbb{Z}_{\Delta^{mm'}Y|S=1}\\
\sqrt{\frac{n}{\omega_{0}}}\left(\Delta^{m,m'}\hat{\mathbb{E}}^{*}\left[Y_{r}|S=1\right]-\Delta^{m,m'}\hat{\mathbb{E}}\left[Y|S=1\right]\right)&\Rightarrow\mathbb{Z}_{\Delta^{mm'}Y|S=1}
\end{align*}
and
\begin{align*}
\sqrt{n}\left(\Delta^{m,m'}\hat{Q}_{Y_{r}|S=1}^{*}\left(\tau\right)-\Delta^{m,m'}Q_{Y|S=1}\left(\tau\right)\right)&\Rightarrow\sqrt{\omega_{0}+1}\mathbb{Z}_{Q|S=1,mm'}\left(\tau\right)\\
\sqrt{\frac{n}{\omega_{0}}}\left(\Delta^{m,m'}\hat{Q}_{Y_{r}|S=1}^{*}\left(\tau\right)-\Delta^{m,m'}\hat{Q}_{Y|S=1}\left(\tau\right)\right)&\Rightarrow\mathbb{Z}_{Q|S=1,mm'}\left(\tau\right)
\end{align*}
uniformly in $\ell^{\infty}\left(\mathcal{T}\right)$.

\subsection{Proof of Theorem~\ref{thm:asymanc}}

By the functional delta method,
\begin{align*}
\sqrt{n}\left(\hat{\mathbb{E}}\left[\pi^{m}|S=1\right]-\mathbb{E}\left[\pi^{m}|S=1\right]\right)&=\sqrt{n}\int_{\mathcal{Z}}\left(\hat{\pi}_{l}\left(z\right)-\pi_{l}\left(z\right)\right)dF_{Z}^{h}\left(z\right)+o_{P}\left(1\right)\\
&\Rightarrow\int_{\mathcal{Z}}\mathbb{Z}_{\pi_{l}}\left(z\right)dF_{Z}^{h}\left(z\right)
\end{align*}

Using the functional delta method once more, it follows that
\begin{align*}
\sqrt{n}\Delta^{m,m'}\left(\hat{\mathbb{E}}\left[\pi|S=1\right]-\mathbb{E}\left[\pi|S=1\right]\right)\Rightarrow\mathbb{Z}_{\pi^{m}|S=1}-\mathbb{Z}_{\pi^{m'}|S=1}\equiv\mathbb{Z}_{\Delta^{mm'}\pi|S=1}
\end{align*}
finishing the proof of the first part of the Theorem.

For the second part, by the functional delta method,
\begin{align*}
\sqrt{n}&\left(\hat{\mathbb{E}}\left[U^{m}|S=1\right]-\mathbb{E}\left[U^{m}|S=1\right]\right)\\
&=\sqrt{n}\int_{\mathcal{Z}}\int_{\varepsilon}^{1-\varepsilon}\tau d\left(\hat{G}_{k,x}\left(\tau|\pi_{l}\left(z\right)\right)-G_{k,x}\left(\tau|\pi_{l}\left(z\right)\right)\right)dF_{Z}^{h}\left(z\right)\\
&+\sqrt{n}\int_{\mathcal{Z}}\int_{\varepsilon}^{1-\varepsilon}\tau d\left(\nabla_{\pi}G_{k,x}\left(\tau|\pi_{l}\left(z\right)\right)\right)\left(\hat{\pi}_{l}\left(z\right)-\pi_{l}\left(z\right)\right)dF_{Z}^{h}\left(z\right)\\
&+\sqrt{n}\int_{\mathcal{Z}}\int_{\varepsilon}^{1-\varepsilon}\tau dG_{k,x}\left(\tau|\pi_{l}\left(z\right)\right)d\left(\hat{F}_{Z}^{h}\left(z\right)-F_{Z}^{h}\left(z\right)\right)+o_{P}\left(1\right)\\
&\Rightarrow\int_{\mathcal{Z}}\int_{\varepsilon}^{1-\varepsilon}\tau \frac{1}{\pi_{l}\left(z\right)}\mathbb{Z}_{c_{k}}\left(\tau,\pi_{l}\left(z\right)\right)d\tau dF_{Z}^{h}\left(z\right)\\
&+\int_{\mathcal{Z}}\int_{\varepsilon}^{1-\varepsilon}\tau d\left(\nabla_{\pi}G_{k,x}\left(\tau|\pi_{l}\left(z\right)\right)\right)\mathbb{Z}_{\pi_{l}}\left(z\right)dF_{Z}^{h}\left(z\right)\\
&+\sqrt{\lambda_{h}}\mathbb{Z}_{Z_{h}}\left(\int_{\varepsilon}^{1-\varepsilon}\tau dG_{k,x}\left(\tau|\pi\left(\cdot\right)\right)\right)\equiv\mathbb{Z}_{U^{m}|S=1}
\end{align*}

Applying the functional delta method again,
\begin{align*}
\sqrt{n}\Delta^{m,m'}\left(\hat{\mathbb{E}}\left[U|S=1\right]-\mathbb{E}\left[U|S=1\right]\right)\Rightarrow\mathbb{Z}_{U^{m}|S=1}-\mathbb{Z}_{U^{m'}|S=1}\equiv\mathbb{Z}_{\Delta^{mm'}U|S=1}
\end{align*}
which finishes the proof.
\section{Auxiliary Lemmas}\label{app:lem}

\subsection{Hadamard Derivatives}

\begin{lem}\label{lem:hadamard}

Let $\kappa\equiv\phi_{\kappa}\left(\psi\left(u\right)\right)=\int_{\varepsilon}^{1-\varepsilon}\overline{\kappa}\left(u\right)\psi\left(u\right)du$ and $\kappa\left(u,h_{n}\right)\equiv\phi_{\kappa}\left(\psi\left(u\right)+t_{n}h_{n}\right)$, where $\phi_{\kappa}:\mathbb{D}_{\phi}\subseteq\ell^{\infty}\left(\mathcal{U}\right)\mapsto\mathbb{R}$. Then, for all sequences $\left\{h_{n}\right\}\subset\ell^{\infty}\left(\mathcal{U}\right)$ and $t_{n}\subset\mathbb{R}$ such that $t_{n}\rightarrow0$, $h_{n}\rightarrow h\in\mathcal{C}\left(\mathcal{U}\right)$ as $n\rightarrow\infty$, $\psi\left(u\right)+t_{n}h_{n}\in\mathbb{D}_{\phi}$ for all $n$, the mapping
\begin{align*}
\phi_{\kappa}'\left(h\right)\equiv\int_{\varepsilon}^{1-\varepsilon}\overline{\kappa}\left(u\right)h\left(u\right)du
\end{align*}
is continuous, linear, and satisfies
\begin{align*}
\lim_{n\rightarrow\infty}\left\|\frac{\phi_{\kappa}\left(\overline{\kappa}\left(u\right)+t_{n}h_{n}\left(u\right)\right)-\phi_{\kappa}\left(\overline{\kappa}\left(u\right)\right)}{t_{n}}-\phi_{\kappa}'\left(h\left(u\right)\right)\right\|_{\mathbb{R}}=0
\end{align*}
\end{lem}

\begin{proof}
\begin{align*}
\frac{\phi_{\kappa}\left(\overline{\kappa}\left(u\right)+t_{n}h_{n}\left(u\right)\right)-\phi_{\kappa}\left(\overline{\kappa}\left(u\right)\right)}{t_{n}}&=\frac{\int_{\varepsilon}^{1-\varepsilon}\overline{\kappa}\left(u\right)\left[\psi\left(u\right)+t_{n}h_{n}\left(u\right)-\psi\left(u\right)\right]du}{t_{n}}\\
&=\frac{\int_{\varepsilon}^{1-\varepsilon}\overline{\kappa}\left(u\right)t_{n}h_{n}\left(u\right)}{t_{n}}\\
&=\int_{\varepsilon}^{1-\varepsilon}\overline{\kappa}\left(u\right)h_{n}\left(u\right)du\\
&\rightarrow\int_{\varepsilon}^{1-\varepsilon}\overline{\kappa}\left(u\right)h\left(u\right)du
\end{align*}
in $\ell^{\infty}\left(\mathcal{U}\right)$ as $n\rightarrow\infty$. The desired result follows.
\end{proof}

\begin{lem}\label{lem:hadamardf}
Let $F_{Y|Z,S=1}\left(y|z\right)\equiv\phi_{g,c,\pi}\left(g\left(x,u\right),c_{x}\left(u,\eta\right),\pi\left(z\right)\right)=\int_{\varepsilon}^{1-\varepsilon}\mathbf{1}\left(g\left(x,u\right)\leq y\right)\frac{1}{\pi\left(z\right)}$ $\cdot c_{x}\left(u,\pi\left(z\right)\right)du$ and $F_{Y|Z,S=1}\left(y|z,h_{n}\right)\equiv\phi_{g,c,\pi}\left(g\left(x,u\right)+t_{n}h_{g,n}\left(u,z\right),c_{x}\left(u,\eta\right)+t_{n}h_{c,h}\left(u,z\right),\right.$ $\left.\pi\left(z\right)+t_{n}h_{\pi,n}\left(u,z\right)\right)$, where $\phi:\mathbb{D}_{\phi}\subseteq\ell^{\infty}\left(\mathcal{Z}\mathcal{U}\mathcal{U}\right)\mapsto\ell^{\infty}\left(\mathcal{Y}\mathcal{Z}\right)\equiv\mathbb{E}$. Then, for all sequences $\left\{h_{n}\right\}\subset\ell^{\infty}\left(\mathcal{Z}\mathcal{U}\mathcal{U}\right)$ and $t_{n}\subset\mathbb{R}$ such that $t_{n}\rightarrow0$, $h_{n}\equiv\left(h_{g,n}',h_{c,n}',h_{\pi,n}'\right)'\rightarrow\left(h_{g}',h_{c}',h_{\pi}'\right)\equiv h\in\mathcal{C}\left(\mathcal{Y}\mathcal{Z}\right)$ as $n\rightarrow\infty$, $\left(g\left(x,u\right)+t_{n}h_{g,n}\left(u,z\right),c_{x}\left(u,\eta\right)+t_{n}h_{c,h}\left(u,z\right),\pi\left(z\right)+t_{n}h_{\pi,n}\left(u,z\right)\right)\in\mathbb{D}_{\phi}$ for all $n$, the mapping
\begin{align*}
\phi_{g,c,\pi}'\left(h\right)&\equiv\int_{\varepsilon}^{1-\varepsilon}\mathbf{1}\left(g\left(x,u\right)\leq y\right)\left[\frac{\nabla_{v}c_{x}\left(u,\pi\left(z\right)\right)}{\pi\left(z\right)}-\frac{c_{x}\left(u,\pi\left(z\right)\right)}{\pi\left(z\right)^{2}}\right]h_{\pi}\left(u,z\right)du\\
&+\int_{\varepsilon}^{1-\varepsilon}\mathbf{1}\left(g\left(x,u\right)\leq y\right)\frac{h_{c}\left(u,z\right)}{\pi\left(z\right)}du\\
&-f_{Y|Z,S=1}\left(y|z\right)h_{g}\left(G_{x}^{-1}\left(F_{Y|Z,S=1}\left(y|z\right),\pi\left(z\right)\right),z\right)
\end{align*}
is continuous, linear, and satisfies
\begin{align*}
\lim_{n\rightarrow\infty}&\left\|\frac{\phi_{g,c,\pi}\left(g\left(x,u\right)+t_{n}h_{g,n}\left(u,z\right),c_{x}\left(u,\eta\right)+t_{n}h_{c,h}\left(u,z\right),\pi\left(z\right)+t_{n}h_{\pi,n}\left(u,z\right)\right)}{t_{n}}\right.\\
&\left.\frac{-\phi_{g,c,\pi}\left(g\left(x,u\right),c_{x}\left(u,\eta\right),\pi\left(z\right)\right)}{t_{n}}-\phi_{g,c,\pi}'\left(h\left(u,z\right)\right)\right\|_{\mathbb{E}}=0
\end{align*}
\end{lem}

\begin{proof}
The first part of the proof is partly based on the proof to Proposition 2 in \cite{Chernozhukov2010}.

For any $\delta>0\exists\epsilon>0$: for $u\in B_{\epsilon}\left(G_{x}^{-1}\left(F_{Y|Z,S=1}\left(y|z\right),\pi\left(z\right)\right)\right)$ and for small enough $t\geq0$
\begin{align*}
\mathbf{1}\left(g\left(x,u\right)+t_{n}h_{g,n}\left(u,z\right)\leq y\right)\leq\mathbf{1}\left(g\left(x,u\right)+t_{n}\left(h_{g}\left(G_{x}^{-1}\left(F_{Y|Z,S=1}\left(y|z\right),\pi\left(z\right)\right),z\right)-\delta\right)\leq y\right)
\end{align*}
whereas $\forall u\notin B_{\epsilon}\left(G_{x}^{-1}\left(F_{Y|Z,S=1}\left(y|z\right),\pi\left(z\right)\right)\right)$
\begin{align*}
\mathbf{1}\left(g\left(x,u\right)+t_{n}h_{g,n}\left(u,z\right)\leq y\right)\leq\mathbf{1}\left(g\left(x,u\right)\leq y\right)
\end{align*}
Therefore, for small enough $t\geq0$
\begin{align}
\frac{1}{t_{n}}&\left[\int_{\varepsilon}^{1-\varepsilon}\mathbf{1}\left(g\left(x,u\right)+t_{n}h_{g,n}\left(u,z\right)\leq y\right)\frac{c_{x}\left(u,\pi\left(z\right)\right)}{\pi\left(z\right)}du\right.\nonumber\\
&-\left.\int_{\varepsilon}^{1-\varepsilon}\mathbf{1}\left(g\left(x,u\right)\leq y\right)\frac{c_{x}\left(u,\pi\left(z\right)\right)}{\pi\left(z\right)}du\right]\label{eq:lemhad1}\\
&\leq\frac{1}{t_{n}}\int_{B_{\epsilon}}\left[\mathbf{1}\left(g\left(x,u\right)+t_{n}h_{g,n}\left(u,z\right)\leq y\right)-\mathbf{1}\left(g\left(x,u\right)\leq y\right)\right]\frac{c_{x}\left(u,\pi\left(z\right)\right)}{\pi\left(z\right)}du\label{eq:lemhad2}
\end{align}
where $B_{\epsilon}$ is shorthand for $B_{\epsilon}\left(G_{x}^{-1}\left(F_{Y|Z,S=1}\left(y|z\right),\pi\left(z\right)\right)\right)$. By the change of variable $\tilde{y}=F_{Y|Z,S=1}^{-1}\left(G_{x}\left(u|\pi\left(z\right)\right)|z\right)$, where $G_{x}^{-1}\left(u,\pi_{d}\left(z\right)\right)$ denotes the inverse of $G_{x}\left(u,\pi_{d}\left(z\right)\right)$ with respect to its first argument, is equal to
\begin{align*}
\frac{1}{t_{n}}\int_{J\cap\left[y,y-t\left(h\left(G_{x}^{-1}\left(F_{Y|Z,S=1}\left(y|z\right),\pi\left(z\right)\right)|z\right)-\delta\right)\right]}f_{Y|Z,S=1}\left(\tilde{y}|z\right)d\tilde{y}
\end{align*}
where $J$ is the image of $B_{\epsilon}\left(G_{x}^{-1}\left(F_{Y|Z,S=1}\left(y|z\right),\pi\left(z\right)\right)\right)$ under $u\mapsto g\left(x,u\right)$. The change of variables is possible because $g_{d}\left(x,C_{d,x}^{-1}\left(u,\pi_{d}\left(z\right)\right)\right)$ is a bijection between $B_{\epsilon}\left(G_{x}^{-1}\left(F_{Y|Z,S=1}\left(y|z\right),\right.\right.$ $\left.\left.\pi\left(z\right)\right)\right)$ and $J$.

Fix $\epsilon>0$ for $t_{n}\rightarrow0$. Then, we have that $J\cap\left[y,y-t\left(h\left(G_{x}^{-1}\left(F_{Y|Z,S=1}\left(y|z\right),\pi\left(z\right)\right)|z\right)-\delta\right)\right]=\left[y,y-t_{n}\left(h\left(G_{x}^{-1}\left(F_{Y|Z,S=1}\left(y|z\right),\pi\left(z\right)\right)|z\right)-\delta\right)\right]$ and $f_{Y|Z,S=1}\left(\tilde{y}|z\right)\rightarrow f_{Y|Z,S=1}\left(y|z\right)$ as $F_{Y|Z,S=1}\left(\tilde{y}|z\right)\rightarrow F_{Y|Z,S=1}\left(y|z\right)$ in $\mathbb{E}$. Thus, Equation~\ref{eq:lemhad2} is no greater than
\begin{align*}
-f_{Y|Z,S=1}\left(y|z\right)\left(h_{g}\left(G_{x}^{-1}\left(F_{Y|Z,S=1}\left(y|z\right),\pi\left(z\right)\right),z\right)-\delta\right)+o\left(1\right)
\end{align*}

By a similar argument,
\begin{align*}
-f_{Y|Z,S=1}\left(y|z\right)\left(h_{g}\left(G_{x}^{-1}\left(F_{Y|Z,S=1}\left(y|z\right),\pi\left(z\right)\right),z\right)+\delta\right)+o\left(1\right)
\end{align*}
bounds Equation~\ref{eq:lemhad1} from below. Since $\delta>0$ can be made arbitrarily small, the desired result follows.

To show that the result holds uniformly in $\left(y,z\right)\in K$, a compact subset of $\mathcal{Y}\mathcal{Z}$, we use Lemma B.4 in \cite{Chernozhukov2013}. Take a sequence $\left(y_{t},z_{t}\right)$ in $K$ that converges to $\left(y,z\right)\in K$. Then, the preceding argument applies to this sequence, since the function $-f_{Y|Z,S=1}\left(y|z\right)\left(h_{g}\left(G_{x}^{-1}\left(F_{Y|Z,S=1}\left(y|z\right),\pi\left(z\right)\right),z\right)\right)$ is uniformly continuous on $K$. This result follows by the assumed continuity of $h_{g}\left(u,x\right)$, $F_{Y|Z,S=1}\left(y|z\right)$, and $f_{Y|Z,S=1}\left(y|z\right)$ in both their arguments, as well as the compactness of $K$.

With some abuse of notation, let $h_{g,n}$, $h_{c,n}$, and $h_{\pi,n}$ respectively denote $h_{g,n}\left(u,z\right)$, $h_{c,n}\left(u,z\right)$, and $h_{\pi,n}\left(u,z\right)$. Doing some algebra, it is possible to show that
\begin{align*}
\frac{1}{t_{n}}&\left[F_{Y|Z,S=1}\left(y|z,h_{n}\right)-F_{Y|Z,S=1}\left(y|z\right)\right]=\\
&\frac{1}{t_{n}}\int_{\varepsilon}^{1-\varepsilon}\mathbf{1}\left(g\left(x,u\right)+t_{n}h_{g,n}\leq y\right)\frac{c_{x}\left(u,\pi\left(z\right)+t_{n}h_{\pi,n}\right)+t_{n}h_{c,n}}{\pi\left(z\right)+t_{n}h_{\pi,n}}du\\
&-\frac{1}{t_{n}}\int_{\varepsilon}^{1-\varepsilon}\mathbf{1}\left(g\left(x,u\right)+t_{n}h_{g,n}\leq y\right)\frac{c_{x}\left(u,\pi\left(z\right)\right)+t_{n}h_{c,n}}{\pi\left(z\right)}du\\
&+\frac{1}{t_{n}}\int_{\varepsilon}^{1-\varepsilon}\mathbf{1}\left(g\left(x,u\right)+t_{n}h_{g,n}\leq y\right)\frac{c_{x}\left(u,\pi\left(z\right)\right)+t_{n}h_{c,n}}{\pi\left(z\right)}du\\
&-\frac{1}{t_{n}}\int_{\varepsilon}^{1-\varepsilon}\mathbf{1}\left(g\left(x,u\right)+t_{n}h_{g,n}\leq y\right)\frac{c_{x}\left(u,\pi\left(z\right)\right)}{\pi\left(z\right)}du\\
&+\frac{1}{t_{n}}\int_{\varepsilon}^{1-\varepsilon}\mathbf{1}\left(g\left(x,u\right)+t_{n}h_{g,n}\leq y\right)\frac{c_{x}\left(u,\pi\left(z\right)\right)}{\pi\left(z\right)}du\\
&-\frac{1}{t_{n}}\int_{\varepsilon}^{1-\varepsilon}\mathbf{1}\left(g\left(x,u\right)\leq y\right)\frac{c_{x}\left(u,\pi\left(z\right)\right)}{\pi\left(z\right)}du\\
&=\frac{1}{t_{n}}\int_{\varepsilon}^{1-\varepsilon}\mathbf{1}\left(g\left(x,u\right)+t_{n}h_{g,n}\leq y\right)\frac{c_{x}\left(u,\pi\left(z\right)+t_{n}h_{\pi,n}\right)-c_{x}\left(u,\pi\left(z\right)\right)}{\pi\left(z\right)+t_{n}h_{\pi,n}}du\\
&-\frac{1}{t_{n}}\int_{\varepsilon}^{1-\varepsilon}\mathbf{1}\left(g\left(x,u\right)+t_{n}h_{g,n}\leq y\right)\frac{t_{n}h_{\pi}c_{x}\left(u,\pi\left(z\right)\right)+t_{n}^{2}h_{c,n}h_{\pi,n}}{\left(\left(\pi\left(z\right)+t_{n}h_{\pi,n}\right)\pi\left(z\right)\right)\pi\left(z\right)}du\\
&+\frac{1}{t_{n}}\int_{\varepsilon}^{1-\varepsilon}\mathbf{1}\left(g\left(x,u\right)+t_{n}h_{g,n}\leq y\right)\frac{t_{n}h_{c,n}}{\pi\left(z\right)}du\\
&+\frac{1}{t_{n}}\int_{\varepsilon}^{1-\varepsilon}\left[\mathbf{1}\left(g\left(x,u\right)+t_{n}h_{g,n}\leq y\right)-\mathbf{1}\left(g\left(x,u\right)\leq y\right)\right]\frac{c_{x}\left(u,\pi\left(z\right)\right)}{\pi\left(z\right)}du\\
&\rightarrow\phi_{g,c,\pi}'\left(h\right)
\end{align*}
in $\ell^{\infty}\left(\mathcal{Y}\mathcal{Z}\right)$ as $n\rightarrow\infty$. The desired result follows.
\end{proof}

\begin{lem}\label{lem:hadamarde}
Let $\mathbb{E}\left[Y|z,S=1\right]\equiv\tilde{\phi}_{g,c,\pi}\left(g\left(x,u\right),c_{x}\left(u,\eta\right),\pi\left(z\right)\right)=\int_{\varepsilon}^{1-\varepsilon}g\left(x,u\right)\frac{1}{\pi\left(z\right)}c_{x}\left(u,\pi\left(z\right)\right)du$ and $\mathbb{E}\left[Y|z,S=1,h_{n}\right]\equiv\tilde{\phi}_{g,c,\pi}\left(g\left(x,u\right)+t_{n}h_{g,n}\left(u,z\right),c_{x}\left(u,\eta\right)+t_{n}h_{c,h}\left(u,z\right),\pi\left(z\right)+t_{n}h_{\pi,n}\left(u,z\right)\right)$, where $\tilde{\phi}:\mathbb{D}_{\tilde{\phi}}\subseteq\ell^{\infty}\left(\mathcal{Z}\mathcal{U}\mathcal{U}\right)\mapsto\ell^{\infty}\left(\mathcal{Z}\right)\equiv\mathbb{E}$. Then, for all sequences $\left\{h_{n}\right\}\subset\ell^{\infty}\left(\mathcal{Z}\mathcal{U}\mathcal{U}\right)$ and $t_{n}\subset\mathbb{R}$ such that $t_{n}\rightarrow0$, $h_{n}\equiv\left(h_{g,n}',h_{c,n}',h_{\pi,n}'\right)'\rightarrow\left(h_{g}',h_{c}',h_{\pi}'\right)\equiv h\in\mathcal{C}\left(\mathcal{Z}\right)$ as $n\rightarrow\infty$, $\left(g\left(x,u\right)+t_{n}h_{g,n}\left(u,z\right),c_{x}\left(u,\eta\right)+t_{n}h_{c,h}\left(u,z\right),\pi\left(z\right)+t_{n}h_{\pi,n}\left(u,z\right)\right)\in\mathbb{D}_{\tilde{\phi}}$ for all $n$, the mapping
\begin{align*}
\tilde{\phi}_{g,c,\pi}'\left(h\right)&\equiv\int_{\varepsilon}^{1-\varepsilon}g\left(x,u\right)\left[\frac{\nabla_{v}c_{x}\left(u,\pi\left(z\right)\right)}{\pi\left(z\right)}-\frac{c_{x}\left(u,\pi\left(z\right)\right)}{\pi\left(z\right)^{2}}\right]h_{\pi}\left(u,z\right)du\\
&+\int_{\varepsilon}^{1-\varepsilon}g\left(x,u\right)\frac{h_{c}\left(u,z\right)}{\pi\left(z\right)}du\\
&+\int_{\varepsilon}^{1-\varepsilon}\frac{c_{x}\left(u,\pi\left(z\right)\right)}{\pi\left(z\right)}h_{g}\left(u,z\right)du
\end{align*}
is continuous, linear, and satisfies
\begin{align*}
\lim_{n\rightarrow\infty}&\left\|\frac{\tilde{\phi}_{g,c,\pi}\left(g\left(x,u\right)+t_{n}h_{g,n}\left(u,z\right),c_{x}\left(u,\eta\right)+t_{n}h_{c,h}\left(u,z\right),\pi\left(z\right)+t_{n}h_{\pi,n}\left(u,z\right)\right)}{t_{n}}\right.\\
&\left.\frac{-\tilde{\phi}_{g,c,\pi}\left(g\left(x,u\right),c_{x}\left(u,\eta\right),\pi\left(z\right)\right)}{t_{n}}-\tilde{\phi}_{g,c,\pi}'\left(h\left(u,z\right)\right)\right\|_{\mathbb{E}}=0
\end{align*}
\end{lem}

\begin{proof}
As in the proof of Lemma~\ref{lem:hadamardf}, after some algebra it is possible to show that
\begin{align*}
\frac{1}{t_{n}}&\left[\mathbb{E}\left[Y|z,S=1\right]-\mathbb{E}\left[Y|z,S=1,h_{n}\right]\right]=\\
&\frac{1}{t_{n}}\int_{\varepsilon}^{1-\varepsilon}\left(g\left(x,u\right)+t_{n}h_{g,n}\right)\frac{c_{x}\left(u,\pi\left(z\right)+t_{n}h_{\pi,n}\right)+t_{n}h_{c,n}}{\pi\left(z\right)+t_{n}h_{\pi,n}}du\\
&-\frac{1}{t_{n}}\int_{\varepsilon}^{1-\varepsilon}\left(g\left(x,u\right)+t_{n}h_{g,n}\right)\frac{c_{x}\left(u,\pi\left(z\right)\right)+t_{n}h_{c,n}}{\pi\left(z\right)}du\\
&+\frac{1}{t_{n}}\int_{\varepsilon}^{1-\varepsilon}\left(g\left(x,u\right)+t_{n}h_{g,n}\right)\frac{c_{x}\left(u,\pi\left(z\right)\right)+t_{n}h_{c,n}}{\pi\left(z\right)}du\\
&-\frac{1}{t_{n}}\int_{\varepsilon}^{1-\varepsilon}\left(g\left(x,u\right)+t_{n}h_{g,n}\right)\frac{c_{x}\left(u,\pi\left(z\right)\right)}{\pi\left(z\right)}du\\
&+\frac{1}{t_{n}}\int_{\varepsilon}^{1-\varepsilon}\left(g\left(x,u\right)+t_{n}h_{g,n}\right)\frac{c_{x}\left(u,\pi\left(z\right)\right)}{\pi\left(z\right)}du\\
&-\frac{1}{t_{n}}\int_{\varepsilon}^{1-\varepsilon}g\left(x,u\right)\frac{c_{x}\left(u,\pi\left(z\right)\right)}{\pi\left(z\right)}du\\
&=\frac{1}{t_{n}}\int_{\varepsilon}^{1-\varepsilon}\left(g\left(x,u\right)+t_{n}h_{g,n}\right)\frac{c_{x}\left(u,\pi\left(z\right)+t_{n}h_{\pi,n}\right)-c_{x}\left(u,\pi\left(z\right)\right)}{\pi\left(z\right)+t_{n}h_{\pi,n}}du\\
&-\frac{1}{t_{n}}\int_{\varepsilon}^{1-\varepsilon}\left(g\left(x,u\right)+t_{n}h_{g,n}\right)\frac{t_{n}h_{\pi}c_{x}\left(u,\pi\left(z\right)\right)+t_{n}^{2}h_{c,n}h_{\pi,n}}{\left(\left(\pi\left(z\right)+t_{n}h_{\pi,n}\right)\pi\left(z\right)\right)\pi\left(z\right)}du\\
&+\frac{1}{t_{n}}\int_{\varepsilon}^{1-\varepsilon}\left(g\left(x,u\right)+t_{n}h_{g,n}\right)\frac{t_{n}h_{c,n}}{\pi\left(z\right)}du\\
&+\frac{1}{t_{n}}\int_{\varepsilon}^{1-\varepsilon}t_{n}h_{g,n}\frac{c_{x}\left(u,\pi\left(z\right)\right)}{\pi\left(z\right)}du\\
&\rightarrow\phi_{g,c,\pi}'\left(h\right)
\end{align*}
in $\ell^{\infty}\left(\mathcal{Z}\right)$ as $n\rightarrow\infty$. The desired result follows.
\end{proof}

\subsection{Solution to the Fredholm Integral Equation}

\begin{lem}\label{lem:fredholm}
Let $L\left(\tau\right)=M_{1}\left(\tau\right)L\left(\tau\right)+M_{2}\left(\tau\right)+\int_{0}^{1}M_{3}\left(u\right)L\left(u\right)du$ be a Fredholm integral equation of the second kind. Moreover, define $\tilde{M}_{2}\left(\tau\right)\equiv\left(I-M_{1}\left(\tau\right)\right)^{-1}M_{2}\left(\tau\right)$ and $\tilde{M}_{3}\left(\tau\right)\equiv M_{3}\left(\tau\right)\left(I-M_{1}\left(\tau\right)\right)^{-1}$. Let

\begin{enumerate}[(i)]
\item $I-M_{1}\left(\tau\right)$ is invertible $\forall\tau\in\left[0,1\right]$
\item $\lim_{n\rightarrow\infty}\left[\int_{0}^{1}\tilde{M}_{3}\left(u\right)du\right]^{n}=0$
\end{enumerate}

Under (i)-(ii), the solution to this equation is given by
\begin{align*}
L\left(\tau\right)=\tilde{M}_{2}\left(\tau\right)+\left(I-M_{1}\left(\tau\right)\right)^{-1}\left(I-\int_{0}^{1}\tilde{M}_{3}\left(u\right)du\right)^{-1}\int_{0}^{1}\tilde{M}_{3}\left(u\right)M_{2}\left(u\right)du
\end{align*}
\end{lem}

\begin{proof}
\begin{align*}
L\left(\tau\right)&=M_{1}\left(\tau\right)L\left(\tau\right)+M_{2}\left(\tau\right)+\int_{0}^{1}M_{3}\left(u\right)L\left(u\right)du\\
&=\tilde{M}_{2}\left(\tau\right)+\left(I-M_{1}\left(\tau\right)\right)^{-1}\int_{0}^{1}M_{3}\left(u\right)L\left(u\right)du\\
&=\tilde{M}_{2}\left(\tau\right)+\left(I-M_{1}\left(\tau\right)\right)^{-1}\sum_{n=0}^{\infty}\left[\int_{0}^{1}\tilde{M}_{3}\left(u\right)du\right]^{n}\int_{0}^{1}\tilde{M}_{3}\left(u\right)M_{2}\left(u\right)du\\
&+\lim_{n\rightarrow\infty}\left(I-M_{1}\left(\tau\right)\right)^{-1}\left[\int_{0}^{1}\tilde{M}_{3}\left(u\right)du\right]^{n}\int_{0}^{1}M_{3}\left(u\right)L\left(u\right)du\\
&=\tilde{M}_{2}\left(\tau\right)+\left(I-M_{1}\left(\tau\right)\right)^{-1}\left(I-\int_{0}^{1}\tilde{M}_{3}\left(u\right)du\right)^{-1}\int_{0}^{1}\tilde{M}_{3}\left(u\right)M_{2}\left(u\right)du
\end{align*}
where the second equality follows by (i), the third one by iteratively substituting $L\left(u\right)$ inside the integral, and the fourth one by (ii) and the following result: define $S\equiv\sum_{n=0}^{\infty}C^{n}$, and $A$, $B$ and $C$ be square matrices. Then, $ASB-ACSB=A\left(I-C\right)SB=AB$. If $I-C$ is invertible, then $S=\left(I-C\right)^{-1}$. Premultiply both sides of the equation by $A$ and postmultiply them by $B$ to obtain the desired result.
\end{proof}

\subsection{Argmax Process}

\begin{lem}\label{lem:argmax}
\citep{Chernozhukov2006} Suppose that uniformly in $\pi$ in a compact set $\Pi$ and for a compact set $K$ $\left(i\right)$ $Z_{n}\left(\pi\right)$ is s.t. $Q_{n}\left(Z_{n}\left(\pi\right)|\pi\right)\geq\sup_{z\in K}Q_{n}\left(z|\pi\right)-\epsilon_{n}$, $\epsilon\searrow0$; $Z_{n}\left(\pi\right)\in K\; wp\rightarrow 1$, $\left(ii\right)$ $Z_{\infty}\left(\pi\right)\equiv\arg\sup_{z\in K}Q_{\infty}\left(z|\pi\right)$ is a uniquely defined continuous process in $\ell^{\infty}\left(\Pi\right)$, $\left(iii\right)$ $Q_{n}\left(\tau|\tau\right)\overset{p}{\rightarrow}Q_{\infty}\left(\tau|\tau\right)$ in $\ell^{\infty}\left(K\times\Pi\right)$, where $Q_{\infty}\left(\tau|\tau\right)$ is continuous. Then $Z_{n}\left(\tau\right)=Z_{\infty}\left(\tau\right)+o_{P}\left(1\right)$ in $\ell^{\infty}\left(\Pi\right)$
\end{lem}

\begin{proof}
See \cite{Chernozhukov2006}.
\end{proof}

\subsection{Stochastic Expansion}

\begin{lem}\label{lem:se}
Under Assumptions~\ref{assum:samp}-\ref{assum:fullrank}, the following statements hold:

\begin{enumerate}
\item $\sup_{\left(\beta,\theta,\gamma,m\right)\in\mathcal{B}\times\Theta\times\Gamma\times\mathcal{M}}\left|\mathbb{E}_{n}\left[q_{m}\left(E,\beta,\theta,\gamma,\tau\right)\right]-\mathbb{E}\left[q_{m}\left(E,\beta,\theta,\gamma,\tau\right)\right]\right|=o_{P}\left(1\right)$
\item $\mathbb{G}_{n}R_{m}\left(E,\beta\left(\cdot\right),\theta,\gamma,\tau\right)\Rightarrow\mathbb{Z}_{R_{m}}\left(\cdot\right)$ in $\ell^{\infty}\left(\mathcal{U}\mathcal{M}\right)$, where $\mathbb{Z}_{R_{m}}\left(\tau\right)$ is a zero-mean tight Gaussian process with covariance $\Sigma_{R_{m}}\left(\tau,\tau'\right)$ defined below in the proof. Moreover, for any $\hat{\vartheta}_{m}\left(\tau\right)$ such that $\sup_{\left(\tau,m\right)\in\mathcal{U}\times\mathcal{M}}\left\|\hat{\vartheta}_{m}\left(\tau\right)-\vartheta_{m}\left(\tau\right)\right\|=o_{P}\left(1\right)$, the following holds:
\begin{align*}
\sup_{\left(\tau,m\right)\in\mathcal{U}\mathcal{M}}\left\|\mathbb{G}_{n}R_{m}\left(E,\hat{\beta}\left(\tau\right),\hat{\theta},\hat{\gamma},\tau\right)-\mathbb{G}_{n}R_{m}\left(E,\beta\left(\tau\right),\theta,\gamma,\tau\right)\right\|=o_{P}\left(1\right)
\end{align*}
\end{enumerate}
\end{lem}

\begin{proof}

Let $\mathcal{F}$ be the class of uniformly smooth functions in $z$ with the uniform smoothness order $\omega>\frac{\dim\left(s,z\right)}{2}$ and $\left\|f\left(\tau',z\right)-f\left(\tau,z\right)\right\|<\overline{K}\left(\tau-\tau'\right)^{a}$ for $\overline{K}>0$, $a>0$, $\forall\left(z,\tau,\tau'\right)\forall f\in\mathcal{F}$. The bracketing number of $\mathcal{F}$, by Corollary 2.7.4 in \cite{VanderVaart1996} satisfies
\begin{align*}
\log N_{\left[\cdot\right]}\left(\epsilon,\mathcal{F},L_{2}\left(P\right)\right)=O\left(\epsilon^{-\frac{\dim\left(z\right)}{\omega}}\right)=O\left(\epsilon^{-2-\delta}\right)
\end{align*}
for some $\delta<0$. Therefore, $\mathcal{F}$ is Donsker with a constant envelope. By Corollary 2.7.4, the bracketing number of
\begin{align*}
\mathcal{J}\equiv\left\{\beta\mapsto X'\beta,\beta\in\mathcal{B}\right\}
\end{align*}
satisfies
\begin{align*}
\log N_{\left[\cdot\right]}\left(\epsilon,\mathcal{J},L_{2}\left(P\right)\right)=O\left(\epsilon^{-\frac{\dim\left(s,x\right)}{\omega}}\right)=O\left(\epsilon^{-2-\delta'}\right)
\end{align*}
for some $\delta'<0$. Since the indicator function is bounded and monotone, and the density functions $f_{Y|Z,S=1}^{d}\left(y|z\right)$ are bounded by Assumption~\ref{assum:bound}, the bracketing number of
\begin{align*}
\mathcal{E}\equiv\left\{\beta\mapsto \mathbf{1}\left(Y<X'\beta\right)\mathbf{1}\left(S=1\right),\beta\in\mathcal{B}\right\}
\end{align*}
satisfies
\begin{align*}
\log N_{\left[\cdot\right]}\left(\epsilon,\mathcal{E},L_{2}\left(P\right)\right)=O\left(\epsilon^{-2-\delta'}\right)
\end{align*}

Since $\mathcal{E}$ has a constant envelope, it is Donsker. Now consider the function $G_{x}$. By Assumptions~\ref{assum:parcop} and~\ref{assum:prop}, the mean value theorem can be applied to show
\begin{align*}
\left\|G_{x}\left(\tau,\pi\left(z,\gamma\right);\theta\right)-G_{x}\left(\tau',\pi\left(z,\gamma\right);\theta\right)\right\|&=\left\|\tau-\tau'\right\|\left\|\frac{\partial}{\partial\tau}G_{x}\left(\tau'',\pi\left(z,\gamma\right);\theta\right)\right\|
\end{align*}
for some $\tau''$ between $\tau$ and $\tau'$. By Assumptions~\ref{assum:parcop} and~\ref{assum:prop}, the second term is bounded $\forall z,\tau''$, so it follows that $G_{x}\in\mathcal{F}$.\footnote{To see this, notice that both $\frac{\partial}{\partial\tau}C_{x}\left(\tau,\pi\right)\in\left[0,1\right]$ and $\pi\left(\tau\right)\in\left[0,1\right]$. Hence, it suffices to show that $\lim_{\pi\rightarrow1}\frac{\partial}{\partial\tau}G_{x}\left(\tau,\pi\right)=\lim_{\pi\rightarrow1}C_{x}\left(\tau,\pi\right)<\infty$, where I have used L'H\^{o}pital rule. Since the derivative is bounded by Assumption~\ref{assum:parcop}, the result follows.} Let $\mathcal{U}\equiv\left\{\tau\mapsto\tau\right\}$ and define
\begin{align*}
\mathcal{W}\equiv&\left\{w=\left(\beta,\theta,\gamma,\tau\right)\mapsto R_{m}\left(E,\beta,\theta,\gamma,\tau\right),\left(\beta,\theta,\gamma\right)\in\mathcal{B}\times\Theta\times\Gamma\right\}
\end{align*}

The first subvector of $\mathcal{W}$ is $\mathcal{E}\times\mathcal{F}-\mathcal{U}\times\mathcal{F}$, the second subvector is $\mathcal{E}\times\mathcal{F}-\mathcal{U}\times\mathcal{F}$, and the third subvector is $\mathcal{F}$. Since $\mathcal{W}$ is Lipschitz over $\left(\mathcal{U},\mathcal{F},\mathcal{E}\right)$, it follows that it is Donsker by Theorem 2.10.6 in \cite{VanderVaart1996}. Define
\begin{align*}
w\equiv\left(\beta,\theta,\gamma,\tau\right)\mapsto\mathbb{G}_{n}R_{m}\left(E,\beta,\theta,\gamma,\tau\right)
\end{align*}

$w$ is Donsker in $\ell^{\infty}\left(\mathcal{W}\right)$. Consider the process
\begin{align*}
\tau\mapsto\mathbb{G}_{n}R_{m}\left(E,\beta,\theta,\gamma,\tau\right)
\end{align*}

By the uniform H\"{o}lder continuity of $\tau\mapsto\left(\tau,\beta\left(\tau\right)\right)$ in $\tau$ with respect to the supremum norm, it is also Donsker in $\ell^{\infty}\left(\mathcal{U}\right)$. This, together with Assumption~\ref{assum:gsize} implies 
\begin{align*}
\mathbb{G}_{n}R_{m}\left(E,\beta\left(\tau\right),\theta,\gamma,\tau\right)
\Rightarrow\mathbb{Z}_{R_{m}}\left(\tau\right)
\end{align*}
with covariate function
\begin{align*}
\Sigma_{R_{m}}\left(\tau,\tau'\right)\equiv\mathbb{E}\left[\mathbb{Z}_{R_{m}}\left(\tau\right)\mathbb{Z}_{R_{m'}}\left(\tau'\right)'\right]=
\begin{bmatrix}\Sigma_{R_{mm'}}^{11}\left(\tau,\tau'\right) & \Sigma_{R_{m'}}^{21}\left(\tau'\right)' & 0\\
\Sigma_{R_{m}}^{21}\left(\tau\right) & \Sigma_{R_{mm'}}^{22} & 0\\
0 & 0 & \Sigma_{R_{mm'}}^{33}\end{bmatrix}
\end{align*}
where
\begin{align*}
\Sigma_{R_{mm'}}^{11}\left(\tau,\tau'\right)&=\mathbb{E}\left[S\left(G_{j,X,\tau}\wedge G_{j',X,\tau'}-G_{j,X,\tau}G_{j,X,\tau'}\right)XX'\mathbf{1}\left(D=j\right)\right]\mathbf{1}\left(j=j'\right)\\
&=\frac{1}{\lambda_{j}}\mathbb{E}\left[S\left(G_{j,X,\tau}\wedge G_{j',X,\tau'}-G_{j,X,\tau}G_{j,X,\tau'}\right)XX'|D=j\right]\mathbf{1}\left(j=j'\right)
\end{align*}
\begin{align*}
\Sigma_{R_{m}}^{21}\left(\tau\right)&=\mathbb{E}\left[S\int_{\varepsilon}^{1-\varepsilon}X\varphi\left(u,Z\right)'\left[G_{k,X,\tau}\wedge G_{k',X,u}-G_{k,X,\tau}G_{k',X,u}\right]du\mathbf{1}\left(D=k\right)\right]\mathbf{1}\left(k=j'\right)\\
&=\frac{1}{\lambda_{k}}\mathbb{E}\left[S\int_{\varepsilon}^{1-\varepsilon}X\varphi\left(u,Z\right)'\left[G_{k,X,\tau}\wedge G_{k',X,u}-G_{k,X,\tau}G_{k',X,u}\right]du|D=k\right]\mathbf{1}\left(k=j'\right)
\end{align*}
\begin{align*}
\Sigma_{R_{mm'}}^{22}&=\mathbb{E}\left[\int_{\varepsilon}^{1-\varepsilon}\int_{\varepsilon}^{1-\varepsilon}\varphi\left(u,Z\right)\varphi\left(v,Z\right)'\left[G_{k,X,u}\wedge G_{k',X,v}-G_{k,X,u}G_{k',X,v}\right]dvdu\mathbf{1}\left(D=k\right)\right]\mathbf{1}\left(k=k'\right)\\
&=\frac{1}{\lambda_{k}}\mathbb{E}\left[\int_{\varepsilon}^{1-\varepsilon}\int_{\varepsilon}^{1-\varepsilon}\varphi\left(u,Z\right)\varphi\left(v,Z\right)'\left[G_{k,X,u}\wedge G_{k',X,v}-G_{k,X,u}G_{k',X,v}\right]dvdu|D=k\right]\mathbf{1}\left(k=k'\right)
\end{align*}
\begin{align*}
\Sigma_{R_{mm'}}^{33}&=\mathbb{E}\left[b_{l}\left(S,Z;\gamma_{l}\right)b_{l'}\left(S,Z;\gamma_{l}\right)'\mathbf{1}\left(D=l\right)\right]\mathbf{1}\left(l=l'\right)\\
&=\frac{1}{\lambda_{l}}\mathbb{E}\left[b_{l}\left(S,Z;\gamma_{l}\right)b_{l'}\left(S,Z;\gamma_{l}\right)'|D=l\right]\mathbf{1}\left(l=l'\right)
\end{align*}
where $\wedge$ denotes the minimum between two variables, and $G_{d,X,\tau}\equiv G_{d,X}\left(\tau,\pi_{d}\left(Z,\gamma_{d}\right);\theta_{d}\right)$. Define $\xi$ as the $L_{2}\left(P\right)$ pseudometric on $\mathcal{W}$:
\begin{align*}
\xi\left(\tilde{w},w\right)\equiv\sqrt{\mathbb{E}\left\|R_{m}\left(E,\tilde{\beta},\tilde{\theta},\tilde{\gamma},\tilde{\tau}\right)-R_{m}\left(E,\beta,\theta,\gamma,\tau\right)\right\|^{2}}
\end{align*}

Define $\delta_{n}\equiv\sup_{\tau\in\mathcal{U}}\left|\xi\left(\tilde{w}\left(\tau\right),w\left(\tau\right)\right)\right|_{\tilde{w}\left(\tau\right)=\hat{w}\left(\tau\right)}$. Since $\hat{\vartheta}_{m}\left(\tau\right)\overset{P}{\rightarrow}\vartheta_{m}\left(\tau\right)$ uniformly in $\tau$, $\forall m\in\mathcal{M}$, by Assumption~\ref{assum:bound}, $\delta_{n}\overset{P}{\rightarrow}0$. Therefore, as $\delta_{n}\overset{P}{\rightarrow}0$,
\begin{align*}
&\sup_{\tau\in\mathcal{U}}\left\|\mathbb{G}_{n}R_{m}\left(E,\hat{\beta},\hat{\theta},\hat{\gamma},\tau\right)-\mathbb{G}_{n}R_{m}\left(E,\beta,\theta,\gamma,\tau\right)\right\|\\
&\leq \sup_{\begin{smallmatrix}\xi\left(\tilde{w},w\right)\leq\delta_{n}\\ \tilde{w},w\in\mathcal{W}\end{smallmatrix}}\left\|\mathbb{G}_{n}R_{m}\left(E,\hat{\beta},\hat{\theta},\hat{\gamma},\tau\right)-\mathbb{G}_{n}R_{m}\left(E,\beta,\theta,\gamma,\tau\right)\right\|=o_{P}\left(1\right)
\end{align*}
by stochastic equicontinuity of $w\mapsto\mathbb{G}_{n}R_{m}\left(E,\beta,\theta,\gamma,\tau\right)$, which proves claim 2. Note that the result is holds uniformly in $m\in\mathcal{M}$ because $\mathcal{M}$ is a finite set. To prove claim 1, define
\begin{align*}
\mathcal{A}\equiv\left\{\left(\beta,\theta,\gamma,\tau\right)\mapsto q_{m}\left(E,\beta,\theta,\gamma,\tau\right)\right\}
\end{align*}

By Assumptions~\ref{assum:samp} and~\ref{assum:compact}, $\mathcal{A}$ is bounded, and it is also uniformly Lipschitz over $\mathcal{B}\times\Theta\times\Gamma\times\mathcal{U}$, so by Theorem 2.10.6 in \cite{VanderVaart1996}, $\mathcal{A}$ is Donsker. Hence, the following ULLN holds:
\begin{align*}
\sup_{w\in\mathcal{W}}\left|\mathbb{E}_{n}q_{m}\left(E,\beta,\theta,\gamma,\tau\right)-\mathbb{E}q_{m}\left(E,\beta,\theta,\gamma,\tau\right)\right|\overset{P}{\rightarrow}0
\end{align*}
which gives
\begin{align*}
\sup_{\left(\beta,\theta,\gamma,\tau\right)\in\mathcal{B}\times\Theta\times\Gamma\times\mathcal{U}}\left|\mathbb{E}_{n}q_{m}\left(E,\beta,\theta,\gamma,\tau\right)-\mathbb{E}q_{m}\left(E,\beta,\theta,\gamma,\tau\right)\right|\overset{P}{\rightarrow}0
\end{align*}

Using the fact that the number of total groups is finite, it implies claim 1.
\end{proof}
\section{Identification under a Parametric Assumption}\label{app:param}

Consider the following assumptions: 
\begin{assum}\label{assum:er}
$\left(U,V\right)$ is jointly statistically independent of $Z_{1}$ given $X=x$.
\end{assum}
\begin{assum}\label{assum:cop}
The bivariate distribution of $\left(U,V\right)$ given $X=x$ is absolutely continuous with respect to the Lebesgue measure, with standard uniform marginals and rectangular support. Denote its cumulative distribution function as $C_{d,x}\left(u,v\right)$.
\end{assum}
\begin{assum}\label{assum:cont}
The conditional cdf $F_{Y^{*}|X}\left(y|x\right)$ and its inverse are strictly increasing. In addition, $C_{d,x}\left(u,v\right)$ is strictly increasing in $u$.
\end{assum}
\begin{assum}\label{assum:prop2}
$\pi_{d}\left(Z\right)\equiv\mathbb{P}\left(S=1|D=d,Z\right)>0$ with probability 1.
\end{assum}

Assumptions~\ref{assum:er}-\ref{assum:prop2} are Assumptions A1-A4 from \cite{Arellano2017}. They are not enough to identify the SQF and the copula, but they can be used to put some restrictions on the conditional copula, as they show in Lemma 1. Specifically, their Equation 6 can be written using this paper's notation as:
\begin{align}\label{eq:fg}
F_{Y|D=d,Z,S=1}\left(F_{Y|D=d,Z,S=1}^{-1}\left(\tau|z'\right)|z\right)=G_{d,x}\left(G_{d,x}^{-1}\left(\tau,\pi_{d}\left(z'\right)\right),\pi_{d}\left(z\right)\right)
\end{align}
where $F_{Y|D=d,Z,S=1}\left(\tau|z\right)$ is the cdf of $Y$, conditional on $D=d$, $Z=z$, and $S=1$. Point identification is achieved when they assume either enough variation in the instrument such that one can use an identification at infinity argument, or that the copula is real analytic and the instrument displays some continuous variation. A feasible alternative would be to impose the following parametric assumption:
\begin{assum}\label{assum:param}
The copula $C_{d,x}\left(\tau,\pi\right)$ is known up to a scalar parameter $\theta_{d,x}\in\Theta_{d,x}$, for $d=0,1$ and $\forall x\in\mathcal{X}$. $C_{d,x}:\left(0,1\right)^{2}\rightarrow\left(0,1\right)$ is uniformly continuous and twice continuously differentiable in its arguments and in $\theta_{d,x}$ \textit{a.e.} Moreover, for any $\theta_{1}<\theta_{2}$, $C_{d,x}\left(\tau,\pi;\theta_{2}\right)$ is strictly more stochastically increasing in joint distribution than $C_{d,x}\left(\tau,\pi;\theta_{1}\right)$.
\end{assum}

The identification under this assumption is established by the following proposition:
\begin{prop}\label{prop:id}
Let Assumptions~\ref{assum:er} to~\ref{assum:param} hold, and $x\in\mathcal{X}$. Then, for $d=0,1$, the functions $\left(\tau,\pi\right)\rightarrow G_{d,x}\left(\tau,\pi\right)$ and $\tau\rightarrow g_{d}\left(x,\tau\right)$ are globally identified.
\end{prop}

\begin{proof}
The proof is split in parts. First, I show the local identification of $\theta_{d,x}$, followed by its global identification, and I conclude by showing the identification of the SQF.

Define the functions $M_{d,x}\left(\tau,\theta_{d,x}\right)\equiv G_{d,x}\left(G_{d,x}^{-1}\left(\tau,\pi_{d}\left(z'\right);\theta_{d,x}\right),\pi_{d}\left(z\right);\theta_{d,x}\right)$ and $\phi_{d,x}\left(\tau\right)\equiv F_{Y|D=d.Z,S=1}\left(F_{Y|D=d,Z,S=1}^{-1}\left(\tau|z'\right)|z\right)$. By Equation~\ref{eq:fg}, $M_{d,x}\left(\tau,\theta_{d,x}\right)=\phi_{d,x}\left(\tau\right)$, $\forall x\in\mathcal{X},d=0,1$. Taking the derivative with respect to the copula parameter for a generic value of $\theta$, and dropping the $\left(d,x\right)$ subscript from the functions $M$ and $\phi$ and $d$ from the propensity score for notational simplicity, yields
\begin{align}\label{eq:jacobian}
\nabla_{\theta}M\left(\tau,\theta\right)&=\nabla_{\theta}G\left(G^{-1}\left(\tau,\pi\left(z'\right);\theta\right),\pi\left(z\right);\theta\right)\nonumber\\
&-\nabla_{u}G\left(G^{-1}\left(\tau,\pi\left(z'\right);\theta\right),\pi\left(z\right);\theta\right)\frac{\nabla_{\theta}G\left(G^{-1}\left(\tau,\pi\left(z'\right);\theta\right),\pi\left(z'\right);\theta\right)}{\nabla_{u}G\left(G^{-1}\left(\tau,\pi\left(z'\right);\theta\right),\pi\left(z'\right);\theta\right)}
\end{align}

Because $M\left(\tau,\theta\right)$ holds for any $\tau\in\left(0,1\right)$, there is an continuum of moments that pin down the parameter $\theta$. Instead, consider a finite number of values of $\tau$, given by $\left\{\tau_{1},...,\tau_{T}\right\}$. Local identification holds when the matrix that collects the Jacobian for all values in this set is of full rank, as required by Theorem 6 in \cite{Rothenberg1971}. Because it is a scalar parameter, full rank is attained if $\nabla_{\theta}M\left(\tau,\theta\right)\neq0$ for any of the values of $\tau$, \textit{i.e.},
\begin{align}\label{eq:cond2}
\frac{\nabla_{\theta}G\left(G^{-1}\left(\tau,\pi\left(z'\right);\theta\right),\pi\left(z\right);\theta\right)}{\nabla_{u}G\left(G^{-1}\left(\tau,\pi\left(z'\right);\theta\right),\pi\left(z\right);\theta\right)}-\frac{\nabla_{\theta}G\left(G^{-1}\left(\tau,\pi\left(z'\right);\theta\right),\pi\left(z'\right);\theta\right)}{\nabla_{u}G\left(G^{-1}\left(\tau,\pi\left(z'\right);\theta\right),\pi\left(z'\right);\theta\right)}\neq0
\end{align}

Let $\tau'\equiv G^{-1}\left(\tau,\pi\left(z'\right);\theta\right)\Leftrightarrow \tau=G\left(\tau',\pi\left(z'\right);\theta\right)$. Then, Equation~\ref{eq:cond2} can be rewritten as
\begin{align}\label{eq:cond3}
\frac{\nabla_{\theta}G\left(\tau',\pi\left(z\right);\theta\right)}{\nabla_{u}G\left(\tau',\pi\left(z\right);\theta\right)}-\frac{\nabla_{\theta}G\left(\tau',\pi\left(z'\right);\theta\right)}{\nabla_{u}G\left(\tau',\pi\left(z'\right);\theta\right)}\neq0
\end{align}

By the definition of the conditional copula, $\sfrac{\nabla_{\theta}G\left(\tau,\pi;\theta\right)}{\nabla_{u}G\left(\tau,\pi;\theta\right)}=\sfrac{\nabla_{\theta}C\left(\tau,\pi;\theta\right)}{\nabla_{u}C\left(\tau,\pi;\theta\right)}$, so Equation~\ref{eq:cond3} is equivalent to
\begin{align}\label{eq:cond4}
\frac{\nabla_{\theta}C\left(\tau',\pi\left(z\right);\theta\right)}{\nabla_{u}C\left(\tau',\pi\left(z\right);\theta\right)}-\frac{\nabla_{\theta}C\left(\tau',\pi\left(z'\right);\theta\right)}{\nabla_{u}C\left(\tau',\pi\left(z'\right);\theta\right)}\neq0
\end{align}

This is equivalent to the first equation in Condition 4.7 in \cite{Han2017}. By Lemma 4.1 in \cite{Han2017}, under Assumption~\ref{assum:param}, if the copula $G_{d,x}\left(\tau,\pi\right)$ satisfies Assumption 6 in \cite{Han2017}, then Equation~\ref{eq:cond4} is strictly decreasing in the second argument of the copula. If $\pi\left(z\right)\neq\pi\left(z'\right)$, \textit{i.e.}, if the instrument does not come from a degenerate distribution, then the copula parameter $\theta_{d,x}$ is locally identifiable by Proposition 4.1 in \cite{Han2017}.

For global identification, I begin by showing that it is possible to apply Lemma 4.2 in \cite{Han2017} on a restricted parameter space, extending it subsequently to the entire parameter space. Note that Equation~\ref{eq:jacobian} can be rewritten as
\begin{align*}
\nabla_{\theta}M\left(\tau,\theta\right)=\nabla_{u}G\left(\tau',\pi\left(z\right);\theta\right)\left[\frac{\nabla_{\theta}G\left(\tau',\pi\left(z\right);\theta\right)}{\nabla_{u}G\left(\tau',\pi\left(z\right);\theta\right)}-\frac{\nabla_{\theta}G\left(\tau',\pi\left(z'\right);\theta\right)}{\nabla_{u}G\left(\tau',\pi\left(z'\right);\theta\right)}\right]
\end{align*}
where $\tau'\equiv G^{-1}\left(\tau,\pi\left(z'\right);\theta\right)$. By Lemma 4.1 in \cite{Han2017}, the term in brackets is positive when $\pi\left(z\right)<\pi\left(z'\right)$. Moreover, $\nabla_{u}G\left(\tau,\pi;\theta\right)=\frac{1}{\pi}\nabla_{u}C\left(\tau,\pi;\theta\right)>0$. Therefore, the Jacobian $\nabla_{\theta}M\left(\tau,\theta\right)$ is positive semidefinite if $\pi\left(z\right)<\pi\left(z'\right)$ and negative semidefinite if $\pi\left(z\right)>\pi\left(z'\right)$. In addition, it has full rank for any $\theta$ as long as $\pi\left(z\right)\neq\pi\left(z'\right)$.

Let $\Theta_{c}\subseteq\Theta$ be a bounded open space with half spaces $\Theta_{c_{1}}\equiv\left\{\theta\in\Theta_{c}:\pi\left(z\right)<\pi\left(z'\right)\right\}$, and $\Theta_{c_{2}}\equiv\left\{\theta\in\Theta_{c}:\pi\left(z\right)>\pi\left(z'\right)\right\}$, which are simply connected. Define $\phi_{c_{1}}\left(\tau\right)=M\left(\tau,\Theta_{c_{1}}\right)$ and $\phi_{c_{2}}\left(\tau\right)=M\left(\tau,\Theta_{c_{2}}\right)$, and let $M|_{\Theta_{c_{1}}}:\Theta_{c_{1}}\rightarrow\phi_{c_{1}}$ and $M|_{\Theta_{c_{2}}}:\Theta_{c_{2}}\rightarrow\phi_{c_{2}}$ be the function $M\left(\tau,\cdot\right)$ on its restricted domains.

Because $M|_{\Theta_{c_{1}}}\left(\tau,\cdot\right)$ and $M|_{\Theta_{c_{2}}}\left(\tau,\cdot\right)$ are continuous, the pre-image of a closed set under $M|_{\Theta_{c_{1}}}\left(\tau,\cdot\right)$ and $M|_{\Theta_{c_{2}}}\left(\tau,\cdot\right)$ is closed. Because $\Theta_{c_{1}}$ and $\Theta_{c_{2}}$ are bounded, the pre-image of a bounded set is bounded. Thus, $M|_{\Theta_{c_{1}}}\left(\tau,\cdot\right)$ and $M|_{\Theta_{c_{2}}}\left(\tau,\cdot\right)$ are proper.

Because $\Theta_{c_{1}}$ and $\Theta_{c_{2}}$ are simply connected, $M|_{\Theta_{c_{1}}}\left(\tau,\cdot\right)$ and $M|_{\Theta_{c_{2}}}\left(\tau,\cdot\right)$ are continuous on $\Theta_{c_{1}}$ and $\Theta_{c_{2}}$, respectively, and the Jacobian $\nabla_{\theta}M\left(\tau,\cdot\right)$ is positive semidefinite and negative semidefinite on $\Theta_{c_{1}}$ and $\Theta_{c_{2}}$, respectively, it follows that $\phi_{c_{1}}$ and $\phi_{c_{2}}$ are simply connected.

Also, $\nabla_{\theta}M\left(\tau,\cdot\right)$ has full rank over $\Theta_{c_{1}}$ and $\Theta_{c_{2}}$. Thus, by Lemma 4.2 in \cite{Han2017}, $\phi\left(\tau\right)=M\left(\tau,\theta\right)$ has a unique solution on $\Theta_{c_{1}}$ and $\Theta_{c_{2}}$, respectively. Because there exist $M|_{\Theta_{c_{1}}}^{-1}\left(\tau,\cdot\right)\in\Theta_{c_{1}}$ for $\phi\in\phi_{c_{1}}$ and $M|_{\Theta_{c_{2}}}^{-1}\left(\tau,\cdot\right)\in\Theta_{c_{2}}$ for $\phi\in\phi_{c_{2}}$, $\theta$ is globally identified.

Now let $\Theta_{1}\equiv\left\{\theta\in\Theta:\pi\left(z\right)<\pi\left(z'\right)\right\}$ and $\Theta_{2}\equiv\left\{\theta\in\Theta:\pi\left(z\right)>\pi\left(z'\right)\right\}$ be two simply connected, possibly unbounded spaces. $\Theta_{1}$ and $\Theta_{2}$ can be represented as a countable union of bounded open simply connected sets. \textit{E.g.}, $\Theta_{j}=\cup_{i=1}^{\infty}\Theta_{ji}$, where $\Theta_{ji}$ is a sequence of bounded open simply connected sets in $\Theta_{j}$ such that $\Theta_{j1}\subset\Theta_{j2}\subset...\subset\Theta_{j}$ for $j=1,2$.

Let $\phi_{ji}\left(\tau\right)\equiv M\left(\tau,\Theta_{ji}\right)$ for $i=1,2,...$ and $j=1,2$. Then, $\phi_{j}\left(\tau\right)=M\left(\tau,\Theta_{j}\right)=M\left(\tau,\cup_{i=1}^{\infty}\Theta_{ji}\right)=\cup_{i=1}^{\infty}M\left(\tau,\Theta_{ij}\right)=\cup_{i=1}^{\infty}\phi_{ij}\left(\tau\right)$, and $\phi_{j1}\subset\phi_{j2}\subset...\subset\phi_{j}$. Then, for any given $\phi\in\phi_{j}$, $\exists q:\phi\in\phi_{ji}\forall i\geq q$, so $M|_{\Theta_{c_{j}}}^{-1}\left(\tau,\phi\right)\in\Theta_{ji}\forall i\geq q$, and therefore $M^{-1}\left(\tau,\phi\right)=M|_{\cup_{i=q}^{\infty}\Theta_{ji}}^{-1}\left(\tau,\phi\right)\in\cup_{i=q}^{\infty}\Theta_{ji}=\Theta_{j}$. Because $M^{-1}\left(\tau,\phi\right)$ is the unique solution on $\Theta_{j}$, it is the unique solution of the full system with $\tau=\left\{\tau_{1},...,\tau_{T}\right\}$. Thus, $\theta$ is globally identified in $\Theta_{j}$. 

To show the identification of the SQF, note that by Equation~\ref{eq:fg}, $F_{Y|D=d,Z,S=1}\left(g_{d}\left(x,\tau\right)|z\right)=G_{d,x}\left(\tau,\pi\left(z\right);\theta_{d,x}\right)$. Therefore, one can solve for $g_{d}$ for $d=0,1$, and express it in terms of either observed or identified functions: $g_{d}\left(x,\tau\right)=F_{Y|D=d,Z,S=1}^{-1}\left(G_{d,x}\left(\tau,\pi\left(z\right);\theta_{d,x}\right)\right)$.
\end{proof}
\section{Particular Cases}\label{sec:cases}

Several particular cases are nested by the general model. I discuss how the different estimands change and which methods can be used to estimation them.

\subsection{Additively separable unobserved term}

An interesting case arises when $Y=\left(g_{D}\left(X\right)+\tilde{U}\right)S$, where $\tilde{U}=Q_{U|D,X}\left(U\right)$ denotes the conditional quantile function of the uniformly distributed random variable $U$. For example, if the additive error term $\tilde{U}$ is homoskedastic and normally distributed, then $Q_{U|D,X}\left(U\right)=\Phi^{-1}\left(U\right)$, where $\Phi\left(\cdot\right)$ is the standard normal cdf. Importantly, as long as $g_{j}\left(x\right)$ contains an intercept, $\mathbb{E}(\tilde{U})=0$. This case yields a substantial simplification to the expected outcome:
\begin{align*}
\mathbb{E}\left[Y^{hjkl}|S=1\right]&=\int_{\mathcal{Z}}\int_{0}^{1}\left[g_{j}\left(x\right)+Q_{U|j,x}\left(u\right)\right]dG_{k,x}\left(u,\pi_{l}\left(z\right)\right)dF_{Z}^{h}\left(z\right)\\
&=\int_{\mathcal{X}}g_{j}\left(x\right)dF_{X}^{h}\left(x\right)+\int_{\mathcal{Z}}\int_{0}^{1}Q_{U|j,x}\left(u\right)dG_{k,x}\left(u,\pi_{l}\left(z\right)\right)dF_{Z}^{h}\left(z\right)
\end{align*}
where $F_{X}^{h}\left(x\right)$ is the cdf of $X$ for group $h\in\mathcal{D}$ and $\mathcal{X}$ its support. Similarly,
\begin{align*}
\mathbb{E}\left[Y^{hjkl}\right]=\int_{\mathcal{X}}g_{j}\left(x\right)\pi_{l}\left(z\right)dF_{Z}^{h}\left(z\right)+\int_{\mathcal{Z}}\int_{0}^{1}Q_{U|j,x}\left(u\right)dC_{k,x}\left(u,\pi_{l}\left(z\right)\right)dF_{Z}^{h}\left(z\right)
\end{align*}

Both means can be split into two terms. One depends exclusively on the separable term $g_{d}\left(x\right)$ and the distribution of the covariates, without the instrument. The other depends on the copula, the propensity score and the distribution of the covariates including the instrument, but not on $g_{d}\left(x\right)$. Such a simplification follows from the linearity of the expectation operator, which does not hold for the unconditional quantiles. Note also that a change in the covariates results in a parallel shift of the conditional distribution function, without affecting its shape, limiting the amount of heterogeneity that the model can display.

The estimation of the SQF is simplified because it is comprised of the sum of a potentially nonlinear term that depends exclusively on the covariates, and the unobserved ability. Thus, it is possible to estimate $g_{d}\left(x\right)$ nonparametrically. For example, \cite{Das2003} proposes a two-step series estimator which, in the first step nonparametrically regresses the selection variable on the covariates and the instruments, and in the second step it nonparametrically regresses the outcome variable on the covariates and the correction term.

\subsection{Linear model}

A particular case of the previous one appears when $g_{d}\left(x,u\right)=x'\beta_{d}+\tilde{u}$, where again $\tilde{u}=Q_{u|d,x}\left(u\right)$. This assumption slightly simplifies the expression of the first term of the mean outcome relative to the additively separable unobserved term model. Specifically, $\int_{\mathcal{X}}g_{j}\left(x\right)dF_{X}^{h}\left(x\right)=\mathbb{E}_{h}\left(X\right)'\beta_{j}$, where $\mathbb{E}_{h}$ stands for the expectation with respect to $F_{X}^{h}$.

Another two-step estimator that can be used in this case is the one proposed by \cite{Newey2009}. This is similar to \cite{Das2003}, with the main difference being the fact that the regression in the second step is linear in the covariates whilst keeping the power series on the correction term. Other alternative estimators based on stronger parametric assumptions are those proposed, \textit{e.g.}, by \cite{Heckman1976,Heckman1979,Lee1983}.\footnote{See \cite{Vella1998} for a review of these and other estimators for data with sample selection.}

\subsection{No self-selection}

Usually, it is assumed that the ability of participants is unrelated to their participation decision. Mathematically, the copula is independent: $C_{d,x}\left(u,v\right)=uv$ and $G_{d,x}\left(u,v\right)=u$. This assumption implies that, conditional on participation, there are no differences in the distribution of ability between the members of both groups, so the selection component vanishes in all decompositions. Moreover, the outcome value of participants is comparable to the potential value of non-participants. Hence, absence of self-selection also implies that the participation component vanishes in the decompositions for participants (Equations~\ref{eq:decommeanys} and~\ref{eq:decomqys}). This can be seen in Equations~\ref{eq:meanys} and~\ref{eq:fys}, as the statistics of interest depend on the propensity score exclusively through the copula. However, the same simplification does not apply to the decompositions for the entire population, as both the mean and the unconditional quantiles directly depend on the propensity score.

The combination of this assumption with the additively linear unobserved term yields another meaningful simplification, as both the participation and selection components equal zero. To see this, note that
\begin{align*}
\int_{0}^{1}Q_{U|j,x}\left(u\right)dG_{k,x}\left(u,\pi_{l}\left(z\right)\right)=\int_{0}^{1}Q_{U|j,x}\left(u\right)du=0
\end{align*}
where the second equality follows by construction.

In this case the estimation of the copula is superfluous: because the distribution of the unobserved ability is the same for participants and non-participants, there is no need to account for sample selection in the estimation of the SQF. This opens the possibility of using methods that are consistent under exogeneity. For example, one can use quantile regression and then apply the decomposition proposed by \cite{Machado2005}, distributional regression and follow \cite{Chernozhukov2013}, or one could use reweighting methods, such as \cite{Dinardo1996} or \cite{Firpo2018}.\footnote{A detailed comparison of the advantages and disadvantages of these methods is provided by \cite{Fortin2011}. Furthermore, \cite{Leorato2015} compares quantile and distributional regression.} Additionally, if the linearity assumption is combined with either no self-selection or all participants, the resulting model can be estimated by OLS. Regardless, the decompositions for the entire population still depend on the propensity score.

\subsection{All participants}

A particular case of the previous one takes place when every considered individual is a participant. In this case, the propensity score equals one for all individuals, so $S_{i}=1$ for all $i=1,...,N$. As a consequence, the participation component vanishes. Also, because $C_{d,x}\left(u,1\right)=G_{d,x}\left(u,1\right)=u$, the selection component also disappears. Moreover, because there are no non-participants, the decomposition for participants and non-participants are the same. Additionally, under linearity the decomposition becomes the Oaxaca-Blinder decomposition.
\section{Additional Theoretical Results}\label{app:additional}

The estimators of Equations~\ref{eq:meany},~\ref{eq:fy} and~\ref{eq:qy} are given by
\begin{align*}\label{eq:meanyshat}
\hat{\mathbb{E}}\left[Y^{m}\right]&=\frac{1}{\tilde{n}_{h}}\sum_{i=1}^{n}\int_{\varepsilon}^{1-\varepsilon}\hat{g}_{j}\left(X_{i},u\right)d\hat{C}_{k,x}\left(u,\hat{\pi}_{l}\left(Z_{i}\right)\right)\mathbf{1}\left(D_{i}=h\right)\\
\hat{F}_{Y}^{m}\left(y\right)&=\frac{1}{\tilde{n}_{h}}\sum_{i=1}^{n}\left[\varepsilon+\int_{\varepsilon}^{1-\varepsilon}\mathbf{1}\left(\hat{g}_{j}\left(X_{i},u\right)\leq y\right)d\hat{C}_{k,x}\left(u,\hat{\pi}_{l}\left(Z_{i}\right)\right)+\left(1-\pi_{l}\left(Z_{i}\right)\right)\right]\mathbf{1}\left(D_{i}=h\right)\\
\hat{Q}_{Y}^{m}\left(\tau\right)&=\inf\left\{y:\tau\leq \hat{F}_{Y}^{m}\left(y\right)\right\}
\end{align*}
where $\tilde{n}_{h}=\sum_{i=1}^{n}\mathbf{1}\left(D_{i}=h\right)$. Under the same condition required for Theorem~\ref{thm:asymgc} to hold, these estimators are consistent and asymptotically Gaussian:
\begin{thm}\label{thm:asymadd}
Let the estimator $\hat{\upsilon}_{m}\left(z,\tau,\eta\right)$ satisfy Condition~\ref{con:genmod}. Under Assumptions~\ref{assum:samp}-\ref{assum:bound}, the following hold for all $\left(m,m'\right)$:
\begin{align*}
\sqrt{n}\Delta^{m,m'}\left(\hat{\mathbb{E}}\left[Y\right]-\mathbb{E}\left[Y\right]\right)\Rightarrow\mathbb{Z}_{\Delta^{mm'}Y}
\end{align*}
where $\mathbb{Z}_{\Delta^{mm'}Y}$ is a zero-mean Gaussian random variable, and
\begin{align*}
\sqrt{n}\Delta^{m,m'}\left(\hat{Q}_{Y}\left(\tau\right)-Q_{Y}\left(\tau\right)\right)\Rightarrow\mathbb{Z}_{Q,mm'}\left(\tau\right)
\end{align*}
where $\mathbb{Z}_{Q,mm'}\left(\tau\right)$ is a zero-mean Gaussian process.
\end{thm}

\begin{proof}
Using the same argument used in Theorem~\ref{thm:asymgc}, it is straightforward to show that
\begin{align*}
\sqrt{n}\Delta^{m,m'}\left(\hat{\mathbb{E}}\left[Y\right]-\mathbb{E}\left[Y\right]\right)\Rightarrow\mathbb{Z}_{Y^{m}}-\mathbb{Z}_{Y^{m'}}\equiv\mathbb{Z}_{\Delta^{mm'}Y}
\end{align*}
where
\begin{align*}
\mathbb{Z}_{\Delta^{mm'}Y}&=\int_{\mathcal{Z}}\int_{\varepsilon}^{1-\varepsilon}\mathbb{Z}_{g_{j}}\left(\tau,x\right)dC_{k,x}\left(\tau|\pi_{l}\left(z\right)\right)dF_{Z}^{h}\left(z\right)\\
&+\int_{\mathcal{Z}}\int_{\varepsilon}^{1-\varepsilon}g_{j}\left(x,\tau\right)\mathbb{Z}_{c_{k}}\left(\tau,\pi_{l}\left(z\right)\right)d\tau dF_{Z}^{h}\left(z\right)\\
&+\int_{\mathcal{Z}}\int_{\varepsilon}^{1-\varepsilon}g_{j}\left(x,\tau\right)d\left(\nabla_{\pi}C_{k,x}\left(\tau|\pi_{l}\left(z\right)\right)\right)\mathbb{Z}_{\pi_{l}}\left(z\right)dF_{Z}^{h}\left(z\right)\\
&+\sqrt{\lambda_{h}}\mathbb{Z}_{Z_{h}}\left(\int_{\varepsilon}^{1-\varepsilon}g_{j}\left(\cdot,\tau\right)dC_{k,x}\left(\tau|\pi_{l}\left(\cdot\right)\right)\right)
\end{align*}

Similarly,
\begin{align*}
\sqrt{n}\Delta^{m,m'}\left(\hat{Q}_{Y}\left(\tau\right)-Q_{Y}\left(\tau\right)\right)\Rightarrow\mathbb{Z}_{Q_{Y}^{m}}\left(\tau\right)-\mathbb{Z}_{Q_{Y}^{m'}}\left(\tau\right)\equiv\mathbb{Z}_{Q,mm'}\left(\tau\right)
\end{align*}
where
\begin{align*}
\mathbb{Z}_{Q_{Y}^{m}}\left(\tau\right)&=-\frac{\mathbb{Z}_{F_{Y}^{m}}\left(Q_{Y}^{m}\left(\tau\right)\right)}{f_{Y}^{m}\left(Q_{Y}^{m}\left(\tau\right)\right)}\\
\mathbb{Z}_{F_{Y}^{m}}\left(y\right)&=\int_{\mathcal{Z}}\left(\hat{F}_{Y|Z}^{m}\left(y|z\right)-F_{Y|Z}^{m}\left(y|z\right)\right)dF_{Z}^{h}\left(z\right)+\int_{\mathcal{Z}}F_{Y|Z}^{m}\left(y|z\right)d\left(\hat{F}_{Z}^{h}\left(z\right)-F_{Z}^{h}\left(z\right)\right)\\
&\Rightarrow\int_{\mathcal{Z}}\mathbb{Z}_{F_{Y}^{m}|X}\left(y|z\right)dF_{Z}^{h}\left(z\right)+\sqrt{\lambda_{h}}\mathbb{Z}_{Z_{h}}\left(F_{Y|Z}^{m}\left(y|z\right)\right)\\
\mathbb{Z}_{F_{Y}^{m}|Z}\left(y,z\right)&=-f_{Y|Z}\left(y|z\right)\frac{1}{\pi_{l}\left(z\right)}c_{l,x}\left(C_{k,x}\left(\tilde{u}_{j}\left(x,y\right),\pi_{l}\left(z\right)\right),\pi_{l}\left(z\right)\right)x'\mathbb{Z}_{\beta_{j}}\left(F_{Y|Z}\left(y|z\right)\right)\\
&+\int_{\varepsilon}^{1-\varepsilon}\mathbf{1}\left(x'\beta_{j}\left(\tau\right)\leq y\right)\mathbb{Z}_{c_{k}}\left(\tau,\pi_{l}\left(z\right)\right)d\tau\\
&+\left[\int_{\varepsilon}^{1-\varepsilon}\mathbf{1}\left(x'\beta_{j}\left(\tau\right)\leq y\right)\nabla_{\pi}c_{k,x}\left(u,\pi_{l}\left(z\right)\right)d\tau-1\right]\mathbb{Z}_{\pi_{l}}\left(z\right)
\end{align*}
and an analog version of Lemma~\ref{lem:hadamardf} for the distribution of the entire population, \textit{i.e.}, $F_{Y|Z}\left(y|z\right)$, is applied.
\end{proof}
\section{Additional Results}\label{app:add}

\subsection{Potential Outcomes for the Entire Population}

Table~\ref{tab:y3g} reports the potential earnings for the full population. Because the selection into employment was negative in the early years of the sample and turned positive afterwards, this distribution tends to be above the distribution of actual earnings for participants at the beginning of the period, and below it towards the end. This is most evident for the distributions of females. The different interquantile range measures (10-90th percentiles, 25-75th percentiles) are of a similar magnitude to those found in Table~\ref{tab:y1g}.

\begin{table}[htbp]
  \centering
  \caption{Potential earnings distributions for the full population by gender (Frank copula)}\label{tab:y3g}
		\begin{threeparttable}
    \begin{tabular}{ccccccc|cccccc}
		\hline
          & \multicolumn{6}{c}{Male}                      & \multicolumn{6}{c}{Female} \\
    Year  & Mean  & P10   & P25   & P50   & P75   & P90   & Mean  & P10   & P25   & P50   & P75   & P90 \\
		\hline
    1976  & 2.81  & 2.12  & 2.49  & 2.84  & 3.15  & 3.41  & 2.38  & 1.80  & 2.08  & 2.39  & 2.68  & 2.93 \\
    1977  & 2.78  & 2.06  & 2.44  & 2.81  & 3.13  & 3.39  & 2.37  & 1.80  & 2.07  & 2.37  & 2.67  & 2.91 \\
    1978  & 2.81  & 2.08  & 2.46  & 2.84  & 3.16  & 3.42  & 2.31  & 1.73  & 2.01  & 2.31  & 2.61  & 2.87 \\
    1979  & 2.82  & 2.10  & 2.47  & 2.85  & 3.18  & 3.44  & 2.31  & 1.74  & 2.02  & 2.31  & 2.61  & 2.87 \\
    1980  & 2.81  & 2.09  & 2.46  & 2.84  & 3.17  & 3.43  & 2.32  & 1.76  & 2.03  & 2.32  & 2.61  & 2.87 \\
    1981  & 2.76  & 2.03  & 2.41  & 2.79  & 3.12  & 3.38  & 2.30  & 1.73  & 2.00  & 2.29  & 2.58  & 2.85 \\
    1982  & 2.73  & 1.98  & 2.36  & 2.76  & 3.10  & 3.37  & 2.26  & 1.67  & 1.95  & 2.25  & 2.56  & 2.83 \\
    1983  & 2.66  & 1.88  & 2.27  & 2.69  & 3.05  & 3.34  & 2.22  & 1.61  & 1.90  & 2.21  & 2.53  & 2.81 \\
    1984  & 2.62  & 1.81  & 2.21  & 2.65  & 3.03  & 3.32  & 2.31  & 1.68  & 1.98  & 2.31  & 2.63  & 2.90 \\
    1985  & 2.71  & 1.91  & 2.32  & 2.75  & 3.11  & 3.40  & 2.28  & 1.63  & 1.93  & 2.27  & 2.61  & 2.89 \\
    1986  & 2.70  & 1.89  & 2.30  & 2.73  & 3.11  & 3.39  & 2.26  & 1.59  & 1.90  & 2.25  & 2.59  & 2.89 \\
    1987  & 2.73  & 1.92  & 2.33  & 2.76  & 3.14  & 3.44  & 2.24  & 1.54  & 1.87  & 2.23  & 2.60  & 2.91 \\
    1988  & 2.66  & 1.82  & 2.24  & 2.69  & 3.08  & 3.38  & 2.30  & 1.58  & 1.93  & 2.30  & 2.65  & 2.95 \\
    1989  & 2.64  & 1.81  & 2.23  & 2.67  & 3.06  & 3.37  & 2.31  & 1.60  & 1.93  & 2.30  & 2.67  & 2.97 \\
    1990  & 2.64  & 1.81  & 2.23  & 2.67  & 3.06  & 3.38  & 2.24  & 1.49  & 1.84  & 2.23  & 2.62  & 2.94 \\
    1991  & 2.62  & 1.79  & 2.20  & 2.64  & 3.04  & 3.36  & 2.30  & 1.60  & 1.93  & 2.30  & 2.67  & 2.97 \\
    1992  & 2.60  & 1.75  & 2.17  & 2.62  & 3.02  & 3.34  & 2.30  & 1.59  & 1.91  & 2.29  & 2.67  & 2.98 \\
    1993  & 2.50  & 1.61  & 2.05  & 2.52  & 2.95  & 3.29  & 2.30  & 1.57  & 1.91  & 2.29  & 2.67  & 2.98 \\
    1994  & 2.51  & 1.63  & 2.06  & 2.53  & 2.96  & 3.30  & 2.27  & 1.52  & 1.86  & 2.26  & 2.65  & 2.98 \\
    1995  & 2.48  & 1.57  & 2.01  & 2.49  & 2.94  & 3.29  & 2.24  & 1.47  & 1.82  & 2.22  & 2.63  & 2.98 \\
    1996  & 2.64  & 1.80  & 2.20  & 2.64  & 3.05  & 3.39  & 2.24  & 1.47  & 1.82  & 2.22  & 2.63  & 2.97 \\
    1997  & 2.63  & 1.79  & 2.19  & 2.62  & 3.03  & 3.37  & 2.25  & 1.49  & 1.83  & 2.24  & 2.64  & 2.98 \\
    1998  & 2.56  & 1.69  & 2.11  & 2.56  & 2.98  & 3.33  & 2.28  & 1.54  & 1.87  & 2.26  & 2.66  & 3.00 \\
    1999  & 2.70  & 1.87  & 2.28  & 2.70  & 3.09  & 3.45  & 2.30  & 1.54  & 1.88  & 2.29  & 2.69  & 3.03 \\
    2000  & 2.71  & 1.87  & 2.28  & 2.71  & 3.12  & 3.48  & 2.37  & 1.61  & 1.96  & 2.36  & 2.75  & 3.09 \\
    2001  & 2.75  & 1.91  & 2.30  & 2.72  & 3.13  & 3.54  & 2.40  & 1.63  & 1.99  & 2.39  & 2.77  & 3.11 \\
    2002  & 2.72  & 1.88  & 2.27  & 2.70  & 3.11  & 3.50  & 2.40  & 1.63  & 1.99  & 2.38  & 2.77  & 3.12 \\
    2003  & 2.72  & 1.89  & 2.28  & 2.71  & 3.12  & 3.51  & 2.34  & 1.55  & 1.91  & 2.32  & 2.73  & 3.10 \\
    2004  & 2.73  & 1.89  & 2.28  & 2.71  & 3.13  & 3.52  & 2.37  & 1.58  & 1.94  & 2.35  & 2.76  & 3.11 \\
    2005  & 2.73  & 1.88  & 2.27  & 2.70  & 3.12  & 3.53  & 2.35  & 1.55  & 1.92  & 2.34  & 2.75  & 3.11 \\
    2006  & 2.81  & 1.96  & 2.34  & 2.78  & 3.21  & 3.64  & 2.37  & 1.58  & 1.94  & 2.35  & 2.76  & 3.12 \\
    2007  & 2.72  & 1.88  & 2.26  & 2.69  & 3.12  & 3.52  & 2.30  & 1.47  & 1.84  & 2.29  & 2.72  & 3.10 \\
    2008  & 2.74  & 1.90  & 2.29  & 2.71  & 3.14  & 3.54  & 2.31  & 1.48  & 1.86  & 2.29  & 2.72  & 3.10 \\
    2009  & 2.65  & 1.77  & 2.19  & 2.62  & 3.06  & 3.47  & 2.32  & 1.49  & 1.88  & 2.30  & 2.73  & 3.10 \\
    2010  & 2.57  & 1.66  & 2.09  & 2.55  & 3.00  & 3.41  & 2.31  & 1.47  & 1.86  & 2.30  & 2.73  & 3.12 \\
    2011  & 2.65  & 1.78  & 2.18  & 2.63  & 3.07  & 3.47  & 2.27  & 1.40  & 1.81  & 2.26  & 2.71  & 3.09 \\
    2012  & 2.50  & 1.55  & 2.00  & 2.48  & 2.96  & 3.39  & 2.28  & 1.44  & 1.82  & 2.26  & 2.70  & 3.09 \\
    2013  & 2.68  & 1.80  & 2.21  & 2.65  & 3.10  & 3.51  & 2.30  & 1.45  & 1.84  & 2.28  & 2.72  & 3.11 \\
		\hline
    \end{tabular}
\end{threeparttable}
\end{table}

The decomposition of mean potential earnings for the entire population is reported in Table~\ref{tab:dy3g}. The gap displays an erratic behavior, driven by the coefficients component, that is more variable than for the two main decompositions. Hence, this gap is more sensitive to the QRS estimates, which may be less robust if the instrument used is weak for men. Regardless, the endowments component switches sign over time, thus helping reduce the gender gap.

\begin{table}[htbp]
  \centering
  \caption{Mean decomposition, potential earnings for the full population (Frank copula)}\label{tab:dy3g}
		\begin{threeparttable}
    \begin{tabular}{cc@{\,}lc@{\,}lc@{\,}l}
		\hline
    Year  & \multicolumn{2}{c}{Total} & \multicolumn{2}{c}{EC} & \multicolumn{2}{c}{CC} \\
		\hline
    1976  & 0.43  & ***   & 0.01  & ***   & 0.42  & *** \\
    1977  & 0.41  & ***   & 0.01  & ***   & 0.40  & *** \\
    1978  & 0.49  & ***   & 0.01  & ***   & 0.48  & *** \\
    1979  & 0.50  & ***   & 0.01  & ***   & 0.49  & *** \\
    1980  & 0.48  & ***   & 0.01  & ***   & 0.48  & *** \\
    1981  & 0.46  & ***   & 0.00  &       & 0.46  & *** \\
    1982  & 0.47  & ***   & 0.01  & **    & 0.46  & *** \\
    1983  & 0.44  & ***   & 0.01  & **    & 0.43  & *** \\
    1984  & 0.31  & ***   & 0.01  & **    & 0.30  & *** \\
    1985  & 0.43  & ***   & 0.01  & ***   & 0.42  & *** \\
    1986  & 0.44  & ***   & 0.01  & ***   & 0.43  & *** \\
    1987  & 0.49  & ***   & 0.01  & **    & 0.48  & *** \\
    1988  & 0.36  & ***   & 0.01  & ***   & 0.35  & *** \\
    1989  & 0.33  & ***   & 0.01  & ***   & 0.32  & *** \\
    1990  & 0.40  & ***   & 0.01  & ***   & 0.39  & *** \\
    1991  & 0.31  & ***   & 0.00  &       & 0.31  & *** \\
    1992  & 0.29  & ***   & 0.00  & *     & 0.29  & *** \\
    1993  & 0.20  & ***   & 0.00  &       & 0.20  & *** \\
    1994  & 0.24  & ***   & 0.00  &       & 0.24  & *** \\
    1995  & 0.24  & ***   & 0.00  &       & 0.24  & *** \\
    1996  & 0.40  & ***   & 0.00  &       & 0.41  & *** \\
    1997  & 0.37  & ***   & -0.01 & **    & 0.38  & *** \\
    1998  & 0.28  & ***   & -0.01 & ***   & 0.29  & *** \\
    1999  & 0.40  & ***   & -0.01 & *     & 0.41  & *** \\
    2000  & 0.35  & ***   & -0.01 & **    & 0.35  & *** \\
    2001  & 0.35  & ***   & -0.01 & ***   & 0.36  & *** \\
    2002  & 0.32  & ***   & -0.01 & ***   & 0.33  & *** \\
    2003  & 0.38  & ***   & -0.02 & ***   & 0.40  & *** \\
    2004  & 0.36  & ***   & -0.02 & ***   & 0.38  & *** \\
    2005  & 0.38  & ***   & -0.02 & ***   & 0.40  & *** \\
    2006  & 0.44  & ***   & -0.03 & ***   & 0.47  & *** \\
    2007  & 0.41  & ***   & -0.03 & ***   & 0.45  & *** \\
    2008  & 0.43  & ***   & -0.03 & ***   & 0.46  & *** \\
    2009  & 0.33  & ***   & -0.04 & ***   & 0.36  & *** \\
    2010  & 0.25  & ***   & -0.04 & ***   & 0.29  & *** \\
    2011  & 0.38  & ***   & -0.04 & ***   & 0.42  & *** \\
    2012  & 0.22  & **    & -0.04 & ***   & 0.25  & *** \\
    2013  & 0.38  & ***   & -0.04 & ***   & 0.42  & *** \\
		\hline
    \end{tabular}\begin{tablenotes}[para,flushleft]
\begin{spacing}{1}
{\footnotesize Notes: Total, EC and CC respectively denote total difference, endowments component and coefficients component.}
\end{spacing}
\end{tablenotes}
\end{threeparttable}
\end{table}

As it was the case for the mean, the decomposition of the unconditional distributions of potential earnings for the entire population are quite similar to those found for actual earnings for participants (Figure~\ref{fig:dec3}; Tables~\ref{tab:dy3p10g}-\ref{tab:dy3p90g}). Regarding its components, the coefficients component is dominant. On the other hand, the endowments component has a slightly increasing shape and is much smaller in magnitude. In the early years it was negative for the left tail and positive for the majority of the distribution, and it has progressively become more negative throughout the entire distribution, now contributing the the reduction of the gap.

\begin{figure}[htbp]
\caption{Unconditional quantiles decompositions, potential earnings for the full population (Frank copula)}
\includegraphics[width=16.5cm]{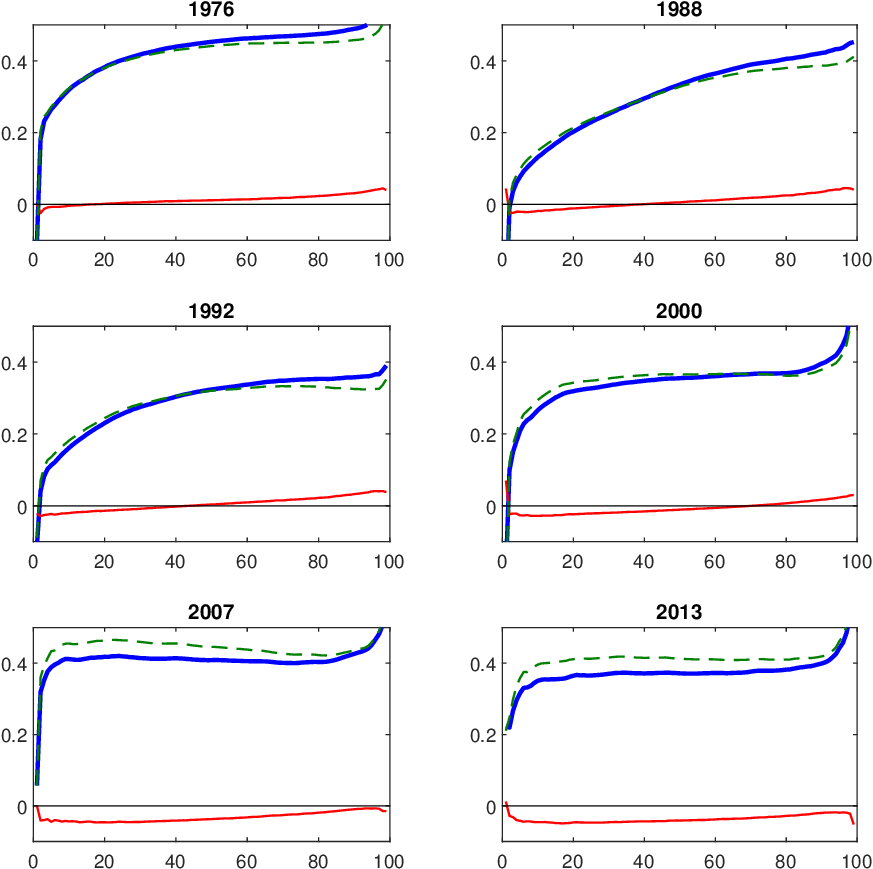}\label{fig:dec3}

{\footnotesize Notes: the solid thick blue line denotes the total gap between male and female workers; the solid thin red line denotes the endowments component; the dashed green line denotes the coefficients component.}
\end{figure}

\begin{table}[htbp]
  \centering
  \caption{10th percentile decomposition, potential earnings for the full population (Frank copula)}\label{tab:dy3p10g}
		\begin{threeparttable}
    \begin{tabular}{cc@{\,}lc@{\,}lc@{\,}l}
		\hline
    Year  & \multicolumn{2}{c}{Total} & \multicolumn{2}{c}{EC} & \multicolumn{2}{c}{CC} \\
		\hline
    1976  & 0.32  & ***   & 0.00  &       & 0.32  & *** \\
    1977  & 0.26  & ***   & -0.01 & **    & 0.28  & *** \\
    1978  & 0.35  & ***   & -0.01 &       & 0.36  & *** \\
    1979  & 0.36  & ***   & 0.00  &       & 0.36  & *** \\
    1980  & 0.32  & ***   & -0.01 & ***   & 0.34  & *** \\
    1981  & 0.30  & ***   & -0.01 & ***   & 0.32  & *** \\
    1982  & 0.31  & ***   & -0.01 & **    & 0.32  & *** \\
    1983  & 0.27  & ***   & -0.02 & ***   & 0.29  & *** \\
    1984  & 0.13  &       & -0.02 & ***   & 0.15  & * \\
    1985  & 0.29  & ***   & -0.01 & ***   & 0.30  & *** \\
    1986  & 0.29  & ***   & -0.01 & ***   & 0.31  & *** \\
    1987  & 0.37  & ***   & -0.02 & ***   & 0.39  & *** \\
    1988  & 0.24  & ***   & -0.02 & ***   & 0.25  & *** \\
    1989  & 0.20  & ***   & -0.01 & ***   & 0.22  & *** \\
    1990  & 0.33  & ***   & -0.02 & ***   & 0.34  & *** \\
    1991  & 0.20  & ***   & -0.02 & ***   & 0.22  & *** \\
    1992  & 0.16  & ***   & -0.02 & ***   & 0.18  & *** \\
    1993  & 0.04  &       & -0.02 & ***   & 0.06  &  \\
    1994  & 0.12  &       & -0.03 & ***   & 0.15  &  \\
    1995  & 0.09  &       & -0.03 & ***   & 0.12  &  \\
    1996  & 0.33  & ***   & -0.03 & ***   & 0.36  & *** \\
    1997  & 0.30  & ***   & -0.03 & ***   & 0.33  & *** \\
    1998  & 0.16  & **    & -0.03 & ***   & 0.19  & ** \\
    1999  & 0.34  & ***   & -0.03 & ***   & 0.36  & *** \\
    2000  & 0.27  & ***   & -0.03 & ***   & 0.29  & *** \\
    2001  & 0.28  & ***   & -0.03 & ***   & 0.32  & *** \\
    2002  & 0.25  & ***   & -0.03 & ***   & 0.28  & *** \\
    2003  & 0.33  & ***   & -0.04 & ***   & 0.37  & *** \\
    2004  & 0.31  & ***   & -0.04 & ***   & 0.35  & *** \\
    2005  & 0.33  & ***   & -0.04 & ***   & 0.37  & *** \\
    2006  & 0.38  & ***   & -0.04 & ***   & 0.42  & *** \\
    2007  & 0.41  & ***   & -0.04 & ***   & 0.46  & *** \\
    2008  & 0.42  & ***   & -0.05 & ***   & 0.47  & *** \\
    2009  & 0.28  & ***   & -0.05 & ***   & 0.32  & *** \\
    2010  & 0.19  &       & -0.05 & ***   & 0.24  &  \\
    2011  & 0.38  & ***   & -0.05 & ***   & 0.43  & *** \\
    2012  & 0.11  &       & -0.05 & ***   & 0.15  &  \\
    2013  & 0.35  & ***   & -0.05 & ***   & 0.40  & *** \\
		\hline
    \end{tabular}\begin{tablenotes}[para,flushleft]
\begin{spacing}{1}
{\footnotesize Notes: Total, EC and CC respectively denote total difference, endowments component and coefficients component.}
\end{spacing}
\end{tablenotes}
\end{threeparttable}
\end{table}

\begin{table}[htbp]
  \centering
  \caption{25th percentile decomposition, potential earnings for the full population (Frank copula)}\label{tab:dy3p25g}
		\begin{threeparttable}
    \begin{tabular}{cc@{\,}lc@{\,}lc@{\,}l}
		\hline
    Year  & \multicolumn{2}{c}{Total} & \multicolumn{2}{c}{EC} & \multicolumn{2}{c}{CC} \\
		\hline
    1976  & 0.40  & ***   & 0.00  &       & 0.40  & *** \\
    1977  & 0.37  & ***   & 0.00  &       & 0.37  & *** \\
    1978  & 0.46  & ***   & 0.00  &       & 0.46  & *** \\
    1979  & 0.45  & ***   & 0.00  &       & 0.45  & *** \\
    1980  & 0.43  & ***   & 0.00  &       & 0.44  & *** \\
    1981  & 0.41  & ***   & -0.01 & *     & 0.42  & *** \\
    1982  & 0.41  & ***   & 0.00  &       & 0.42  & *** \\
    1983  & 0.37  & ***   & -0.01 & **    & 0.38  & *** \\
    1984  & 0.23  & ***   & -0.01 & *     & 0.24  & *** \\
    1985  & 0.39  & ***   & 0.00  &       & 0.39  & *** \\
    1986  & 0.40  & ***   & -0.01 &       & 0.40  & *** \\
    1987  & 0.46  & ***   & -0.01 & **    & 0.47  & *** \\
    1988  & 0.31  & ***   & 0.00  &       & 0.32  & *** \\
    1989  & 0.29  & ***   & 0.00  &       & 0.30  & *** \\
    1990  & 0.39  & ***   & -0.01 &       & 0.40  & *** \\
    1991  & 0.27  & ***   & -0.01 & ***   & 0.28  & *** \\
    1992  & 0.26  & ***   & -0.01 & ***   & 0.27  & *** \\
    1993  & 0.14  & *     & -0.01 & ***   & 0.16  & ** \\
    1994  & 0.20  & **    & -0.02 & ***   & 0.22  & ** \\
    1995  & 0.19  & **    & -0.02 & ***   & 0.22  & ** \\
    1996  & 0.39  & ***   & -0.02 & ***   & 0.41  & *** \\
    1997  & 0.36  & ***   & -0.02 & ***   & 0.38  & *** \\
    1998  & 0.24  & ***   & -0.03 & ***   & 0.27  & *** \\
    1999  & 0.39  & ***   & -0.02 & ***   & 0.42  & *** \\
    2000  & 0.33  & ***   & -0.02 & ***   & 0.35  & *** \\
    2001  & 0.31  & ***   & -0.03 & ***   & 0.34  & *** \\
    2002  & 0.29  & ***   & -0.02 & ***   & 0.31  & *** \\
    2003  & 0.37  & ***   & -0.03 & ***   & 0.40  & *** \\
    2004  & 0.34  & ***   & -0.04 & ***   & 0.38  & *** \\
    2005  & 0.36  & ***   & -0.04 & ***   & 0.40  & *** \\
    2006  & 0.40  & ***   & -0.04 & ***   & 0.44  & *** \\
    2007  & 0.42  & ***   & -0.05 & ***   & 0.46  & *** \\
    2008  & 0.43  & ***   & -0.05 & ***   & 0.48  & *** \\
    2009  & 0.31  & ***   & -0.05 & ***   & 0.36  & *** \\
    2010  & 0.23  & *     & -0.05 & ***   & 0.28  & ** \\
    2011  & 0.37  & ***   & -0.05 & ***   & 0.42  & *** \\
    2012  & 0.17  &       & -0.05 & ***   & 0.22  & * \\
    2013  & 0.37  & ***   & -0.05 & ***   & 0.41  & *** \\
		\hline
    \end{tabular}\begin{tablenotes}[para,flushleft]
\begin{spacing}{1}
{\footnotesize Notes: Total, EC and CC respectively denote total difference, endowments component and coefficients component.}
\end{spacing}
\end{tablenotes}
\end{threeparttable}
\end{table}

\begin{table}[htbp]
  \centering
  \caption{50th percentile decomposition, potential earnings for the full population (Frank copula)}\label{tab:dy3p50g}
		\begin{threeparttable}
    \begin{tabular}{cc@{\,}lc@{\,}lc@{\,}l}
		\hline
    Year  & \multicolumn{2}{c}{Total} & \multicolumn{2}{c}{EC} & \multicolumn{2}{c}{CC} \\
		\hline
    1976  & 0.45  & ***   & 0.01  & ***   & 0.44  & *** \\
    1977  & 0.44  & ***   & 0.01  & **    & 0.43  & *** \\
    1978  & 0.53  & ***   & 0.01  & ***   & 0.52  & *** \\
    1979  & 0.54  & ***   & 0.01  & ***   & 0.53  & *** \\
    1980  & 0.53  & ***   & 0.01  & ***   & 0.52  & *** \\
    1981  & 0.51  & ***   & 0.00  &       & 0.50  & *** \\
    1982  & 0.51  & ***   & 0.00  &       & 0.50  & *** \\
    1983  & 0.48  & ***   & 0.00  &       & 0.47  & *** \\
    1984  & 0.33  & ***   & 0.01  & *     & 0.33  & *** \\
    1985  & 0.47  & ***   & 0.01  & **    & 0.46  & *** \\
    1986  & 0.48  & ***   & 0.01  & **    & 0.47  & *** \\
    1987  & 0.54  & ***   & 0.01  &       & 0.53  & *** \\
    1988  & 0.39  & ***   & 0.01  & ***   & 0.37  & *** \\
    1989  & 0.37  & ***   & 0.01  & ***   & 0.36  & *** \\
    1990  & 0.43  & ***   & 0.01  & **    & 0.43  & *** \\
    1991  & 0.34  & ***   & 0.00  &       & 0.33  & *** \\
    1992  & 0.32  & ***   & 0.00  &       & 0.32  & *** \\
    1993  & 0.23  & ***   & 0.00  &       & 0.23  & *** \\
    1994  & 0.27  & ***   & 0.00  &       & 0.27  & *** \\
    1995  & 0.27  & ***   & 0.00  &       & 0.27  & *** \\
    1996  & 0.42  & ***   & -0.01 & **    & 0.43  & *** \\
    1997  & 0.38  & ***   & -0.01 & ***   & 0.39  & *** \\
    1998  & 0.30  & ***   & -0.01 & ***   & 0.31  & *** \\
    1999  & 0.41  & ***   & -0.01 & ***   & 0.42  & *** \\
    2000  & 0.36  & ***   & -0.01 & ***   & 0.37  & *** \\
    2001  & 0.34  & ***   & -0.02 & ***   & 0.35  & *** \\
    2002  & 0.31  & ***   & -0.01 & ***   & 0.33  & *** \\
    2003  & 0.38  & ***   & -0.02 & ***   & 0.40  & *** \\
    2004  & 0.36  & ***   & -0.03 & ***   & 0.39  & *** \\
    2005  & 0.37  & ***   & -0.03 & ***   & 0.40  & *** \\
    2006  & 0.42  & ***   & -0.03 & ***   & 0.45  & *** \\
    2007  & 0.41  & ***   & -0.04 & ***   & 0.45  & *** \\
    2008  & 0.42  & ***   & -0.04 & ***   & 0.46  & *** \\
    2009  & 0.33  & ***   & -0.04 & ***   & 0.37  & *** \\
    2010  & 0.25  & ***   & -0.04 & ***   & 0.29  & *** \\
    2011  & 0.38  & ***   & -0.04 & ***   & 0.42  & *** \\
    2012  & 0.22  & ***   & -0.04 & ***   & 0.26  & *** \\
    2013  & 0.37  & ***   & -0.04 & ***   & 0.41  & *** \\
		\hline
    \end{tabular}\begin{tablenotes}[para,flushleft]
\begin{spacing}{1}
{\footnotesize Notes: Total, EC and CC respectively denote total difference, endowments component and coefficients component.}
\end{spacing}
\end{tablenotes}
\end{threeparttable}
\end{table}

\begin{table}[htbp]
  \centering
  \caption{75th percentile decomposition, potential earnings for the full population (Frank copula)}\label{tab:dy3p75g}
		\begin{threeparttable}
    \begin{tabular}{cc@{\,}lc@{\,}lc@{\,}l}
		\hline
    Year  & \multicolumn{2}{c}{Total} & \multicolumn{2}{c}{EC} & \multicolumn{2}{c}{CC} \\
		\hline
    1976  & 0.47  & ***   & 0.02  & ***   & 0.45  & *** \\
    1977  & 0.47  & ***   & 0.02  & ***   & 0.45  & *** \\
    1978  & 0.55  & ***   & 0.02  & ***   & 0.53  & *** \\
    1979  & 0.57  & ***   & 0.02  & ***   & 0.55  & *** \\
    1980  & 0.56  & ***   & 0.02  & ***   & 0.54  & *** \\
    1981  & 0.54  & ***   & 0.01  & ***   & 0.53  & *** \\
    1982  & 0.54  & ***   & 0.02  & ***   & 0.52  & *** \\
    1983  & 0.52  & ***   & 0.02  & ***   & 0.51  & *** \\
    1984  & 0.40  & ***   & 0.02  & ***   & 0.38  & *** \\
    1985  & 0.50  & ***   & 0.02  & ***   & 0.48  & *** \\
    1986  & 0.51  & ***   & 0.02  & ***   & 0.49  & *** \\
    1987  & 0.54  & ***   & 0.02  & ***   & 0.52  & *** \\
    1988  & 0.43  & ***   & 0.03  & ***   & 0.40  & *** \\
    1989  & 0.39  & ***   & 0.02  & ***   & 0.37  & *** \\
    1990  & 0.45  & ***   & 0.03  & ***   & 0.42  & *** \\
    1991  & 0.37  & ***   & 0.02  & ***   & 0.35  & *** \\
    1992  & 0.35  & ***   & 0.02  & ***   & 0.33  & *** \\
    1993  & 0.28  & ***   & 0.02  & ***   & 0.26  & *** \\
    1994  & 0.31  & ***   & 0.02  & ***   & 0.29  & *** \\
    1995  & 0.30  & ***   & 0.02  & ***   & 0.29  & *** \\
    1996  & 0.42  & ***   & 0.01  & **    & 0.41  & *** \\
    1997  & 0.39  & ***   & 0.01  &       & 0.38  & *** \\
    1998  & 0.32  & ***   & 0.00  &       & 0.31  & *** \\
    1999  & 0.40  & ***   & 0.01  & *     & 0.40  & *** \\
    2000  & 0.37  & ***   & 0.00  &       & 0.37  & *** \\
    2001  & 0.36  & ***   & 0.00  &       & 0.36  & *** \\
    2002  & 0.34  & ***   & 0.00  &       & 0.34  & *** \\
    2003  & 0.39  & ***   & 0.00  &       & 0.39  & *** \\
    2004  & 0.37  & ***   & -0.01 & ***   & 0.38  & *** \\
    2005  & 0.37  & ***   & -0.01 & ***   & 0.39  & *** \\
    2006  & 0.46  & ***   & -0.02 & ***   & 0.47  & *** \\
    2007  & 0.40  & ***   & -0.02 & ***   & 0.42  & *** \\
    2008  & 0.42  & ***   & -0.02 & ***   & 0.44  & *** \\
    2009  & 0.33  & ***   & -0.03 & ***   & 0.36  & *** \\
    2010  & 0.27  & ***   & -0.03 & ***   & 0.31  & *** \\
    2011  & 0.37  & ***   & -0.03 & ***   & 0.39  & *** \\
    2012  & 0.26  & ***   & -0.03 & ***   & 0.29  & *** \\
    2013  & 0.38  & ***   & -0.03 & ***   & 0.41  & *** \\
		\hline
    \end{tabular}\begin{tablenotes}[para,flushleft]
\begin{spacing}{1}
{\footnotesize Notes: Total, EC and CC respectively denote total difference, endowments component and coefficients component.}
\end{spacing}
\end{tablenotes}
\end{threeparttable}
\end{table}

\begin{table}[htbp]
  \centering
  \caption{90th percentile decomposition, potential earnings for the full population (Frank copula)}\label{tab:dy3p90g}
		\begin{threeparttable}
    \begin{tabular}{cc@{\,}lc@{\,}lc@{\,}l}
		\hline
    Year  & \multicolumn{2}{c}{Total} & \multicolumn{2}{c}{EC} & \multicolumn{2}{c}{CC} \\
		\hline
    1976  & 0.49  & ***   & 0.03  & ***   & 0.46  & *** \\
    1977  & 0.48  & ***   & 0.03  & ***   & 0.45  & *** \\
    1978  & 0.56  & ***   & 0.03  & ***   & 0.52  & *** \\
    1979  & 0.57  & ***   & 0.03  & ***   & 0.54  & *** \\
    1980  & 0.56  & ***   & 0.03  & ***   & 0.53  & *** \\
    1981  & 0.53  & ***   & 0.02  & ***   & 0.51  & *** \\
    1982  & 0.54  & ***   & 0.03  & ***   & 0.52  & *** \\
    1983  & 0.53  & ***   & 0.03  & ***   & 0.50  & *** \\
    1984  & 0.42  & ***   & 0.03  & ***   & 0.39  & *** \\
    1985  & 0.51  & ***   & 0.03  & ***   & 0.47  & *** \\
    1986  & 0.50  & ***   & 0.04  & ***   & 0.47  & *** \\
    1987  & 0.53  & ***   & 0.04  & ***   & 0.49  & *** \\
    1988  & 0.43  & ***   & 0.04  & ***   & 0.39  & *** \\
    1989  & 0.40  & ***   & 0.04  & ***   & 0.36  & *** \\
    1990  & 0.44  & ***   & 0.04  & ***   & 0.40  & *** \\
    1991  & 0.39  & ***   & 0.03  & ***   & 0.35  & *** \\
    1992  & 0.36  & ***   & 0.03  & ***   & 0.33  & *** \\
    1993  & 0.30  & ***   & 0.03  & ***   & 0.27  & *** \\
    1994  & 0.32  & ***   & 0.03  & ***   & 0.29  & *** \\
    1995  & 0.32  & ***   & 0.03  & ***   & 0.29  & *** \\
    1996  & 0.42  & ***   & 0.02  & ***   & 0.39  & *** \\
    1997  & 0.38  & ***   & 0.02  & ***   & 0.36  & *** \\
    1998  & 0.33  & ***   & 0.02  & ***   & 0.32  & *** \\
    1999  & 0.42  & ***   & 0.02  & ***   & 0.40  & *** \\
    2000  & 0.40  & ***   & 0.02  & ***   & 0.38  & *** \\
    2001  & 0.42  & ***   & 0.02  & ***   & 0.41  & *** \\
    2002  & 0.38  & ***   & 0.01  & ***   & 0.37  & *** \\
    2003  & 0.42  & ***   & 0.01  & ***   & 0.41  & *** \\
    2004  & 0.41  & ***   & 0.00  &       & 0.41  & *** \\
    2005  & 0.42  & ***   & 0.00  &       & 0.42  & *** \\
    2006  & 0.52  & ***   & 0.00  &       & 0.52  & *** \\
    2007  & 0.43  & ***   & -0.01 & ***   & 0.44  & *** \\
    2008  & 0.44  & ***   & -0.01 & ***   & 0.45  & *** \\
    2009  & 0.37  & ***   & -0.02 & ***   & 0.38  & *** \\
    2010  & 0.30  & ***   & -0.02 & ***   & 0.32  & *** \\
    2011  & 0.38  & ***   & -0.02 & ***   & 0.40  & *** \\
    2012  & 0.29  & ***   & -0.02 & ***   & 0.31  & *** \\
    2013  & 0.40  & ***   & -0.02 & ***   & 0.42  & *** \\
		\hline
    \end{tabular}\begin{tablenotes}[para,flushleft]
\begin{spacing}{1}
{\footnotesize Notes: Total, EC and CC respectively denote total difference, endowments component and coefficients component.}
\end{spacing}
\end{tablenotes}
\end{threeparttable}
\end{table}

\clearpage

\subsection{Generalized Entropy Measures of the Gap}

I report the generalized entropy measures recommended by \cite{Maasoumi2019} for the three populations considered. For the sake of brevity, I only comment the normalization of the Bhattacharya-Matusita-Hellinger measure, denoted by $S_{\rho}$, and the normalized and symmetrized Kullback-Leibler-Theil measure, denoted by Theil.

Table~\ref{tab:GE1} shows the estimates of these measures for the distributions of actual earnings for participants, and both of them experience a steady decrease over time, suggesting an important convergence between the distributions of both genders. These results are similar if one considers the entire population (Table~\ref{tab:GE2}), with two main differences: the values of these measures are larger when one considers the entire population, and the reduction of the measures has also been more pronounced. The explanation for these differences lies in the large reduction of the participation rates between genders.

\begin{table}[htbp]
  \centering
  \caption{Generalized entropy measures of the gap, actual earnings for participants (Frank copula)}\label{tab:GE1}
		\begin{threeparttable}
    \begin{tabular}{cccccccccccc}
		\hline
    Year  & $S_{\rho}$ & Theil & \footnotesize{$k=0.1$} & \footnotesize{$k=0.2$} & \footnotesize{$k=0.3$} & \footnotesize{$k=0.4$} & \footnotesize{$k=0.5$} & \footnotesize{$k=0.6$} & \footnotesize{$k=0.7$} & \footnotesize{$k=0.8$} & \footnotesize{$k=0.9$} \\
		\hline
    1976  & 13.6  & 133.1 & 8.5   & 13.4  & 17.6  & 22.1  & 27.1  & 33.2  & 41.2  & 53.4  & 76.7 \\
    1977  & 12.3  & 109.6 & 7.3   & 11.8  & 15.8  & 20.0  & 24.6  & 30.0  & 36.9  & 47.2  & 65.9 \\
    1978  & 10.8  & 51.4  & 4.8   & 9.1   & 13.3  & 17.5  & 21.7  & 26.2  & 31.0  & 36.4  & 43.0 \\
    1979  & 10.9  & 51.0  & 4.8   & 9.2   & 13.4  & 17.6  & 21.9  & 26.4  & 31.2  & 36.7  & 43.2 \\
    1980  & 10.5  & 50.2  & 4.6   & 8.8   & 12.8  & 16.9  & 21.0  & 25.3  & 29.9  & 35.1  & 41.3 \\
    1981  & 9.8   & 54.4  & 4.6   & 8.4   & 12.1  & 15.8  & 19.6  & 23.6  & 28.2  & 33.6  & 41.1 \\
    1982  & 9.6   & 47.0  & 4.4   & 8.2   & 11.9  & 15.6  & 19.4  & 23.4  & 27.7  & 32.7  & 39.2 \\
    1983  & 8.4   & 40.3  & 4.1   & 7.5   & 10.8  & 14.1  & 17.6  & 21.2  & 25.2  & 30.0  & 37.0 \\
    1984  & 7.8   & 48.7  & 4.1   & 7.2   & 10.1  & 13.0  & 16.1  & 19.5  & 23.5  & 28.7  & 37.3 \\
    1985  & 6.9   & 32.0  & 3.1   & 5.9   & 8.6   & 11.3  & 14.1  & 17.0  & 20.1  & 23.6  & 28.2 \\
    1986  & 6.2   & 27.2  & 3.1   & 5.6   & 8.1   & 10.7  & 13.3  & 16.0  & 19.0  & 22.6  & 27.9 \\
    1987  & 5.4   & 22.5  & 2.9   & 5.2   & 7.5   & 9.8   & 12.1  & 14.6  & 17.4  & 20.8  & 26.3 \\
    1988  & 5.1   & 21.0  & 2.9   & 5.0   & 7.2   & 9.4   & 11.6  & 14.0  & 16.8  & 20.2  & 26.2 \\
    1989  & 4.9   & 20.9  & 2.7   & 4.8   & 6.8   & 8.9   & 11.1  & 13.4  & 16.0  & 19.2  & 24.7 \\
    1990  & 4.2   & 17.5  & 2.7   & 4.6   & 6.4   & 8.3   & 10.3  & 12.5  & 15.0  & 18.3  & 24.7 \\
    1991  & 3.8   & 15.5  & 2.2   & 3.7   & 5.3   & 7.0   & 8.6   & 10.4  & 12.5  & 15.0  & 19.5 \\
    1992  & 3.4   & 13.9  & 2.1   & 3.5   & 5.0   & 6.5   & 8.0   & 9.7   & 11.6  & 14.0  & 18.6 \\
    1993  & 3.1   & 12.7  & 2.2   & 3.6   & 4.9   & 6.4   & 7.9   & 9.5   & 11.5  & 14.3  & 19.9 \\
    1994  & 2.6   & 10.8  & 2.1   & 3.3   & 4.5   & 5.8   & 7.1   & 8.6   & 10.5  & 13.1  & 18.9 \\
    1995  & 2.5   & 10.2  & 2.3   & 3.5   & 4.7   & 6.0   & 7.4   & 9.0   & 11.0  & 14.0  & 21.1 \\
    1996  & 2.3   & 9.4   & 2.2   & 3.3   & 4.4   & 5.6   & 6.9   & 8.4   & 10.4  & 13.3  & 20.1 \\
    1997  & 2.2   & 9.1   & 2.5   & 3.5   & 4.7   & 5.9   & 7.3   & 8.9   & 10.9  & 14.2  & 22.2 \\
    1998  & 2.3   & 9.2   & 2.8   & 3.9   & 5.1   & 6.4   & 7.8   & 9.6   & 11.9  & 15.6  & 24.8 \\
    1999  & 2.2   & 9.0   & 2.5   & 3.6   & 4.7   & 5.9   & 7.3   & 8.9   & 11.0  & 14.3  & 22.3 \\
    2000  & 2.1   & 8.5   & 2.4   & 3.4   & 4.5   & 5.7   & 7.0   & 8.5   & 10.5  & 13.7  & 21.5 \\
    2001  & 1.8   & 7.4   & 3.1   & 4.1   & 5.2   & 6.4   & 7.8   & 9.6   & 12.1  & 16.3  & 27.5 \\
    2002  & 1.6   & 6.6   & 3.0   & 4.0   & 5.0   & 6.2   & 7.6   & 9.3   & 11.7  & 15.9  & 27.3 \\
    2003  & 1.6   & 6.6   & 3.4   & 4.4   & 5.4   & 6.7   & 8.1   & 10.0  & 12.7  & 17.4  & 30.3 \\
    2004  & 1.4   & 5.8   & 2.7   & 3.6   & 4.5   & 5.5   & 6.7   & 8.2   & 10.4  & 14.2  & 24.5 \\
    2005  & 1.3   & 5.3   & 2.7   & 3.5   & 4.4   & 5.4   & 6.6   & 8.1   & 10.3  & 14.1  & 24.5 \\
    2006  & 1.5   & 6.0   & 2.5   & 3.4   & 4.3   & 5.3   & 6.4   & 7.9   & 9.9   & 13.5  & 22.8 \\
    2007  & 1.3   & 5.1   & 2.8   & 3.6   & 4.5   & 5.5   & 6.7   & 8.3   & 10.5  & 14.5  & 25.5 \\
    2008  & 1.2   & 5.0   & 2.8   & 3.6   & 4.4   & 5.4   & 6.6   & 8.1   & 10.3  & 14.3  & 25.3 \\
    2009  & 1.2   & 4.9   & 2.9   & 3.7   & 4.6   & 5.6   & 6.8   & 8.4   & 10.7  & 14.8  & 26.3 \\
    2010  & 1.1   & 4.5   & 3.2   & 4.0   & 4.9   & 6.0   & 7.2   & 8.9   & 11.5  & 16.1  & 29.2 \\
    2011  & 1.2   & 4.9   & 2.9   & 3.7   & 4.5   & 5.5   & 6.7   & 8.2   & 10.5  & 14.6  & 26.0 \\
    2012  & 1.0   & 4.2   & 3.3   & 4.1   & 5.0   & 6.0   & 7.3   & 9.0   & 11.6  & 16.3  & 29.8 \\
    2013  & 1.0   & 3.9   & 2.7   & 3.4   & 4.2   & 5.1   & 6.1   & 7.6   & 9.7   & 13.6  & 24.6 \\
		\hline
    \end{tabular}
\end{threeparttable}
\end{table}

\begin{table}[htbp]
  \centering
  \caption{Generalized entropy measures of the gap, actual earnings for the full population (Frank copula)}\label{tab:GE2}
		\begin{threeparttable}
    \begin{tabular}{cccccccccccc}
		\hline
    Year  & $S_{\rho}$ & Theil & \footnotesize{$k=0.1$} & \footnotesize{$k=0.2$} & \footnotesize{$k=0.3$} & \footnotesize{$k=0.4$} & \footnotesize{$k=0.5$} & \footnotesize{$k=0.6$} & \footnotesize{$k=0.7$} & \footnotesize{$k=0.8$} & \footnotesize{$k=0.9$} \\
		\hline
    1976  & 11.3  & 50.2  & 5.3   & 9.9   & 14.4  & 18.9  & 23.5  & 28.4  & 33.6  & 39.7  & 47.9 \\
    1977  & 11.0  & 51.4  & 5.4   & 9.8   & 14.2  & 18.6  & 23.0  & 27.8  & 33.1  & 39.4  & 48.4 \\
    1978  & 10.9  & 52.4  & 5.4   & 9.9   & 14.2  & 18.5  & 22.9  & 27.7  & 33.1  & 39.5  & 48.9 \\
    1979  & 10.9  & 51.5  & 5.4   & 9.8   & 14.1  & 18.4  & 22.9  & 27.6  & 32.9  & 39.3  & 48.6 \\
    1980  & 9.8   & 42.9  & 4.8   & 8.9   & 12.8  & 16.8  & 20.9  & 25.2  & 29.9  & 35.4  & 43.4 \\
    1981  & 9.1   & 42.6  & 4.5   & 8.2   & 11.8  & 15.4  & 19.1  & 23.1  & 27.5  & 32.7  & 40.3 \\
    1982  & 8.4   & 40.4  & 4.4   & 7.8   & 11.2  & 14.5  & 18.0  & 21.8  & 26.1  & 31.3  & 39.5 \\
    1983  & 7.0   & 33.3  & 3.8   & 6.7   & 9.5   & 12.4  & 15.3  & 18.5  & 22.2  & 26.7  & 34.0 \\
    1984  & 6.2   & 27.8  & 3.3   & 5.9   & 8.4   & 11.0  & 13.6  & 16.4  & 19.6  & 23.5  & 29.9 \\
    1985  & 6.0   & 27.6  & 3.3   & 5.8   & 8.2   & 10.7  & 13.3  & 16.1  & 19.2  & 23.1  & 29.5 \\
    1986  & 5.8   & 25.2  & 3.3   & 5.7   & 8.1   & 10.5  & 13.0  & 15.8  & 18.8  & 22.7  & 29.4 \\
    1987  & 5.3   & 22.0  & 3.1   & 5.3   & 7.6   & 9.8   & 12.2  & 14.8  & 17.6  & 21.3  & 27.8 \\
    1988  & 4.8   & 20.1  & 3.0   & 5.1   & 7.2   & 9.3   & 11.5  & 13.9  & 16.7  & 20.4  & 27.2 \\
    1989  & 4.6   & 19.6  & 2.9   & 4.8   & 6.8   & 8.8   & 11.0  & 13.3  & 15.9  & 19.3  & 25.7 \\
    1990  & 4.5   & 18.8  & 2.9   & 4.8   & 6.8   & 8.8   & 10.9  & 13.2  & 15.9  & 19.4  & 26.1 \\
    1991  & 4.1   & 16.8  & 2.6   & 4.3   & 6.1   & 7.9   & 9.8   & 11.8  & 14.2  & 17.3  & 23.1 \\
    1992  & 3.5   & 14.5  & 2.4   & 3.9   & 5.4   & 7.0   & 8.7   & 10.6  & 12.7  & 15.6  & 21.2 \\
    1993  & 3.2   & 13.2  & 2.2   & 3.6   & 5.1   & 6.5   & 8.1   & 9.8   & 11.8  & 14.6  & 20.2 \\
    1994  & 3.0   & 12.4  & 2.2   & 3.5   & 4.9   & 6.3   & 7.8   & 9.5   & 11.4  & 14.2  & 19.9 \\
    1995  & 3.1   & 12.5  & 2.3   & 3.7   & 5.1   & 6.5   & 8.1   & 9.8   & 11.9  & 14.8  & 21.0 \\
    1996  & 2.8   & 11.6  & 2.5   & 3.8   & 5.2   & 6.6   & 8.1   & 9.9   & 12.0  & 15.2  & 22.5 \\
    1997  & 2.7   & 11.0  & 2.6   & 3.9   & 5.2   & 6.6   & 8.1   & 9.9   & 12.1  & 15.5  & 23.4 \\
    1998  & 2.7   & 11.1  & 2.7   & 4.0   & 5.3   & 6.7   & 8.3   & 10.1  & 12.4  & 15.9  & 24.1 \\
    1999  & 2.8   & 11.2  & 2.7   & 4.0   & 5.4   & 6.9   & 8.5   & 10.3  & 12.6  & 16.2  & 24.6 \\
    2000  & 2.6   & 10.5  & 2.7   & 4.0   & 5.3   & 6.7   & 8.2   & 10.0  & 12.4  & 16.0  & 24.7 \\
    2001  & 2.4   & 9.8   & 3.2   & 4.5   & 5.8   & 7.3   & 8.9   & 10.9  & 13.5  & 17.9  & 29.0 \\
    2002  & 2.3   & 9.4   & 3.1   & 4.3   & 5.6   & 7.0   & 8.6   & 10.5  & 13.1  & 17.3  & 28.2 \\
    2003  & 2.2   & 8.7   & 3.1   & 4.2   & 5.5   & 6.8   & 8.3   & 10.2  & 12.7  & 17.0  & 27.9 \\
    2004  & 2.0   & 7.9   & 2.7   & 3.7   & 4.8   & 6.0   & 7.4   & 9.0   & 11.2  & 14.9  & 24.4 \\
    2005  & 1.9   & 7.7   & 2.7   & 3.7   & 4.8   & 6.0   & 7.3   & 9.0   & 11.2  & 14.9  & 24.5 \\
    2006  & 1.9   & 7.8   & 2.7   & 3.7   & 4.8   & 6.0   & 7.3   & 9.0   & 11.2  & 14.9  & 24.4 \\
    2007  & 1.9   & 7.5   & 2.8   & 3.8   & 4.8   & 6.0   & 7.4   & 9.0   & 11.3  & 15.1  & 24.9 \\
    2008  & 1.7   & 6.7   & 2.7   & 3.6   & 4.6   & 5.7   & 7.0   & 8.6   & 10.7  & 14.5  & 24.3 \\
    2009  & 1.6   & 6.5   & 2.6   & 3.5   & 4.5   & 5.6   & 6.8   & 8.4   & 10.5  & 14.2  & 23.8 \\
    2010  & 1.3   & 5.1   & 2.5   & 3.3   & 4.1   & 5.0   & 6.1   & 7.5   & 9.5   & 13.0  & 22.5 \\
    2011  & 1.3   & 5.1   & 2.3   & 3.1   & 3.9   & 4.7   & 5.8   & 7.1   & 9.0   & 12.2  & 21.0 \\
    2012  & 1.4   & 5.5   & 2.5   & 3.3   & 4.2   & 5.2   & 6.3   & 7.8   & 9.8   & 13.4  & 22.9 \\
    2013  & 1.3   & 5.4   & 2.5   & 3.3   & 4.1   & 5.1   & 6.2   & 7.6   & 9.6   & 13.1  & 22.4 \\
		\hline
    \end{tabular}
\end{threeparttable}
\end{table}

Finally, the behavior of the generalized entropy measures is also more volatile for the potential outcomes than for the actual outcomes. Regardless, the long-term trend is a marked decrease, pointing at a reduction in earnings differences between men and women.

\begin{table}[htbp]
  \centering
  \caption{Generalized entropy measures of the gap, potential earnings for the full population (Frank copula)}\label{tab:GE3}
		\begin{threeparttable}
    \begin{tabular}{cccccccccccc}
		\hline
    Year  & $S_{\rho}$ & Theil & \footnotesize{$k=0.1$} & \footnotesize{$k=0.2$} & \footnotesize{$k=0.3$} & \footnotesize{$k=0.4$} & \footnotesize{$k=0.5$} & \footnotesize{$k=0.6$} & \footnotesize{$k=0.7$} & \footnotesize{$k=0.8$} & \footnotesize{$k=0.9$} \\
		\hline
    1976  & 11.4  & 50.2  & 5.1   & 9.7   & 14.2  & 18.7  & 23.2  & 28.0  & 33.1  & 38.9  & 46.1 \\
    1977  & 10.7  & 52.1  & 4.9   & 9.2   & 13.4  & 17.5  & 21.8  & 26.3  & 31.2  & 36.9  & 44.5 \\
    1978  & 14.2  & 70.2  & 6.6   & 12.3  & 17.7  & 23.2  & 28.8  & 34.7  & 41.3  & 49.1  & 59.3 \\
    1979  & 14.6  & 71.7  & 6.8   & 12.6  & 18.3  & 23.9  & 29.6  & 35.8  & 42.6  & 50.5  & 60.9 \\
    1980  & 13.3  & 58.1  & 6.0   & 11.3  & 16.6  & 21.8  & 27.2  & 32.8  & 38.7  & 45.4  & 53.7 \\
    1981  & 12.5  & 60.4  & 5.6   & 10.5  & 15.3  & 20.1  & 25.1  & 30.2  & 35.8  & 42.1  & 50.1 \\
    1982  & 12.2  & 60.1  & 5.8   & 10.7  & 15.4  & 20.2  & 25.1  & 30.3  & 36.0  & 42.7  & 52.0 \\
    1983  & 10.3  & 48.6  & 4.6   & 8.8   & 12.7  & 16.7  & 20.8  & 25.1  & 29.7  & 35.0  & 41.7 \\
    1984  & 6.0   & 27.3  & 2.7   & 5.1   & 7.5   & 9.8   & 12.2  & 14.7  & 17.4  & 20.4  & 24.1 \\
    1985  & 9.4   & 43.7  & 4.3   & 8.1   & 11.8  & 15.5  & 19.2  & 23.2  & 27.5  & 32.4  & 39.0 \\
    1986  & 8.9   & 39.3  & 4.2   & 7.8   & 11.3  & 14.9  & 18.6  & 22.4  & 26.5  & 31.2  & 37.5 \\
    1987  & 9.7   & 40.7  & 4.5   & 8.5   & 12.4  & 16.3  & 20.3  & 24.5  & 28.9  & 33.9  & 40.7 \\
    1988  & 5.8   & 24.3  & 2.8   & 5.2   & 7.6   & 10.0  & 12.4  & 15.0  & 17.7  & 20.9  & 25.5 \\
    1989  & 5.1   & 21.8  & 2.4   & 4.5   & 6.5   & 8.6   & 10.7  & 12.9  & 15.3  & 17.9  & 21.5 \\
    1990  & 6.3   & 26.3  & 2.9   & 5.5   & 8.1   & 10.6  & 13.3  & 16.0  & 18.8  & 22.0  & 26.3 \\
    1991  & 4.4   & 18.2  & 2.1   & 3.9   & 5.6   & 7.4   & 9.3   & 11.2  & 13.2  & 15.4  & 18.5 \\
    1992  & 3.8   & 15.8  & 1.8   & 3.4   & 4.9   & 6.5   & 8.1   & 9.8   & 11.5  & 13.5  & 16.3 \\
    1993  & 2.5   & 10.3  & 1.2   & 2.2   & 3.2   & 4.2   & 5.2   & 6.3   & 7.4   & 8.7   & 10.4 \\
    1994  & 2.7   & 11.2  & 1.3   & 2.4   & 3.5   & 4.7   & 5.8   & 7.0   & 8.3   & 9.7   & 11.8 \\
    1995  & 2.4   & 9.8   & 1.2   & 2.2   & 3.2   & 4.2   & 5.3   & 6.4   & 7.5   & 8.9   & 11.1 \\
    1996  & 5.2   & 21.5  & 3.1   & 5.3   & 7.5   & 9.8   & 12.2  & 14.7  & 17.5  & 21.2  & 27.8 \\
    1997  & 4.5   & 18.5  & 2.9   & 4.8   & 6.8   & 8.8   & 11.0  & 13.3  & 15.9  & 19.4  & 26.1 \\
    1998  & 2.8   & 11.2  & 2.0   & 3.2   & 4.4   & 5.7   & 7.1   & 8.6   & 10.3  & 12.8  & 17.9 \\
    1999  & 5.0   & 20.5  & 3.4   & 5.5   & 7.7   & 10.0  & 12.4  & 15.0  & 18.0  & 22.1  & 30.2 \\
    2000  & 3.8   & 15.6  & 3.0   & 4.8   & 6.5   & 8.4   & 10.4  & 12.6  & 15.2  & 19.1  & 27.4 \\
    2001  & 3.6   & 14.6  & 3.8   & 5.6   & 7.4   & 9.3   & 11.5  & 14.0  & 17.3  & 22.3  & 34.6 \\
    2002  & 3.0   & 12.2  & 3.3   & 4.8   & 6.3   & 8.0   & 9.8   & 11.9  & 14.7  & 19.1  & 29.7 \\
    2003  & 4.3   & 17.7  & 3.9   & 5.9   & 7.9   & 10.1  & 12.5  & 15.2  & 18.5  & 23.4  & 34.7 \\
    2004  & 3.9   & 16.1  & 3.4   & 5.2   & 7.1   & 9.0   & 11.1  & 13.5  & 16.5  & 20.8  & 30.7 \\
    2005  & 3.9   & 15.9  & 3.5   & 5.3   & 7.1   & 9.1   & 11.2  & 13.6  & 16.7  & 21.1  & 31.4 \\
    2006  & 5.2   & 21.7  & 4.8   & 7.3   & 9.8   & 12.4  & 15.3  & 18.7  & 22.8  & 29.1  & 43.6 \\
    2007  & 4.9   & 20.1  & 3.8   & 6.0   & 8.3   & 10.6  & 13.1  & 15.9  & 19.3  & 24.1  & 34.6 \\
    2008  & 5.1   & 21.9  & 4.1   & 6.4   & 8.7   & 11.2  & 13.8  & 16.8  & 20.4  & 25.6  & 37.0 \\
    2009  & 2.9   & 11.7  & 2.7   & 4.1   & 5.5   & 7.0   & 8.6   & 10.5  & 12.8  & 16.4  & 24.7 \\
    2010  & 1.7   & 7.1   & 2.1   & 2.9   & 3.8   & 4.8   & 5.9   & 7.2   & 8.9   & 11.6  & 18.5 \\
    2011  & 4.0   & 18.0  & 3.1   & 4.9   & 6.7   & 8.6   & 10.7  & 12.9  & 15.7  & 19.7  & 28.3 \\
    2012  & 1.3   & 5.3   & 1.8   & 2.5   & 3.2   & 4.0   & 4.9   & 6.0   & 7.5   & 9.9   & 16.2 \\
    2013  & 3.6   & 14.8  & 3.4   & 5.0   & 6.8   & 8.6   & 10.6  & 12.9  & 15.8  & 20.2  & 30.2 \\
		\hline
    \end{tabular}
\end{threeparttable}
\end{table}

\clearpage

\subsection{Specification as in Maasoumi and Wang (2019)}

The main specification difference relative to \cite{Maasoumi2019} is the inclusion of the covariate number of children in the quantile regressions. Tables~\ref{tab:y1old}-\ref{tab:dy1p90old} show the estimates when this variable is not included. Relative to the baseline specification, the distributions of actual earnings for participants are almost unaltered for both genders. However, the decomposition presents some differences. Specifically, the coefficients components has been between 0.05 and 0.15 points smaller than in the baseline specification, presenting some volatility across time. In contrast, the selection and participation components are larger. The former accounts to about two thirds of the magnitude of the change in the coefficients component, whereas the latter accounts for approximately the remaining third. Finally, the endowments component has remained unaltered.

\begin{table}[htbp]
  \centering
  \caption{Actual earnings distributions for participants by gender (Frank copula)}\label{tab:y1old}
		\begin{threeparttable}
\begin{tabular}{ccccccc|cccccc}
		\hline
          & \multicolumn{6}{c}{Male}                      & \multicolumn{6}{c}{Female} \\
    Year  & Mean  & P10   & P25   & P50   & P75   & P90   & Mean  & P10   & P25   & P50   & P75   & P90 \\
		\hline
    1976  & 2.72  & 2.03  & 2.40  & 2.77  & 3.09  & 3.35  & 2.31  & 1.76  & 2.01  & 2.31  & 2.61  & 2.87 \\
    1977  & 2.73  & 2.03  & 2.40  & 2.78  & 3.10  & 3.37  & 2.33  & 1.79  & 2.04  & 2.33  & 2.63  & 2.89 \\
    1978  & 2.73  & 2.03  & 2.40  & 2.78  & 3.12  & 3.38  & 2.33  & 1.78  & 2.04  & 2.33  & 2.64  & 2.90 \\
    1979  & 2.75  & 2.05  & 2.40  & 2.79  & 3.13  & 3.39  & 2.34  & 1.80  & 2.05  & 2.34  & 2.64  & 2.91 \\
    1980  & 2.73  & 2.03  & 2.39  & 2.78  & 3.12  & 3.38  & 2.33  & 1.80  & 2.05  & 2.33  & 2.63  & 2.90 \\
    1981  & 2.68  & 1.97  & 2.34  & 2.73  & 3.07  & 3.34  & 2.30  & 1.76  & 2.01  & 2.30  & 2.60  & 2.87 \\
\hline
    \end{tabular}
\end{threeparttable}
\end{table}

\begin{table}[htbp]
  \centering
  \caption{Mean decomposition, actual earnings distributions for participants by gender (Frank copula)}\label{tab:dy1old}
		\begin{threeparttable}
    \begin{tabular}{cr@{\,}lr@{\,}lr@{\,}lr@{\,}lr@{\,}l}
		\hline
    Year  & \multicolumn{2}{c}{Total} & \multicolumn{2}{c}{EC} & \multicolumn{2}{c}{CC} & \multicolumn{2}{c}{SC} & \multicolumn{2}{c}{PC} \\
		\hline
    1976  & 0.41  & ***   & 0.01  & ***   & 0.32  & ***   & 0.00  &       & 0.08  & *** \\
    1977  & 0.40  & ***   & 0.01  & ***   & 0.31  & ***   & 0.01  &       & 0.07  & *** \\
    1978  & 0.41  & ***   & 0.01  & ***   & 0.37  & ***   & -0.02 &       & 0.04  & *** \\
    1979  & 0.41  & ***   & 0.01  & ***   & 0.38  & ***   & -0.03 &       & 0.04  & *** \\
    1980  & 0.40  & ***   & 0.01  & ***   & 0.36  & ***   & -0.02 &       & 0.04  & *** \\
    1981  & 0.39  & ***   & 0.00  & *     & 0.35  & ***   & -0.02 &       & 0.04  & *** \\
    1982  & 0.39  & ***   & 0.01  & ***   & 0.35  & ***   & -0.01 &       & 0.03  & *** \\
\hline
    \end{tabular}\begin{tablenotes}[para,flushleft]
\begin{spacing}{1}
{\footnotesize Notes: Total, EC and CC respectively denote total difference, endowments component and coefficients component.}
\end{spacing}
\end{tablenotes}
\end{threeparttable}
\end{table}

\begin{table}[htbp]
  \centering
  \caption{10th percentile decomposition, actual earnings for participants(Frank copula)}\label{tab:dy1p10old}
		\begin{threeparttable}
    \begin{tabular}{cr@{\,}lr@{\,}lr@{\,}lr@{\,}lr@{\,}l}
		\hline
    Year  & \multicolumn{2}{c}{Total} & \multicolumn{2}{c}{EC} & \multicolumn{2}{c}{CC} & \multicolumn{2}{c}{SC} & \multicolumn{2}{c}{PC} \\
		\hline
    1976  & 0.27  & ***   & -0.01 & ***   & 0.22  & ***   & 0.00  &       & 0.06  & *** \\
    1977  & 0.24  & ***   & -0.01 & ***   & 0.19  & ***   & 0.01  &       & 0.05  & *** \\
    1978  & 0.25  & ***   & -0.01 & ***   & 0.24  & ***   & -0.02 &       & 0.04  & *** \\
    1979  & 0.25  & ***   & -0.01 &       & 0.24  & ***   & -0.02 &       & 0.03  & *** \\
    1980  & 0.23  & ***   & -0.01 & ***   & 0.22  & ***   & -0.02 &       & 0.04  & *** \\
    1981  & 0.21  & ***   & -0.02 & ***   & 0.21  & ***   & -0.01 &       & 0.03  & *** \\
    1982  & 0.22  & ***   & -0.01 & **    & 0.21  & ***   & -0.01 &       & 0.03  & *** \\
\hline
    \end{tabular}\begin{tablenotes}[para,flushleft]
\begin{spacing}{1}
{\footnotesize Notes: Total, EC and CC respectively denote total difference, endowments component and coefficients component.}
\end{spacing}
\end{tablenotes}
\end{threeparttable}
\end{table}

\begin{table}[htbp]
  \centering
  \caption{25th percentile decomposition, actual earnings for participants(Frank copula)}\label{tab:dy1p25old}
		\begin{threeparttable}
    \begin{tabular}{cr@{\,}lr@{\,}lr@{\,}lr@{\,}lr@{\,}l}
		\hline
    Year  & \multicolumn{2}{c}{Total} & \multicolumn{2}{c}{EC} & \multicolumn{2}{c}{CC} & \multicolumn{2}{c}{SC} & \multicolumn{2}{c}{PC} \\
		\hline
    1976  & 0.38  & ***   & 0.00  & ***   & 0.30  & ***   & 0.00  &       & 0.08  & *** \\
    1977  & 0.36  & ***   & 0.00  & ***   & 0.29  & ***   & 0.01  &       & 0.06  & *** \\
    1978  & 0.36  & ***   & 0.00  & ***   & 0.34  & ***   & -0.02 &       & 0.04  & *** \\
    1979  & 0.35  & ***   & 0.00  & *     & 0.34  & ***   & -0.02 &       & 0.04  & *** \\
    1980  & 0.35  & ***   & 0.00  & ***   & 0.33  & ***   & -0.02 &       & 0.04  & *** \\
    1981  & 0.33  & ***   & -0.01 & ***   & 0.31  & ***   & -0.02 &       & 0.04  & *** \\
    1982  & 0.33  & ***   & 0.00  & ***   & 0.31  & ***   & -0.01 &       & 0.03  & *** \\
\hline
    \end{tabular}\begin{tablenotes}[para,flushleft]
\begin{spacing}{1}
{\footnotesize Notes: Total, EC and CC respectively denote total difference, endowments component and coefficients component.}
\end{spacing}
\end{tablenotes}
\end{threeparttable}
\end{table}

\begin{table}[htbp]
  \centering
  \caption{50th percentile decomposition, actual earnings for participants(Frank copula)}\label{tab:dy1p50old}
		\begin{threeparttable}
    \begin{tabular}{cr@{\,}lr@{\,}lr@{\,}lr@{\,}lr@{\,}l}
		\hline
    Year  & \multicolumn{2}{c}{Total} & \multicolumn{2}{c}{EC} & \multicolumn{2}{c}{CC} & \multicolumn{2}{c}{SC} & \multicolumn{2}{c}{PC} \\
		\hline
    1976  & 0.46  & ***   & 0.01  & ***   & 0.35  & ***   & 0.00  &       & 0.10  & *** \\
    1977  & 0.44  & ***   & 0.01  & ***   & 0.34  & ***   & 0.01  &       & 0.08  & *** \\
    1978  & 0.45  & ***   & 0.01  & ***   & 0.41  & ***   & -0.02 &       & 0.05  & *** \\
    1979  & 0.45  & ***   & 0.01  & *     & 0.43  & ***   & -0.03 &       & 0.04  & ** \\
    1980  & 0.45  & ***   & 0.01  & ***   & 0.41  & ***   & -0.02 &       & 0.05  & *** \\
    1981  & 0.43  & ***   & 0.00  & ***   & 0.40  & ***   & -0.02 &       & 0.05  & *** \\
    1982  & 0.43  & ***   & 0.01  & ***   & 0.40  & ***   & -0.01 &       & 0.04  & ** \\
\hline
    \end{tabular}\begin{tablenotes}[para,flushleft]
\begin{spacing}{1}
{\footnotesize Notes: Total, EC and CC respectively denote total difference, endowments component and coefficients component.}
\end{spacing}
\end{tablenotes}
\end{threeparttable}
\end{table}

\begin{table}[htbp]
  \centering
  \caption{75th percentile decomposition, actual earnings for participants(Frank copula)}\label{tab:dy1p75old}
		\begin{threeparttable}
    \begin{tabular}{cr@{\,}lr@{\,}lr@{\,}lr@{\,}lr@{\,}l}
		\hline
    Year  & \multicolumn{2}{c}{Total} & \multicolumn{2}{c}{EC} & \multicolumn{2}{c}{CC} & \multicolumn{2}{c}{SC} & \multicolumn{2}{c}{PC} \\
		\hline
    1976  & 0.48  & ***   & 0.02  & ***   & 0.36  & ***   & 0.00  &       & 0.10  & *** \\
    1977  & 0.47  & ***   & 0.02  & ***   & 0.36  & ***   & 0.01  &       & 0.08  & *** \\
    1978  & 0.48  & ***   & 0.02  & ***   & 0.43  & ***   & -0.02 &       & 0.05  & *** \\
    1979  & 0.49  & ***   & 0.02  & **    & 0.46  & ***   & -0.03 &       & 0.04  & ** \\
    1980  & 0.49  & ***   & 0.02  & ***   & 0.44  & ***   & -0.02 &       & 0.05  & *** \\
    1981  & 0.48  & ***   & 0.02  & ***   & 0.43  & ***   & -0.02 &       & 0.05  & *** \\
    1982  & 0.47  & ***   & 0.02  & ***   & 0.43  & ***   & -0.01 &       & 0.04  & ** \\
\hline
    \end{tabular}\begin{tablenotes}[para,flushleft]
\begin{spacing}{1}
{\footnotesize Notes: Total, EC and CC respectively denote total difference, endowments component and coefficients component.}
\end{spacing}
\end{tablenotes}
\end{threeparttable}
\end{table}

\begin{table}[htbp]
  \centering
  \caption{90th percentile decomposition, actual earnings for participants(Frank copula)}\label{tab:dy1p90old}
		\begin{threeparttable}
    \begin{tabular}{cr@{\,}lr@{\,}lr@{\,}lr@{\,}lr@{\,}l}
		\hline
    Year  & \multicolumn{2}{c}{Total} & \multicolumn{2}{c}{EC} & \multicolumn{2}{c}{CC} & \multicolumn{2}{c}{SC} & \multicolumn{2}{c}{PC} \\
		\hline
    1976  & 0.49  & ***   & 0.04  & ***   & 0.36  & ***   & 0.00  &       & 0.09  & *** \\
    1977  & 0.48  & ***   & 0.03  & ***   & 0.37  & ***   & 0.01  &       & 0.07  & *** \\
    1978  & 0.49  & ***   & 0.03  & ***   & 0.43  & ***   & -0.02 &       & 0.05  & *** \\
    1979  & 0.48  & ***   & 0.03  & ***   & 0.45  & ***   & -0.03 &       & 0.04  & ** \\
    1980  & 0.49  & ***   & 0.03  & ***   & 0.43  & ***   & -0.02 &       & 0.05  & *** \\
    1981  & 0.47  & ***   & 0.03  & ***   & 0.42  & ***   & -0.02 &       & 0.05  & *** \\
    1982  & 0.47  & ***   & 0.03  & ***   & 0.42  & ***   & -0.01 &       & 0.03  & *** \\
\hline
    \end{tabular}\begin{tablenotes}[para,flushleft]
\begin{spacing}{1}
{\footnotesize Notes: Total, EC and CC respectively denote total difference, endowments component and coefficients component.}
\end{spacing}
\end{tablenotes}
\end{threeparttable}
\end{table}

\clearpage

\subsection{Stata vs Matlab estimates}

The QRS estimates presented in this paper differ from those in \cite{Maasoumi2019}. As mentioned in Section~\ref{sec:emp}, this is partly due to some differences in the implementation of the estimator. However, an important factor to explain the differences in the results is the worse performance of Stata relative to Matlab.\footnote{Specifically, the Stata estimates are obtained with the codes available at \url{https://www.journals.uchicago.edu/doi/suppl/10.1086/701788/suppl_file/2012616data.zip}, which are based on the Stata package qreg, while the Matlab estimates are obtained using the rq.m code available at \url{http://www.econ.uiuc.edu/~roger/research/rq/rq.m}.} To see this, consider the estimation of the QRS estimates using the exact specification and implementation used by \cite{Maasoumi2019} in both statistical packages.

Both sets of estimates of the Frank copula for both genders are shown in Figure~\ref{fig:svm}. The original estimates are slightly smaller in general, which means that the estimated amount of self-selection is slightly more positive.

\begin{figure}[htbp]
\caption{Frank copula parameter estimates: Matlab vs Stata}
\includegraphics[width=16.5cm]{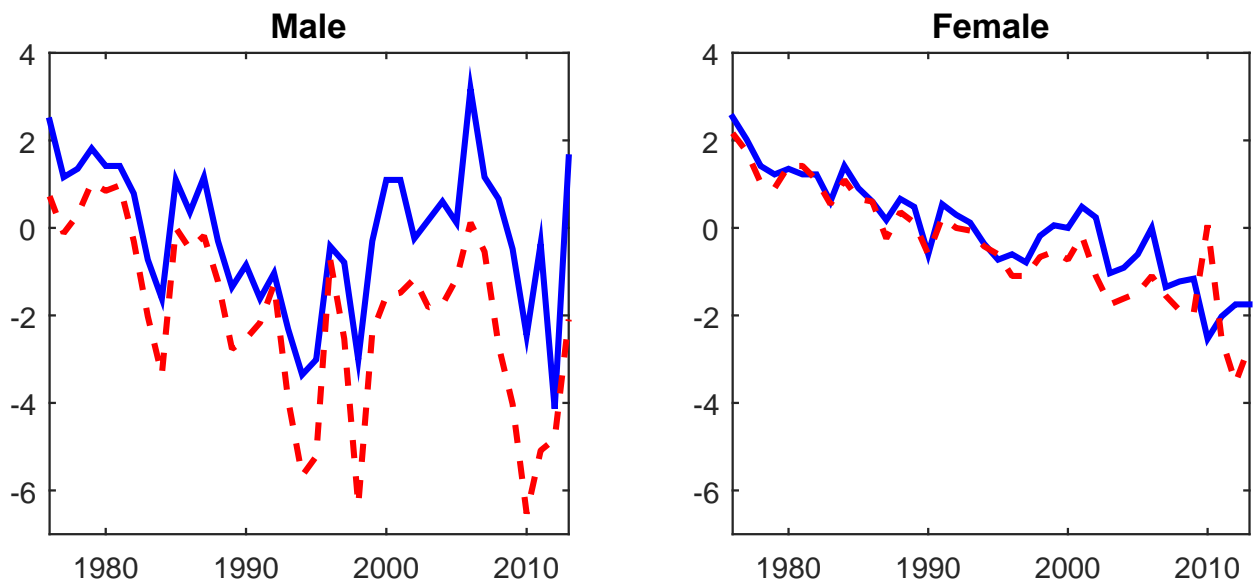}\label{fig:svm}

{\footnotesize Notes: the solid thick blue lines denote the Matlab estimates and the dashed red lines denote the Stata estimates.}
\end{figure}

To understand why the two sets of estimates are different, let us focus on the estimates for male workers in 1976. These estimates minimize the value of the criterion function defined in Equation~\ref{eq:thetahat}, which in turn depend on the estimates $\hat{\beta}_{d}\left(\tau;t\right)$ given by Equation~\ref{eq:betahat}. Fix the value of the Frank copula to be the first one considered in the estimation grid: $t=-42.889$. The minimized values of Equation~\ref{eq:betahat} for the quantiles used in the estimation of $\theta_{d}$ for both sets of estimates are shown in Table~\ref{tab:svm}. These are larger for the Stata estimates than for the Matlab ones. As such, when the slope parameter estimates are plugged into Equation~\ref{eq:thetahat}, this is different for the estimates with both statistical packages. As a consequence, the estimates of the copula parameters are different.

\begin{table}[htbp]
  \centering
  \caption{Minimized value of check function}\label{tab:svm}
		\begin{threeparttable}
    \begin{tabular}{cccccc}
		\hline
          & \multicolumn{5}{c}{$\tau$} \\
          & 0.3   & 0.4   & 0.5   & 0.6   & 0.7 \\
		\hline
    Matlab estimates & 1216.086 & 2044.821 & 2672.171 & 3051.696 & 3090.995 \\
    Stata estimates & 1216.519 & 2045.354 & 2673.005 & 3052.189 & 3091.509 \\
		\hline
        \end{tabular}\begin{tablenotes}[para,flushleft]
\begin{spacing}{1}
{\footnotesize Notes: male sample, year 1976, Frank copula with $t=-42.889$.}
\end{spacing}
\end{tablenotes}
\end{threeparttable}
\end{table}

One final aspect regards the computational time required to obtain the estimates with each statistical package. Whereas those in Stata are obtained using the simplex method, the Matlab estimates are based on the interior point method, which are substantially faster \citep{Portnoy1997,Chernozhukov2022}. Therefore, this allows to obtain a larger set of estimates, including more flexible models, such as those with heterogeneous copulas.
\clearpage
\section{Additional Tables and Figures}\label{app:tabs}

\begin{table}[htbp]
  \centering
  \caption{Kendall's $\tau$ correlation coefficients for the Gaussian copula}\label{tab:heck}
		\begin{threeparttable}
    \begin{tabular}{cr@{\,}lr@{\,}lr@{\,}lr@{\,}l}
		\hline
							& \multicolumn{4}{c}{Frank copula} & \multicolumn{4}{c}{Gaussian copula} \\
    Year  & \multicolumn{2}{c}{Male} & \multicolumn{2}{c}{Female} & \multicolumn{2}{c}{Male} & \multicolumn{2}{c}{Female} \\
		\hline
    1976  & 0.28  & ***   & 0.13  & ***   & 0.31  & **    & 0.10  & ** \\
    1977  & 0.16  &       & 0.08  & *     & 0.17  &       & 0.07  & ** \\
    1978  & 0.21  & *     & -0.01 &       & 0.23  & *     & -0.02 &  \\
    1979  & 0.24  & **    & -0.02 &       & 0.22  & **    & -0.02 &  \\
    1980  & 0.24  & **    & 0.01  &       & 0.22  & *     & -0.01 &  \\
    1981  & 0.23  & ***   & 0.01  &       & 0.22  & *     & 0.02  &  \\
    1982  & 0.15  &       & -0.03 &       & 0.14  &       & -0.02 &  \\
    1983  & 0.02  &       & -0.09 & **    & -0.01 &       & -0.10 & ** \\
    1984  & -0.08 &       & 0.02  &       & -0.08 &       & 0.01  &  \\
    1985  & 0.12  &       & -0.05 &       & 0.09  &       & -0.05 &  \\
    1986  & 0.09  &       & -0.10 & ***   & 0.10  &       & -0.10 & *** \\
    1987  & 0.11  &       & -0.17 & ***   & 0.06  &       & -0.19 & *** \\
    1988  & -0.05 &       & -0.10 & ***   & -0.04 &       & -0.10 & *** \\
    1989  & -0.12 &       & -0.10 & **    & -0.12 &       & -0.03 &  \\
    1990  & -0.09 &       & -0.23 & ***   & -0.09 &       & -0.22 & *** \\
    1991  & -0.05 &       & -0.10 & **    & -0.07 &       & -0.11 & ** \\
    1992  & -0.09 &       & -0.11 & ***   & -0.07 &       & -0.10 & *** \\
    1993  & -0.28 & **    & -0.13 & ***   & -0.22 &       & -0.13 & *** \\
    1994  & -0.22 & *     & -0.17 & ***   & -0.24 & *     & -0.19 & *** \\
    1995  & -0.27 & **    & -0.21 & ***   & -0.28 & **    & -0.22 & *** \\
    1996  & 0.07  &       & -0.22 & ***   & 0.06  &       & -0.22 & *** \\
    1997  & 0.01  &       & -0.21 & ***   & -0.10 &       & -0.17 & *** \\
    1998  & -0.22 &       & -0.21 & ***   & -0.34 & **    & -0.23 & *** \\
    1999  & 0.06  &       & -0.22 & ***   & -0.01 &       & -0.26 & *** \\
    2000  & 0.10  &       & -0.13 & *     & 0.06  &       & -0.10 &  \\
    2001  & 0.13  &       & -0.12 & **    & 0.13  &       & -0.14 & ** \\
    2002  & 0.04  &       & -0.13 & **    & -0.01 &       & -0.11 & * \\
    2003  & 0.06  &       & -0.25 & ***   & 0.02  &       & -0.23 & *** \\
    2004  & 0.10  &       & -0.21 & ***   & 0.03  &       & -0.21 & *** \\
    2005  & 0.08  &       & -0.22 & ***   & 0.08  &       & -0.25 & *** \\
    2006  & 0.29  & ***   & -0.17 & ***   & 0.19  & **    & -0.18 & *** \\
    2007  & 0.10  &       & -0.30 & ***   & 0.03  &       & -0.28 & *** \\
    2008  & 0.13  &       & -0.31 & ***   & 0.06  &       & -0.37 & *** \\
    2009  & -0.06 &       & -0.26 & ***   & -0.14 &       & -0.26 & *** \\
    2010  & -0.18 &       & -0.28 & ***   & -0.23 & **    & -0.25 & *** \\
    2011  & -0.01 &       & -0.32 & ***   & -0.10 &       & -0.34 & *** \\
    2012  & -0.26 & *     & -0.28 & ***   & -0.32 & **    & -0.36 & *** \\
    2013  & 0.06  &       & -0.27 & ***   & -0.05 &       & -0.31 & *** \\
		\hline
    \end{tabular}\begin{tablenotes}[para,flushleft]
\begin{spacing}{1}
{\footnotesize Notes: Kendall's $\tau$ correlation coefficients of the copula estimates by year and gender. QRS and H2S respectively denote the estimates of the quantile regression with selection and Heckman 2-stage estimators.}
\end{spacing}
\end{tablenotes}
\end{threeparttable}
\end{table}

\begin{figure}[htbp]
\caption{Mean value of $u$ for participants, MA}
\includegraphics[width=16.5cm]{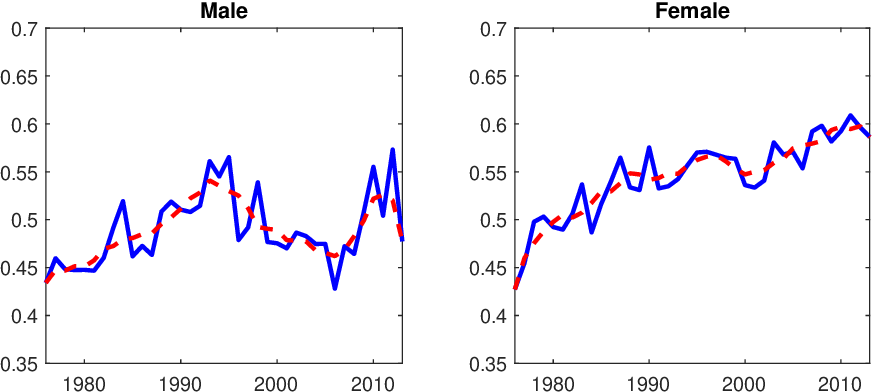}\label{fig:meanuma}

{\footnotesize Notes: the solid blue line denotes the estimate with the Frank copula for all participants; the dashed red line denotes the estimate with the Moving Average Frank copula for all participants; coefficients scaled by 100.}
\end{figure}

\begin{figure}[htbp]
\caption{Mean decomposition, actual earnings for participants, MA}
\includegraphics[width=16.5cm]{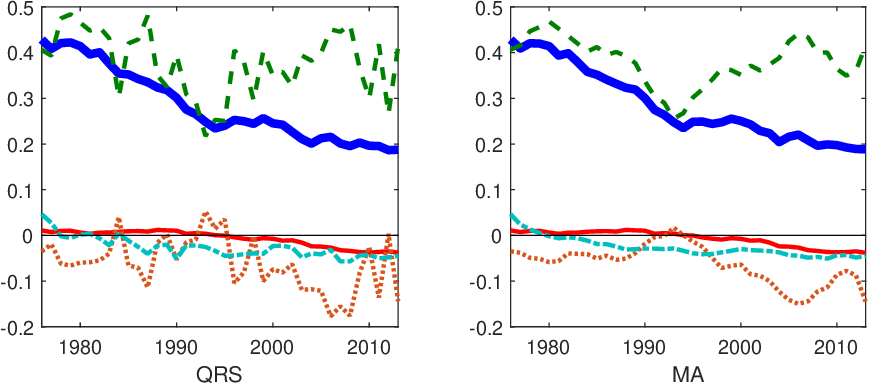}\label{fig:meanyma1}

{\footnotesize Notes: QRS and MA stand for the standard Quantile Regression with Selection estimator and its Moving Average counterpart; the solid thick blue line denotes the total gap between male and female workers; the solid thin red line denotes the endowments component; the dashed thin green line denotes the coefficients component; the dotted thin orange line denotes the selection component; the dashed-dotted thin cyan line denotes the participation component.}
\end{figure}

\begin{figure}[htbp]
\caption{Mean decomposition, actual earnings for the entire population, MA}
\includegraphics[width=16.5cm]{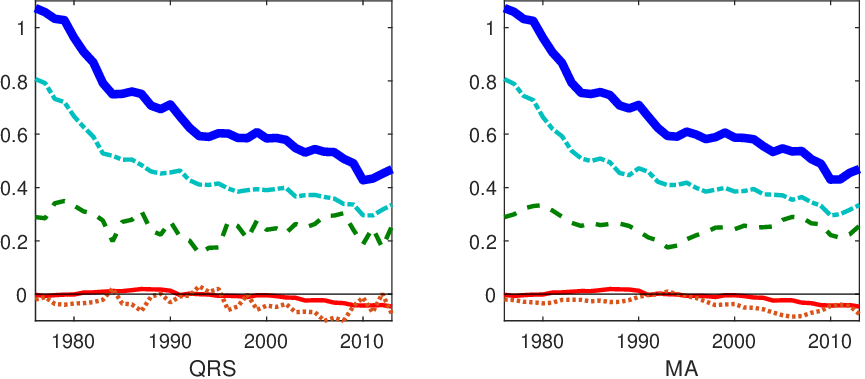}\label{fig:meanyma2}

{\footnotesize Notes: QRS and MA stand for the standard Quantile Regression with Selection estimator and its Moving Average counterpart; the solid thick blue line denotes the total gap between male and female workers; the solid thin red line denotes the endowments component; the dashed thin green line denotes the coefficients component; the dotted thin orange line denotes the selection component; the dashed-dotted thin cyan line denotes the participation component.}
\end{figure}

\clearpage

\begin{table}[htbp]
  \centering
  \caption{10th percentile decomposition, actual earnings for participants (Frank copula)}\label{tab:dy1p10g}
		\begin{threeparttable}
\begin{tabular}{cr@{\,}lr@{\,}lr@{\,}lr@{\,}lr@{\,}l}
		\hline
    Year  & \multicolumn{2}{c}{Total} & \multicolumn{2}{c}{EC} & \multicolumn{2}{c}{CC} & \multicolumn{2}{c}{SC} & \multicolumn{2}{c}{PC} \\
		\hline
    1976  & 0.29  & ***   & -0.01 & ***   & 0.30  & ***   & -0.03 &       & 0.04  & *** \\
    1977  & 0.25  & ***   & -0.02 & ***   & 0.26  & ***   & -0.02 &       & 0.02  & ** \\
    1978  & 0.27  & ***   & -0.01 & ***   & 0.35  & ***   & -0.06 & **    & 0.00  &  \\
    1979  & 0.27  & ***   & -0.01 & **    & 0.35  & ***   & -0.07 & ***   & -0.01 &  \\
    1980  & 0.25  & ***   & -0.01 & ***   & 0.32  & ***   & -0.06 & ***   & 0.00  &  \\
    1981  & 0.23  & ***   & -0.02 & ***   & 0.30  & ***   & -0.06 & *     & 0.00  &  \\
    1982  & 0.23  & ***   & -0.01 & ***   & 0.31  & ***   & -0.06 & *     & -0.01 &  \\
\hline
    \end{tabular}\begin{tablenotes}[para,flushleft]
\begin{spacing}{1}
{\footnotesize Notes: Total, EC, CC, SC and PC respectively denote total difference, endowments component, coefficients component, selection component and participation component.}
\end{spacing}
\end{tablenotes}
\end{threeparttable}
\end{table}

\clearpage

\begin{table}[htbp]
  \centering
  \caption{25th percentile decomposition, actual earnings for participants (Frank copula)}\label{tab:dy1p25g}
		\begin{threeparttable}
    \begin{tabular}{cr@{\,}lr@{\,}lr@{\,}lr@{\,}lr@{\,}l}
		\hline
    Year  & \multicolumn{2}{c}{Total} & \multicolumn{2}{c}{EC} & \multicolumn{2}{c}{CC} & \multicolumn{2}{c}{SC} & \multicolumn{2}{c}{PC} \\
		\hline
    1976  & 0.39  & ***   & 0.00  &       & 0.38  & ***   & -0.03 &       & 0.05  & *** \\
    1977  & 0.37  & ***   & -0.01 &       & 0.37  & ***   & -0.02 &       & 0.03  & * \\
    1978  & 0.37  & ***   & 0.00  &       & 0.44  & ***   & -0.06 & **    & 0.00  &  \\
    1979  & 0.36  & ***   & 0.00  &       & 0.43  & ***   & -0.07 & **    & -0.01 &  \\
    1980  & 0.36  & ***   & -0.01 & *     & 0.42  & ***   & -0.06 & ***   & 0.00  &  \\
    1981  & 0.34  & ***   & -0.01 & **    & 0.40  & ***   & -0.06 & *     & 0.00  &  \\
    1982  & 0.34  & ***   & -0.01 &       & 0.41  & ***   & -0.06 & *     & -0.01 &  \\
\hline
    \end{tabular}\begin{tablenotes}[para,flushleft]
\begin{spacing}{1}
{\footnotesize Notes: Total, EC, CC, SC and PC respectively denote total difference, endowments component, coefficients component, selection component and participation component.}
\end{spacing}
\end{tablenotes}
\end{threeparttable}
\end{table}

\clearpage

\begin{table}[htbp]
  \centering
  \caption{50th percentile decomposition, actual earnings for participants (Frank copula)}\label{tab:dy1p50g}
		\begin{threeparttable}
    \begin{tabular}{cr@{\,}lr@{\,}lr@{\,}lr@{\,}lr@{\,}l}
		\hline
    Year  & \multicolumn{2}{c}{Total} & \multicolumn{2}{c}{EC} & \multicolumn{2}{c}{CC} & \multicolumn{2}{c}{SC} & \multicolumn{2}{c}{PC} \\
		\hline
    1976  & 0.47  & ***   & 0.01  & ***   & 0.44  & ***   & -0.04 &       & 0.05  & *** \\
    1977  & 0.45  & ***   & 0.01  & ***   & 0.43  & ***   & -0.02 &       & 0.03  & * \\
    1978  & 0.46  & ***   & 0.01  & ***   & 0.52  & ***   & -0.07 & **    & 0.00  &  \\
    1979  & 0.46  & ***   & 0.01  & ***   & 0.53  & ***   & -0.07 & **    & -0.01 &  \\
    1980  & 0.46  & ***   & 0.01  & ***   & 0.51  & ***   & -0.06 & **    & 0.00  &  \\
    1981  & 0.44  & ***   & 0.00  &       & 0.50  & ***   & -0.06 & *     & 0.00  &  \\
    1982  & 0.44  & ***   & 0.01  & *     & 0.50  & ***   & -0.06 & *     & -0.01 &  \\
\hline
    \end{tabular}\begin{tablenotes}[para,flushleft]
\begin{spacing}{1}
{\footnotesize Notes: Total, EC, CC, SC and PC respectively denote total difference, endowments component, coefficients component, selection component and participation component.}
\end{spacing}
\end{tablenotes}
\end{threeparttable}
\end{table}

\clearpage

\begin{table}[htbp]
  \centering
  \caption{75th percentile decomposition, actual earnings for participants (Frank copula)}\label{tab:dy1p75g}
		\begin{threeparttable}
    \begin{tabular}{cr@{\,}lr@{\,}lr@{\,}lr@{\,}lr@{\,}l}
		\hline
    Year  & \multicolumn{2}{c}{Total} & \multicolumn{2}{c}{EC} & \multicolumn{2}{c}{CC} & \multicolumn{2}{c}{SC} & \multicolumn{2}{c}{PC} \\
		\hline
    1976  & 0.49  & ***   & 0.02  & ***   & 0.45  & ***   & -0.04 &       & 0.05  & *** \\
    1977  & 0.48  & ***   & 0.02  & ***   & 0.45  & ***   & -0.02 &       & 0.03  & * \\
    1978  & 0.49  & ***   & 0.02  & ***   & 0.54  & ***   & -0.07 & *     & 0.00  &  \\
    1979  & 0.50  & ***   & 0.02  & ***   & 0.56  & ***   & -0.07 & **    & -0.01 &  \\
    1980  & 0.50  & ***   & 0.02  & ***   & 0.54  & ***   & -0.07 & **    & 0.00  &  \\
    1981  & 0.48  & ***   & 0.01  & ***   & 0.53  & ***   & -0.07 & *     & 0.00  &  \\
    1982  & 0.48  & ***   & 0.02  & ***   & 0.53  & ***   & -0.06 & *     & -0.01 &  \\
\hline
    \end{tabular}\begin{tablenotes}[para,flushleft]
\begin{spacing}{1}
{\footnotesize Notes: Total, EC, CC, SC and PC respectively denote total difference, endowments component, coefficients component, selection component and participation component.}
\end{spacing}
\end{tablenotes}
\end{threeparttable}
\end{table}

\clearpage

\begin{table}[htbp]
  \centering
  \caption{90th percentile decomposition, actual earnings for participants (Frank copula)}\label{tab:dy1p90g}
		\begin{threeparttable}
    \begin{tabular}{cr@{\,}lr@{\,}lr@{\,}lr@{\,}lr@{\,}l}
		\hline
    Year  & \multicolumn{2}{c}{Total} & \multicolumn{2}{c}{EC} & \multicolumn{2}{c}{CC} & \multicolumn{2}{c}{SC} & \multicolumn{2}{c}{PC} \\
		\hline
    1976  & 0.49  & ***   & 0.03  & ***   & 0.45  & ***   & -0.04 &       & 0.05  & *** \\
    1977  & 0.49  & ***   & 0.03  & ***   & 0.45  & ***   & -0.02 &       & 0.03  & * \\
    1978  & 0.49  & ***   & 0.03  & ***   & 0.53  & ***   & -0.07 & *     & 0.00  &  \\
    1979  & 0.49  & ***   & 0.03  & ***   & 0.54  & ***   & -0.08 & **    & -0.01 &  \\
    1980  & 0.49  & ***   & 0.03  & ***   & 0.53  & ***   & -0.07 & **    & 0.00  &  \\
    1981  & 0.47  & ***   & 0.03  & ***   & 0.51  & ***   & -0.07 & *     & 0.00  &  \\
    1982  & 0.48  & ***   & 0.03  & ***   & 0.51  & ***   & -0.05 &       & -0.01 &  \\
\hline
    \end{tabular}\begin{tablenotes}[para,flushleft]
\begin{spacing}{1}
{\footnotesize Notes: Total, EC, CC, SC and PC respectively denote total difference, endowments component, coefficients component, selection component and participation component.}
\end{spacing}
\end{tablenotes}
\end{threeparttable}
\end{table}

\clearpage

\begin{table}[htbp]
  \centering
  \caption{10th percentile decomposition, actual earnings for the full population (Frank copula)}\label{tab:dy2p10g}
		\begin{threeparttable}
    \begin{tabular}{cr@{\,}lr@{\,}lr@{\,}lr@{\,}lr@{\,}l}
		\hline
    Year  & \multicolumn{2}{c}{Total} & \multicolumn{2}{c}{EC} & \multicolumn{2}{c}{CC} & \multicolumn{2}{c}{SC} & \multicolumn{2}{c}{PC} \\
		\hline
    1976  & 0.00  &       & 0.00  &       & 0.00  &       & 0.00  &       & 0.00  &  \\
    1977  & 0.00  &       & 0.00  &       & 0.00  &       & 0.00  &       & 0.00  &  \\
    1978  & 0.00  &       & 0.00  &       & 0.00  &       & 0.00  &       & 0.00  &  \\
    1979  & 0.00  &       & 0.00  &       & 0.00  &       & 0.00  &       & 0.00  &  \\
    1980  & 0.00  &       & 0.00  &       & 0.00  &       & 0.00  &       & 0.00  &  \\
    1981  & 0.00  &       & 0.00  &       & 0.00  &       & 0.00  &       & 0.00  &  \\
    1982  & 0.00  &       & 0.00  &       & 0.00  &       & 0.00  &       & 0.00  &  \\
\hline
    \end{tabular}\begin{tablenotes}[para,flushleft]
\begin{spacing}{1}
{\footnotesize Notes: Total, EC, CC, SC and PC respectively denote total difference, endowments component, coefficients component, selection component and participation component.}
\end{spacing}
\end{tablenotes}
\end{threeparttable}
\end{table}

\clearpage

\begin{table}[htbp]
  \centering
  \caption{25th percentile decomposition, actual earnings for the full population (Frank copula)}\label{tab:dy2p25g}
		\begin{threeparttable}
    \begin{tabular}{cr@{\,}lr@{\,}lr@{\,}lr@{\,}lr@{\,}l}
		\hline
    Year  & \multicolumn{2}{c}{Total} & \multicolumn{2}{c}{EC} & \multicolumn{2}{c}{CC} & \multicolumn{2}{c}{SC} & \multicolumn{2}{c}{PC} \\
		\hline
    1976  & 0.00  &       & 0.00  &       & 0.00  &       & 0.00  &       & 0.00  &  \\
    1977  & 0.00  &       & 0.00  &       & 0.00  &       & 0.00  &       & 0.00  &  \\
    1978  & 0.00  &       & 0.00  &       & 0.00  &       & 0.00  &       & 0.00  &  \\
    1979  & 0.00  &       & 0.00  &       & 0.00  &       & 0.00  &       & 0.00  &  \\
    1980  & 0.00  &       & 0.00  &       & 0.00  &       & 0.00  &       & 0.00  &  \\
    1981  & 0.00  &       & 0.00  &       & 0.00  &       & 0.00  &       & 0.00  &  \\
    1982  & 0.00  &       & 0.00  &       & 0.00  &       & 0.00  &       & 0.00  &  \\
\hline
    \end{tabular}\begin{tablenotes}[para,flushleft]
\begin{spacing}{1}
{\footnotesize Notes: Total, EC, CC, SC and PC respectively denote total difference, endowments component, coefficients component, selection component and participation component.}
\end{spacing}
\end{tablenotes}
\end{threeparttable}
\end{table}

\clearpage

\begin{table}[htbp]
  \centering
  \caption{50th percentile decomposition, actual earnings for the full population (Frank copula)}\label{tab:dy2p50g}
		\begin{threeparttable}
    \begin{tabular}{cr@{\,}lr@{\,}lr@{\,}lr@{\,}lr@{\,}l}
		\hline
    Year  & \multicolumn{2}{c}{Total} & \multicolumn{2}{c}{EC} & \multicolumn{2}{c}{CC} & \multicolumn{2}{c}{SC} & \multicolumn{2}{c}{PC} \\
		\hline
    1976  & 2.50  & ***   & 0.00  &       & 0.42  & ***   & -0.03 &       & 2.11  & *** \\
    1977  & 2.51  & ***   & 0.00  &       & 0.41  & ***   & -0.02 &       & 2.12  & *** \\
    1978  & 2.50  & ***   & 0.00  &       & 0.49  & ***   & -0.06 & **    & 2.08  & *** \\
    1979  & 2.52  & ***   & 0.00  &       & 0.48  & ***   & -0.06 & **    & 2.10  & *** \\
    1980  & 2.50  & ***   & 0.00  &       & 0.46  & ***   & -0.05 & **    & 2.09  & *** \\
    1981  & 2.42  & ***   & 0.00  &       & 0.43  & ***   & -0.05 & *     & 2.04  & *** \\
    1982  & 2.34  & ***   & 0.00  &       & 0.42  & ***   & -0.05 & *     & 1.96  & *** \\
\hline
    \end{tabular}\begin{tablenotes}[para,flushleft]
\begin{spacing}{1}
{\footnotesize Notes: Total, EC, CC, SC and PC respectively denote total difference, endowments component, coefficients component, selection component and participation component.}
\end{spacing}
\end{tablenotes}
\end{threeparttable}
\end{table}

\clearpage

\begin{table}[htbp]
  \centering
  \caption{75th percentile decomposition, actual earnings for the full population (Frank copula)}\label{tab:dy2p75g}
		\begin{threeparttable}
    \begin{tabular}{cr@{\,}lr@{\,}lr@{\,}lr@{\,}lr@{\,}l}
		\hline
    Year  & \multicolumn{2}{c}{Total} & \multicolumn{2}{c}{EC} & \multicolumn{2}{c}{CC} & \multicolumn{2}{c}{SC} & \multicolumn{2}{c}{PC} \\
		\hline
    1976  & 0.89  & ***   & 0.02  & ***   & 0.45  & ***   & -0.03 &       & 0.45  & *** \\
    1977  & 0.86  & ***   & 0.02  & ***   & 0.46  & ***   & -0.02 &       & 0.40  & *** \\
    1978  & 0.84  & ***   & 0.02  & ***   & 0.54  & ***   & -0.06 & **    & 0.34  & *** \\
    1979  & 0.81  & ***   & 0.02  & ***   & 0.56  & ***   & -0.07 & **    & 0.30  & *** \\
    1980  & 0.77  & ***   & 0.02  & ***   & 0.54  & ***   & -0.06 & **    & 0.27  & *** \\
    1981  & 0.75  & ***   & 0.02  & ***   & 0.53  & ***   & -0.06 & *     & 0.27  & *** \\
    1982  & 0.76  & ***   & 0.02  & ***   & 0.52  & ***   & -0.06 & *     & 0.28  & *** \\
\hline
    \end{tabular}\begin{tablenotes}[para,flushleft]
\begin{spacing}{1}
{\footnotesize Notes: Total, EC, CC, SC and PC respectively denote total difference, endowments component, coefficients component, selection component and participation component.}
\end{spacing}
\end{tablenotes}
\end{threeparttable}
\end{table}

\clearpage

\begin{table}[htbp]
  \centering
  \caption{90th percentile decomposition, actual earnings for the full population (Frank copula)}\label{tab:dy2p90g}
		\begin{threeparttable}
    \begin{tabular}{cr@{\,}lr@{\,}lr@{\,}lr@{\,}lr@{\,}l}
		\hline
    Year  & \multicolumn{2}{c}{Total} & \multicolumn{2}{c}{EC} & \multicolumn{2}{c}{CC} & \multicolumn{2}{c}{SC} & \multicolumn{2}{c}{PC} \\
		\hline
    1976  & 0.68  & ***   & 0.03  & ***   & 0.45  & ***   & -0.04 &       & 0.24  & *** \\
    1977  & 0.67  & ***   & 0.03  & ***   & 0.45  & ***   & -0.02 &       & 0.21  & *** \\
    1978  & 0.67  & ***   & 0.03  & ***   & 0.53  & ***   & -0.06 & *     & 0.17  & *** \\
    1979  & 0.67  & ***   & 0.03  & ***   & 0.55  & ***   & -0.07 & **    & 0.16  & *** \\
    1980  & 0.65  & ***   & 0.03  & ***   & 0.53  & ***   & -0.06 & **    & 0.15  & *** \\
    1981  & 0.63  & ***   & 0.03  & ***   & 0.52  & ***   & -0.06 & *     & 0.15  & *** \\
    1982  & 0.63  & ***   & 0.03  & ***   & 0.52  & ***   & -0.05 & *     & 0.14  & *** \\
\hline
    \end{tabular}\begin{tablenotes}[para,flushleft]
\begin{spacing}{1}
{\footnotesize Notes: Total, EC, CC, SC and PC respectively denote total difference, endowments component, coefficients component, selection component and participation component.}
\end{spacing}
\end{tablenotes}
\end{threeparttable}
\end{table}

\begin{figure}[htbp]
\caption{Actual earnings distributions for male participants}
\includegraphics[width=16.5cm]{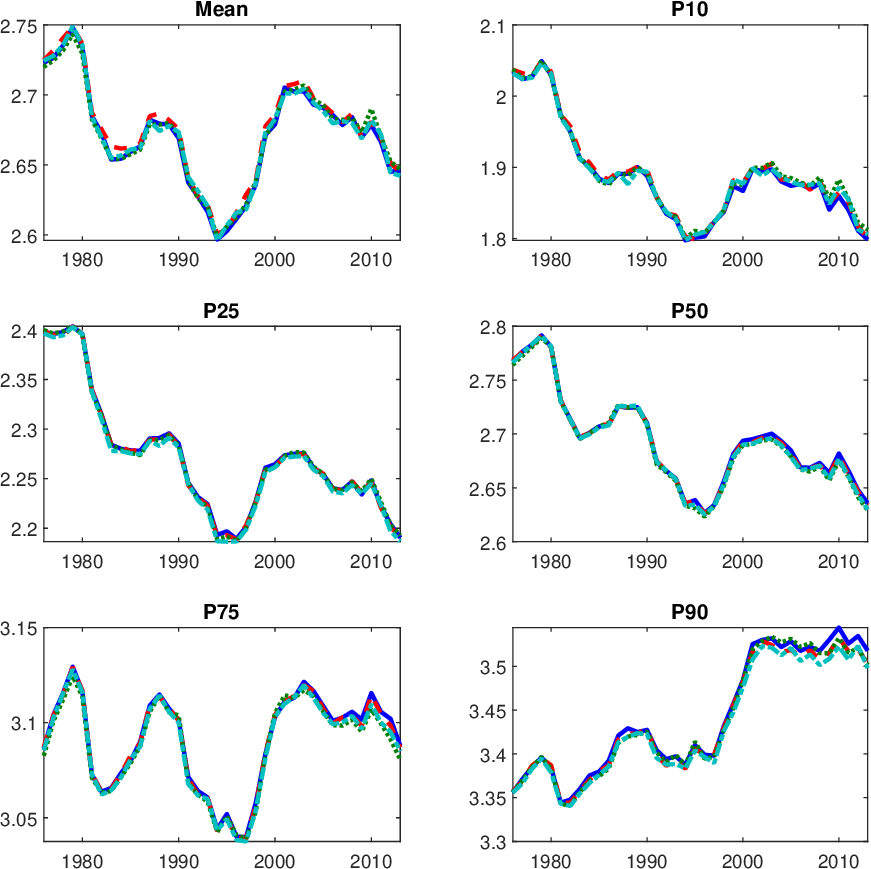}\label{fig:y1m}

{\footnotesize Notes: the solid thick blue line denotes the estimate with the Frank copula for all individuals; the solid thin red line denotes the estimate with the Gaussian copula for all individuals; the dashed green line denotes the estimate with the Frank copula, heterogeneous across race; the dotted orange line denotes the estimate with the Frank copula, heterogeneous across education level; the dashed-dotted cyan line denotes the estimate with the Frank copula, heterogeneous across marital status.}
\end{figure}

\begin{figure}[htbp]
\caption{Actual earnings distributions for the full male population}
\includegraphics[width=16.5cm]{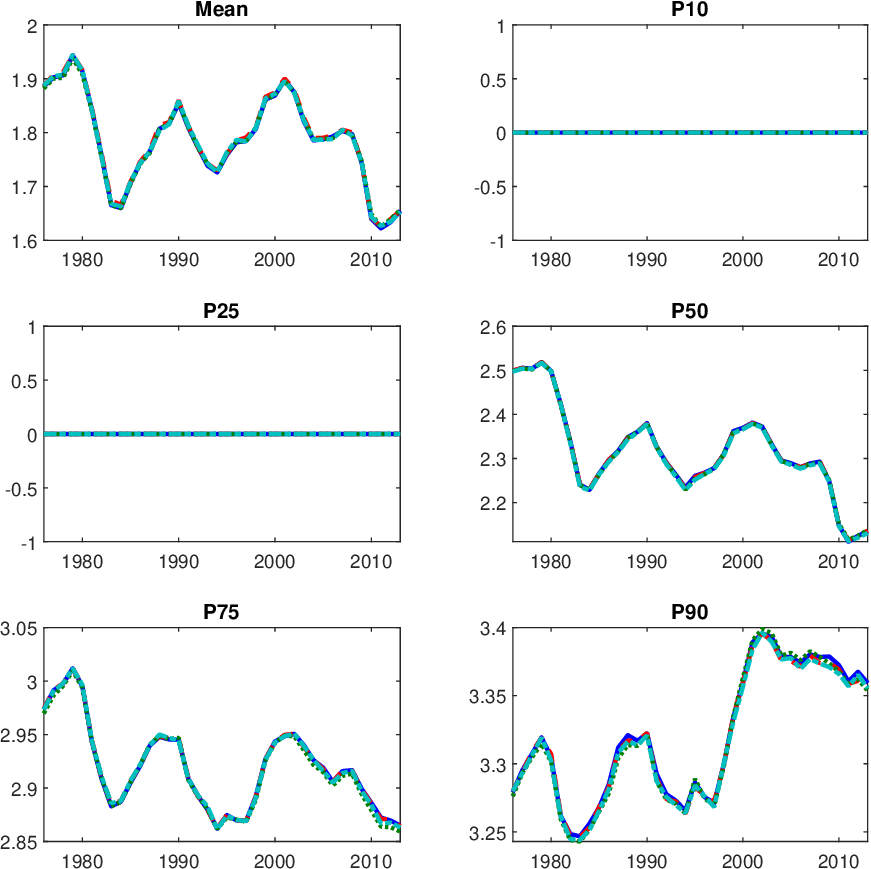}\label{fig:y2m}

{\footnotesize Notes: the solid thick blue line denotes the estimate with the Frank copula for all individuals; the solid thin red line denotes the estimate with the Gaussian copula for all individuals; the dashed green line denotes the estimate with the Frank copula, heterogeneous across race; the dotted orange line denotes the estimate with the Frank copula, heterogeneous across education level; the dashed-dotted cyan line denotes the estimate with the Frank copula, heterogeneous across marital status.}
\end{figure}

\begin{figure}[htbp]
\caption{Potential earnings distributions for the full male population}
\includegraphics[width=16.5cm]{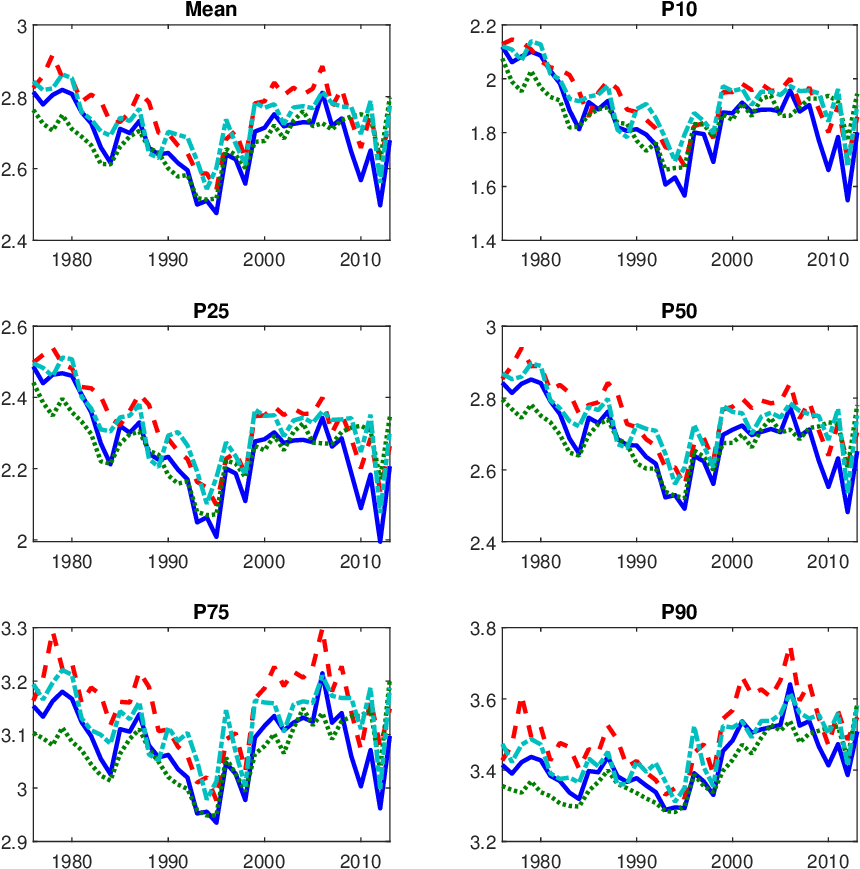}\label{fig:y3m}

{\footnotesize Notes: the solid thick blue line denotes the estimate with the Frank copula for all individuals; the solid thin red line denotes the estimate with the Gaussian copula for all individuals; the dashed green line denotes the estimate with the Frank copula, heterogeneous across race; the dotted orange line denotes the estimate with the Frank copula, heterogeneous across education level; the dashed-dotted cyan line denotes the estimate with the Frank copula, heterogeneous across marital status.}
\end{figure}

\begin{figure}[htbp]
\caption{Actual earnings distributions for female participants}
\includegraphics[width=16.5cm]{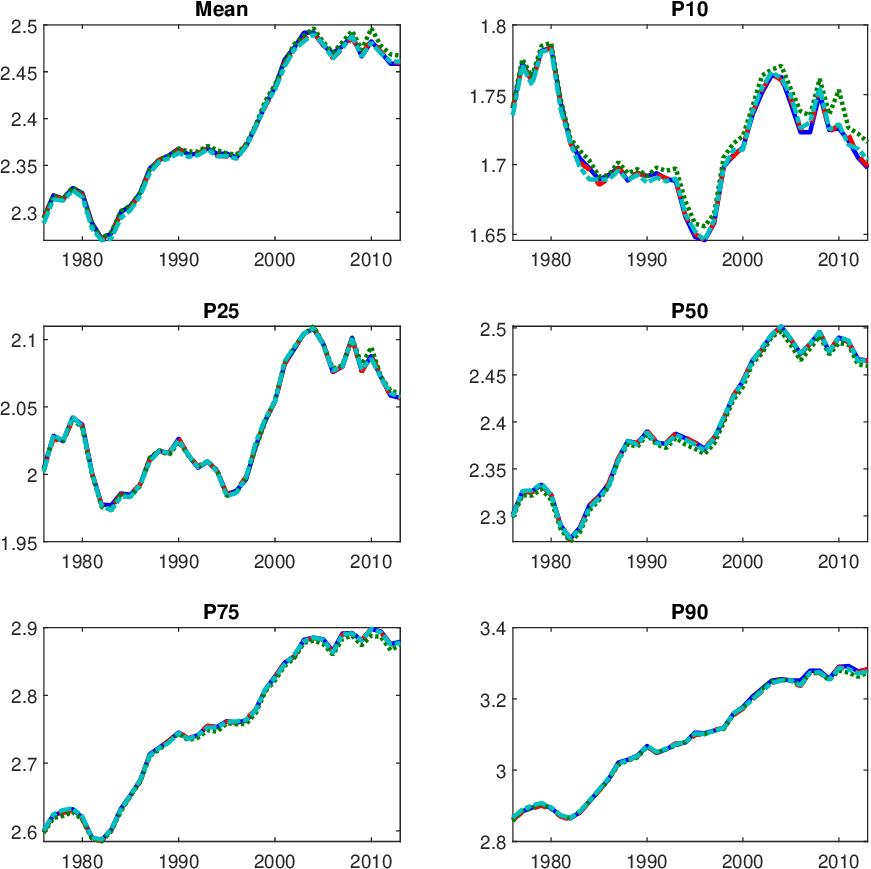}\label{fig:y1f}

{\footnotesize Notes: the solid thick blue line denotes the estimate with the Frank copula for all individuals; the solid thin red line denotes the estimate with the Gaussian copula for all individuals; the dashed green line denotes the estimate with the Frank copula, heterogeneous across race; the dotted orange line denotes the estimate with the Frank copula, heterogeneous across education level; the dashed-dotted cyan line denotes the estimate with the Frank copula, heterogeneous across marital status.}
\end{figure}

\begin{figure}[htbp]
\caption{Actual earnings distributions for the full female population}
\includegraphics[width=16.5cm]{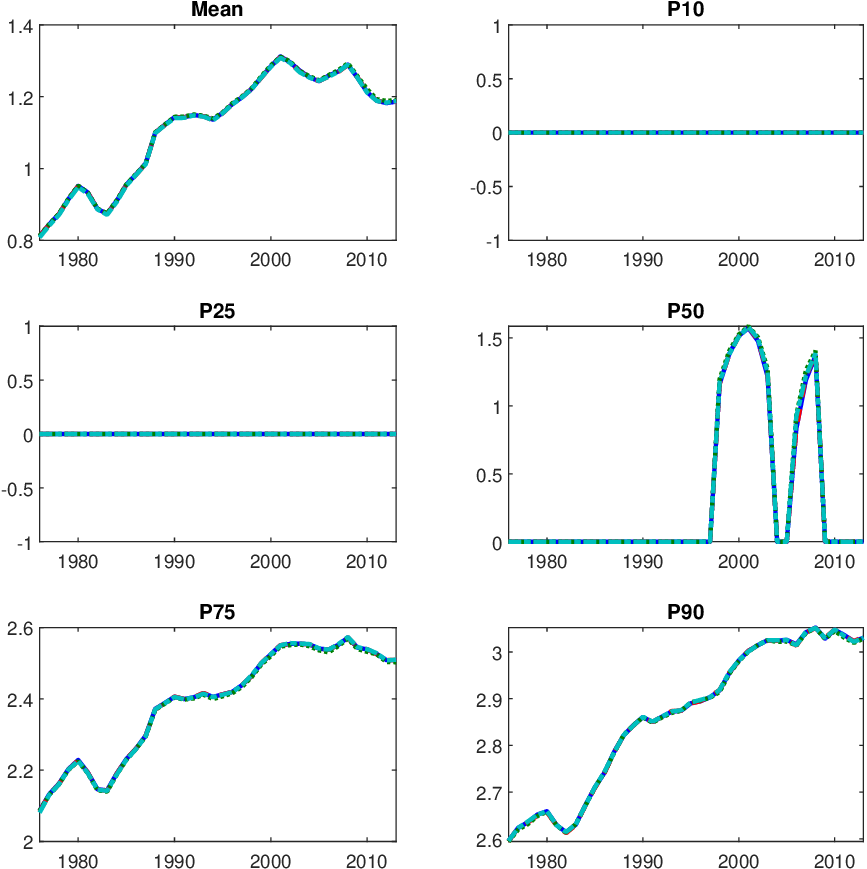}\label{fig:y2f}

{\footnotesize Notes: the solid thick blue line denotes the estimate with the Frank copula for all individuals; the solid thin red line denotes the estimate with the Gaussian copula for all individuals; the dashed green line denotes the estimate with the Frank copula, heterogeneous across race; the dotted orange line denotes the estimate with the Frank copula, heterogeneous across education level; the dashed-dotted cyan line denotes the estimate with the Frank copula, heterogeneous across marital status.}
\end{figure}

\begin{figure}[htbp]
\caption{Potential earnings distributions for the full female population}
\includegraphics[width=16.5cm]{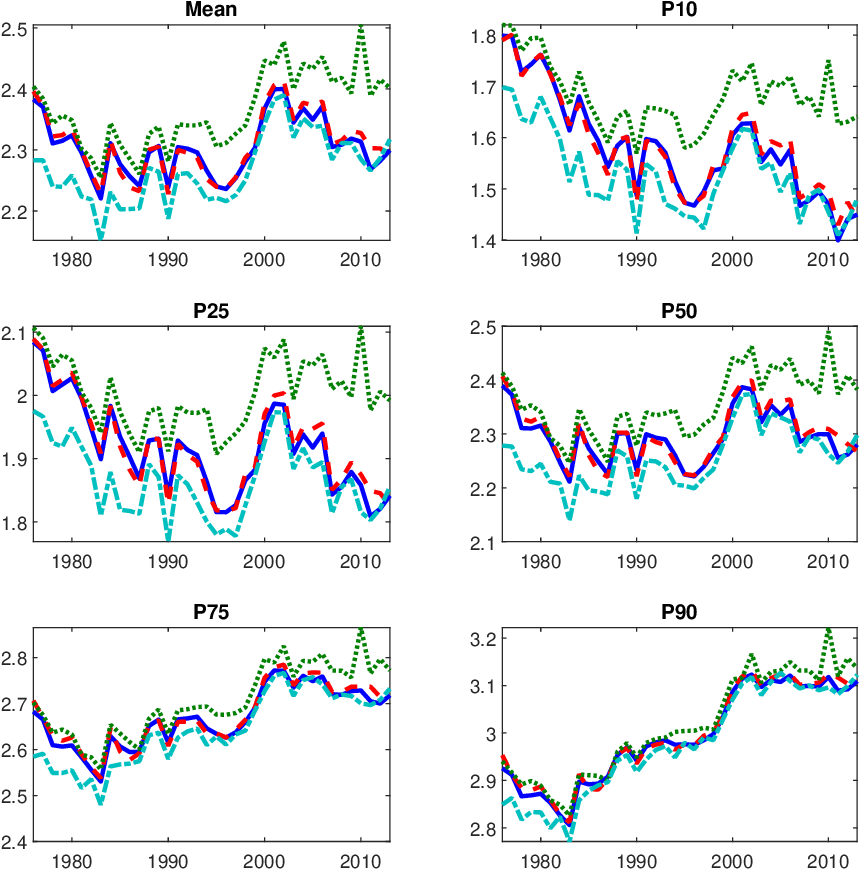}\label{fig:y3f}

{\footnotesize Notes: the solid thick blue line denotes the estimate with the Frank copula for all individuals; the solid thin red line denotes the estimate with the Gaussian copula for all individuals; the dashed green line denotes the estimate with the Frank copula, heterogeneous across race; the dotted orange line denotes the estimate with the Frank copula, heterogeneous across education level; the dashed-dotted cyan line denotes the estimate with the Frank copula, heterogeneous across marital status.}
\end{figure}

\clearpage

\begin{table}[htbp]
  \centering
  \caption{Self-selection decomposition (Gaussian copula)}\label{tab:duf}
		\begin{threeparttable}
\begin{tabular}{cc@{\,}lc@{\,}lc@{\,}lc@{\,}l|cc@{\,}lc@{\,}lc@{\,}lc@{\,}l}
		\hline
     Year  & \multicolumn{2}{c}{Total} & \multicolumn{2}{c}{EC} & \multicolumn{2}{c}{SC} & \multicolumn{2}{c}{PC} & Year  & \multicolumn{2}{c}{Total} & \multicolumn{2}{c}{EC} & \multicolumn{2}{c}{SC} & \multicolumn{2}{c}{PC} \\
		\hline
    1976  & 0.1   &       & -0.1  & **    & -3.0  &       & 3.2   & ***   & 1995  & -1.1  &       & 0.1   &       & 1.3   &       & -2.5  & *** \\
    1977  & 0.3   &       & -0.1  &       & -1.6  &       & 1.9   & *     & 1996  & -8.8  &       & 0.0   &       & -6.3  & **    & -2.5  & *** \\
    1978  & -4.7  &       & -0.1  &       & -4.5  & *     & -0.1  &       & 1997  & -7.3  &       & 0.0   &       & -5.0  & *     & -2.3  & *** \\
    1979  & -5.5  &       & -0.1  & **    & -5.1  & **    & -0.4  &       & 1998  & -1.9  &       & 0.0   &       & 0.3   &       & -2.3  & *** \\
    1980  & -4.6  &       & -0.1  & *     & -4.7  & **    & 0.1   &       & 1999  & -8.4  &       & 0.0   &       & -5.9  & **    & -2.4  & *** \\
    1981  & -4.7  &       & 0.0   &       & -4.8  & *     & 0.1   &       & 2000  & -6.2  &       & 0.0   &       & -4.8  & *     & -1.4  & * \\
    1982  & -4.6  &       & 0.0   &       & -4.2  & *     & -0.5  &       & 2001  & -6.5  &       & 0.0   &       & -5.3  & ***   & -1.3  & ** \\
\hline
    \end{tabular}\begin{tablenotes}[para,flushleft]
\begin{spacing}{1}
{\footnotesize Notes: Total, EC, SC and PC respectively denote total difference, endowments component, selection component and participation component.}
\end{spacing}
\end{tablenotes}
\end{threeparttable}
\end{table}

\clearpage

\begin{table}[htbp]
  \centering
  \caption{Actual earnings distributions for participants by gender (Gaussian copula)}\label{tab:y1f}
		\begin{threeparttable}
    \begin{tabular}{ccccccc|cccccc}
		\hline
          & \multicolumn{6}{c}{Male}                      & \multicolumn{6}{c}{Female} \\
    Year  & Mean  & P10   & P25   & P50   & P75   & P90   & Mean  & P10   & P25   & P50   & P75   & P90 \\
		\hline
    1976  & 2.72  & 2.03  & 2.40  & 2.77  & 3.09  & 3.36  & 2.30  & 1.74  & 2.00  & 2.30  & 2.60  & 2.86 \\
    1977  & 2.73  & 2.03  & 2.40  & 2.78  & 3.10  & 3.37  & 2.32  & 1.77  & 2.03  & 2.33  & 2.62  & 2.88 \\
    1978  & 2.74  & 2.03  & 2.40  & 2.78  & 3.12  & 3.39  & 2.31  & 1.76  & 2.02  & 2.33  & 2.63  & 2.89 \\
    1979  & 2.75  & 2.05  & 2.40  & 2.79  & 3.13  & 3.39  & 2.33  & 1.78  & 2.04  & 2.33  & 2.63  & 2.90 \\
    1980  & 2.74  & 2.03  & 2.40  & 2.78  & 3.12  & 3.39  & 2.32  & 1.78  & 2.04  & 2.32  & 2.62  & 2.89 \\
    1981  & 2.69  & 1.97  & 2.34  & 2.73  & 3.07  & 3.35  & 2.29  & 1.74  & 2.00  & 2.29  & 2.59  & 2.87 \\
\hline
    \end{tabular}
\end{threeparttable}
\end{table}

\clearpage

\begin{table}[htbp]
  \centering
  \caption{Actual earnings distributions for the full population by gender (Gaussian copula)}\label{tab:y2f}
		\begin{threeparttable}
    \begin{tabular}{ccccccc|cccccc}
		\hline
          & \multicolumn{6}{c}{Male}                      & \multicolumn{6}{c}{Female} \\
    Year  & Mean  & P10   & P25   & P50   & P75   & P90   & Mean  & P10   & P25   & P50   & P75   & P90 \\
		\hline
    1976  & 1.89  & 0.00  & 0.00  & 2.50  & 2.97  & 3.28  & 0.81  & 0.00  & 0.00  & 0.00  & 2.09  & 2.60 \\
    1977  & 1.91  & 0.00  & 0.00  & 2.51  & 2.99  & 3.29  & 0.85  & 0.00  & 0.00  & 0.00  & 2.13  & 2.62 \\
    1978  & 1.91  & 0.00  & 0.00  & 2.50  & 3.00  & 3.31  & 0.87  & 0.00  & 0.00  & 0.00  & 2.16  & 2.64 \\
    1979  & 1.95  & 0.00  & 0.00  & 2.52  & 3.01  & 3.32  & 0.92  & 0.00  & 0.00  & 0.00  & 2.20  & 2.65 \\
    1980  & 1.92  & 0.00  & 0.00  & 2.50  & 3.00  & 3.31  & 0.95  & 0.00  & 0.00  & 0.00  & 2.23  & 2.66 \\
    1981  & 1.85  & 0.00  & 0.00  & 2.42  & 2.94  & 3.26  & 0.93  & 0.00  & 0.00  & 0.00  & 2.19  & 2.63 \\
\hline
    \end{tabular}
\end{threeparttable}
\end{table}

\clearpage

\begin{table}[htbp]
  \centering
  \caption{Potential earnings distributions for the full population by gender (Gaussian copula)}\label{tab:y3f}
		\begin{threeparttable}
    \begin{tabular}{ccccccc|cccccc}
		\hline
          & \multicolumn{6}{c}{Male}                      & \multicolumn{6}{c}{Female} \\
    Year  & Mean  & P10   & P25   & P50   & P75   & P90   & Mean  & P10   & P25   & P50   & P75   & P90 \\
		\hline
    1976  & 2.82  & 2.12  & 2.49  & 2.84  & 3.16  & 3.43  & 2.38  & 1.80  & 2.08  & 2.38  & 2.67  & 2.93 \\
    1977  & 2.78  & 2.07  & 2.44  & 2.81  & 3.13  & 3.40  & 2.37  & 1.80  & 2.07  & 2.37  & 2.66  & 2.91 \\
    1978  & 2.81  & 2.08  & 2.46  & 2.84  & 3.16  & 3.43  & 2.31  & 1.73  & 2.01  & 2.31  & 2.61  & 2.87 \\
    1979  & 2.82  & 2.11  & 2.47  & 2.86  & 3.19  & 3.45  & 2.31  & 1.74  & 2.02  & 2.31  & 2.61  & 2.87 \\
    1980  & 2.82  & 2.09  & 2.47  & 2.85  & 3.18  & 3.45  & 2.32  & 1.76  & 2.03  & 2.32  & 2.61  & 2.87 \\
    1981  & 2.77  & 2.04  & 2.41  & 2.80  & 3.13  & 3.40  & 2.29  & 1.72  & 1.99  & 2.28  & 2.58  & 2.85 \\
\hline
    \end{tabular}
\end{threeparttable}
\end{table}

\clearpage

\begin{table}[htbp]
  \centering
  \caption{Self-selection decomposition, heterogeneous copula by race}\label{tab:dug2}
		\begin{threeparttable}
    \begin{tabular}{cc@{\,}lc@{\,}lc@{\,}lc@{\,}l|cc@{\,}lc@{\,}lc@{\,}lc@{\,}l}
		\hline
     Year  & \multicolumn{2}{c}{Total} & \multicolumn{2}{c}{EC} & \multicolumn{2}{c}{SC} & \multicolumn{2}{c}{PC} & Year  & \multicolumn{2}{c}{Total} & \multicolumn{2}{c}{EC} & \multicolumn{2}{c}{SC} & \multicolumn{2}{c}{PC} \\
		\hline
    1976  & 1.1   &       & -0.2  &       & -3.1  &       & 4.3   & ***   & 1995  & -3.8  &       & 0.1   &       & -1.0  &       & -2.9  & *** \\
    1977  & -3.7  &       & -0.1  &       & -5.8  & **    & 2.2   & **    & 1996  & -11.1 & ***   & -0.2  & *     & -8.0  & ***   & -2.9  & *** \\
    1978  & -9.3  & ***   & -0.4  & ***   & -9.5  & ***   & 0.6   &       & 1997  & -11.1 & ***   & -0.3  & **    & -8.0  & ***   & -2.9  & *** \\
    1979  & -6.8  & ***   & -0.3  & **    & -6.6  & ***   & 0.1   &       & 1998  & -5.9  & *     & -0.1  &       & -3.4  &       & -2.4  & *** \\
    1980  & -5.6  & **    & -0.4  & ***   & -6.0  & ***   & 0.8   &       & 1999  & -12.3 & ***   & 0.1   &       & -9.8  & ***   & -2.6  & *** \\
    1981  & -5.6  & **    & -0.1  &       & -6.0  & ***   & 0.5   &       & 2000  & -9.1  & ***   & 0.0   &       & -7.8  & ***   & -1.4  &  \\
    1982  & -8.8  & ***   & -0.2  & *     & -8.3  & ***   & -0.3  &       & 2001  & -10.0 & ***   & -0.1  &       & -8.7  & ***   & -1.3  & ** \\
\hline
    \end{tabular}\begin{tablenotes}[para,flushleft]
\begin{spacing}{1}
{\footnotesize Notes: Total, EC, SC and PC respectively denote total difference, endowments component, selection component and participation component; coefficients scaled by 100.}
\end{spacing}
\end{tablenotes}
\end{threeparttable}
\end{table}

\clearpage

\begin{table}[htbp]
  \centering
  \caption{Self-selection decomposition, heterogeneous copula by education}\label{tab:dug3}
		\begin{threeparttable}
    \begin{tabular}{cc@{\,}lc@{\,}lc@{\,}lc@{\,}l|cc@{\,}lc@{\,}lc@{\,}lc@{\,}l}
		\hline
     Year  & \multicolumn{2}{c}{Total} & \multicolumn{2}{c}{EC} & \multicolumn{2}{c}{SC} & \multicolumn{2}{c}{PC} & Year  & \multicolumn{2}{c}{Total} & \multicolumn{2}{c}{EC} & \multicolumn{2}{c}{SC} & \multicolumn{2}{c}{PC} \\
		\hline
    1976  & 5.3   & **    & 0.4   & *     & 0.5   &       & 4.4   & ***   & 1995  & 0.7   &       & 0.4   & **    & 1.6   &       & -1.3  & ** \\
    1977  & 4.2   &       & 0.6   & ***   & 0.9   &       & 2.7   & ***   & 1996  & -5.5  & **    & 0.1   &       & -4.5  & **    & -1.0  &  \\
    1978  & 3.2   &       & 0.6   & ***   & 1.5   &       & 1.1   &       & 1997  & -4.0  &       & 0.2   &       & -3.3  &       & -0.9  &  \\
    1979  & 1.9   &       & 0.1   &       & 0.5   &       & 1.2   &       & 1998  & -1.3  &       & 0.2   & ***   & -0.5  &       & -1.0  &  \\
    1980  & 2.9   &       & 0.5   & ***   & 1.2   &       & 1.1   &       & 1999  & -2.6  &       & 0.3   & ***   & -2.3  &       & -0.5  &  \\
    1981  & -0.3  &       & 0.7   & ***   & -1.6  &       & 0.5   &       & 2000  & 0.2   &       & 0.2   & ***   & -0.4  &       & 0.3   &  \\
    1982  & 0.8   &       & 0.6   & ***   & -0.3  &       & 0.5   &       & 2001  & -2.6  &       & 0.2   & **    & -2.4  &       & -0.4  &  \\
\hline
    \end{tabular}\begin{tablenotes}[para,flushleft]
\begin{spacing}{1}
{\footnotesize Notes: Total, EC, SC and PC respectively denote total difference, endowments component, selection component and participation component; coefficients scaled by 100.}
\end{spacing}
\end{tablenotes}
\end{threeparttable}
\end{table}

\clearpage

\begin{table}[htbp]
  \centering
  \caption{Self-selection decomposition, heterogeneous copula by marital status}\label{tab:dug4}
		\begin{threeparttable}
    \begin{tabular}{cc@{\,}lc@{\,}lc@{\,}lc@{\,}l|cc@{\,}lc@{\,}lc@{\,}lc@{\,}l}
		\hline
     Year  & \multicolumn{2}{c}{Total} & \multicolumn{2}{c}{EC} & \multicolumn{2}{c}{SC} & \multicolumn{2}{c}{PC} & Year  & \multicolumn{2}{c}{Total} & \multicolumn{2}{c}{EC} & \multicolumn{2}{c}{SC} & \multicolumn{2}{c}{PC} \\
		\hline
    1976  & -8.2  & **    & -0.2  & **    & -6.4  & **    & -1.7  &       & 1995  & -7.5  & **    & 0.0   &       & -3.0  &       & -4.4  & *** \\
    1977  & -8.2  & **    & -0.1  & *     & -5.0  & *     & -3.1  & ***   & 1996  & -14.7 & ***   & 0.0   &       & -10.7 & ***   & -3.9  & *** \\
    1978  & -10.9 & ***   & -0.1  & **    & -7.1  & ***   & -3.7  & ***   & 1997  & -10.5 & ***   & 0.0   &       & -5.9  & **    & -4.5  & *** \\
    1979  & -13.0 & ***   & -0.1  & **    & -8.4  & ***   & -4.5  & ***   & 1998  & -6.6  & **    & 0.0   &       & -2.5  &       & -4.1  & *** \\
    1980  & -11.6 & ***   & -0.1  & *     & -7.9  & ***   & -3.6  & ***   & 1999  & -12.7 & ***   & 0.0   &       & -9.2  & ***   & -3.6  & *** \\
    1981  & -9.2  & ***   & 0.0   &       & -6.0  & **    & -3.2  & ***   & 2000  & -9.9  & ***   & 0.0   &       & -7.0  & ***   & -3.0  & *** \\
    1982  & -7.8  & ***   & 0.0   &       & -4.4  & *     & -3.3  & ***   & 2001  & -8.4  & ***   & 0.0   &       & -5.6  & ***   & -2.8  & *** \\
\hline
    \end{tabular}\begin{tablenotes}[para,flushleft]
\begin{spacing}{1}
{\footnotesize Notes: Total, EC, SC and PC respectively denote total difference, endowments component, selection component and participation component; coefficients scaled by 100.}
\end{spacing}
\end{tablenotes}
\end{threeparttable}
\end{table}

\clearpage

\begin{table}[htbp]
  \centering
  \caption{Mean decomposition, actual earnings for participants, heterogeneous copula by race}\label{tab:dy12}
		\begin{threeparttable}
    \begin{tabular}{cr@{\,}lr@{\,}lr@{\,}lr@{\,}lr@{\,}l}
		\hline
    Year  & \multicolumn{2}{c}{Total} & \multicolumn{2}{c}{EC} & \multicolumn{2}{c}{CC} & \multicolumn{2}{c}{SC} & \multicolumn{2}{c}{PC} \\
		\hline
    1976  & 0.43  & ***   & 0.01  & ***   & 0.40  & ***   & -0.04 &       & 0.06  & *** \\
    1977  & 0.41  & ***   & 0.01  & **    & 0.45  & ***   & -0.07 & **    & 0.03  & ** \\
    1978  & 0.43  & ***   & 0.01  & ***   & 0.53  & ***   & -0.12 & ***   & 0.01  &  \\
    1979  & 0.42  & ***   & 0.01  & ***   & 0.50  & ***   & -0.08 & ***   & 0.00  &  \\
    1980  & 0.42  & ***   & 0.01  & ***   & 0.48  & ***   & -0.08 & ***   & 0.01  &  \\
    1981  & 0.40  & ***   & 0.00  &       & 0.47  & ***   & -0.08 & ***   & 0.01  &  \\
    1982  & 0.40  & ***   & 0.01  & **    & 0.51  & ***   & -0.11 & ***   & 0.00  &  \\
\hline
    \end{tabular}\begin{tablenotes}[para,flushleft]
\begin{spacing}{1}
{\footnotesize Notes: Total, EC, CC, SC and PC respectively denote total difference, endowments component, coefficients component, selection component and participation component.}
\end{spacing}
\end{tablenotes}
\end{threeparttable}
\end{table}

\clearpage

\begin{table}[htbp]
  \centering
  \caption{Mean decomposition, actual earnings for participants, heterogeneous copula by education}\label{tab:dy13}
		\begin{threeparttable}
    \begin{tabular}{cr@{\,}lr@{\,}lr@{\,}lr@{\,}lr@{\,}l}
		\hline
    Year  & \multicolumn{2}{c}{Total} & \multicolumn{2}{c}{EC} & \multicolumn{2}{c}{CC} & \multicolumn{2}{c}{SC} & \multicolumn{2}{c}{PC} \\
		\hline
    1976  & 0.30  & ***   & -0.01 & **    & 0.31  & ***   & -0.05 & *     & 0.04  & *** \\
    1977  & 0.26  & ***   & -0.02 & ***   & 0.32  & ***   & -0.07 & ***   & 0.02  & ** \\
    1978  & 0.27  & ***   & -0.02 & ***   & 0.39  & ***   & -0.10 & ***   & 0.00  &  \\
    1979  & 0.27  & ***   & -0.01 & *     & 0.35  & ***   & -0.07 & ***   & 0.00  &  \\
    1980  & 0.25  & ***   & -0.01 & ***   & 0.31  & ***   & -0.05 & **    & 0.01  &  \\
    1981  & 0.23  & ***   & -0.02 & ***   & 0.31  & ***   & -0.07 & ***   & 0.00  &  \\
    1982  & 0.24  & ***   & -0.01 & **    & 0.35  & ***   & -0.10 & ***   & 0.00  &  \\
\hline
    \end{tabular}\begin{tablenotes}[para,flushleft]
\begin{spacing}{1}
{\footnotesize Notes: Total, EC, CC, SC and PC respectively denote total difference, endowments component, coefficients component, selection component and participation component.}
\end{spacing}
\end{tablenotes}
\end{threeparttable}
\end{table}

\clearpage

\begin{table}[htbp]
  \centering
  \caption{Mean decomposition, actual earnings for participants, heterogeneous copula by marital status}\label{tab:dy14}
		\begin{threeparttable}
    \begin{tabular}{cr@{\,}lr@{\,}lr@{\,}lr@{\,}lr@{\,}l}
		\hline
    Year  & \multicolumn{2}{c}{Total} & \multicolumn{2}{c}{EC} & \multicolumn{2}{c}{CC} & \multicolumn{2}{c}{SC} & \multicolumn{2}{c}{PC} \\
		\hline
    1976  & 0.40  & ***   & 0.00  & ***   & 0.39  & ***   & -0.04 & **    & 0.05  & *** \\
    1977  & 0.37  & ***   & -0.01 & ***   & 0.42  & ***   & -0.07 & ***   & 0.03  & * \\
    1978  & 0.37  & ***   & 0.00  & ***   & 0.47  & ***   & -0.11 & ***   & 0.01  &  \\
    1979  & 0.36  & ***   & 0.00  & *     & 0.44  & ***   & -0.08 & ***   & 0.00  &  \\
    1980  & 0.36  & ***   & 0.00  & ***   & 0.42  & ***   & -0.06 & **    & 0.01  &  \\
    1981  & 0.34  & ***   & -0.01 & ***   & 0.41  & ***   & -0.07 & ***   & 0.01  &  \\
    1982  & 0.34  & ***   & -0.01 & ***   & 0.45  & ***   & -0.10 & ***   & 0.00  &  \\
\hline
    \end{tabular}\begin{tablenotes}[para,flushleft]
\begin{spacing}{1}
{\footnotesize Notes: Total, EC, CC, SC and PC respectively denote total difference, endowments component, coefficients component, selection component and participation component.}
\end{spacing}
\end{tablenotes}
\end{threeparttable}
\end{table}

\clearpage

\begin{figure}[htbp]
\caption{Actual earnings decomposition for participants, H2S}
\includegraphics[width=16.5cm]{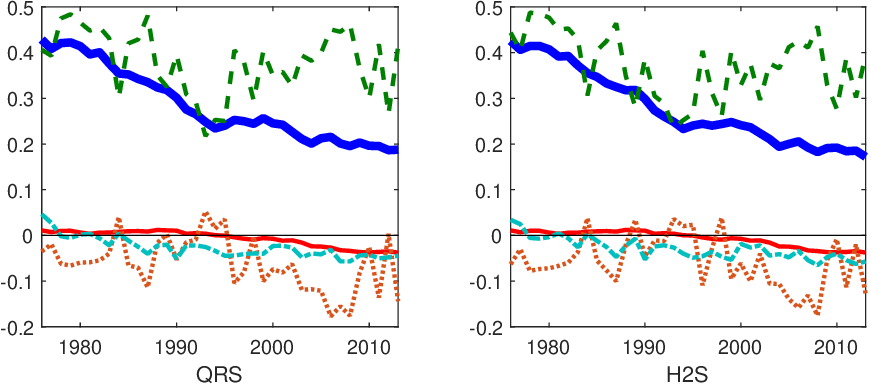}\label{fig:dech1}

{\footnotesize Notes: QRS and H2S stand for Quantile Regression with Selection and Heckman 2-Stage estimators; the solid thick blue line denotes the total gap between male and female workers; the solid thin red line denotes the endowments component; the dashed green line denotes the coefficients component; the dotted orange line denotes the selection component; the dashed-dotted cyan line denotes the participation component.}
\end{figure}

\begin{figure}[htbp]
\caption{Actual earnings decomposition for the entire population, H2S}
\includegraphics[width=16.5cm]{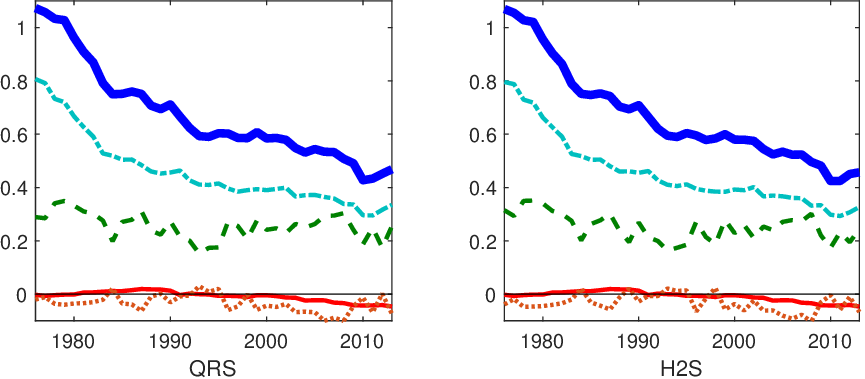}\label{fig:dech2}

{\footnotesize Notes: QRS and H2S stand for Quantile Regression with Selection and Heckman 2-Stage estimators; the solid thick blue line denotes the total gap between male and female workers; the solid thin red line denotes the endowments component; the dashed green line denotes the coefficients component; the dotted orange line denotes the selection component; the dashed-dotted cyan line denotes the participation component.}
\end{figure}

\begin{figure}[htbp]
\caption{Unconditional quantiles decompositions, actual earnings for participants, heterogeneous copula by race}
\includegraphics[width=16.5cm]{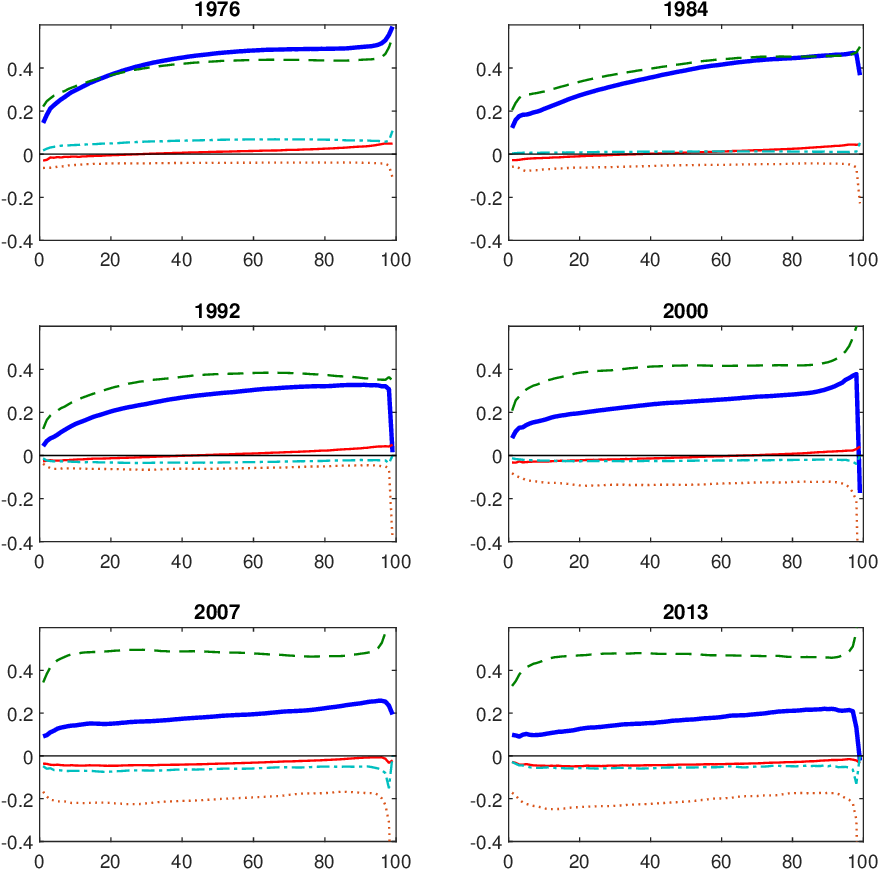}\label{fig:dec12}

{\footnotesize Notes: the solid thick blue line denotes the total gap between male and female workers; the solid thin red line denotes the endowments component; the dashed green line denotes the coefficients component; the dotted orange line denotes the selection component; the dashed-dotted cyan line denotes the participation component.}
\end{figure}

\begin{figure}[htbp]
\caption{Unconditional quantiles decompositions, actual earnings for participants, heterogeneous copula by education}
\includegraphics[width=16.5cm]{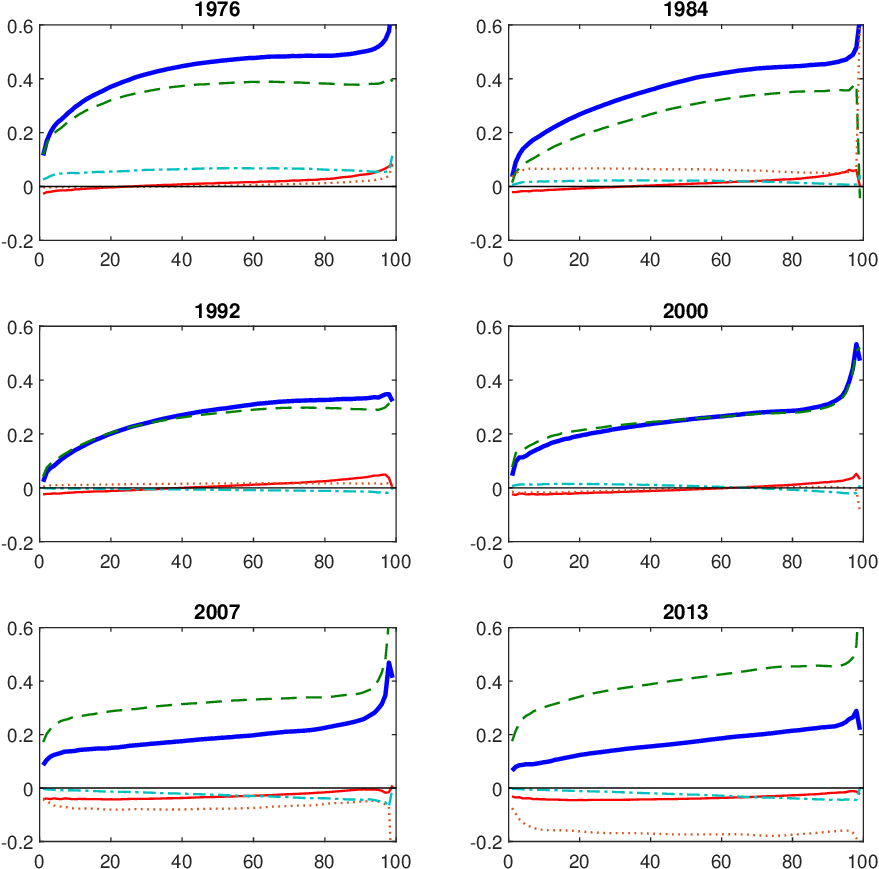}\label{fig:dec13}

{\footnotesize Notes: the solid thick blue line denotes the total gap between male and female workers; the solid thin red line denotes the endowments component; the dashed green line denotes the coefficients component; the dotted orange line denotes the selection component; the dashed-dotted cyan line denotes the participation component.}
\end{figure}

\begin{figure}[htbp]
\caption{Unconditional quantiles decompositions, actual earnings for participants, heterogeneous copula by marital status}
\includegraphics[width=16.5cm]{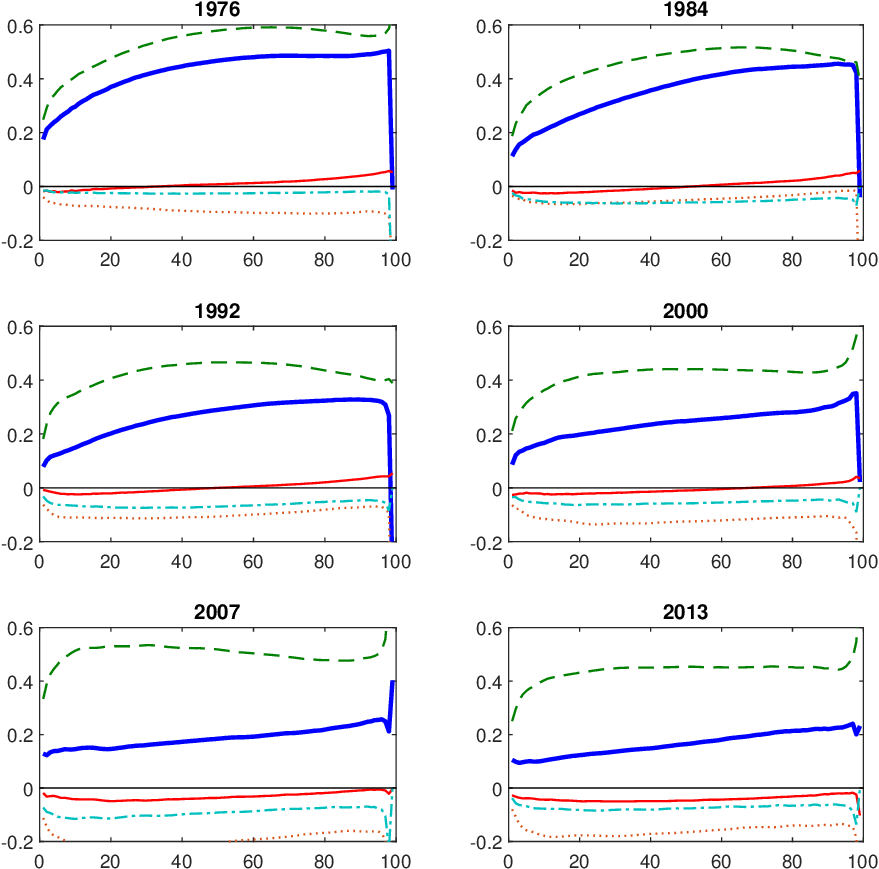}\label{fig:dec14}

{\footnotesize Notes: the solid thick blue line denotes the total gap between male and female workers; the solid thin red line denotes the endowments component; the dashed green line denotes the coefficients component; the dotted orange line denotes the selection component; the dashed-dotted cyan line denotes the participation component.}
\end{figure}

\end{document}